%% file: arxiv-thesis-abhi.tex
\documentclass[openright,twoside]{iitbthesis}

\usepackage{times}

\usepackage[dvips]{graphicx,color}
\usepackage{bmpsize}
\DeclareGraphicsExtensions{.pdf,.jpeg,.png,.jpg,.eps}
\usepackage[cmex10]{amsmath}
\allowdisplaybreaks[3]
\usepackage{amsthm}
\usepackage{amssymb}
\usepackage{amsmath}
\usepackage{latexsym}
\usepackage{float}
\usepackage{amsfonts} 
\usepackage{enumerate} 
\usepackage{array}
\usepackage[shortlabels]{enumitem}

\usepackage{idxlayout}
\usepackage[useindex]{splitidx}
\makeindex
\newindex[Index to symbols]{symb}
\newindex[Index]{term}

\usepackage[small,bf,hang]{caption}
\usepackage{subcaption}

\usepackage{booktabs}
\usepackage{framed}
\usepackage{verbatim}
\usepackage{paralist}
\usepackage{xspace} 
\usepackage{latexsym}
\usepackage{appendix}

\usepackage{multirow}
\usepackage{paralist}

\usepackage{tabularx}

\def\nocolor{1}

\ifdefined\nocolor
  \def\linkcolor{black}
  \def\citecolor{black}
\else
  \def\linkcolor{purple}
  \def\citecolor{blue}
\fi

\usepackage[hyphens]{url}

\usepackage[linktocpage=true, colorlinks=true, pagebackref=true,
linkcolor=\linkcolor,
citecolor=\citecolor,
plainpages=false,
pdfpagelabels=true, breaklinks]{hyperref}
\usepackage{breakurl}

\def\UrlBreaks{\do\/\do h}

\usepackage{booktabs}

\usepackage{framed}

\usepackage{verbatim}

\usepackage[stable,splitrule]{footmisc}

\renewcommand{\footnoterule}{%
  \kern 3pt
  \hrule width \textwidth height 0.5pt
  \kern 8pt
}

\usepackage[numbers, sort, square]{natbib}

\usepackage{zref-abspage}
\newwrite\figpages
\openout\figpages=\jobname.fpg
\makeatletter

\usepackage{fancyhdr}
\usepackage{etoolbox}

\makeatletter
\def\cleardoublepage{\clearpage\if@twoside \ifodd\c@page\else
  \hbox{}
  \thispagestyle{empty}
  \newpage
  \if@twocolumn\hbox{}\newpage\fi\fi\fi}
\makeatother

\graphicspath{{figs/}}

\newcommand{\nocontentsline}[3]{}
\newcommand{\tocless}[2]{\bgroup\let\addcontentsline=\nocontentsline#1{#2}\egroup}

\input{macros}

\newcounter{remark-tracker}

\theoremstyle{plain}

\newcounter{thm} 
\numberwithin{thm}{section}

\newcounter{thmfuturework} 
\numberwithin{thmfuturework}{chapter}

\newcounter{thmchapter} 
\numberwithin{thmchapter}{chapter}

\newtheorem{theorem}[thm]{Theorem}
\newtheorem{corollary}[thm]{Corollary}
\newtheorem{lemma}[thm]{Lemma}

\newtheorem{hypothesis}[thm]{Hypothesis}

\newtheorem{proposition}[thm]{Proposition}

\newtheorem{futureworktheorem}[thmfuturework]{Theorem}
\newtheorem{futureworklemma}[thmfuturework]{Lemma}
\newtheorem{futureworkconj}[thmfuturework]{Conjecture}

\newtheorem{chaptertheorem}[thmchapter]{Theorem}

\theoremstyle{definition}
\newtheorem{defn}[thm]{Definition}
\newtheorem{example}[thm]{Example}
\newtheorem{remark}[thm]{Remark}

\newtheorem{chapterdefn}[thmchapter]{Definition}
\newtheorem{chapterremark}[thmchapter]{Remark}
\newtheorem{futureworkremark}[thmfuturework]{Remark}
\newtheorem{futureworkdefn}[thmfuturework]{Definition}

\begin{document}

\input{arxiv-prelude}

\pagestyle{fancy}
\renewcommand{\chaptermark}[1]%
{\markboth{#1}{}}
\renewcommand{\sectionmark}[1]%
{\markright{#1}}
\renewcommand{\headrulewidth}{0pt}
\renewcommand{\footrulewidth}{0pt}
\newcommand{\helv}{%
\fontfamily{phv}\fontsize{9}{12}\selectfont}
\fancyhf{}
\fancyhead[RE]{\helv \nouppercase{\leftmark}}
\fancyhead[LE]{\helv Chapter~\thechapter}

\fancyhead[RO]{
\ifnum\value{section}=0
 {\helv \nouppercase{\leftmark}} 
\else
 {\helv \nouppercase{\rightmark}}
\fi
}

\fancyhead[LO]{ 
\ifnum\value{section}=0
 {\helv Chapter~\thechapter}
\else
 {\helv Section~\thesection}
\fi
}
\fancyfoot[C]{\thepage}
\pretocmd{\bibliography}{\cleardoublepage \fancyhead[LE,LO]{}}{}{}

\setcounter{secnumdepth}{5}

\input{intro-5.0}

\part{Classical Model Theory}
\input{classical-model-theory-results}

\part{Finite Model Theory}
\input{finite-model-theory-results}
\input{conclusion}



\Urlmuskip=0mu plus 2mu\relax
\bibliographystyle{plainnat}
\cleardoublepage
\phantomsection
\addcontentsline{toc}{chapter}{Bibliography}
\bibliography{refs}



{\cleardoublepage \fancyhead[RE,RO]{\helv \nouppercase{Acknowledgments}}}{}{}
\begin{acknowledgments}
\input{ack-4.0}
\end{acknowledgments}

{\cleardoublepage \fancyhead[RE,RO]{\helv \nouppercase{\rightmark}}}{}{}

\phantomsection
\addcontentsline{toc}{chapter}{Index to symbols}
\printindex[symb]
{\cleardoublepage}
\phantomsection
\addcontentsline{toc}{chapter}{Index}
\printindex[term]

\end{document}

%% file: macros.tex
\newcommand{\tbf}[1]{\textbf{#1}}
\newcommand{\mc}[1]{\mathcal{#1}}
\newcommand{\ul}[1]{\underline{#1}}
\newcommand{\mf}[1]{\mathfrak{#1}}

\newcommand{\fequiv}[1]{\ensuremath{\equiv_{#1, \fo}}}
\newcommand{\mequiv}[1]{\ensuremath{\equiv_{#1, \mso}}}
\newcommand{\lequiv}[1]{\ensuremath{\equiv_{#1, \mc{L}}}}

\newcommand{\ftp}[3]{\ensuremath{\mathsf{tp}_{#1, #2, #3, \fo}}}

\newcommand{\ltp}[3]{\ensuremath{\mathsf{tp}_{#1, #2, #3, \mc{L}}}}

\newcommand{\fotp}[2]{\ensuremath{\mathsf{tp}_{#1, #2}}}
\newcommand{\pitp}[2]{\ensuremath{\mathsf{tp}_{\Pi, #1, #2}}}

\newcommand{\lth}[2]{\ensuremath{\mathsf{Th}_{#1,
      \mc{L}}(#2)}}

\newcommand{\foth}{\ensuremath{\mathsf{Th}}}

\newcommand{\true}{\ensuremath{\mathsf{True}}}
\newcommand{\false}{\ensuremath{\mathsf{False}}}

\newcommand{\fef}{\ensuremath{\fo\text{-}\text{EF}}}
\newcommand{\mef}{\ensuremath{\mso\text{-}\text{EF}}}

\newcommand{\fo}{\ensuremath{\text{FO}}}
\newcommand{\mso}{\ensuremath{\text{MSO}}}

\newcommand{\rank}[1]{\ensuremath{\text{rank}(#1)}}
\newcommand{\univ}[1]{\ensuremath{\mathsf{U}_{#1}}}
\newcommand{\diag}[1]{\ensuremath{\text{Diag}(#1)}}
\newcommand{\eldiag}[1]{\ensuremath{\text{El-diag}(#1)}}

\newcommand{\aset}{\{a_1,\ldots,a_k\}}
\newcommand{\M}[2]{\ensuremath{{#1}({#2})}}

\newcommand{\hpc}[1]{\ensuremath{h\text{-}PC({#1})}}

\newcommand{\lt}{{{\L}o{\'s}-Tarski}}

\newcommand{\glt}[1]{\ensuremath{\mathsf{GLT}({#1})}}

\newcommand{\hpt}{\ensuremath{\mathsf{HPT}}}

\newcommand{\lghpt}[1]{\ensuremath{\mc{L}\text{-}\mathsf{GHPT}({#1})}}
\newcommand{\fghpt}[1]{\ensuremath{\fo\text{-}\mathsf{GHPT}({#1})}}

\newcommand{\lebsp}[2]{\ensuremath{\mc{L}\text{-}\mathsf{EBSP}({#1},
    {#2})}}
\newcommand{\febsp}[2]{\ensuremath{\fo\text{-}\mathsf{EBSP}({#1},
    {#2})}}
\newcommand{\mebsp}[2]{\ensuremath{\mso\text{-}\mathsf{EBSP}({#1},
    {#2})}}
\newcommand{\lebspcond}{\ensuremath{\mc{L}\text{-}\mathsf{EBSP}\text{-}\mathsf{condition}}}
\newcommand{\febspcond}{\ensuremath{\fo\text{-}\mathsf{EBSP}\text{-}\mathsf{condition}}}

\newcommand{\reflebsp}{\ensuremath{\mc{L}\text{-}\mathsf{EBSP}}}

\newcommand{\witfn}[3]{\ensuremath{\theta_{(#1, #2, #3)}}}
\newcommand{\hlebsp}[2]{\ensuremath{h\text{-}\mc{L}\text{-}\mathsf{EBSP}({#1},
    {#2})}}
\newcommand{\hfebsp}[2]{\ensuremath{h\text{-}\fo\text{-}\mathsf{EBSP}({#1},
    {#2})}}
\newcommand{\hmebsp}[2]{\ensuremath{h\text{-}\mso\text{-}\mathsf{EBSP}({#1},
    {#2})}}
\newcommand{\wqo}[2]{\ensuremath{\mathsf{WQO}({#1}, {#2})}}

\newcommand{\ls}{L\"owenheim-Skolem}
\newcommand{\dls}{downward L\"owenheim-Skolem}
\newcommand{\dlsfull}{downward L\"owenheim-Skolem theorem}

\newcommand{\ldlsp}[2]{\ensuremath{\mc{L}\text{-}\mathsf{DLSP}({#1},
    {#2})}}
\newcommand{\ldlspm}[2]{\ensuremath{\mc{L}\text{-}\mathsf{DLSP}_{\mathsf{M}}({#1},
    {#2})}}

\newcommand{\ldlsps}[2]{\ensuremath{\mc{L}\text{-}\mathsf{DLSP}_{\mathsf{S}}({#1},
    {#2})}}

\newcommand{\cl}[1]{\ensuremath{\mathcal{#1}}}

\newcommand{\words}{\mathsf{Words}}

\newcommand{\eat}[1]{}

\newcommand{\tree}[1]{\ensuremath{\mathsf{#1}}}

\newcommand{\nesword}[1]{\ensuremath{\mathsf{#1}}}
\newcommand{\nestedwords}[1]{\mathsf{Nested}\text{-}\mathsf{words}(#1)}

\newcommand{\ncographs}{\ensuremath{n\text{-}\mathsf{partite}\text{-}\mathsf{cographs}}}

%% file: arxiv-prelude.tex

\clearpage\pagenumbering{roman}  
\thispagestyle{empty}

\title{A Generalization of the {\lt} Preservation Theorem}

\author{Abhisekh Sankaran}

\date{2016}

\rollnum{07405001} 

\iitbdegree{Doctor of Philosophy}

\thesis



\department{Department of Computer Science and Engineering}

\setguide{Prof. Supratik Chakraborty}
\setcoguide{Prof. Bharat Adsul}

\maketitle

\begin{dedication}
\large{To my parents}
\end{dedication}
%

\makeapproval




%

\makedeclaration

%



\begin{abstract}
  \input{abstract-abhi}
\end{abstract}

\tableofcontents
\listoffigures
\input{mypubs}
\begingroup
\let\clearpage\relax
\endgroup

%

\cleardoublepage\pagenumbering{arabic} 

%% file: abstract-abhi.tex
Preservation theorems are amongst the earliest areas of study in
classical model theory. One of the first preservation theorems to be
proven is the {\lt} theorem that provides over arbitrary structures
and for arbitrary finite vocabularies, semantic characterizations of
the $\forall^*$ and $\exists^*$ prefix classes of first order logic
(FO) sentences, via the properties of preservation under substructures
and preservation under extensions respectively.  In the classical
model theory part of this thesis, we present new parameterized
preservation properties that provide for each natural number $k$,
semantic characterizations of the $\exists^k \forall^*$ and $\forall^k
\exists^*$ prefix classes of FO sentences, over the class of all
structures and for arbitrary finite vocabularies. These properties,
that we call \emph{preservation under substructures modulo $k$-cruxes}
and \emph{preservation under $k$-ary covered extensions} respectively,
correspond exactly to the properties of preservation under
substructures and preservation under extensions, when $k$ equals 0.
As a consequence, we get a parameterized generalization of the
{\L}o{\'s}-Tarski theorem for sentences, in both its substructural and
extensional forms.  We call our characterizations collectively the
\emph{generalized {\L}o{\'s}-Tarski theorem for sentences at level
  $k$}, abbreviated $\glt{k}$.  To the best of our knowledge,
$\glt{k}$ is the first to relate \emph{counts} of quantifiers
appearing in the sentences of the $\Sigma^0_2$ and $\Pi^0_2$ prefix
classes of FO, to natural quantitative properties of models, and hence
provides new semantic characterizations of these sentences.  We
generalize $\glt{k}$ to theories, by showing that theories that are
preserved under $k$-ary covered extensions are characterized by
theories of $\forall^k
\exists^*$ sentences, and theories that are preserved under
substructures modulo $k$-cruxes, are equivalent, under a
well-motivated model-theoretic hypothesis, to theories of $\exists^k
\forall^*$ sentences. We also present natural variants of our
preservation properties in which, instead of natural numbers $k$, we
consider infinite cardinals $\lambda$, and show that these variants
provide new semantic characterizations of $\Sigma^0_2$ and $\Pi^0_2$
theories.

In contrast to existing preservation properties in the literature that
characterize $\Sigma^0_2$ and $\Pi^0_2$ sentences, our preservation
properties are combinatorial and finitary in nature, and stay
non-trivial over finite structures as well.  Hence, in the finite
model theory part of the thesis, we investigate $\glt{k}$ over finite
structures. Like most preservation theorems, $\glt{k}$ fails over the
class of all finite structures. To ``recover'' $\glt{k}$, we identify
a new logic based combinatorial property of classes $\cl{S}$ of finite
structures, that we call the \emph{$\mc{L}$-equivalent bounded
substructure property}, abbreviated $\lebsp{\cl{S}}{k}$, where
$\mc{L}$ is either $\fo$ or $\mso$. We show that $\lebsp{\cl{S}}{k}$
entails $\glt{k}$, and even an effective version of the latter, under
suitable ``computability'' assumptions.  A variety of classes of
finite structures of interest in computer science turn out to satisfy
$\lebsp{\cl{S}}{k}$, and the just mentioned computability assumptions
as well, whereby all of these classes satisfy an effective version of
$\glt{k}$.  Examples include the classes of words, trees (unordered,
ordered or ranked), nested words, cographs, graph classes of bounded
tree-depth, graph classes of bounded shrub-depth and $n$-partite
cographs. These classes were earlier not known to even satisfy the
{\lt} theorem. All of the aforesaid classes have received significant
attention due to their excellent logical and algorithmic properties,
and moreover, many of these are recently defined (in the last 10
years).  We go further to give ways to construct new classes of
structures satisfying $\lebsp{\cdot}{\cdot}$ by showing the closure of
the latter under set-theoretic operations and special kinds of
translation schemes. As a consequence, we get that $\lebsp{\cdot}{k}$
is closed under unary operations like complementation, transpose and
the line-graph operation, and binary ``sum-like'' operations like
disjoint union and join, while $\febsp{\cdot}{\cdot}$ is additionally
closed under ``product-like'' operations like the cartesian, tensor,
lexicographic and strong products.  On studying the $\lebsp{\cdot}{k}$
property further, it turns out that any class of structures that is
well-quasi-ordered under embedding satisfies $\lebsp{\cdot}{0}$, that
$\lebsp{\cdot}{k}$ classes, under the aforementioned computability
assumptions, admit decidability of the satisfiability problem for
$\mc{L}$, and that $\lebsp{\cdot}{k}$ entails the homomorphism
preservation theorem.  Finally, we find it worth mentioning that
$\lebsp{\cl{S}}{k}$ has a remarkably close resemblance to the
classical downward L\"owenheim-Skolem property, and can very well be
regarded as a finitary analogue of the latter.  It is pleasantly
surprising that while the downward L\"owenheim-Skolem property is by
itself meaningless over finite structures, a natural finitary analogue
of it is satisfied by a wide spectrum of classes of finite structures,
that are of interest and importance in computer science.

In summary, the properties introduced in this thesis are interesting
in both the classical and finite model theory contexts, and yield in
both these contexts, a new parameterized generalization of the {\lt}
preservation theorem.

%% file: mypubs.tex
\chapter*{List of Publications}
\addcontentsline{toc}{chapter}{List of publications}
{\large{I. Classical model theory}}

\begin{enumerate}
\item Abhisekh Sankaran, Bharat Adsul, and Supratik
  Chakraborty. A generalization of the {\lt} preservation
    theorem. \emph{Annals of Pure and Applied Logic}, 167(3):189 - 210, 2016.
\end{enumerate}

{\large{II. Finite model theory}}
\begin{enumerate}
\item
Abhisekh Sankaran, Bharat Adsul, and Supratik Chakraborty. A
  generalization of the {\lt} preservation theorem over classes of
  finite structures. In \emph{Proceedings of the 39th International
Symposium on Mathematical Foundations of Computer Science, MFCS 2014,
Budapest, Hungary, August 25-29, 2014, Part I}, pages 474 - 485, 2014.

\item
A. Sankaran, B. Adsul, V. Madan, P. Kamath, and
S. Chakraborty. Preservation under substructures modulo bounded
cores. In \emph{Proceedings of the 19th International Workshop on
  Logic, Language, Information and Computation, WoLLIC 2012, Buenos
  Aires, Argentina, September 3-6, 2012}, pages 291 - 305, 2012.

\end{enumerate}

\newpage

%% file: intro-5.0.tex
\chapter{Introduction}\label{chapter:introduction}


Classical model theory is a subject within mathematical logic, that
studies the relationship between a formal language and its
interpretations, also called structures or models~\cite{chang-keisler,
  hodges}. The most well-studied formal language in classical model
theory is \emph{first order logic} (henceforth called FO), a language
that is built up from predicates, functions and constant symbols using
boolean connectives, and existential and universal quantifications.
Classical model theory largely studies the correspondence between the
\emph{syntax} of a description in FO with the \emph{semantics} of the
description, where the latter is the class of all structures that
satisfy the description~\cite{tarski-truth}.


Amongst the earliest areas of study in classical model theory, is a
class of results called \emph{preservation theorems}. A preservation
theorem identifies syntactic features that capture a
\emph{preservation property}, which is a special kind of semantics
that defines classes of arbitrary structures (that is, structures that
could be finite or infinite) that are closed or \emph{preserved} under
some model-theoretic operation.  For instance, the class of all
cliques (graphs in which any two vertices are adjacent) is preserved
under the operation of taking substructures (which are induced
subgraphs in this context). The class of all cliques is defined by the
FO sentence that says ``for all (vertices) $x$ and forall all
(vertices) $y$, (there is an) edge between $x$ and $y$''. The latter
is a description in FO having the special syntactic feature that it
contains only universal quantifications and no existential
quantifications. One of the earliest preservation theorems of
classical model theory, the {\lt} theorem, proven by Jerzy {\L}o{\'s}
and Alfred Tarski in 1954-55~\cite{hodges}, says that the
aforementioned syntactic feature is indeed \emph{expressively
  complete} for the semantics of preservation under substructures. In
other words, a class of arbitrary structures that is defined by an FO
sentence, is preserved under substructures if, and only if, it is
definable by a universal sentence, the latter being an FO sentence in
which only universal quantifications appear (Theorem 3.2.2
in~\cite{chang-keisler}).  In ``dual'' form, the {\lt} theorem states
that a class of arbitrary structures that is defined by an FO
sentence, is preserved under extensions if, and only if, it is
definable by an existential sentence which is an FO sentence that uses
only existential quantifications.  The theorem extends to theories
(sets of sentences) as well: a class of arbitrary structures that is
defined by an FO theory, is preserved under substructures
(respectively, extensions) if, and only if, it is definable by a
theory of universal (respectively, existential) sentences.
Historically speaking, the study of preservation theorems began with
Marczewski asking in 1951, which FO definable classes of structures
are preserved under surjective
homomorphisms~\cite{hodges-history}. This question triggered off an
extensive study of preservation theorems in which a variety of
model-theoretic operations like substructures, extensions,
homomorphisms, unions of chains, direct products, reduced products,
etc. were taken up and preservation theorems for these operations were
proven.

The {\lt} theorem holds a special place amongst preservation theorems,
for its significance from at least two points of view: historical and
technical. From the historical point of view, the theorem was amongst
the earliest applications of the \emph{compactness theorem} (G\"odel
1930, Mal'tsev 1936)~\cite{hodges}, a result that is now regarded as
one of the pillars of classical model theory. Further, the method of
proof of the {\lt} theorem lent itself to adaptations that enabled
proving the various other preservation theorems mentioned above. This
extensive research into preservation theorems from the '50s to the
'70s (indeed these theorems were subsequently also proven for
extensions of FO, like infinitary
logics~\cite{keisler:infinitary-logic-book}) contributed much to the
development of classical model theory~\cite{hodges-history}.  From the
technical point of view, the property of preservation under
substructures that the {\lt} theorem characterizes, has been studied
substantially in the literature of various mathematical disciplines,
under the name of \emph{hereditariness}. A property is hereditary if
for any structure satisfying the property, any substructure of it also
satisfies the property.\sindex[term]{hereditary} Hereditary properties
or their variants have been of significant interest in topology, set
theory, graph theory and poset theory, to name a few areas. In more
detail, in topology, the notions of second countability and
metrisability are hereditary, while those of sequentiality and
Hausdorff compactness are what are called \emph{weakly
  hereditary}~\cite{kelley}. In set theory, the notions of hereditary
sets, hereditarily finite sets and hereditarily countable sets are all
hereditary properties~\cite{kunen}. Various classes of graphs of
interest in graph theory are hereditary; examples include cliques,
forests, $n$-partite graphs, planar graphs, graphs of bounded degree,
graphs that exclude any fixed finite set of graphs as subgraphs or
induced subgraphs~\cite{diestel}. In poset theory, a landmark result
in the sub-area of well-quasi-orders~\cite{basicwqotheory}, namely the
Robertson-Seymour theorem~\cite{robertson-seymour}, characterizes a
variant of hereditariness, called \emph{minor-hereditariness}. The
{\lt} theorem studies hereditariness from the point of view of logic,
and specifically, provides a syntactic characterization of hereditary
classes of structures that are FO definable.

While a preservation theorem can be seen as providing a syntactic
characterization of a preservation property, the same theorem, flipped
around, can also be seen as providing a semantic characterization (and
furthermore, via a preservation property) of a syntactic class of FO
theories. Thus, the {\lt} theorem provides semantic characterizations
of existential and universal theories, in terms of preservation under
extensions and preservation under substructures
respectively. Existential and universal theories are equivalent
respectively to what are known in the literature as $\Sigma^0_1$ and
$\Pi^0_1$ theories. For $n \ge 1$, a $\Sigma^0_n$ theory is a set of
$\Sigma^0_n$ sentences, where a $\Sigma^0_n$ sentence is a
\emph{prenex} FO sentence in which from left to right, there is a
\emph{quantifier prefix} consisting of $n$ blocks of quantifiers
(equivalently, $n-1$ alternations of quantifiers) beginning with a
block of existential quantifiers, followed by a quantifier-free
formula.  Likewise a $\Pi^0_n$ theory is a set of $\Pi^0_n$ sentences,
where a $\Pi^0_n$ sentence is a prenex FO sentence in which from left
to right, there is a quantifier prefix consisting of $n$ blocks of
quantifiers beginning with a block of universal quantifiers, followed
by a quantifier-free formula. The {\lt} theorem provides semantic
characterizations of $\Sigma^0_1$ and $\Pi^0_1$ theories. For
$\Sigma^0_n$ sentences and $\Pi^0_n$ theories for $n \ge 2$, semantic
characterizations were proven using preservation properties defined in
terms of \emph{ascending chains} and \emph{descending chains} (Theorem
5.2.8 in~\cite{chang-keisler}). Finally in 1960, Keisler proved the
\emph{$n$-sandwich theorem}~\cite{keisler-sandwich} that provides a
characterization of $\Sigma^0_n$ and $\Pi^0_n$ sentences and theories,
for each $n \ge 1$, using preservation properties defined uniformly in
terms of the notion of \emph{$n$-sandwiches}.  It is important to note
that all of the characterizations mentioned above are over arbitrary
structures, and make important use of the presence of infinite
structures.




In 1973, Fagin proved a remarkable syntax-semantics correspondence
over \emph{finite} structures. He showed that an isomorphism-closed
class of finite structures has the (algorithmic) semantic property of
being in the complexity class \textsf{NP} (Non-deterministic
Polynomial time) if, and only if, it is definable in an extension of
{\fo} called \emph{existential second order logic} (Theorem 3.2.4
in~\cite{gradel-vardi-et-al}). This result gave birth to the area of
\emph{finite model theory}, whose aims are similar to classical model
theory (i.e. study of the expressive power of formal languages) but
now the structures under consideration are only finite. Finite model
theory~\cite{libkin, gradel-vardi-et-al} is closely connected with
computer science since many disciplines within the latter use formal
languages, such as programming languages, database query languages or
specification languages, and further the structures that arise in
these disciplines are often finite, such as data structures, databases
or program models respectively. It is natural to ask if the results
and techniques of classical model theory can be carried over to the
finite.  Unfortunately, it turns out that many important results and
methods of classical model theory fail in the context of finite
structures.  The most stark failure is that of the compactness
theorem, whereby, all proofs based on the compactness theorem --
indeed this includes the proofs of almost every preservation theorem
-- fail when restricted to only finite structures. But worse still,
the statements of most preservation theorems fail too. The {\lt}
theorem fails in the finite; Tait~\cite{tait} showed there is an FO
sentence that is preserved under substructures over the class of all
finite structures, but that is not equivalent over this class, to any
universal sentence. The other preservation theorems from the classical
model theory literature mentioned earlier, namely those characterizing
$\Sigma^0_n$ and $\Pi^0_n$ sentences/theories for $n \ge 2$, fail in
the finite too; this is simply because the characterizing notions
become trivial over finite structures.
Two rare theorems that survive passage into the finite are the modal
characterization theorem and the homomorphism preservation
theorem\sindex[term]{homomorphism preservation theorem} -- the former
was shown by Rosen~\cite{rosen-van-benthem}, and the latter was a
striking result due to Rossman~\cite{rossman-hom}, that settled a long
standing open problem in finite model theory concerning the status of
this theorem in the finite. But then, these results are
exceptions. (See~\cite{gurevich84, gurevich-ajtai,
  ajtai-gurevich94, rosen-thesis, stolboushkin, rosen-weinstein,
  gurevich-alechina, gradel-rosen} for more on the investigations of
results from classical model theory in the context of all finite
structures. See~\cite{rosen} for an excellent survey of these.)

To ``recover'' classical preservation theorems in the finite model
theory setting, recent
research~\cite{dawar-pres-under-ext,duris-ext,nicole-lmcs-15,dawar-hom,dawar-quasi-wide,dawar-survey}
in the last ten years, has focussed attention on studying these
theorems over ``well-behaved'' classes of finite structures.  In
particular, Atserias, Dawar and Grohe showed
in~\cite{dawar-pres-under-ext} that under suitable closure
assumptions, classes of structures that are acyclic or of bounded
degree admit the {\lt} theorem for sentences. Likewise, the class of
all structures of tree-width at most $k$ also admits the {\lt}
theorem, for each natural number $k$.  These classes of structures are
well-behaved in the sense that they have proved especially important
in modern graph structure theory as also from an algorithmic point of
view~\cite{downey-fellows}. A classic result from graph structure
theory states that a minor-hereditary class of graphs has bounded
treewidth if, and only if, the class has a finite set of forbidden
minors that includes a planar graph~\cite{graphminortreewidth}.  From
an algorithmic point of view, many computational problems that are
otherwise intractable (such as 3-colorability), become tractable when
restricted to structures of bounded
treewidth~\cite{courcelle}. Likewise, over structures of bounded
degree, many problems that are polynomial time solvable in general
(such as checking if a graph is triangle-free) become solvable in
linear time~\cite{seese}.  Atserias, Dawar and Kolatis showed that the
homomorphism preservation theorem also holds over the aforesaid
classes of structures~\cite{dawar-hom}. (Note that this theorem being
true over all finite structures does not imply that it would be true
over subclasses of finite structures; restricting attention to a
subclass weakens both the hypothesis and the consequent of the
statement of the theorem).  Subsequently, Harwath, Heimberg and
Schweikardt~\cite{nicole-lmcs-15} studied the bounds for an effective
version of the {\lt} theorem and the homomorphism preservation theorem
over bounded degree structures. In~\cite{duris-ext}, Duris showed that
the {\lt} theorem holds for structures that are acyclic in a more
general sense.  All of the classes of structures mentioned above are
thus ``well-behaved'' from the model-theoretic point of view as well
(in that, these classes admit theorems from classical model
theory). The investigation of such model-theoretic well-behavedness is
a current and active area of research in finite model theory.

\input{results-2.0}

%% file: results-2.0.tex
\section{Our results}\label{section:our-results}


The properties in the classical model theory literature that
characterize $\Sigma^0_n$ and $\Pi^0_n$ sentences or theories,
characterize these syntactic classes ``as a whole''. None of these
characterize $\Sigma^0_n$ and $\Pi^0_n$ sentences/theories in which
for some given block, the number of quantifiers in that block is
\emph{fixed} to a given natural number $k$.  Further, all of these
properties are in terms of notions that are ``infinitary'',
i.e. notions that are non-trivial only when arbitrary (i.e. finite and
infinite) structures are considered, and that become trivial when
restricted to only finite structures.  Given the active interest in
preservation theorems in the finite model theory context, none of the
properties mentioned above can be used to characterize $\Sigma^0_n$
and $\Pi^0_n$ sentences in the finite, for $n \ge 2$. Further, the
preservation theorems that have been investigated over well-behaved
classes in the finite model theory literature, namely the {\lt}
theorem and the homomorphism preservation theorem, are those that
characterize only $\Sigma^0_1$ or $\Pi^0_1$ sentences, or subclasses
of these.

The observations above raise the following two natural questions:

\begin{enumerate}[nosep]
\item[\textsf{Q1}.] Are there properties that semantically characterize, over
  arbitrary structures,  $\Sigma^0_n$ and $\Pi^0_n$ sentences/theories in
  which the number of quantifiers appearing in a given block(s) is
  fixed to a given natural number(s)?  

\item[\textsf{Q2}.] Are there properties that semantically characterize, over
  classes of finite structures, $\Sigma^0_n$ and $\Pi^0_n$ sentences
  in which the number of quantifiers appearing in a given block(s) is
  fixed to a given natural number(s)? If so, what are these classes?


\end{enumerate}

In this thesis, we consider the case when $n = 2$, and present our
partial results towards addressing the above questions. Specifically,
for the case of $\Sigma^0_2$ and $\Pi^0_2$ sentences, in which the
number of quantifiers in the \emph{leading block} is fixed to a given
natural number, we identify preservation properties that uniformly
answer both \textsf{Q1} and (the first part of) \textsf{Q2} in the
affirmative. In other words, we present quantitative dual
parameterized preservation properties that are finitary and
combinatorial, and that characterize over arbitrary structures and
over a variety of interesting classes of finite structures,
$\Sigma^0_2$ and $\Pi^0_2$ sentences whose quantifier prefixes are
respectively of the form $\exists^k \forall^*$ or $\forall^k
\exists^*$ (i.e. $k$ quantifiers in the first quantifier block
followed by zero or more quantifiers in the second quantifier block).
Our properties, that we call \emph{preservation under substructures
  modulo $k$-cruxes} and \emph{preservation under $k$-ary covered
  extensions} are exactly the classical properties of preservation
under substructures and preservation under extensions for the case of
$k = 0$. Whereby, our characterizations of $\exists^k \forall^*$ and
$\forall^k \exists^*$ sentences yield the {\lt} theorem for sentences
for the case of $k = 0$. We hence call our characterizations
collectively as the \emph{generalized {\lt} theorem for sentences at
  level $k$}, and denote it as $\glt{k}$. To the best of our
knowledge, our characterizations are the first to relate natural
quantitative properties of models of sentences in a semantic class to
counts of leading quantifiers in equivalent $\Sigma^0_2$ or $\Pi^0_2$
sentences. Before we present our results in more detail and provide
our answer to the second part of \textsf{Q2}, we briefly describe the
importance of the $\Sigma^0_2$ and $\Pi^0_2$ classes of sentences.



After Hilbert posed the \emph{Entscheidungsproblem} in 1928, namely
the problem of deciding if a given FO sentence is satisfiable,
abbreviated the SAT problem, one of the first classes of FO sentences
for which SAT was shown to be decidable, was the $\Sigma^0_2$
class. This was shown by Bernays and Sch\"onfinkel for $\Sigma^0_2$
sentences without equality, and later extended to full $\Sigma^0_2$ by
Ramsey ~\cite{ramsey} (on a historical note: it was in showing this
result that Ramsey proved the famous \emph{Ramsey's theorem}). In a
subsequent extensive research of about 70 years on the SAT problem for
prefix classes, it was shown~\cite{classical-decision-problem} that
$\Sigma^0_2$ is indeed one of the \emph{maximal} prefix classes for
which the SAT problem is decidable. Interestingly, on the other hand,
various subclasses of $\Pi^0_2$ class turn out to be \emph{minimal}
prefix classes for which the SAT problem is undecidable; for instance,
the class of $\Pi^0_2$ sentences with only two universal quantifiers
and over a vocabulary containing just one binary relation symbol, is
undecidable for SAT, when equality is allowed.  With the growth of
parameterized complexity theory~\cite{downey-fellows}, it became
interesting to study the computational complexity of the
satisfiability problem for the $\Sigma^0_2$ class, in terms of
\emph{counts of quantifiers} as parameters.  As shown
in~\cite{classical-decision-problem}, satisfiability for the
$\Sigma^0_2$ class is in $\mathsf{NTIME}((n\cdot k^m)^c)$, where $n$
is the length of the input sentence, $k$ and $m$ are the number of
existential and universal quantifiers respectively in the sentence,
and $c$ is a suitable constant.  In recent years, there has been
significant interest in the $\Sigma^0_2$ class from the program
verification and program synthesis communities as
well~\cite{srivastava,emmer,gulwani,bjorner}. Here, the $\Sigma^0_2$
class is also referred to as \emph{effectively propositional
  logic}. For the $\Pi^0_2$ class on the other hand, the database
community has shown a lot of active interest in this class in the
context of data exchange, data integration and data
interoperability~\cite{data-exch,data-integ,datalog,forall-exis-rules},
and much more recently, in the context of query answering over RDF and
OWL knowledge~\cite{mathew-1,mathew-2}.






In the remainder of this section, we describe our preservation
properties, and our main results and techniques. All of these in the
classical model theory setting are described in
Section~\ref{section:intro-results-in-CMT}, and these in the finite
model theory setting are described in
Section~\ref{section:intro-results-in-FMT}. The latter section also
contains our answer to the second part of \textsf{Q2} raised above.
The results that we present here contain, and generalize
significantly, the results
in~\cite{abhisekh-apal,abhisekh-mfcs,abhisekh-wollic}.


\input{intro-results-in-CMT-2.0}

\input{intro-results-in-FMT-2.0}

\input{organization-of-thesis}

%% file: intro-results-in-CMT-2.0.tex
\subsection{Results in the classical model theory context}\label{section:intro-results-in-CMT}

Our property of preservation under substructures modulo $k$-cruxes
($PSC(k)$), is a natural parameterized generalization of preservation
under substructures, as can be seen from its definition
(Definition~\ref{defn:PSC(k)}): A sentence $\phi$ is $PSC(k)$ if every
model $\mf{A}$ of $\phi$ contains a set $C$ of at most $k$ elements
such that any substructure of $\mf{A}$, \emph{that contains $C$},
satisfies $\phi$. It is evident that preservation under substructures
is a special case of $PSC(k)$ when $k$ equals 0. The property of
preservation under $k$-ary covered extensions ($PCE(k)$) is defined as
the dual of $PSC(k)$, whereby it generalizes the property of
preservation under extensions (Definition~\ref{defn:PCE(k)}). The
generalized Los-Tarski theorem for sentences at level $k$ ($\glt{k}$)
gives syntactic characterizations of $PSC(k)$ and $PCE(k)$ as follows
(Theorem~\ref{theorem:glt(k)}): (i) an FO sentence is $PSC(k)$ if, and
only if, it is equivalent to an $\exists^k \forall^*$ sentence, and
(ii) an FO sentence is $PCE(k)$ if, and only if, it is equivalent to a
$\forall^k \exists^*$ sentence. We call the former the
\emph{substructural version} of $\glt{k}$, and the latter the
\emph{extensional version} of $\glt{k}$.  The {\lt} theorem for
sentences is indeed a special case of $\glt{k}$ when $k$ equals 0.

Towards extending $\glt{k}$ to the case of theories (sets of
sentences), we first extend the notions of $PSC(k)$ and $PCE(k)$ to
theories, and consider separately the substructural and extensional
versions of $\glt{k}$. The extensional version of $\glt{k}$ lifts
naturally: a theory is $PCE(k)$ if, and only if, it is equivalent to a
theory of $\forall^k \exists^*$ sentences
(Theorem~\ref{theorem:ext-chars}(\ref{theorem:char-of-PCE(k)-theories})). The
substructural version of $\glt{k}$ however does not lift to theories,
as is witnessed by an intriguing counterexample that shows that there
is a theory of $\exists \forall^*$ sentences, i.e. $\Sigma^0_2$
sentences with just one existential quantifier, that is not $PSC(k)$
for any $k$. Nevertheless, we show that $PSC(k)$ theories are always
equivalent to $\Sigma^0_2$ theories, and as a (conditional) refinement
of this result, we show that under a well-motivated model-theoretic
hypothesis, $PSC(k)$ theories are equivalent to theories of $\exists^k
\forall^*$ sentences (Theorems
~\ref{theorem:subst-char-for-theories}(\ref{theorem:PSC-and-PSC(aleph_0)-for-theories})
and ~\ref{theorem:conditional-refinement}).


The above results give new semantic characterizations of the classes
of $\Sigma^0_2$ and $\Pi^0_2$ sentences: the properties of ``is
$PSC(k)$ for some $k$'' and ``is $PCE(k)$ for some $k$'' respectively
characterize these sentences. The situation however becomes different
when these characterizations are considered in the context of
theories: $\Pi^0_2$ theories turn out to be more general than $PCE(k)$
theories for any $k$, and $\Sigma^0_2$ theories, indeed even $\exists
\forall^*$ theories, turn out to be, as mentioned earlier, more
general than $PSC(k)$ theories for any $k$.  To get a characterization
of $\Sigma^0_2$ and $\Pi^0_2$ theories by staying within the ambit of
the flavour of our preservation properties, we introduce the
properties of $PSC(\lambda)$ and $PCE(\lambda)$ as natural
``infinitary'' extensions of $PSC(k)$ and $PCE(k)$ respectively, in
which the sizes of cruxes and arities of covers are now less than
$\lambda$, for an infinite cardinal $\lambda$. Indeed, these
extensions characterize $\Sigma^0_2$ and $\Pi^0_2$ theories
(Theorems~\ref{theorem:ext-chars}(\ref{theorem:char-of-PCE(lambda)-theories})
and~\ref{theorem:subst-char-for-theories}(\ref{theorem:PSC(lambda)-char-for-theories})),
thereby giving new characterizations of the latter. We apply these
characterizations to give new and simple proofs of well-known
inexpressibility results in FO such as the inexpressibility of
acyclicity, connectedness, bipartiteness, etc.

This completes the description of our results in the classical model
theory context. We present various directions for future work, and
sketch how natural generalizations of the properties of $PSC(k)$ and
$PCE(k)$ can be used to get finer characterizations of $\Sigma^0_n$
and $\Pi^0_n$ sentences/theories for $n > 2$, analogous to the finer
characterizations of $\Sigma^0_2$ and $\Pi^0_2$ sentences/theories by
$PSC(k)$ and $PCE(k)$.

We conclude this subsection by briefly describing the techniques we
use in proving our results described above. For $\glt{k}$, we give two
proofs, one via a special class of structures called
\emph{$\lambda$-saturated structures}, and the other via ascending
chains of structures (a similar proof works for the characterization
of $PCE(k)$ and $PCE(\lambda)$ theories). To very quickly describe the
former, we first show $\glt{k}$ over the class of $\lambda$-saturated
structures, and then using the fact that any arbitrary structure has
an ``FO-similar'' structure (i.e. a structure that satisfies the same
FO sentences) that is $\lambda$-saturated, we ``transfer'' the truth
of $\glt{k}$ over the class of $\lambda$-saturated structures, to that
over the class of all structures.  To show that $PSC(k)$ and
$PSC(\lambda)$ theories are equivalent to $\Sigma^0_2$ theories, we
use Keisler's characterization of $\Sigma^0_2$ theories in terms of a
preservation property defined in terms of 1-sandwiches, and show that
any theory that is $PSC(k)$ or $PSC(\lambda)$ satisfies this
preservation property. The proof of our result showing that under the
well-motivated model-theoretic hypothesis alluded to earlier, a
$PSC(k)$ theory is equivalent to an $\exists^k \forall^*$ theory, is
the most involved of all our proofs. It introduces a novel technique
of getting a syntactically defined FO theory equivalent to a given FO
theory satisfying a semantic property, \emph{by going outside of
  FO}. Specifically, for the case of $PSC(k)$ theories, under the
aforementiond model-theoretic hypothesis, we first ``go up'' into an
\emph{infinitary logic} and show that a $PSC(k)$ theory can be
characterized by syntactically defined sentences of this logic
(Lemma~\ref{lemma:inf-characterization-of-PSC-var(k)}). We then ``come
down'' back to FO by providing a translation of the aforesaid
infinitary sentences, to their equivalent FO theories, whenever these
sentences are known to be equivalent to FO theories
(Proposition~\ref{prop:char-of-inf-formulae-defining-elem-classes}).
The FO theories are obtained from suitable \emph{finite
  approximations} of the infinitary sentences, and turn out to be
theories of $\exists^k \forall^*$ sentences.  The ``coming down''
process can be seen as a ``compilation'' process (in the sense of
compilers used in computer science) in which a ``high level''
description -- via infinitary sentences that are known to be
equivalent to FO theories -- is translated into an equivalent ``low
level'' description -- via FO theories.  We believe this technique of
accessing the descriptive power of an infinitary logic followed by
accessing the translation power of ``compiler results'' of the kind
just mentioned, may have other applications as well.

%% file: intro-results-in-FMT-2.0.tex
\subsection{Results in the finite model theory context}\label{section:intro-results-in-FMT}

While the failure of the {\lt} theorem in the finite shows that
universal sentences cannot capture in the finite, the property of
preservation under substructures, we show a stronger result: that for
any $k \ge 0$, the class of $\exists^k \forall^*$ also cannot capture
in the finite, the property of preservation under substructures, and
hence (not capture) $PSC(l)$ for any $l \ge 0$
(Proposition~\ref{prop:failure-of-glt(k)-in-the-finite}).  This
therefore shows the failure of $\glt{k}$ over all finite structures,
for all $k \ge 0$.  What happens to $\glt{k}$ over the well-behaved
classes that have been identified by Atserias, Dawar and Grohe to
admit the {\lt} theorem?  It unfortunately turns out that none of the
above classes, in general, admits $\glt{k}$ for any $k \ge 2$.  We
show that the existence of induced paths of unbounded length in a
class is, under reasonable assumptions on the class, the reason for
the failure of $\glt{k}$ over the class
(Theorem~\ref{theorem:glt(k)-implies-bounded-ind-path-lengths}).
Since these assumptions are satisfied by the aforesaid well-behaved
classes and the latter allow unbounded induced path lengths in
general, $\glt{k}$ fails over these classes in general.

To ``recover'' $\glt{k}$ in the face of the above failures, we define
a new logic based combinatorial property of classes of finite
structures, that we call the \emph{$\mc{L}$-equivalent bounded
  substructure property}, denoted $\lebsp{\cl{S}}{k}$, where $\mc{L}$
is either $\fo$ or an extension of $\fo$, called \emph{monadic second
  order} logic ($\mso$), $\cl{S}$ is a class of finite structures, and
$k$ is a natural number (Definition~\ref{defn:lebsp}).  Intuitively,
this property says that any structure $\mf{A}$ in $\cl{S}$ contains a
small substructure $\mf{B}$ that is in $\cl{S}$ and that is
``logically similar'' to $\mf{A}$. More precisely, $\mf{B}$ is ``$(m,
\mc{L})$-similar'' to $\mf{A}$, in that $\mf{B}$ and $\mf{A}$ agree on
all $\mc{L}$ sentences of quantifier rank $m$, where $m$ is a given
number. The bound on the size of $\mf{B}$ depends only on $m$ (if
$\cl{S}$ and $k$ are fixed).  Further, such a small and $(m, \mc{L})$-similar
substructure can always be found ``around'' any given set of at most
$k$ elements of $\mf{A}$.

We show that $\lebsp{\cl{S}}{k}$ indeed entails
$\glt{k}$. Interestingly, it also entails the homomorphism
preservation theorem ($\hpt$). (In fact, more general versions of
$\glt{k}$ and $\hpt$ are entailed by $\lebsp{\cl{S}}{k}$; see
Theorem~\ref{theorem:lebsp-implies-glt(k)} and
Theorem~\ref{theorem:h-ebsp-implies-ghpt(k)}.) Furthermore, if
$\lebsp{\cl{S}}{k}$ holds with ``computable bounds'', i.e. if the
bound on the size of the small substructure as referred to in the
$\reflebsp$ definition, is computable , then effective versions of the
$\glt{k}$ and $\hpt$ are entailed by $\lebsp{\cl{S}}{k}$.

It turns out that a variety of classes of finite structures, that are
of interest in computer science and finite model theory, satisfy
$\lebsp{\cdot}{k}$, and moreover, with computable bounds. The classes
that we consider are broadly of two kinds: special kinds of labeled
posets and special kinds of graphs. For the case of labeled posets, we
show $\lebsp{\cdot}{k}$ for the cases of words, trees (of various
kinds such as unordered, ordered and ranked), and nested words over a
finite alphabet, and regular subclasses of these
(Theorem~\ref{theorem:words-and-trees-and-nested-words-satisfy-lebsp}).
While words and trees have had a long history of studies in the
literature, nested words are much more recent~\cite{alur-madhu}, and
have attracted a lot of attention as they admit a seamless
generalization of the theory of regular languages and are also closely
connected with visibly pushdown languages~\cite{visibly-pushdown}. For
the case of graphs, we show $\lebsp{\cdot}{k}$ holds for a very
general, and again very recently defined, class of graphs called
\emph{$n$-partite cographs}, and all hereditary subclasses of this
class (Theorem~\ref{theorem:n-partite-cographs-satisfy-lebsp}). This
class of graphs, introduced in~\cite{shrub-depth}, jointly generalizes
the classes of cographs, graph classes of bounded tree-depth and those
of bounded shrub-depth. The latter graph classes have various
interesting finiteness properties, and have become very prominent in
the context of fixed parameter tractability of MSO model checking, and
in the context of investigating when FO equals MSO in its expressive
power~\cite{gajarsky-MSO-FPT-bdd-tree-depth, lampis,
  shrub-depth-FO-equals-MSO,tree-depth-FO-equals-MSO}. Being
hereditary subclasses of the class of $n$-partite cographs, all these
graph classes satisfy $\lebsp{\cdot}{k}$.

We go further to give many methods to construct new classes of
structures satisfying $\lebsp{\cdot}{\cdot}$ (with computable bounds)
from classes known to satisfy $\lebsp{\cdot}{\cdot}$ (with computable
bounds). We show that $\lebsp{\cdot}{\cdot}$ is closed under taking
subclasses that are hereditary or $\mc{L}$-definable, and is also
closed under finite intersections and finite unions
(Lemma~\ref{lemma:lebsp-closure-under-set-theoretic-ops}).  We show
that $\lebsp{\cdot}{\cdot}$ remains preserved under various operations
on structures, that have been well-studied in the literature: unary
operations like complementation, transpose and the line graph
operation, binary ``sum-like'' operations~\cite{makowsky} such as
disjoint union and join, and binary ``product-like'' operations that
include various kinds of products like cartesian, tensor,
lexicographic and strong products. All of these are examples of
operations that can be implemented using what are called
\emph{quantifier-free translation
  schemes}~\cite{makowsky,model-theoretic-methods}.  We show that
$\febsp{\cdot}{\cdot}$ is always closed under such operations, and
$\mebsp{\cdot}{k}$ is closed under such operations, provided that they
are unary or sum-like. It follows that finite unions of classes
obtained by finite compositions of the aforesaid operations also
satisfies $\lebsp{\cdot}{\cdot}$. However, many interesting classes of
structures can be obtained only by taking infinite unions of the kind
just described, a notable example being the class of hamming graphs of
the $n$-clique~\cite{hamming-graphs}. We show that if the
aforementioned infinite unions are ``regular'', in a sense we make
precise, then $\lebsp{\cdot}{0}$ is preserved under these unions,
under reasonable assumptions on the operations
(Theorem~\ref{theorem:lebsp(S,0)-pres-under-MSO-def-op-tree-lang}). As
applications of this result, we get that the class of hamming graphs
of the $n$-clique satisfies $\febsp{\cdot}{0}$, as does the class of
$p$-dimensional grid posets, where $p$ belongs to any {\mso} definable
(using a linear order) class of natural numbers (like, even numbers).

The proofs of the above results rely on tree-representations of
structures, and proceed by performing appropriate ``prunings'' of, and
``graftings'' within, these trees, in a manner that preserves the
substructure and ``$(m, \mc{L})$-similarity'' relations between the
structures represented by these trees. The process eventually yields
small subtrees that represent bounded structures that are
substructures of, and are $(m, \mc{L})$-similar to, the original
structures.  Two key technical elements that are employed to perform
the aforementioned prunings and graftings are the finiteness of the
index of the ``$(m, \mc{L})$-similarity'' relation (which is an
equivalence relation) and the \emph{type-transfer property} of the
tree-representations. The latter means that the $(m,
\mc{L})$-similarity type of the structure represented by a tree
$\tree{t}$ is \emph{determined} by the multi-set of the 
$(m, \mc{L})$-similarity types of the structures represented by the
subtrees rooted at the children of the root of $\tree{t}$, and
further, determined only by a threshold number of appearances of each
$(m, \mc{L})$-similarity type in the multi-set, with the threshold
depending solely on $m$. (Thus any change in the multi-set with
respect to the $(m, \mc{L})$-similarity types in it that appear less
than threshold number of times, gets ``transferred'' to the $(m,
\mc{L})$-similarity type of the structure represented by the (changed)
tree $\tree{t}$.)  These techniques have been incorporated into a
single abstract result concerning tree representations,
(Theorem~\ref{theorem:abstract-tree-theorem}), which we believe might
be of independent interest.

Finally, we present three additional findings about the
$\lebsp{\cdot}{k}$ property.  We show that the $\mc{L}$-SAT problem
(the problem of deciding if a given $\mc{L}$ sentence is satisfiable)
is decidable over any class satisfying $\lebsp{\cdot}{k}$ with
computable bounds (Lemma~\ref{lemma:lebsp-and-lSAT}).  We next show
that any class of structures that is well-quasi-ordered under the
embedding relation satisfies $\lebsp{\cdot}{0}$
(Theorem~\ref{theorem:wqo-implies-lebsp}). The notion of
well-quasi-orders is very well-studied in the literature~\cite{higman,
  wqosurveykruskal, basicwqotheory} and has great algorithmic
implications. For instance, checking membership in any hereditary
subclass of a well-quasi-ordered class can be done efficiently
(i.e. in polynomial time).  Our result above not only gives a
technique to show the $\lebsp{\cdot}{0}$ property for a class (by
showing the class to be well-quasi-ordered) but also, flipped around,
gives a ``logic-based'' tool to show that a class of structures is not
w.q.o. under embedding (by showing that the class does not satisfy
$\lebsp{\cdot}{0}$).  Finally, we show that $\lebsp{\cdot}{k}$ can
very well be seen to be a \emph{finitary analogue} of the
model-theoretic property that the classical downward
L\"owenheim-Skolem theorem (one the first results of classical model
theory and a widely used tool in the subject, along with the
compactness theorem) states of FO and arbitrary structures.  This
theorem says that an infinite structure $\mf{A}$ over a countable
vocabulary always contains a countable substructure $\mf{B}$ that is
``FO-similar'' to $\mf{A}$, in that $\mf{B}$ and $\mf{A}$ agree on all
FO sentences. Further such a countable and FO-similar substructure can
always be found ``around'' any given countable set of elements of
$\mf{A}$.  The importance of the downward L\"owenheim-Skolem theorem
in classical model theory can be gauged from the fact that this
theorem, along with the compactness theorem, characterizes
FO~\cite{lindstrom}.  It indeed is pleasantly surprising that while
the downward L\"owenheim-Skolem theorem is by itself meaningless over
finite structures, a natural finitary analogue of the model-theoretic
property that this theorem talks about, is satisfied by a wide
spectrum of classes of finite structures, that are of interest and
importance in computer science and finite model theory.

This  answers the second part of \textsf{Q2} raised at the outset of
Section~\ref{section:our-results}.

We conclude this part of the thesis with several directions for future
work, two of which we highlight here. The first asks for an
investigation of a \emph{structural characterization} of
$\lebsp{\cdot}{k}$ motivated by the observation that any hereditary
class of graphs satisfying $\lebsp{\cdot}{k}$ has bounded induced path
lengths. The second of these is a conjecture. Though the failure of
$\glt{k}$ over the class of all finite structures shows that $PSC(k)$
sentences are more expressive than $\exists^k \forall^*$ sentences
over this class, it is still possible that all $PSC(k)$ sentences for
all $k \ge 0$, taken together, are just as expressive as $\Sigma^0_2$
sentences, over all finite structures. We conjecture that this is
indeed the case.

\vspace{5pt}
In summary, the properties and notions introduced in this thesis are
interesting in both the classical and finite model theory contexts,
and yield in both these contexts, a new, natural and parameterized
generalization of the {\lt} preservation theorem.

%% file: organization-of-thesis.tex
\section{Organization of the thesis}\label{section:thesis-organization}

The thesis is in two parts: the first concerning classical model
theory, and the second concerning finite model theory. We describe the
organization within each of these parts below. For (the statements of)
our results in the classical model theory part of the thesis (that are
contained in
Chapters~\ref{chapter:the-case-of-sentences},~\ref{chapter:the-case-of-theories}
and~\ref{chapter:open-questions-classical-model-theory}), we assume
\emph{arbitrary finite} vocabularies, unless explicitly stated
otherwise (though the proofs of these results may resort to infinite
vocabularies).  In the finite model theory part of the thesis, we
always consider \emph{finite relational} vocabularies, unless we
explicitly state otherwise.

\vspace{10pt}\tbf{\ul{Part I: Classical model theory}}

\vspace{3pt} Chapter~\ref{chapter:background-CMT}: We recall relevant
notions and notation from classical model theory literature. We also
recall the {\lt} theorem, and other results that we use in the
subsequent chapters.

Chapter~\ref{chapter:our-properties}: We define the properties of
$PSC(k)$ and $PCE(k)$, and formally show their duality.

Chapter~\ref{chapter:the-case-of-sentences}: We characterize our
properties and their variants for the case of sentences.  We present
$\glt{k}$, and provide two proofs of it using $\lambda$-saturated
structures and ascending chains of structures
(Section~\ref{section:glt-for-sentences}). We define the properties of
$PSC(\lambda)$ and $PCE(\lambda)$ for infinite cardinals $\lambda$,
provide characterizations of these properties, and present
applications of these characterizations in getting new proofs of known
inexpressibility results in FO
(Section~\ref{section:lambda-cruxes-and-lambda-ary-covers}).  We make
further observations about our results so far, and prove an
uncomputability result in connection with $PSC(k)$
(Section~\ref{section:uncomputability}).

Chapter~\ref{chapter:the-case-of-theories}: We characterize all the
properties introduced thus far, for the case of theories. 

Chapter~\ref{chapter:open-questions-classical-model-theory}: We
present directions for future work (in the classical model theory
context).
\vspace{10pt}

\tbf{\ul{Part II: Finite model theory}}

\vspace{3pt}Chapter~\ref{chapter:background-FMT}: We recall basic
notions, notation, and results from the finite model theory
literature. 

\vspace{3pt}Chapter~\ref{chapter:need-for-new-classes-for-GLT}: We
discuss the need to investigate new classes of finite structures for
$\glt{k}$.

Chapter~\ref{chapter:lebsp}: We define the property
$\lebsp{\cdot}{k}$, show that it entails $\glt{k}$
(Section~\ref{section:lebsp-and-glt}) and show precisely, its
connections with the {\dls} property
(Section~\ref{section:lebsp-and-dlsp}).

Chapter~\ref{chapter:classes-satisfying-lebsp}: We show that various
classes of structures satisfy $\lebsp{\cdot}{k}$. We first prove an
abstract result concerning tree representations
(Section~\ref{section:abstract-tree-theorem}), and then demonstrate
its applications in showing the $\lebsp{\cdot}{k}$ property for
classes of words, trees (unordered, ordered, ranked) and nested words
(Section~\ref{section:words-trees-nestedwords}), and the class of
$n$-partite cographs and its various important subclasses
(Section~\ref{section:n-partite-cographs}). We then give ways of
constructing new classes satisfying $\lebsp{\cdot}{\cdot}$, by
presenting various closure properties of the latter
(Section~\ref{section:closure-prop-of-lebsp}).

Chapter~\ref{chapter:additional-studies-on-lebsp}: We present
additional studies on $\lebsp{\cdot}{k}$. We show the decidability of
the $\mc{L}$-satisfiability problem over classes satisfying
$\lebsp{\cdot}{k}$ with computable bounds
(Section~\ref{section:lebsp-and-SAT}), and the connections of
$\lebsp{\cdot}{k}$ with well-quasi-orders
(Section~\ref{section:wqo-and-lebsp}). We also show that
$\lebsp{\cdot}{k}$ entails a parameterized generalization of the
homomorphism preservation theorem
(Section~\ref{section:lebsp-entails-the-hpt}).

Chapter~\ref{chapter:open-questions-finite-model-theory}: We present
directions for future work (in the finite model theory context).

Chapter~\ref{chapter:conclusion}: We summarize the contributions of
this thesis on both fronts, of classical model theory and finite model
theory.

%% file: classical-model-theory-results.tex
\input{background-CMT}

\input{properties}
\input{the-case-of-sentences}

\input{the-case-of-theories}

\input{open-questions-classical-model-theory}

%% file: background-CMT.tex
\chapter{Background and preliminaries}\label{chapter:background-CMT}

In this part of the thesis, we shall be concerned with arbitrary
structures (i.e. structures that are finite or infinite), and with the
logic $\fo$.  The forthcoming sections of this chapter introduce the
notation and terminology that we use throughout this part of the
thesis. We also recall relevant results from the literature that we
use in our proofs in the subsequent chapters.  The classic references
for all of the background that we set up in this chapter
are~\cite{chang-keisler, hodges, libkin}.

We denote ordinals and cardinals using the letters $\lambda, \mu,
\kappa$ or $\eta$\sindex[symb]{$\lambda, \mu, \kappa, \eta$}.  We let
$\mathbb{N}$ denote the set of natural numbers (zero included), and
typically denote the elements of $\mathbb{N}$ by the letters $i, j$
etc. The cardinality of a set $A$ is denoted as $|A|$; likewise the
length of a tuple $\bar{a}$ is denoted as $|\bar{a}|$.  We denote
$|\mathbb{N}|$ by either $\omega$ or $\aleph_0$.  We abbreviate in the
standard way, some English phrases that commonly appear in
mathematical literature. Specifically, `w.l.o.g' stands for `without
loss of generality', `iff' stands for `if and only if', `w.r.t.'
stands for `with respect to' and `resp.'  stands for `respectively'.

\section{Syntax and semantics of $\fo$}\label{section:background:FO-syntax-semantics}\sindex[term]{first order logic}

\tbf{Syntax}: A \emph{vocabulary}\sindex[term]{vocabulary}, denoted by
$\tau$\sindex[symb]{$\tau, \sigma$} or $\sigma$, is a (possibly
infinite) set of predicate, function and constant symbols.  We denote
variables by $x, y, z$, etc., possibly with numbers as subscripts. We
denote a sequence of variables by $\bar{x}, \bar{y}, \bar{z}$, etc.,
again possibly with numbers as subscripts. We define below the notions
of term, atomic formula and $\fo$ formula over $\tau$.
\begin{enumerate}[leftmargin=*,nosep]
\item A \emph{term}\sindex[term]{term} over $\tau$, or simply term if
  $\tau$ is clear from context, denoted using the letter `$t$'
  typically along with numbers as subscripts, is defined inductively
  as follows:
\begin{enumerate}[leftmargin=*,nosep]
\item A constant (of $\tau$) and a variable are terms each.
\item If $t_1, \ldots, t_n$ are terms over $\tau$, then $f(t_1,
  \ldots, t_n)$ is also a term over $\tau$ where $f$ is an $n$-ary
  function symbol of $\tau$.
\end{enumerate}
\item An \emph{atomic} formula\sindex[term]{atomic formula} over
  $\tau$ is one of the following:
\begin{enumerate}[leftmargin=*,nosep]
\item The formula $t_1 = t_2$ where $t_1$ and $t_2$ are terms over
  $\tau$ and `=' is a special predicate symbol not a part of $\tau$
  which is interpreted always as the equality relation.
\item The formula $R(t_1, \ldots, t_n)$ where $R$ is an $n$-ary
  relation symbol of $\tau$, and $t_1, \ldots, t_n$ are terms over
  $\tau$.
\end{enumerate}
\item An \emph{$\fo$ formula}\sindex[term]{formula} over $\tau$, also
  called an $\fo(\tau)$ formula, or simply formula if $\tau$ is clear
  from context, is defined inductively as follows:
\begin{enumerate}[leftmargin=*,nosep]
\item An atomic formula over $\tau$ is an $\fo(\tau)$ formula. 
\item If $\varphi_1$ and $\varphi_2$ are $\fo(\tau)$ formulae, then
  each of $\varphi_1 \wedge \varphi_2$, $\varphi_1 \vee \varphi_2$ and
  $\neg \varphi_1$ are also $\fo(\tau)$ formulae. Here, the symbols
  $\wedge, \vee$ and $\neg$ denote the usual boolean connectives 
  `and', `or' and `not' respectively.
\item If $\varphi_1$ is an $\fo(\tau)$ formula, then $\exists x
  \varphi_1$ and $\forall x \varphi_1$ are also $\fo(\tau)$ formulae.
  Here the symbols $\exists$ and $\forall$ denote respectively, the
  existential and universal quantifiers.
\end{enumerate}
\end{enumerate}

In addition to the letter $\varphi$, we use other Greek letters like
$\phi, \psi, \chi, \xi, \gamma, \alpha$ and $\beta$ to denote
formulae\sindex[symb]{$\alpha, \beta, \gamma, \xi, \phi, \varphi,
  \chi, \psi$}.  A formula without any quantifiers is called
\emph{quantifier-free}\sindex[term]{quantifier-free}.  We abbreviate
a block of quantifiers of the form $Q x_1 \ldots Q x_k$ by $Q^k
\bar{x}$ or $Q \bar{x}$\sindex[symb]{$\exists^k \bar{x}, \exists
  \bar{x}, \forall^k \bar{x}, \forall \bar{x}, \exists^*, \forall^*$}
(depending on what is better suited for the context), where $Q \in
\{\forall, \exists\}$ and $k \in \mathbb{N}$.  By $Q^*$, we mean a
block of $k$ $Q$ quantifiers, for some $k \in \mathbb{N}$.

We now define the notion of \emph{free variables} 
\sindex[term]{free variable} of a term or a formula. The term $x$,
where $x$ is a variable, has only one free variable, which is $x$
itself. A term that is a constant has no free variables. The set of
free variables of the term $f(t_1, \ldots, t_n)$ is the union of the
sets of free variables of $t_1, \ldots, t_n$.  The latter is also the
set of free variables of the atomic formula $R(t_1, \ldots, t_n)$. Any
free variable of the atomic formula $t_1 = t_2$ is a free variable of
either $t_1$ or $t_2$. The set of free variables of the formula
$\varphi_1 \wedge \varphi_2$, and of the formula $\varphi_1 \vee
\varphi_2$, is the union of the sets of free variables of $\varphi_1$
and $\varphi_2$. Negation preserves the free variables of a
formula. Finally, the free variables of $\exists x \varphi$ and
$\forall x \varphi$ are the free variables of $\varphi$ except for
$x$. We let $t(\bar{x})$, resp.\ $\varphi(\bar{x})$, denote a term
$t$, resp.\ formula $\varphi$, whose free variables are \emph{among}
$\bar{x}$.  A formula with no free variables is called a
\emph{sentence}.

\vspace{3pt} \tbf{Semantics}: Let $\tau = \tau_C \sqcup \tau_R \sqcup
\tau_F$ where $\tau_C, \tau_R$ and $\tau_F$ are respectively the set
of constant, relation and function symbols of $\tau$.  A
\emph{$\tau$-structure}\sindex[term]{$\tau$-structure} $\mf{A} =
(\mathsf{U}_{\mf{A}}, (c^{\mf{A}})_{c \in \tau_C}, (R^{\mf{A}})_{R \in
  \tau_R}, (f^{\mf{A}})_{f \in \tau_F})$ consists of a set
$\mathsf{U}_{\mf{A}}$ x\sindex[symb]{$\univ{\mf{A}}, c^{\mf{A}},
  R^{\mf{A}}, f^{\mf{A}}, t^{\mf{A}}$} called the
\emph{universe}\sindex[term]{universe} or the \emph{domain} of
$\mf{A}$, along with interpretations $c^{\mf{A}}, R^{\mf{A}}$ and
$f^{\mf{A}}$ of each of the symbols $c, R$ and $f$ of $\tau_C, \tau_R$
and $\tau_F$ respectively, such that
\begin{itemize}[nosep]
\item the constant symbol $c$ is interpreted as an element $c^{\mf{A}}
  \in \mathsf{U}_{\mf{A}}$
\item the $n$-ary relation symbol $R$ is interpreted as a set
  $R^{\mf{A}}$ of $n$-tuples of $\mf{A}$, i.e. $R^{\mf{A}} \subseteq
  (\mathsf{U}_{\mf{A}})^n$
\item the $n$-ary function symbol $f$ is interpreted as a function
  $f^{\mf{A}}: (\mathsf{U}_{\mf{A}})^n \rightarrow
  \mathsf{U}_{\mf{A}}$
\end{itemize}

When $\tau$ is clear from context, we refer to a $\tau$-structure as
simply a structure. We denote structures by $\mf{A}, \mf{B} $
\sindex[symb]{$\mf{A}, \mf{B}$}etc.

Towards the semantics of $\fo$, we first define, for a given structure
$\mf{A}$ and a term $t(x_1, \ldots, x_n)$, the \emph{value} in
$\mf{A}$, of $t(x_1, \ldots, x_n)$ for a given assignment $a_1,
\ldots, a_n$ of elements of $\mf{A}$, to the variables $x_1, \ldots,
x_n$. We denote this value as $t^{\mf{A}}(\bar{a})$ where $\bar{a} =
(a_1, \ldots, a_n)$.
\begin{itemize}[nosep]
\item If $t$ is a constant symbol, then $t^{\mf{A}}(\bar{a}) =
  c^{\mf{A}}$.
\item If $t$ is the variable $x_i$, then $t^{\mf{A}}(\bar{a}) = a_i$.
\item If $t = f(t_1, \ldots, t_n)$, then $t^{\mf{A}}(\bar{a}) =
  f^{\mf{A}}(t_1^{\mf{A}}(\bar{a}), \ldots, t_n^{\mf{A}}(\bar{a}))$.
\end{itemize}

We now define, for a given structure $\mf{A}$ and formula
$\varphi(\bar{x})$, the notion of the \emph{truth} of
$\varphi(\bar{x})$ in $\mf{A}$\sindex[term]{truth in a structure}
given an assignment $\bar{a}$ of elements of $\mf{A}$, to the
variables $\bar{x}$.  We denote by $(\mf{A}, \bar{a}) \models
\varphi(\bar{x})$\sindex[symb]{$\models$}, that $\varphi(\bar{x})$ is
true in $\mf{A}$ for the assignment $\bar{a}$ to $\bar{x}$. We then
call $(\mf{A}, \bar{a})$ a \emph{model} of
$\varphi(\bar{x})$\sindex[term]{model}.
\begin{itemize}[nosep]
\item If $\varphi(\bar{x})$ is the formula $t_1 = t_2$, then $(\mf{A},
  \bar{a}) \models \varphi(\bar{x})$ iff $t_1^{\mf{A}}(\bar{a}) =
  t_2^{\mf{A}} (\bar{a})$.
\item If $\varphi(\bar{x})$ is the formula $R(t_1, \ldots, t_n)$, then
  $(\mf{A}, \bar{a}) \models \varphi(\bar{x})$ iff
  $(t_1^{\mf{A}}(\bar{a}), \ldots, t_n^{\mf{A}}(\bar{a})) \in
  R^{\mf{A}}$.
\item If $\varphi(\bar{x})$ is the formula $\varphi_1(\bar{x}) \wedge
  \varphi_2(\bar{x})$, then $(\mf{A}, \bar{a}) \models
  \varphi(\bar{x})$ iff $(\mf{A}, \bar{a}) \models \varphi_1(\bar{x})$
  and \linebreak $(\mf{A}, \bar{a}) \models \varphi_2(\bar{x})$.
\item If $\varphi(\bar{x})$ is the formula $\varphi_1(\bar{x}) \vee
  \varphi_2(\bar{x})$, then $(\mf{A}, \bar{a}) \models
  \varphi(\bar{x})$ iff $(\mf{A}, \bar{a}) \models \varphi_1(\bar{x})$
  or \linebreak $(\mf{A}, \bar{a}) \models \varphi_2(\bar{x})$.
\item If $\varphi(\bar{x})$ is the formula $\neg \varphi_1(\bar{x})$,
  then $(\mf{A}, \bar{a}) \models \varphi(\bar{x})$ iff it is not the
  case that\linebreak $(\mf{A}, \bar{a}) \models \varphi_1(\bar{x})$.
\item If $\varphi(\bar{x})$ is the formula $\exists y
  \varphi_1(\bar{x}, y)$, then $(\mf{A}, \bar{a}) \models
  \varphi(\bar{x})$ iff there exists an element $b \in
  \mathsf{U}_{\mf{A}}$ such that $(\mf{A}, \bar{a}, b) \models
  \varphi_1(\bar{x}, y)$.
\item If $\varphi(\bar{x})$ is the formula $\forall y
  \varphi_1(\bar{x}, y)$, then $(\mf{A}, \bar{a}) \models
  \varphi(\bar{x})$ iff for all elements $b \in \mathsf{U}_{\mf{A}}$,
  $(\mf{A}, \bar{a}, b) \models \varphi_1(\bar{x}, y)$.
\end{itemize}

If $\varphi$ is a sentence, we denote the truth of $\varphi$ in
$\mf{A}$ simply as $\mf{A} \models \varphi$, and call $\mf{A}$ a model
of $\varphi$.

Given an $\fo(\tau)$ formula $\varphi(x_1, \ldots, x_n)$ and distinct
constants $c_1, \ldots, c_n$ not appearing in $\tau$, let $\varphi'$
be the $\fo$ sentence over the vocabulary $\tau \cup \{c_1, \ldots,
c_n\}$, obtained by substituting $c_i$ for the free occurrences of
$x_i$ in $\varphi(x_1, \ldots, x_n)$ for each $i \in \{1, \ldots,
n\}$. The following lemma connects the notions of truth of
$\varphi(x_1, \ldots, x_n)$ in a model and the truth of $\varphi'$ in
a model.

\begin{lemma}\label{lemma:connecting-truth-of-formulae-and-sentences}
$(\mf{A}, a_1, \ldots, a_n) \models \varphi(x_1, \ldots, x_n)$ iff
  $(\mf{A}, a_1, \ldots, a_n) \models \varphi'$.
\end{lemma}

Note the distinction between the two occurrences of ``$(\mf{A}, a_1,
\ldots, a_n)$'' in the lemma above. The occurrence on the left denotes
that $a_1, \ldots, a_n$ is an assignment to $x_1, \ldots, x_n$ in the
$\tau$-structure $\mf{A}$, whereas the occurrence on the right denotes
a $\tau_n$-structure.

\vspace{5pt} \tbf{Extending syntax and semantics to theories}: In
classical model theory, one frequently talks about $\fo$ theories. We
define the syntax and semantics of these now.  A \emph{theory},
resp.\ \emph{$\fo(\tau)$ theory},\sindex[term]{theory} is simply a set
of sentences, resp.\ a set of $\fo(\tau)$ sentences. An $\fo(\tau)$
theory is also referred to as a \emph{theory over $\tau$}.  We
typically denote theories using capital letters like $T, V, Y, Z$,
possibly with numbers as subscripts. A theory, resp.\ $\fo(\tau)$
theory, \emph{whose free variables are among $\bar{x}$}, is a set of
formulae, resp.\ $\fo(\tau)$ formulae, all of whose free variables are
among $\bar{x}$. We denote such theories as $T(\bar{x}), V(\bar{x})$
etc. Given a theory $T(\bar{x})$, a structure $\mf{A}$, and a tuple
$\bar{a}$ from $\mf{A}$ such that $|\bar{a}| = |\bar{x}|$, we denote
by $(\mf{A}, \bar{a}) \models T(\bar{x})$ that $(\mf{A}, \bar{a})
\models \varphi(\bar{x})$\sindex[symb]{$\models$} for each formula
$\varphi(\bar{x}) \in T(\bar{x})$. In such a case, we say $T(\bar{x})$
is \emph{true} in $\mf{A}$\sindex[term]{truth in a structure} for the
assignment $\bar{a}$ to $\bar{x}$, and that $(\mf{A}, \bar{a})$ is a
\emph{model} of $T(\bar{x})$\sindex[term]{model}.  If $T$ has no free
variables, then we denote the truth of $T$ in $\mf{A}$ as $\mf{A}
\models T$, and say that $\mf{A}$ is a model of $T$. Observe that if
$T$ is empty, then trivially, every structure $\mf{A}$ is a model of
$T$.

Given an $\fo(\tau)$ theory $T(x_1, \ldots, x_n)$ and distinct
constants $c_1, \ldots, c_n$ not appearing in $\tau$, let $T'$ be the
$\fo$ theory without free variables over the vocabulary $\tau \cup
\{c_1, \ldots, c_n\}$, obtained by substituting $c_i$ for the free
occurrences of $x_i$ in $T(x_1, \ldots, x_n)$ for each $i \in \{1,
\ldots, n\}$. Analogous to
Lemma~\ref{lemma:connecting-truth-of-formulae-and-sentences}, we have
the following lemma.

\begin{lemma}
$(\mf{A}, a_1, \ldots, a_n) \models T(x_1, \ldots, x_n)$ iff $(\mf{A},
  a_1, \ldots, a_n) \models T'$.
\end{lemma}

\tbf{Consistency, validity, entailment and equivalence}: Let
$T(\bar{x})$ be a given theory and $\varphi(\bar{x})$ be a given
formula.  We say $T(\bar{x})$ is \emph{consistent} or
\emph{satisfiable}\sindex[term]{consistent}\sindex[term]{satisfiable}
\footnote{\noindent In the literature, consistency has a
  proof-theoretic definition. However G\"odel's completeness theorem
  shows that for $\fo$, consistency is the same as satisfiability, the
  latter meaning the existence of a model. Hence, we do not make a
  distinction between consistency and satisfiability in this thesis.}
if it has a model, i.e. if there exists a structure $\mf{A}$ and tuple
$\bar{a}$ of $\mf{A}$ such that $|\bar{a}| = |\bar{x}|$ and $(\mf{A},
\bar{a}) \models T(\bar{x})$. If $T(\bar{x})$ is not consistent, then
we say it is \emph{inconsistent} or \emph{unsatisfiable}.  We say
$T(\bar{x})$ is \emph{valid}\sindex[term]{valid} if $(\mf{A}, \bar{a})
\models T(\bar{x})$ for every structure $\mf{A}$ and every tuple
$\bar{a}$ of $\mf{A}$ such that $|\bar{a}| = |\bar{x}|$. The notions
above have natural adaptations to formulae.  We say $\varphi(\bar{x})$
is satisfiable, unsatisfiable, or valid\sindex[term]{satisfiable} if
$\{\varphi(\bar{x})\}$ is satisfiable, unsatisfiable, or valid,
respectively. It is easy to see that $\varphi(\bar{x})$ is valid iff
$\neg \varphi(\bar{x})$ is unsatisfiable, and that $\varphi(\bar{x})$
and $\neg \varphi(\bar{x})$ can both be satisfiable.

We say $T(\bar{x})$ \emph{entails} $\varphi(\bar{x})$, denoted
$T(\bar{x}) \vdash
\varphi(\bar{x})$\sindex[symb]{$\vdash$}\sindex[term]{entails}, if
every model $(\mf{A}, \bar{a})$ of $T(\bar{x})$ is also a model of
$\varphi(\bar{x})$. For a formula $\psi(\bar{x})$, we denote by
$\psi(\bar{x}) \rightarrow \varphi(\bar{x})$, that the theory
$\{\psi(\bar{x})\}$ entails $\varphi(\bar{x})$. It is easy to verify
that $\psi(\bar{x}) \rightarrow \varphi(\bar{x})$ iff $\neg
\psi(\bar{x}) \vee \varphi(\bar{x})$ is valid. Given a theory
$Y(\bar{x})$, we say $T(\bar{x})$ is \emph{equivalent to} $Y(\bar{x})$
\sindex[term]{equivalent} if $T(\bar{x})$ entails every formula of
$Y(\bar{x})$, and vice-versa. We denote by $\psi(\bar{x})
\leftrightarrow \varphi(\bar{x})$\sindex[symb]{$\leftrightarrow$} that
$\{\psi(\bar{x})\}$ is equivalent to $\{\varphi(\bar{x})\}$.

We now adapt all the notions above to versions of these \emph{modulo
  theories}. Given a consistent theory $V$, we say $T(\bar{x})$ is
\emph{consistent or satisfiable modulo $V$} if $(V \cup T(\bar{x}))$
is consistent, and say $T(\bar{x})$ is \emph{inconsistent or
  unsatisfiable modulo $V$} if $(V \cup T(\bar{x}))$ is
inconsistent. We say $T(\bar{x})$ \emph{entails $\varphi(\bar{x})$
  modulo $V$}\sindex[term]{modulo a theory!entails} if $(V \cup
T(\bar{x})) \vdash \varphi(\bar{x})$, and say $T(\bar{x})$ and
$Y(\bar{x})$ are \emph{equivalent modulo $V$} 
\sindex[term]{modulo a theory!equivalent} if $(V \cup T(\bar{x}))$ is
equivalent to $(V \cup Y(\bar{x}))$.  This last notion is particularly
relevant for this part of the thesis, and stated in other words, it
says that for every model $\mf{A}$ of $V$, and for every tuple
$\bar{a}$ from $\mf{A}$ such that $|\bar{a}| = |\bar{x}|$, it is the
case that $(\mf{A}, \bar{a}) \models T(\bar{x})$ iff $(\mf{A},
\bar{a}) \models Y(\bar{x})$.  One can define all the notions just
mentioned, for formulae, analogously as in the previous paragraphs.

\section{$\Sigma^0_n$ and $\Pi^0_n$ formulae}\label{section:background:sigma-n-and-pi-n}

An $\fo$ formula in which all quantifiers appear first (from left to
right) followed by a quantifier-free formula, is said to be in
\emph{prenex normal form}. For such a formula, the sequence of
quantifiers is called the \emph{quantifier
  prefix}\sindex[term]{quantifier prefix}, and the quantifier-free
part is called the \emph{matrix} of the formula.  For every non-zero
$n \in \mathbb{N}$, we denote by $\Sigma^0_n$,
resp.\ $\Pi^0_n$\sindex[symb]{$\Sigma^0_n, \Pi^0_n$}, the class of all
$\fo$ formulae in prenex normal form, whose quantifier prefix begins
with $\exists$, resp.\ $\forall$, and consists of $n-1$ alternations
of quantifiers. We call $\Sigma^0_1$ formulae \emph{existential} and
$\Pi^0_1$ formulae \emph{universal}.  We call $\Sigma^0_2$ formulae
having $k$ existential quantifiers \emph{$\exists^k \forall^*$
  formulae}, and $\Pi^0_2$ formulae having $k$ universal quantifiers
\emph{$\forall^k \exists^*$ formulae}\sindex[symb]{$\exists^k
  \forall^*, \forall^k \exists^*$}. The aforementioned notions for
formulae, have natural liftings to theories: A $\Sigma^0_n$ theory,
resp.\ $\Pi^0_n$ theory, is a theory all of whose formulae are
$\Sigma^0_n$, resp.\ $\Pi^0_n$; an existential theory,
resp.\ universal theory, is a $\Sigma^0_1$ theory, resp.\ $\Pi^0_1$
theory; an $\exists^k \forall^*$ theory, resp.\ $\forall^k \exists^*$
theory, is a theory all of whose formulae are $\exists^k \forall^*$,
resp.\ $\forall^k \exists^*$.  We now have the following lemma.

\begin{lemma}
Every $\fo$ formula is equivalent to an $\fo$ formula in prenex normal
form. By extension, every $\fo$ theory is equivalent to a theory of
$\fo$ formulae, all of which are in prenex normal form.
\end{lemma}

\section{Notions concerning structures}\label{section:background:structures}

The \emph{size (or power)} of a structure $\mf{A}$ is the cardinality
of $\univ{\mf{A}}$.  A structure is called \emph{finite} if its size
is finite, else it is called \emph{infinite}.  A \emph{substructure of
  $\mf{A}$ induced by} a subset $B$ of $\mathsf{U}_{\mf{A}}$ is a
structure $\mf{B}$ such that (i) $\mathsf{U}_{\mf{B}} = \{
t^{\mf{A}}(\bar{a}) \mid t(x_1, \ldots, x_n)~\text{is a term
  over}~\tau,$ $ \bar{a}~\text{is an}~n\text{-tuple from}~B\}$ (ii)
$c^{\mf{B}} = c^{\mf{A}}$ for each constant symbol $c \in \tau$, (iii)
$R^{\mf{B}} = R^{\mf{A}} \cap (\mathsf{U}_{\mf{B}})^n$ for each
$n$-ary relation symbol $R \in \tau$, and (iv) $f^{\mf{B}}$ is the
restriction of $f^{\mf{A}}$ to $\mathsf{U}_{\mf{B}}$, for each $n$-ary
function symbol $f \in \tau$.  A
\emph{substructure}\sindex[term]{substructure} $\mf{B}$ of $\mf{A}$,
denoted $\mf{B} \subseteq \mf{A}$,\sindex[symb]{$\subseteq$} is a
substructure of $\mf{A}$ induced by some subset of
$\mathsf{U}_{\mf{A}}$. If $\mf{B} \subseteq \mf{A}$, we say $\mf{A}$
is an \emph{extension} of $\mf{B}$.\sindex[term]{extension} It is easy
to see that if $\mf{B} \subseteq \mf{A}$, then for all quantifier-free
formulae $\varphi(\bar{x})$ and all $n$-tuples $\bar{a}$ from $\mf{B}$
where $n = |\bar{x}|$, we have that $(\mf{B}, \bar{a}) \models
\varphi(\bar{x})$ iff $(\mf{A}, \bar{a}) \models \varphi(\bar{x})$. If
$(\mf{B}, \bar{a})$ and $(\mf{A}, \bar{a})$ agree on \emph{all} $\fo$
formulae $\varphi(\bar{x})$ (instead of only quantifier-free formulae)
for all $n$-tuples $\bar{a}$ from $\mf{B}$ where $n = |\bar{x}|$, then
we say $\mf{B}$ is an \emph{elementary
  substructure}\sindex[term]{substructure!elementary} of $\mf{A}$, or
$\mf{A}$ is an \emph{elementary
  extension}\sindex[term]{extension!elementary} of $\mf{B}$, and
denote it as $\mf{B} \preceq \mf{A}$.\sindex[symb]{$\preceq$} A notion
related to the notion of elementary substructure is that of elementary
equivalence: We say two structures $\mf{A}$ and $\mf{C}$ are
\emph{elementarily equivalent},\sindex[term]{elementarily!equivalent}
denoted $\mf{A} \equiv \mf{C}$,\sindex[symb]{$\equiv$} if they agree
on all $\fo$ sentences.

Given $\tau$-structures $\mf{A}$ and $\mf{B}$, an \emph{isomorphism}
from $\mf{A}$ to $\mf{B}$\sindex[term]{isomorphism}, denoted $h:
\mf{A} \rightarrow \mf{B}$, is a bijection $h: \mathsf{U}_{\mf{A}}
\rightarrow \mathsf{U}_{\mf{B}}$ such that (i) $c^{\mf{B}} =
h(c^{\mf{A}})$ for every constant symbol $c \in \tau$ (ii) $(a_1,
\ldots, a_n) \in R^{\mf{A}}$ iff $(h(a_1), \ldots, h(a_n)) \in
R^{\mf{B}}$ for every $n$-ary relation symbol $R \in \tau$, and (iii)
$f^{\mf{A}}(a_1, \ldots, a_n) = a$ iff $f^{\mf{B}}(h(a_1), \ldots,
h(a_n)) = h(a)$ for every $n$-ary function symbol $f \in \tau$. If an
isomorphism from $\mf{A}$ to $\mf{B}$ exists, then so does an
isomorphism from $\mf{B}$ to $\mf{A}$ (namely, the inverse of the
former isomorphism), and we say $\mf{A}$ and $\mf{B}$ are
\emph{isomorphic}, and denote it as $\mf{A} \cong \mf{B}$.
\sindex[symb]{$\cong$} We say $\mf{A}$ is \emph{(isomorphically)
  embeddable} in $\mf{B}$, or simply embeddable in $\mf{B}$, denoted
as $\mf{A} \hookrightarrow \mf{B}$,\sindex[symb]{$\hookrightarrow$} if
there exists a substructure $\mf{C}$ of $\mf{B}$ such that there is an
isomorphism $h: \mf{A} \rightarrow \mf{C}$. In such a case, we say $h$
is an \emph{(isomorphic) embedding}, or simply an
embedding\sindex[term]{embedding}, of $\mf{A}$ in $\mf{B}$.  We say
$\mf{A}$ is \emph{elementarily
  embeddable}\sindex[term]{elementarily!embeddable} in $\mf{B}$ if
there exists an elementary substructure $\mf{C}$ of $\mf{B}$ such that
there is an isomorphism $h: \mf{A} \rightarrow \mf{C}$. In such a
case, we say $h$ is an \emph{elementary
  embedding}\sindex[term]{elementary!embedding} of $\mf{A}$ in
$\mf{B}$.

Given vocabularies $\tau, \tau'$, we say $\tau'$ is an
\emph{expansion}\sindex[term]{expansion!of a vocabulary} of $\tau$ if
$\tau \subseteq \tau'$. Given a $\tau$-structure $\mf{A}$ and a
$\tau'$-structure $\mf{A}'$, we say $\mf{A}'$ is a
\emph{$\tau$'-expansion} of $\mf{A}$, or simply an expansion of
$\mf{A}$\sindex[term]{expansion!of a structure} (if $\tau'$ is clear
from context), if the universe of $\mf{A}'$ and the interpretations in
$\mf{A}'$, of the constant, predicate and function symbols of $\tau$
are exactly the same as those in $\mf{A}$ respectively. In such a
case, we also say that $\mf{A}$ is a
\emph{$\tau$-reduct}\sindex[term]{reduct of a structure} of
$\mf{A}'$. In this thesis, we will mostly consider expansions $\tau'$
of $\tau$ in which all the symbols of $\tau' \setminus \tau$ are
constants. Given a cardinal $\lambda$, we denote by
$\tau_\lambda$\sindex[symb]{$\tau_\lambda$}, a fixed expansion of
$\tau$ such that $\tau_\lambda \setminus \tau$ consists only of the
constants $c_1, \ldots c_\lambda$ that are distinct and do not appear
in $\tau$, and say that $\tau_\lambda$ is an expansion of $\tau$
\emph{with} constants $c_1, \ldots c_\lambda$.  Given a
$\tau$-structure $\mf{A}$ and a $\lambda$-tuple (i.e. a tuple of
length $\lambda$) $\bar{a} = (a_1, \ldots, a_\lambda)$ of elements of
$\mf{A}$, we denote by $(\mf{A}, \bar{a})$\sindex[symb]{$(\mf{A},
  \bar{a})$} the $\tau_\lambda$-structure whose $\tau$-reduct is
$\mf{A}$, and in which $c_i$ is interpreted as $a_i$, for each $i <
\lambda$.  Given a $\tau$-structure $\mf{A}$ and a subset $X$ of
$\mathsf{U}_{\mf{A}}$, we denote by $\tau_X$, the expansion of $\tau$
with $|X|$ many fresh and distinct constants, one constant per element
of $X$. We denote by $(\mf{A}, (a)_{a \in X})$ the $\tau_X$-expansion
of $\mf{A}$ in which the constant in $\tau_X \setminus \tau$
corresponding to an element $a$ of $X$, is interpreted as $a$ itself.
Given a $\tau$-structure $\mf{B}$ such that $\mf{B} \subseteq \mf{A}$,
if $X = \univ{\mf{B}}$, then we denote $\tau_X$ as
$\tau_{\mf{B}}$\sindex[symb]{$\tau_{\mf{B}}$}, and the structure
$(\mf{A}, (a)_{a \in X})$ as
$\mf{A}_{\mf{B}}$\sindex[symb]{$\mf{A}_{\mf{B}}$}. The
\emph{diagram}\sindex[term]{diagram} of $\mf{A}$, denoted
$\diag{\mf{A}}$\sindex[symb]{$\diag{\mf{A}}$}, is the set of all
quantifier-free $\fo(\tau_{\mf{A}})$ sentences that are true in
$\mf{A}_{\mf{A}}$.  The \emph{elementary
  diagram}\sindex[term]{elementary!diagram} of $\mf{A}$, denoted
$\eldiag{\mf{A}}$\sindex[symb]{$\eldiag{\mf{A}}$}, is the set of all
$\fo(\tau_{\mf{A}})$ sentences that are true in $\mf{A}_{\mf{A}}$. The
following lemma connects the notions described in the previous and
current paragraphs.
\begin{lemma}\label{lemma:prop-of-structures}
Let $\mf{A}$ and $\mf{B}$ be given $\tau$-structures. Then the
following are true.
\begin{enumerate}[nosep]
\item $\mf{A} \cong \mf{B}$ implies $\mf{A} \equiv \mf{B}$.
\item $\mf{A} \preceq \mf{B}$ implies $\mf{A} \equiv \mf{B}$.
\item $\mf{A} \preceq \mf{B}$ iff $\mf{A}_{\mf{A}} \equiv
  \mf{B}_{\mf{A}}$.
\item $\mf{A}$ is embeddable in $\mf{B}$ iff for some
  $\tau_{\mf{A}}$-expansion $\mf{B}'$ of $\mf{B}$, it is the case that
  $\mf{B}' \models \diag{\mf{A}}$.
\item $\mf{A}$ is elementarily embeddable in $\mf{B}$ iff for some
  $\tau_{\mf{A}}$-expansion $\mf{B}'$ of $\mf{B}$, it is the case that
  $\mf{B}' \models \eldiag{\mf{A}}$.
\item If $\mf{A}$ is finite and $\mf{A} \equiv \mf{B}$, then $\mf{A}
  \cong \mf{B}$.\label{lemma:prop-of-structures:finite-A}
\end{enumerate}
\end{lemma}

A class of structures is said to be
\emph{elementary}\sindex[term]{elementary!class of structures} if it
is the class of models of an $\fo$ theory. It is easy to see that an
elementary class of structures is closed under elementary equivalence,
and hence under isomorphisms.

We conclude this section by recalling some important results from the
literature~\cite{chang-keisler}. The first two of these below are
arguably the most important theorems\footnote{As an aside, by a
  celebrated result of Lindstr\"om~\cite{lindstrom}, $\fo$ is the only
  logic having certain well-defined and reasonable closure properties,
  that satisfies Theorem~\ref{theorem:compactness} and
  Theorem~\ref{theorem:DLS}.} of classical model theory
(see Theorem 1.3.22 in~\cite{chang-keisler} and Corollary 3.1.4
in~\cite{hodges}).

\begin{theorem}[Compactness theorem, G\"odel 1930, Mal'tsev 1936]\label{theorem:compactness}\sindex[term]{compactness theorem}
A theory $T(\bar{x})$ is consistent iff every finite subset of it is
consistent.
\end{theorem}

\begin{theorem}[Downward L\"owenheim-Skolem theorem, L\"owenheim 1915, Skolem 1920s, Mal'tsev 1936]\label{theorem:DLS}\sindex[term]{downward L\"owenheim-Skolem theorem}
Let $\mf{A}$ be a structure over a countable vocabulary and $W$ be a
set of at most $\lambda$ elements of $\mf{A}$, where $\lambda$ is an
infinite cardinal. Then there exists an elementary substructure
$\mf{B}$ of $\mf{A}$, that contains $W$ and that has size at most
$\lambda$.
\end{theorem}

An easy but important corollary of the compactness theorem is the
following.

\begin{lemma}[Corollary 5.4.2, ref.~\cite{hodges}]\label{lemma:exis-amalgam}
Let $\mf{A}$ and $\mf{B}$ be structures such that every existential
sentence that is true in $\mf{B}$ is true in $\mf{A}$. Then $\mf{B}$
is embeddable in an elementary extension of $\mf{A}$.
\end{lemma}

To state the final result that we recall here from literature, we need
some terminology.  Given a cardinal $\lambda$, an \emph{ascending
  chain}\sindex[term]{ascending chain}, or simply chain,
$(\mf{A}_\eta)_{\eta < \lambda}$\sindex[symb]{$(\mf{A}_\eta)_{\eta <
    \lambda}$} of structures is a sequence $\mf{A}_0, \mf{A}_1,
\ldots$ of structures such that $\mf{A}_{0} \subseteq \mf{A}_{1}
\subseteq \ldots$. The \emph{union}\sindex[term]{union of a chain} of
the chain $(\mf{A}_\eta)_{\eta < \lambda}$ is a structure $\mf{A}$
defined as follows: (i) $\univ{\mf{A}} = \bigcup_{\eta < \lambda}
\univ{\mf{A}_\eta}$ (ii) $c^{\mf{A}} = c^{\mf{A}_\eta}$ for every
constant symbol $c \in \tau$ and every $\eta < \lambda$ (observe that
$c^{\mf{A}}$ is well-defined) (iii) $R^{\mf{A}} = \bigcup_{\eta <
  \lambda} R^{\mf{A}_\eta}$ for every relation symbol $R \in \tau$
(iv) $f^{\mf{A}} = \bigcup_{\eta < \lambda} f^{\mf{A}_\eta}$ for every
function symbol $f \in \tau$ (here, in taking the union of functions,
we view an $n$-ary function as its corresponding $(n+1)$-ary
relation). It is clear that $\mf{A}$ is well-defined. We denote
$\mf{A}$ as $\bigcup_{\eta < \lambda}
\mf{A}_\eta$\sindex[symb]{$\bigcup_{\eta < \lambda} \mf{A}_\eta$}. A
chain $(\mf{A}_\eta)_{\eta < \lambda}$ with the property that
$\mf{A}_{0} \preceq \mf{A}_{1} \preceq \ldots$ is said to be an
\emph{elementary chain}. We now have the following result (Theorem
3.1.9 of~\cite{chang-keisler}).

\begin{theorem}[Elementary chain theorem, Tarski-Vaught]\label{theorem:elem-chain-theorem}\sindex[term]{elementary!chain theorem}
Let $(\mf{A}_\eta)_{\eta < \lambda}$ be an elementary chain of
structures. Then $\bigcup_{\eta < \lambda} \mf{A}_\eta$ is an
elementary extension of $\mf{A}_\eta$ for each $\eta < \lambda$.
\end{theorem}


\section{Types and $\mu$-saturation}\label{section:background:saturation}

Given a vocabulary $\tau$, a set $\Gamma(x_1, \ldots, x_k)$ of
$\fo(\tau)$ formulae, all of whose free variables are among $x_1,
\ldots, x_k$, is said to be an \emph{$\fo$-type of $\tau$}, or simply
a type of $\tau$, if it is \emph{maximally consistent}, i.e. if it is
consistent and for any $\fo(\tau)$ formula $\varphi(x_1, \ldots,
x_k)$, exactly one of $\varphi(x_1, \ldots, x_k)$ and $\neg
\varphi(x_1, \ldots, x_k)$ belongs to $\Gamma(x_1, \ldots, x_k)$.
Given a $\tau$-structure $\mf{A}$ and a $k$-tuple $\bar{a}$ of
$\mf{A}$, we let $\fotp{\mf{A}}{\bar{a}}(x_1, \ldots,
x_k)$\sindex[symb]{$\fotp{\mf{A}}{\bar{a}}(x_1, \ldots, x_k)$} denote
the \emph{type of $\bar{a}$ in $\mf{A}$}\sindex[term]{type of a
  tuple}, i.e. the set of all $\fo(\tau)$ formulae $\varphi(x_1,
\ldots, x_k)$ such that $(\mf{A}, \bar{a}) \models \varphi(x_1,
\ldots, x_k)$.  It is clear that $\fotp{\mf{A}}{\bar{a}}(x_1, \ldots,
x_k)$ is a type of $\tau$. The converse is clear too: if $\Gamma(x_1,
\ldots, x_k)$ is a type of $\tau$, then for some $\tau$-structure
$\mf{A}$ and some $k$-tuple $\bar{a}$ of $\mf{A}$, it is the case that
$\Gamma(x_1, \ldots, x_k)$ is the type of $\bar{a}$ in $\mf{A}$. In
such a case, we say that $\mf{A}$ \emph{realizes} $\Gamma(x_1, \ldots,
x_k)$, and that $\bar{a}$ \emph{satisfies}, or
\emph{realizes}\sindex[term]{tuple realizing a type}, $\Gamma(x_1,
\ldots, x_k)$ in $\mf{A}$.  It is easy to see for given structures
$\mf{A}$ and $\mf{B}$, and given $k$-tuples $\bar{a}$ and $\bar{b}$
from $\mf{A}$ and $\mf{B}$ resp., that $\fotp{\mf{A}}{\bar{a}}(x_1,
\ldots, x_k)$ = $\fotp{\mf{A}}{\bar{a}}(x_1, \ldots, x_k)$ iff
$(\mf{A}, \bar{a}) \equiv (\mf{B}, \bar{b})$.  By
$\pitp{\mf{A}}{\bar{a}}(x_1, \ldots,
x_k)$\sindex[symb]{$\pitp{\mf{A}}{\bar{a}}(x_1, \ldots, x_k)$}, we
denote the \emph{$\Pi^0_1$-type of $\bar{a}$ in
  $\mf{A}$}\sindex[term]{$\Pi^0_1$-type of a tuple}, i.e. the subset
of $\fotp{\mf{A}}{\bar{a}}(x_1, \ldots, x_k)$ that consists of all
$\Pi^0_1$ formulae of the latter. We denote by
$\foth(\mf{A})$\sindex[symb]{$\foth(\mf{A})$}, the \emph{theory of
  $\mf{A}$}, i.e. the set of all $\fo(\tau)$ sentences that are true
in $\mf{A}$.

We now recall the important notion of $\mu$-saturated structures from
the literature. 



\begin{defn}[Chp. 5, ref.~\cite{chang-keisler}]\label{defn:lambda-saturated-structure}\sindex[term]{$\mu$-saturation}
Let $\mu$ be a cardinal. A $\tau$-structure $\mf{A}$ is said to be
\emph{$\mu$-saturated} if for every subset $X$ of $\mathsf{U}_{\mf{A}}$,
of cardinality less than $\mu$, if $\mf{A}'$ is the
$\tau_X$-expansion $(\mf{A}, (a)_{a \in X})$ of $\mf{A}$, then
$\mf{A}'$ realizes every type $\Gamma(x)$ of the vocabulary $\tau_X$,
that is consistent modulo $\foth(\mf{A}')$.
\end{defn}

Following are some results in connection with $\mu$-saturated
structures, that we crucially use in many proofs in the forthcoming
chapters.

\begin{proposition}\label{prop:saturation-results}
The following are true for any vocabulary $\tau$ and any
$\tau$-structure $\mf{A}$.
\begin{enumerate}[nosep]
\item \,[Proposition 5.1.1, ref.~\cite{chang-keisler}]~ $\mf{A}$ is
  $\mu$-saturated if and only if for every ordinal $\eta < \mu$ and
  every $\eta$-tuple $\bar{a}$ of $\mf{A}$, the expansion $(\mf{A},
  \bar{a})$ is $\mu$-saturated.
\label{prop:saturation-results:expansion-preserves-saturation}
\item \,[Proposition 5.1.2, ref.~\cite{chang-keisler}]~$\mf{A}$ is
  finite if and only if $\mf{A}$ is $\mu$-saturated for all cardinals
  $\mu$.\label{prop:saturation-results:finite-is-always-saturated}
\item \,[Lemma 5.1.4, ref.~\cite{chang-keisler}]~There exists a
  $\mu$-saturated elementary extension of $\mf{A}$, for some cardinal
  $\mu \ge
  |\tau|$.\label{prop:saturation-results:exis-of-sat-elem-ext}
\item \,[Lemma 5.1.10, ref.~\cite{chang-keisler}]~If $\mf{A}$ is
  $\mu$-saturated, $\mf{A} \equiv \mf{B}$ and $\bar{b}$ is an
  $\eta$-tuple of $\mf{B}$ where $\eta < \mu$, then there exists an
  $\eta$-tuple $\bar{a}$ of $\mf{A}$ such that $(\mf{A}, \bar{a})
  \equiv (\mf{B},
  \bar{b})$.\label{prop:saturation-results:type-realization}
\item \,[Lemma 5.2.1, ref.~\cite{chang-keisler}]~Suppose every
  existential sentence that holds in $\mf{A}$ also holds in $\mf{B}$,
  where $\mf{B}$ is $\mu$-saturated for $\mu \ge |\mf{A}|$. Then
  $\mf{A}$ is embeddable in
  $\mf{B}$.\label{prop:saturation-results:isomorphic-embedding}
\end{enumerate}
\end{proposition}


\input{PS-and-PE}

%% file: PS-and-PE.tex
\section{Two classical preservation properties}\label{section:background-classical-properties}

We first recall the classical dual notions of preservation under
substructures and preservation under extensions. We fix a finite
vocabulary $\tau$ in our discussion below.

\begin{defn}\label{defn:PS-and-PE}
Let $\cl{S}$ be a class of structures. 
\begin{enumerate}[nosep]
\item A subclass $\cl{U}$ of $\cl{S}$ is said to be \emph{preserved
  under substructures over $\cl{S}$}, abbreviated as \emph{$\cl{U}$ is
  $PS$ over $\cl{S}$}, if for each structure $\mf{A} \in \cl{U}$, if
  $\mf{B} \subseteq \mf{A}$ and $\mf{B} \in \cl{S}$, then $\mf{B} \in
  \cl{U}$.\sindex[term]{preservation under!substructures}
\item A subclass $\cl{U}$ of $\cl{S}$ is said to be \emph{preserved
  under extensions over $\cl{S}$}, abbreviated as \emph{$\cl{U}$ is
  $PE$ over $\cl{S}$}, if for each structure $\mf{A} \in \cl{U}$, if
  $\mf{A} \subseteq \mf{B}$ and $\mf{B} \in \cl{S}$, then $\mf{B} \in
  \cl{U}$.\sindex[term]{preservation under!extensions}
\end{enumerate}
If $V$ and $T$ are theories, then we say \emph{$T$ is $PS$ modulo $V$}
(resp.\ \emph{$T$ is $PE$ modulo $V$}) if the class of models of $T
\cup V$ is $PS$ (resp.\ $PE$) over the class of models of $V$. For a
sentence $\phi$, we say \emph{$\phi$ is $PS$ modulo $V$}
(resp.\ \emph{$\phi$ is $PE$ modulo $V$}) if the theory $\{\phi\}$ is
$PS$ (resp.\ $PE$) modulo $V$.
\end{defn}

As an example, let $\tau = \{E\}$ be the vocabulary consisting of a
single relation symbol $E$ that is binary, and let $\cl{S}$ be the
class of all $\tau$-structures in which $E$ is interpreted as a
symmetric binary relation. The class $\cl{S}$ can be seen as the class
of all undirected graphs. Let $\cl{U}_1$ be the subclass of $\cl{S}$
consisting of all undirected graphs that are acyclic. Let $\cl{U}_2$
be the subclass of $\cl{S}$ consisting of all undirected graphs that
contain a triangle as a subgraph. It is easy to see that $\cl{U}_1$ is
$PS$ over $\cl{S}$, and $\cl{U}_2$ is $PE$ over $\cl{S}$. Observe that
$\cl{S}$ is defined by the theory $V = \{\forall x \forall y \,(E(x,
y) \rightarrow E(y, x))\}$.  Let $\psi_n$ be the universal sentence
that asserts the absence of a cycle of length $n$ as a subgraph. Then
$\cl{U}_1$ is exactly the class of models in $\cl{S}$, of the theory
$T = \{\psi_n \mid n \ge 3\}$, and $\cl{U}_2$ is exactly the class of
models in $\cl{S}$, of the sentence $\phi = \neg \psi_3$.  Whereby,
$T$ is $PS$ modulo $V$, and $\phi$ is $PE$ modulo $V$.

\vspace{3pt}
The following lemma establishes the duality between $PS$ and $PE$.

\begin{lemma}[$PS$-$PE$ duality]\label{lemma:PS-PE-duality}
Let $\cl{S}$ be a class of structures, $\cl{U}$ be a subclass of
$\cl{S}$ and $\overline{\cl{U}}$ be the complement of $\cl{U}$ in
$\cl{S}$. Then $\cl{U}$ is $PS$ over $\cl{S}$ iff $\overline{\cl{U}}$
is $PE$ over $\cl{S}$. In particular, if $\,\cl{S}$ is defined by a
theory $V$, then a sentence $\phi$ is $PS$ modulo $V$ iff $\neg \phi$
is $PE$ modulo $V$.
\end{lemma}

The notion of a theory being $PS$ modulo $V$ or $PE$ modulo $V$ can be
extended to theories with free variables in a natural manner.  Given
$n \in \mathbb{N}$, recall from
Section~\ref{section:background:structures} that $\tau_n$ is the
vocabulary obtained by expanding $\tau$ with $n$ fresh and distinct
constants symbols $c_1, \ldots, c_n$.  Let $T(\bar{x})$ be an
$\fo(\tau)$ theory with free variables among $\bar{x} = (x_1, \ldots,
x_n)$, and let $T'$ be the $\fo(\tau_n)$ theory obtained by
substituting $c_i$ for the free occurrences of $x_i$ in $T(\bar{x})$,
for each $i \in \{1, \ldots, n\}$. Given a theory $V$, we say
\emph{$T(\bar{x})$ is $PS$ modulo $V$} if $T'$ is $PS$ modulo $V$,
where $V$ is treated as an $\fo(\tau_n)$ theory. The notion
\emph{$T(\bar{x})$ is $PE$ modulo $V$} is defined similarly.

\subsection{The {\lt} preservation theorem}\label{subsection:los-tarski-theorem}

In the mid 1950s, Jerzy {\L}o{\'s} and Alfred Tarski provided syntactic
characterizations of theories that are $PS$ and theories that are $PE$
via the following preservation theorem. This result (Theorem 3.2.2. in
Chapter 3 of~\cite{chang-keisler}) and its proof set the trend for
various other preservation theorems to follow.

\begin{theorem}[{\lt}, 1954-55]\label{theorem:los-tarski-subst-version}\sindex[term]{{\lt} theorem}
Let $T(\bar{x})$ be a theory whose free variables are among
$\bar{x}$. Given a theory $V$, each of the following is true.
\begin{enumerate}[nosep]
\item $T(\bar{x})$ is $PS$ modulo $V$ iff $T(\bar{x})$ is equivalent
  modulo $V$ to a theory $Y(\bar{x})$ of universal formulae, all of
  whose free variables are among $\bar{x}$. If $T(\bar{x})$ is a
  singleton, then so is $Y(\bar{x})$.
\item $T(\bar{x})$ is $PE$ modulo $V$ iff $T(\bar{x})$ is equivalent
  modulo $V$ to a theory $Y(\bar{x})$ of existential formulae, all of
  whose free variables are among $\bar{x}$. If $T(\bar{x})$ is a
  singleton, then so is $Y(\bar{x})$.
\end{enumerate}
 \end{theorem}


In the remainder of the thesis, if $\cl{S}$, as mentioned in the
definitions above, is clear from context, then we skip mentioning its
associated qualifier, namely, `over $\cl{S}$'. Likewise, we skip
mentioning `modulo $V$' when $V$ is clear from context.

%% file: properties.tex
\chapter{New parameterized preservation properties}\label{chapter:our-properties}
We fix a finite vocabulary $\tau$ in our discussion in this and in all
the subsequent chapters of this part of the thesis. By formula,
theory, and structure, we always mean respectively an $\fo(\tau)$
formula, an $\fo(\tau)$ theory and a $\tau$-structure, unless
explicitly stated otherwise.

\section{Preservation under substructures modulo $k$-cruxes}\label{section:PSC(k)}

\begin{defn}\label{defn:PSC(k)}\sindex[symb]{$PSC(k)$}
Let $\cl{S}$ be a class of structures and $k \in \mathbb{N}$. A
subclass $\cl{U}$ of $\cl{S}$ is said to be \emph{preserved under
  substructures modulo $k$-cruxes over
  $\cl{S}$}\sindex[term]{preservation under!substructures modulo
  $k$-cruxes}, abbreviated as \emph{$\cl{U}$ is $PSC(k)$ over
  $\cl{S}$}\sindex[symb]{$PSC(k)$}, if for every structure $\mf{A} \in
\cl{U}$, there exists a subset $C$ of the universe of $\mf{A}$, of
size at most $k$, such that if $\mf{B} \subseteq \mf{A}$, $\mf{B}$
contains $C$ and $\mf{B} \in \cl{S}$, then $\mf{B} \in \cl{U}$.  The
set $C$ is called a \emph{$k$-crux of $\mf{A}$ w.r.t. $\cl{U}$ over
  $\cl{S}$}\sindex[term]{$k$-crux of a model}. Any substructure
$\mf{B}$ of $\mf{A}$, that contains $C$ is called a \emph{substructure
  of $\mf{A}$ modulo the $k$-crux
  $C$}\sindex[term]{substructure!modulo a $k$-crux}. Given theories
$V$ and $T$, we say \emph{$T$ is $PSC(k)$ modulo $V$}, if the class of
models of $T \cup V$ is $PSC(k)$ over the class of models of $V$. For
a sentence $\phi$, we say \emph{$\phi$ is $PSC(k)$ modulo
  $V$}\sindex[term]{modulo a theory!$PSC(k)$} if the theory $\{\phi\}$
is $PSC(k)$ modulo $V$.

\end{defn}

Let $\cl{U}, \cl{S}, \mf{A}, C, V, T$ and $\phi$ be as above. If
$\cl{S}$ is defined by $V$ and $\cl{U}$ is defined by $T$ over
$\cl{S}$, then we say $C$ is a \emph{$k$-crux of $\mf{A}$ w.r.t. $T$
  modulo $V$}. If $\cl{U}$ is defined by $\phi$ over $\cl{S}$, then we
say $C$ is a \emph{$k$-crux of $\mf{A}$ w.r.t. $\phi$ modulo $V$}. In
many occasions in this thesis, the set $C$ is the set of elements of a
tuple $\bar{a}$ of elements of $\mf{A}$. Hence, we use the phrase
\emph{$\bar{a}$ is a $k$-crux of $\mf{A}$ w.r.t. $T$ modulo $V$} or
\emph{$\bar{a}$ is a $k$-crux of $\mf{A}$ w.r.t. $\phi$ modulo $V$} to
mean the corresponding statements with $C$ in place of $\bar{a}$.  As
in Section~\ref{section:background-classical-properties}, if any of
$\cl{U}, \cl{S}, T, V$ or $\phi$ is clear from context, then we skip
mentioning its associated qualifier (viz., `w.r.t. $\cl{U}$', `over
$\cl{S}$', `w.r.t. $T$', `modulo $V$' and `w.r.t. $\phi$'
respectively) in the definitions above.

\begin{remark}
Definition~\ref{defn:PSC(k)} is an adapted version of related
definitions in~\cite{abhisekh-wollic} and~\cite{arxiv-self-2010}. The
notion of `core' in Definition 1 of~\cite{abhisekh-wollic} is exactly
the notion of `crux' defined above, where the underlying class
$\cl{U}$ in the definition above, is the class of all structures. We
avoid using the word `core' for a crux to prevent confusion with
existing notions of cores in the literature~\cite{rossman-hom,
  dawar-hom}.
\end{remark}

Given an $\text{FO}(\tau)$ theory $T(\bar{x})$ and an FO$(\tau)$
formula $\phi(\bar{x})$ each of whose free variables are among
$\bar{x} = (x_1, \ldots, x_n)$, we can define the notion of
$T(\bar{x})$, resp.\ $\phi(\bar{x})$, being $PSC(k)$ modulo a theory
$V$ analogously to the notion of $T(\bar{x})$, resp.\ $\phi(\bar{x})$,
being $PS$ modulo $V$ as defined in
Section~\ref{section:background-classical-properties}. Specifically,
let $c_1, \ldots, c_n$ be the distinct constant symbols of $\tau_n
\setminus \tau$, and let $T'$ be the $\text{FO}(\tau_n)$ theory
obtained by substituting $c_i$ for the free occurrences of $x_i$ in
$T(\bar{x})$, for each $i \in \{1, \ldots, n\}$. Then we say
\emph{$T(\bar{x})$ is $PSC(k)$ modulo $V$} if $T'$ is $PSC(k)$ modulo
$V$, where $V$ is treated as an $\text{FO}(\tau_n)$ theory. The notion
\emph{$\phi(\bar{x})$ is $PSC(k)$ modulo $V$} is defined similarly.

\begin{example}
Let $\cl{S}$ be the class of all undirected graphs. Given $k \in
\mathbb{N}$, consider the class $\cl{U}_k$ of all graphs of $\cl{S}$
containing a cycle of length $k$ as a subgraph. Clearly, for any graph
$G$ in $\cl{U}_k$, the vertices of any cycle of length $k$ in $G$ form
a $k$-crux of $G$ w.r.t. $\cl{U}_k$.  Hence $\cl{U}_k$ is $PSC(k)$.
It is easy to see that $\cl{U}_k$ is definable by an FO sentence, call
it $\phi$, whereby $\phi$ is $PSC(k)$.
\end{example}

Fix a class $\cl{S}$ of structures. For properties $P_1$ and $P_2$ of
subclasses of $\cl{S}$, we denote by $P_1 \Rightarrow P_2$ that any
subclass of $\cl{S}$ satisfying $P_1$ also satisfies $P_2$. We denote
by $P_1 \Leftrightarrow P_2$ that $P_1 \Rightarrow P_2$ and $P_2
\Rightarrow P_1$.  It is now easy to check the following facts
concerning the $PSC(k)$ subclasses of $\cl{S}$: (i) $PSC(0)$ coincides
with the property of preservation under substructures, so $PSC(0)
\Leftrightarrow PS$ (ii) $PSC(l) \Rightarrow PSC(k)$ for $l \leq
k$. If $\cl{S}$ is any substructure-closed class of structures over a
purely relational vocabulary (a vocabulary that contains only relation
symbols), that contains infinitely many finite structures, then for
each $l$, there exists $k > l$ and a $PSC(k)$ subclass $\cl{U}$ of
$\cl{S}$ such that $\cl{U}$ is not $PSC(l)$ over $\cl{S}$. This is
seen as follows. Given $l$, let $k > l$ be such that there is some
structure of size $k$ in $\cl{S}$, and let $\phi_k$ be the sentence
asserting that there are at least $k$ elements in any model. Clearly
$\phi_k$ is $PSC(k)$ over $\cl{S}$ but not $PSC(l)$ over $\cl{S}$.

Define the property $PSC$ of subclasses of $\cl{S}$ as follows: A
subclass $\cl{U}$ of $\cl{S}$ is \emph{$PSC$ over
  $\cl{S}$}\sindex[symb]{$PSC$} if it is $PSC(k)$ over $\cl{S}$ for
some $k \in \mathbb{N}$. Notationally, $PSC \Leftrightarrow \bigvee_{k
  \ge 0} PSC(k)$.  If $\cl{S}$ is defined by a theory $V$, then the
notions of `a sentence is $PSC$ modulo $V$' and `a theory is $PSC$
modulo $V$' are defined similarly as in Definition~\ref{defn:PSC(k)},
and these notions are extended to formulae and theories with free
variables similarly as done above for $PSC(k)$. The implications
mentioned in the previous paragraph show that $PSC$ generalizes $PS$.
If $\cl{S}$ is any substructure-closed class of purely relational
structures, that contains infinitely many finite structures, then the
strict implications mentioned above show a strictly infinite hierarchy
within $PSC$; whence $PSC$ provides a strict generalization of $PS$.

Suppose that $\cl{S}$ is defined by a theory $V$. Given a $\Sigma^0_2$
sentence $\phi = \exists x_1 \ldots \exists x_k$ $ \forall
\bar{y}~\varphi(x_1, \ldots, x_k, \bar{y})$ and a structure $\mf{A}$
of $\cl{S}$ such that $\mf{A} \models \phi$, any set of
\emph{witnesses}\sindex[term]{witness} in $\mf{A}$ of the existential
quantifiers of $\phi$, forms a $k$-crux of $\mf{A}$. In particular, if
$a_1, \ldots, a_k$ are witnesses in $\mf{A}$, of the quantifiers
associated with $x_1, \ldots, x_k$ (whence $\mf{A} \models \forall
\bar{y}~\varphi(a_1,a_2,\ldots, a_k,\bar{y})$), then given any
substructure $\mf{B}$ of $\mf{A}$ containing $a_1, \ldots, a_k$, the
latter elements can again be chosen as witnesses in $\mf{B}$, to make
$\phi$ true in $\mf{B}$. Therefore, $\phi$ is $PSC(k)$ (modulo
$V$). It follows that $\Sigma^0_2$ formulae with $k$ existential
quantifiers are also $PSC(k)$ (modulo $V$).

\begin{remark}
Contrary to intuition, witnesses and $k$-cruxes cannot always be
equated!  Consider the sentence $\phi = \exists x \forall y E(x,y)$
and the structure $\mf{A} = (\mathbb{N}, \leq)$, i.e.\ the natural
numbers with the usual ordering. Let $\cl{S}$ be the class of all
structures. Clearly, $\phi$ is $PSC(1)$, $\mf{A} \models \phi$ and the
only witness of the existential quantifier of $\phi$ in $\mf{A}$ is
the minimum element $0 \in \mathbb{N}$.  In contrast, every singleton
subset of $\mathbb{N}$ is a 1-crux of $\mf{A}$, since each
substructure of $\mf{A}$ contains a minimum element under the induced
order; this in turn is due to $\mathbb{N}$ being well-ordered by
$\leq$.  This example shows that there can be models having many more
(even infinitely more) cruxes than witnesses.
\end{remark}

Since $\Sigma^0_1$ and $\Pi^0_1$ formulae are also $\Sigma^0_2$
formulae and the latter are $PSC$, the former are also $PSC$. However,
$\Pi^0_2$ formulae are not necessarily $PSC$. Consider $\phi = \forall
x \exists y E(x, y)$ and consider the model $\mf{A}$ of $\phi$ given
by $\mf{A} = (\mathbb{N}, E^\mf{A} = \{(i, i+1) \mid i \in
\mathbb{N}\})$. It is easy to check that no finite substructure of
$\mf{A}$ models $\phi$; then $\mf{A}$ does not have any $k$-crux for
any $k \in \mathbb{N}$, whence $\phi$ is not $PSC(k)$ for any $k$, and
hence is not $PSC$.

\section{Preservation under $k$-ary covered extensions}\label{section:PCE(k)}
The classical notion of ``extension of a structure'' has a natural
generalization to\linebreak the notion of \emph{extension of a
  collection of structures} as follows.  A structure $\mf{A}$ is said
to be an extension of a collection $R$ of structures if for each
$\mf{B} \in R$, we have $\mf{B} \subseteq \mf{A}$.  We now define a
special kind of extensions of a collection of structures.
\begin{defn}\label{defn:k-ary-covered-ext}
For $k \in \mathbb{N}$, a structure $\mf{A}$ is said to be a
\emph{$k$-ary covered extension}
\sindex[term]{extension!$k$-ary covered} of a non-empty collection $R$
of structures if (i) $\mf{A}$ is an extension of $R$, and (ii) for
every subset $C$ of the universe of $\mf{A}$, of size at most $k$,
there is a structure in $R$ that contains $C$.  We call $R$ a
\emph{$k$-ary cover}\sindex[term]{$k$-ary cover} of $\mf{A}$.
\end{defn}

\begin{example}
Let $\mf{A}$ be a graph on $n$ vertices and let $R$ be the collection
of all $r$ sized induced subgraphs of $\mf{A}$, where $1 \leq r \leq
n$. Then $\mf{A}$ is a $k$-ary covered extension of $R$ for every $k$
in $\{0, \ldots, r\}$.
\end{example}

\begin{remark}\label{remark:some-obs-about-k-ary-covered-ext}
Note that a $0$-ary covered extension of $R$ is simply an extension of
$R$.  For $k > 0$, the universe of a $k$-ary covered extension of $R$
is necessarily the union of the universes of the structures in $R$.
However, different $k$-ary covered extensions of $R$ can differ in the
interpretation of predicates (if any) of arity greater than $k$.  Note
also that a $k$-ary covered extension of $R$ is an $l$-ary covered
extension of $R$ for every $l \in \{0, \ldots, k\}$.
\end{remark}

\begin{defn}\label{defn:PCE(k)}
Let $\cl{S}$ be a class of structures and $k \in \mathbb{N}$. A
subclass $\cl{U}$ of $\cl{S}$ is said to be \emph{preserved under
  $k$-ary covered extensions}\sindex[term]{preservation under!$k$-ary
  covered extensions} over $\cl{S}$, abbreviated as \emph{$\cl{U}$ is
  $PCE(k)$ over $\cl{S}$}, if for every collection $R$ of structures
of $\cl{U}$, if $\mf{A}$ is a $k$-ary covered extension of $R$ and
$\mf{A} \in \cl{S}$, then $\mf{A} \in \cl{U}$. Given theories $V$ and
$T$, we say \emph{$T$ is $PCE(k)$ modulo $V$}\sindex[symb]{$PCE(k)$}
if the class of models of $T \cup V$ is $PCE(k)$ over the class of
models of $V$. For a sentence $\phi$, we say \emph{$\phi$ is $PCE(k)$
  modulo $V$} if the theory $\{\phi\}$ is $PCE(k)$ modulo $V$.
\end{defn}

As in the previous subsection, if any of $\cl{S}$ or $V$ is clear from
context, then we skip mentioning its associated qualifier.  Again as
in the previous section, given a theory $T(\bar{x})$ and a formula
$\phi(\bar{x})$ each of whose free variables is among $\bar{x}$, we
can define the notions of `$T(\bar{x})$ is $PCE(k)$ modulo $V$' and
`$\phi(\bar{x})$ is $PCE(k)$ modulo $V$' analogously to the
corresponding notions in the context of $PSC(k)$.

The following lemma establishes the duality between $PSC(k)$ and
$PCE(k)$, generalizing the duality between $PS$ and $PE$ given by
Lemma \ref{lemma:PS-PE-duality}.

\begin{lemma}[$PSC(k)$-$PCE(k)$ duality]\label{lemma:PSC(k)-PCE(k)-duality}
Let $\cl{S}$ be a class of structures, $\cl{U}$ be a subclass of
$\cl{S}$ and $\overline{\cl{U}}$ be the complement of $\cl{U}$ in
$\cl{S}$. Then $\cl{U}$ is $PSC(k)$ over $\cl{S}$ iff
$\overline{\cl{U}}$ is $PCE(k)$ over $\cl{S}$, for each $k \in
\mathbb{N}$. In particular, if $\,\cl{S}$ is defined by a theory $V$,
then a sentence $\phi$ is $PSC(k)$ modulo $V$ iff $\neg \phi$ is
$PCE(k)$ modulo $V$.
\end{lemma}
\begin{proof} 

\underline{If:}~ Suppose $\overline{\cl{U}}$ is $PCE(k)$ over $\cl{S}$
but $\cl{U}$ is not $PSC(k)$ over $\cl{S}$. Then there exists $\mf{A}
\in \cl{U}$ such that for every set $C$ of at most $k$ elements of
$\mf{A}$, there is a substructure $\mf{B}_C$ of $\mf{A}$ that (i)
contains $C$, and (ii) belongs to $\cl{S} \setminus \cl{U}$,
i.e.\ belongs to $\overline{\cl{U}}$. Then $R = \{\mf{B}_C \mid C
~\text{is a subset of }\mf{A}, ~\text{of size at most }k\}$ is a
$k$-ary cover of $\mf{A}$. Since $\overline{\cl{U}}$ is $ PCE(k)$ over
$\cl{S}$, it follows that $\mf{A} \in \overline{\cl{U}}$ -- a
contradiction.

\underline{Only If}:~ Suppose $\cl{U}$ is $PSC(k)$ over $\cl{S}$ but
$\overline{\cl{U}}$ is not $PCE(k)$ over $\cl{S}$. Then there exists
$\mf{A} \in \cl{U}$ and a $k$-ary cover $R$ of $\mf{A}$ such that
every structure $\mf{B}$ of $R$ belongs to $\overline{\cl{U}}$. Since
$\cl{U}$ is $PSC(k)$ over $\cl{S}$, there exists a $k$-crux $C$ of
$\mf{A}$ w.r.t. $\cl{U}$ over $\cl{S}$. Consider the structure
$\mf{B}_C \in R$ that contains $C$ -- this exists since $R$ is a
$k$-ary cover of $\mf{A}$. Then $\mf{B}_C \in \cl{U}$ since $C$ is a
$k$-crux of $\mf{A}$ -- a contradiction.
\end{proof}

Fix a class $\cl{S}$ of structures.  Analogous to the notion of $PSC$,
we define the notion of $PCE$\sindex[symb]{$PCE$} as $PCE
\Leftrightarrow \bigvee_{k \ge 0} PCE(k)$.  The notions of a class, a
sentence, a formula, a theory (without free variables) and a theory
with free variables being $PCE$ are defined analogously to
corresponding notions for $PSC$.  Then from the discussion in
Section~\ref{section:PSC(k)}, and from
Remark~\ref{remark:some-obs-about-k-ary-covered-ext} and
Lemma~\ref{lemma:PSC(k)-PCE(k)-duality} above, we see that (i) $PCE(0)
\Leftrightarrow PE$, (ii) $PCE(l) \Rightarrow PCE(k)$ for $l \leq k$,
and (iii) a subclass $\cl{U}$ of $\cl{S}$ is $PSC$ over $\cl{S}$ iff
the complement $\overline{\cl{U}}$ of $\cl{U}$ in $\cl{S}$, is $PCE$
over $\cl{S}$.  Further, if $\cl{S}$ is defined by a theory $V$, then
all $\Pi^0_2$ formulae having at most $k$ universal quantifiers are
$PCE(k)$ (modulo $V$) and hence $PCE$, whereby all $\Sigma^0_1$ and
$\Pi^0_1$ formulae are $PCE$ as well. However $\Sigma^0_2$ formulae,
in general, are not $PCE$ since, as seen towards the end of
Section~\ref{section:PSC(k)}, $\Pi^0_2$ formulae are, in general, not
$PSC$.

In the next two chapters, we present characterizations of the $PSC(k)$
and $PCE(k)$ properties, and some natural variants of these. Our
results and methods of proof are in general very different in the case
of sentences vis-\'a-vis the case of theories. Hence, we deal with the
two cases separately.

%% file: the-case-of-sentences.tex
\chapter{Characterizations: the case of sentences}\label{chapter:the-case-of-sentences}

\input{glt-for-sentences}

\input{variants-of-our-properties}

\input{uncomputability}

%% file: glt-for-sentences.tex
\section{The generalized {\lt} theorem for sentences -- $\glt{k}$}\label{section:glt-for-sentences}

The central result of this section is as follows. 
This result, for the case of sentences, is called \emph{the
  generalized {\lt} theorem for sentences at level $k$}, or simply
\emph{the generalized {\lt} theorem for
  sentences}\sindex[term]{generalized {\lt} theorem}, and is denoted
as $\glt{k}$\sindex[symb]{$\glt{k}$}. Observe that for $k = 0$ below,
we get exactly the {\lt} theorem for sentences.

\begin{theorem}\label{theorem:glt(k)}
Given a theory $V$, the following are true for each $k \in
\mathbb{N}$.
\begin{enumerate}[nosep]
\item A formula $\phi(\bar{x})$ is $PSC(k)$ modulo $V$ iff
  $\phi(\bar{x})$ is equivalent modulo $V$ to a $\Sigma^0_2$ formula
  whose free variables are among $\bar{x}$, and that has $k$
  existential quantifiers.\label{theorem:glt(k)-subst}
\item A formula $\phi(\bar{x})$ is $PCE(k)$ modulo $V$ iff
  $\phi(\bar{x})$ is equivalent modulo $V$ to a $\Pi^0_2$ formula
  whose free variables are among $\bar{x}$, and that has $k$ universal
  quantifiers.\label{theorem:glt(k)-ext}
\end{enumerate}
\end{theorem}

Recall that $PSC \Leftrightarrow \bigvee_{k \ge 0} PSC(k)$ and $PCE
\Leftrightarrow \bigvee_{k \ge 0} PCE(k)$. We then have the following
corollary.

\begin{corollary}\label{corollary:glt-PSC-PCE}
Given a theory $V$, the following are true.
\begin{enumerate}[nosep]
\item A formula $\phi(\bar{x})$ is $PSC$ modulo $V$ iff
  $\phi(\bar{x})$ is equivalent modulo $V$ to a $\Sigma^0_2$ formula
  whose free variables are among $\bar{x}$.
\item A formula $\phi(\bar{x})$ is $PCE$ modulo $V$ iff
  $\phi(\bar{x})$ is equivalent modulo $V$ to a $\Pi^0_2$ formula
  whose free variables are among $\bar{x}$.
\end{enumerate}
\end{corollary}


\vspace{2pt} The rest of this section is devoted to proving
Theorem~\ref{theorem:glt(k)}. We present two proofs of this result,
one that uses $\lambda$-saturated structures (Section
\ref{subsection:proof-of-glt-using-saturated-structures}), and the
other that uses ascending chains of structures (Section
\ref{subsection:proof-of-glt-using-chains}).


\subsection{Proof of $\glt{k}$ using $\lambda$-saturated structures}\label{subsection:proof-of-glt-using-saturated-structures}

Given theories $T$ and $V$, we say that $\Gamma$ is the \emph{set of
  $\forall^k \exists^*$ consequences of $T$ modulo $V$} if $\Gamma =
\{\varphi \mid \varphi~\text{is a}~\forall^k \exists^*~\text{sentence,
  and}~T~\text{entails}~\varphi~\text{modulo}~V\}$. The following
lemma is key to the proof.

\begin{lemma}\label{lemma:k-ary-covers-for-saturated-models}
Let $V$ and $T$ be consistent theories, and $k \in \mathbb{N}$.  Let
$\Gamma$ be the set of $\forall^k \exists^*$ consequences of $T$
modulo $V$.  Then for all infinite cardinals $\mu$, for every
$\mu$-saturated structure $\mf{A}$ that models $V$, we have that
$\mf{A} \models \Gamma$ iff there exists a $k$-ary cover $R$ of
$\mf{A}$ such that \linebreak $\mf{B} \models (V \cup T)\,$ for every
$\mf{B} \in R$.
\end{lemma}
\begin{proof} 
The `If' direction is easy: for each $\mf{B} \in R$, since $\mf{B}
\models (V \cup T)$, we have $\mf{B} \models \varphi$ for each
$\varphi \in \Gamma$. From the discussion towards the end of
Section~\ref{section:PCE(k)}, any $\forall^k \exists^*$ sentence is
$PCE(k)$ modulo $V$. Then since $R$ is a $k$-ary cover of $\mf{A}$, we
have $\mf{A} \models \varphi$ for each $\varphi \in \Gamma$.

For the `Only If' direction, let the vocabulary of $V$ and $T$ be
$\tau$. We show that for every $k$-tuple $\bar{a}$ of $\mf{A}$, there
is a substructure $\mf{A}_{\bar{a}}$ of $\mf{A}$ containing (the
elements of) $\bar{a}$ such that $\mf{A}_{\bar{a}} \models (V \cup
T)$. Then the set $R = \{\mf{A}_{\bar{a}} \mid \bar{a}~\text{is
  a}~k\text{-tuple of}~\mf{A}\}$ forms the desired $k$-ary cover of
$\mf{A}$.  To show the existence of $\mf{A}_{\bar{a}}$, it suffices to
show that there exists a $\tau$-structure $\mf{B}$ such that (i)
$|\mf{B}| \le \mu$, (ii) $\mf{B} \models (V \cup T)$, and (iii)
the $\Pi^0_1$-type of $\bar{a}$ in $\mf{A}$, i.e. $\mathsf{tp}_{\Pi,
  \mf{A}, \bar{a}}(x_1, \ldots, x_k)$, is realized in $\mf{B}$ by some
$k$-tuple, say $\bar{b}$. Then every $\Sigma^0_1$ sentence of
FO$(\tau_k)$ true in $(\mf{B}, \bar{b})$ is also true in $(\mf{A},
\bar{a})$. Since $\mf{A}$ is $\mu$-saturated, we have by
Proposition~\ref{prop:saturation-results}(\ref{prop:saturation-results:expansion-preserves-saturation}),
that $(\mf{A}, \bar{a})$ is also $\mu$-saturated.  There exists
then, an isomorphic embedding $f: (\mf{B}, \bar{b}) \rightarrow
(\mf{A}, \bar{a})$ by
Proposition~\ref{prop:saturation-results}(\ref{prop:saturation-results:isomorphic-embedding}).
Whereby the $\tau$-reduct of the image of $(\mf{B},
\bar{b})$ under $f$ can serve as $\mf{A}_{\bar{a}}$.  The proof is
therefore completed by showing the existence of $\mf{B}$ with the
above properties.

Suppose $Z(x_1, \ldots, x_k) = V \cup T \cup \mathsf{tp}_{\Pi, \mf{A},
  \bar{a}}(x_1, \ldots, x_k)$ is inconsistent. By the compactness
theorem, there is a finite subset of $Z(x_1, \ldots, x_k)$ that is
inconsistent. Since $\mathsf{tp}_{\Pi, \mf{A}, \bar{a}}(x_1, \ldots,
x_k)$ is closed under taking finite conjunctions and since each of
$\mathsf{tp}_{\Pi, \mf{A}, \bar{a}}(x_1, \ldots, x_k)$, $V$ and $T$ is
consistent, there is a formula $\psi(x_1, \ldots, x_k)$ in
$\mathsf{tp}_{\Pi, \mf{A}, \bar{a}}(x_1, \ldots, x_k)$ such that $V
\cup T \cup \{\psi(x_1, \ldots, x_k)\}$ is inconsistent. In other
words, $(V \cup T) \vdash \neg \psi(x_1, \ldots, x_k)$.  By
$\forall$-introduction, we have $(V \cup T) \vdash \varphi$, where
$\varphi = \forall x_1 \ldots \forall x_k \neg \psi(x_1, \ldots,
x_k)$.  Observe that $\varphi$ is a $\forall^k \exists^*$ sentence;
then by the definition of $\Gamma$, we have $\varphi \in \Gamma$, and
hence $\mf{A} \models \varphi$. Instantiating the $k$-tuple $(x_1,
\ldots, x_k)$ as $\bar{a}$, we have $(\mf{A}, \bar{a}) \models \neg
\psi(x_1, \ldots, x_k)$, contradicting the fact that $\psi(x_1,
\ldots, x_k) \in \mathsf{tp}_{\Pi, \mf{A}, \bar{a}}(x_1, \ldots,
x_k)$.  Then $Z(x_1, \ldots, x_k)$ must be consistent.  By the
downward L\"owenheim-Skolem theorem, there is a model $(\mf{B},
\bar{b})$ of $Z(x_1, \ldots, x_k)$ of power at most $\mu$; then
$\mf{B}$ is as desired.
\end{proof}

\begin{proof}[Proof of Theorem~\ref{theorem:glt(k)}]

We prove part~(\ref{theorem:glt(k)-ext}) of
Theorem~\ref{theorem:glt(k)}. Part~(\ref{theorem:glt(k)-subst}) of
Theorem~\ref{theorem:glt(k)} follows from the duality of $PSC(k)$ and
$PCE(k)$ given by Lemma~\ref{lemma:PSC(k)-PCE(k)-duality}. Also, we
prove part~(\ref{theorem:glt(k)-ext}) of Theorem~\ref{theorem:glt(k)}
for the case of sentences; the result for formulae follows from
definitions.

Suppose $\phi$ is equivalent modulo $V$ to a $\forall^k \exists^*$
sentence $\varphi$.  That $\varphi$ is $PCE(k)$ modulo $V$ follows
from the discussion towards the end of
Chapter~\ref{chapter:our-properties}. Whereby $\phi$ is $PCE(k)$
modulo $V$.

In the converse direction, suppose $\phi$ is $PCE(k)$ modulo $V$. If
$V \cup \{\phi\}$ is unsatisfiable, we are trivially done. Otherwise,
let $\Gamma$ be the set of $\forall^k \exists^*$ consequences of
$\{\phi\}$ modulo $V$. Then $(V \cup \{\phi\}) \vdash \Gamma$. We show
below that $(V \cup \Gamma) \vdash \phi$, thereby showing that $\phi$
is equivalent to $\Gamma$ modulo V. Then by the compactness theorem,
we have $\phi$ is equivalent to a finite subset of $\Gamma$ modulo
$V$.  Since a finite conjunction of $\forall^k \exists^*$ sentences is
equivalent to a single $\forall^k \exists^*$ sentence, it follows that
$\phi$ is equivalent to a $\forall^k \exists^*$ sentence, completing
the proof.

Suppose $\mf{A} \models (V \cup \Gamma)$. Consider a $\mu$-saturated
elementary extension $\mf{A}^+$ of $\mf{A}$, for some $\mu \ge \omega$
($\mf{A}^+$ exists by
Proposition~\ref{prop:saturation-results}(\ref{prop:saturation-results:exis-of-sat-elem-ext})).
Then $\mf{A}^+ \models (V \cup \Gamma)$.  By
Lemma~\ref{lemma:k-ary-covers-for-saturated-models}, there exists a
$k$-ary cover $R$ of $\mf{A}^+$ such that $\mf{B} \models (V \cup
\{\phi\})$ for every $\mf{B} \in R$. Since $\phi$ is $PCE(k)$ modulo
$V$, it follows that $\mf{A}^+ \models \phi$.  Then since $\mf{A}
\preceq \mf{A}^+$, we have $\mf{A} \models \phi$.
\end{proof}


\subsection{Proof of $\glt{k}$ using ascending chains of structures}\label{subsection:proof-of-glt-using-chains}

We first define the notion of a \emph{$k$-ary cover of a structure
  $\mf{A}$ in an elementary extension of $\mf{A}$}. This notion
generalizes the notion of $k$-ary cover seen earlier in
Definition~\ref{defn:k-ary-covered-ext} -- the latter corresponds to
the notion in Definition~\ref{defn:k-ary-cover-inside-an-elem-ext}
below, with $\mf{A}^+$ being the same as $\mf{A}$.

\begin{defn}\label{defn:k-ary-cover-inside-an-elem-ext}
Let $\mf{A}$ be a structure and $\mf{A}^+$ be an elementary extension
of $\mf{A}$. A non-empty collection $R$ of substructures of $\mf{A}^+$
is said to be a \emph{$k$-ary cover of $\mf{A}$ in $\mf{A}^+$} if for
every $k$-tuple $\bar{a}$ of elements of $\mf{A}$, there exists a
structure in $R$ containing $\bar{a}$.
\end{defn}



The following lemma is key to the proof.

\begin{lemma}\label{lemma:exis-of-k-ary-cover-in-an-elem-ext-satisfying-T}
Let $V$ and $T$ be consistent theories and $k \in \mathbb{N}$. Let
$\Gamma$ be the set of $\forall^k \exists^*$ consequences of $T$
modulo $V$. Then for every structure $\mf{A}$ that models $V$, we have
that $\mf{A} \models \Gamma$ iff there exists an elementary extension
$\mf{A}^+$ of $\mf{A}$ and a $k$-ary cover $R$ of $\mf{A}$ in
$\mf{A}^+$ such that \linebreak $\mf{B} \models (V \cup T)$ for every
$\mf{B} \in R$.
\end{lemma}

\begin{proof} 
\ul{If}: We show that $\mf{A} \models \varphi$ for each sentence
$\varphi$ of $\Gamma$. Let $\varphi = \forall^k \bar{x} \psi(\bar{x})$
for a $\Sigma^0_1$ formula $\psi(\bar{x})$.  Let $\bar{a}$ be a
$k$-tuple of $\mf{A}$. Since $R$ is a $k$-ary cover of $\mf{A}$ in
$\mf{A}^+$, there exists $\mf{B}_{\bar{a}} \in R$ such that
$\mf{B}_{\bar{a}}$ contains $\bar{a}$. Since $\mf{B}_{\bar{a}} \models
(V \cup T)$, we have $\mf{B}_{\bar{a}} \models \Gamma$. Then
$\mf{B}_{\bar{a}} \models \varphi$ and hence $(\mf{B}_{\bar{a}},
\bar{a}) \models \psi(\bar{x})$. Since $\psi(\bar{x})$ is a
$\Sigma^0_1$ formula and $\mf{B}_{\bar{a}} \subseteq \mf{A}^+$, we
have $(\mf{A}^+, \bar{a}) \models \psi(\bar{x})$, whence $(\mf{A},
\bar{a}) \models \psi(\bar{x})$ since $\mf{A} \preceq \mf{A}^+$. Since
$\bar{a}$ is arbitrary, $\mf{A} \models \varphi$.

\ul{Only If}: We have two cases here depending on whether $\mf{A}$ is
finite or infinite. Before considering these cases, we present the
following observation, call it $\dagger$. Let the vocabulary of
$\mf{A}$ be $\tau$.

\vspace{5pt} $(\dagger)$ Given an elementary extension $\mf{A}'$ of
$\mf{A}$ and a $k$-tuple $\bar{a}$ of $\mf{A}$, there exist an
elementary extension $\mf{A}''$ of $\mf{A}'$ and a substructure
$\mf{B}$ of $\mf{A}''$ such that (i) $\mf{B}$ contains $\bar{a}$ and
(ii) $\mf{B} \models (V \cup T)$.

\vspace{5pt} This is seen as follows.  Let $Z(\bar{x})$ be the theory
given by $Z(\bar{x}) = V \cup T \cup \mathsf{tp}_{\Pi, \mf{A},
  \bar{a}}(\bar{x})$. We can show that $Z(\bar{x})$ is satisfiable by
following the same argument as in the last paragraph of the proof of
Lemma~\ref{lemma:k-ary-covers-for-saturated-models}.  Whereby if
$(\mf{D}, \bar{d}) \models Z(\bar{x})$, then every existential
sentence that is true in $(\mf{D}, \bar{d})$ is also true in $(\mf{A},
\bar{a})$, and hence in $(\mf{A}', \bar{a})$. Then by
Lemma~\ref{lemma:exis-amalgam}, there is an isomorphic embedding
$f$ of $(\mf{D}, \bar{d})$ in an elementary extension $(\mf{A}'',
\bar{a})$ of $(\mf{A}, \bar{a})$. Taking $\mf{B}$ to be the
$\tau$-reduct of the image of $(\mf{D}, \bar{d})$ under $f$, we see
that $\mf{B}$ and $\mf{A}''$ are indeed as desired.

We now consider the two cases mentioned above.

(1) $\mf{A}$ is finite: Given a $k$-tuple $\bar{a}$ of $\mf{A}$, by
$(\dagger)$, there exists an elementary extension $\mf{A}''$ of
$\mf{A}$ and a substructure $\mf{B}_{\bar{a}}$ of $\mf{A}''$ such that
(i) $\mf{B}_{\bar{a}}$ contains $\bar{a}$ and (ii) $\mf{B}_{\bar{a}}
\models (V \cup T)$. Since $\mf{A}$ is finite, it follows from
Lemma~\ref{lemma:prop-of-structures}, that $\mf{A}'' =
\mf{A}$. Whereby, taking $\mf{A}^+ = \mf{A}$ and $R =
\{\mf{B}_{\bar{a}} \mid \bar{a} \in \mathsf{U}_{\mf{A}}^k\}$, we see
that $\mf{A}^+$ and $R$ are respectively indeed the desired elementary
extension of $\mf{A}$ and $k$-ary cover of $\mf{A}$ in $\mf{A}^+$.

(2) $\mf{A}$ is infinite: The proof for this case is along the lines
of the proof of the characterization of $\Pi^0_2$ sentences in terms
of the property of preservation under unions of chains (see proof of
Theorem 3.2.3 in Chapter 3 of~\cite{chang-keisler}). Let $\lambda$ be
the successor cardinal of $|\mf{A}|$ and $(\bar{a}_\kappa)_{\kappa <
  \lambda}$ be an enumeration of the $k$-tuples of $\mf{A}$. For $\eta
\leq \lambda$, given sequences $(\mf{E}_\kappa)_{\kappa < \eta}$ and
$(\mf{F}_\kappa)_{\kappa < \eta}$ of structures, we say that
$\mc{P}((\mf{E}_\kappa)_{\kappa < \eta}, (\mf{F}_\kappa)_{\kappa <
  \eta})$ is true iff $(\mf{E}_\kappa)_{\kappa < \eta}$ is an
ascending elementary chain and $\mf{A} \preceq \mf{E}_0$, and for each
$\kappa < \eta$, we have (i) $\mf{F}_\kappa \subseteq \mf{E}_\kappa$
(ii) $\mf{F}_{\kappa}$ contains $\bar{a}_{\kappa}$ and (iii)
$\mf{F}_\kappa \models (V \cup T)$. We then show the existence of
sequences $(\mf{A}_\kappa)_{\kappa < \lambda}$ and
$(\mf{B}_\kappa)_{\kappa < \lambda}$ of structures such that
$\mc{P}((\mf{A}_\kappa)_{\kappa < \lambda}, (\mf{B}_\kappa)_{\kappa <
  \lambda})$ is true. Then by
Theorem~\ref{theorem:elem-chain-theorem}, taking $\mf{A}^+ =
\bigcup_{\kappa < \lambda} \mf{A}_\kappa$ and $R = \{\mf{B}_\kappa
\mid \kappa < \lambda\}$, we see that $\mf{A}^+$ and $R$ are
respectively indeed the desired elementary extension of $\mf{A}$ and
$k$-ary cover of $\mf{A}$ in $\mf{A}^+$.

We construct the sequences $(\mf{A}_\kappa)_{\kappa < \lambda}$ and
$(\mf{B}_\kappa)_{\kappa < \lambda}$ by constructing for each $\eta
\leq \lambda$, the partial (initial) sequences
$(\mf{A}_\kappa)_{\kappa < \eta}$ and $(\mf{B}_\kappa)_{\kappa <
  \eta}$ and showing that $\mc{P}((\mf{A}_\kappa)_{\kappa < \eta},
(\mf{B}_\kappa)_{\kappa < \eta})$ is true. We do this by (transfinite)
induction on $\eta$.  For the base case of $\eta = 1$, we see by
$(\dagger)$ above that if $\mf{A}' = \mf{A}$, then there exists an
elementary extension $\mf{A}''$ of $\mf{A}$ and a substructure
$\mf{B}$ of $\mf{A}''$ such that (i) $\mf{B}$ contains $\bar{a}_0$ and
(ii) $\mf{B} \models (V \cup T)$. Then taking $\mf{A}_0 = \mf{A}''$
and $\mf{B}_0 = \mf{B}$, we see that $\mc{P}((\mf{A}_0), (\mf{B}_0))$
is true.  As the induction hypothesis, assume that we have constructed
sequences $(\mf{A}_\kappa)_{\kappa < \eta}$ and
$(\mf{B}_\kappa)_{\kappa < \eta}$ such that
$\mc{P}((\mf{A}_\kappa)_{\kappa < \eta}, (\mf{B}_\kappa)_{\kappa <
  \eta})$ is true.  Then by Theorem~\ref{theorem:elem-chain-theorem},
the structure $\mf{A}' = \bigcup_{\kappa < \eta} \mf{A}_\kappa$ is
such that $\mf{A} \preceq \mf{A}'$. Then for the tuple $\bar{a}_\eta$
of $\mf{A}$, by $(\dagger)$, there exists an elementary extension
$\mf{C}$ of $\mf{A}'$ and a substructure $\mf{D}$ of $\mf{C}$ such
that (i) $\mf{D}$ contains $\bar{a}_\eta$ and (ii) $\mf{D} \models (V
\cup T)$.  Then taking $\mf{A}_\eta = \mf{C}$ and $\mf{B}_\eta =
\mf{D}$, and letting $\mu$ be the successor ordinal of $\eta$, we see
that $\mc{P}((\mf{A}_\kappa)_{\kappa < \mu}, (\mf{B}_\kappa)_{\kappa <
  \mu})$ is indeed true, completing the induction.
\end{proof}

\begin{proof}[Proof of Theorem~\ref{theorem:glt(k)}]
We prove part~(\ref{theorem:glt(k)-ext}) of
Theorem~\ref{theorem:glt(k)}. Part~(\ref{theorem:glt(k)-subst}) of
Theorem~\ref{theorem:glt(k)} follows from the duality of $PSC(k)$ and
$PCE(k)$ given by Lemma~\ref{lemma:PSC(k)-PCE(k)-duality}.  We prove
part~(\ref{theorem:glt(k)-ext}) of Theorem~\ref{theorem:glt(k)} for
the case of sentences; the result for formulae follows. The `If'
direction of part~(\ref{theorem:glt(k)-ext}) of
Theorem~\ref{theorem:glt(k)} is proved exactly as the proof of this
part of Theorem~\ref{theorem:glt(k)}, as presented in the Section
\ref{subsection:proof-of-glt-using-saturated-structures}. We hence
prove the `Only if' direction below.

Suppose $\phi$ is $PCE(k)$ modulo $V$. If $V \cup \{\phi\}$ is
unsatisfiable, we are trivially done. Otherwise, let $\Gamma$ be the
set of $\forall^k \exists^*$ consequences of $\{\phi\}$ modulo
$V$. Then $(V \cup \{\phi\}) \vdash \Gamma$. We show below that $(V
\cup \Gamma) \vdash \phi$, thereby showing that $\phi$ is equivalent
to $\Gamma$ modulo V. Suppose $\mf{A} \models (V \cup \Gamma)$.
Consider the sequence $(\mf{A}_i)_{i \ge 0}$ of structures and the
sequence $(R_i)_{i \ge 0}$ of collections of structures with the
following properties.
\begin{enumerate}[nosep]
\item $(\mf{A}_i)_{i \ge 0}$ is an ascending elementary chain such
  that $\mf{A} \preceq \mf{A}_0$ (whereby $\mf{A}_i \models (V \cup
  \Gamma)$ for each $i \ge 0$) and for each $i \ge 0$, $\mf{A}_{i+1}$
  is the elementary extension of $\mf{A}_i$ as given by
  Lemma~\ref{lemma:exis-of-k-ary-cover-in-an-elem-ext-satisfying-T}.
\item For each $i \ge 0$, $R_i$ is the $k$-ary cover of $\mf{A}_i$ in
  $\mf{A}_{i+1}$ as given by
  Lemma~\ref{lemma:exis-of-k-ary-cover-in-an-elem-ext-satisfying-T}.
\end{enumerate}

Consider the structure $\mf{A}^+ = \bigcup_{i \ge 0}
\mf{A}_i$. Consider any $k$-tuple $\bar{a}$ of $\mf{A}^+$; it is clear
that there must exist $j \ge 0$ such $\bar{a}$ is contained in
$\mf{A}_j$. Then there exists a structure $\mf{B}_{\bar{a}} \in R_j$
such that (i) $\mf{B}_{\bar{a}}$ contains $\bar{a}$ and (ii)
$\mf{B}_{\bar{a}} \models (V \cup \{\phi\})$. Since $\mf{B}_{\bar{a}}
\in R_j$, we have $\mf{B}_{\bar{a}} \subseteq \mf{A}_{j+1}$ and since
$\mf{A}_{j+1} \preceq \mf{A}^+$ (by
Theorem~\ref{theorem:elem-chain-theorem}), we have $\mf{B}_{\bar{a}}
\subseteq \mf{A}^+$. Then $R = \{\mf{B}_{\bar{a}} \mid
\bar{a}~\text{is a}~k\text{-tuple from}~\mf{A}^+\}$ is a $k$-ary cover
of $\mf{A}^+$ (or equivalently, a $k$-ary cover of $\mf{A}^+$ in
$\mf{A}^+$) such that $\mf{B} \models (V \cup \{\phi\})$ for each
$\mf{B} \in R$. Since $\phi$ is $PCE(k)$ modulo $V$, it follows that
$\mf{A}^+ \models \phi$. Then since $\mf{A} \preceq \mf{A}^+$, we have
that $\mf{A} \models \phi$, completing the proof.
\end{proof}




%% file: variants-of-our-properties.tex
\section{Variants of our properties and their characterizations}\label{section:lambda-cruxes-and-lambda-ary-covers}

In this section, we present natural generalizations of the $PSC(k)$
and $PCE(k)$ properties in which, rather than insisting on bounded
sized cruxes and bounded arity covers, we allow cruxes of sizes, and
covers of arities, less than $\lambda$, where $\lambda$ is an infinite
cardinal.  We first define the notion of \emph{$\lambda$-ary covered
  extensions}.

\begin{defn}\label{defn:lambda-ary-covered-extension}
Given an infinite cardinal $\lambda$, a structure $\mf{A}$ is called a
\emph{$\lambda$-ary covered extension} of a collection $R$ of
structures if (i) $\mf{A}$ is an extension of $R$ (ii) for each subset
$C$ of the universe of $\mf{A}$, \emph{of size less than $\lambda$},
there is a structure in $R$ containing $C$. We call $R$ a
\emph{$\lambda$-ary cover} of $\mf{A}$.
\end{defn}


Observe that in the definition above, $\mf{A}$ must be unique such
since all relation symbols and function symbols have finite arity.

\begin{defn}\label{defn:PSC(lambda)-and-PCE(lambda)}
Let $\cl{S}$ be a class of structures and $\cl{U}$ be a subclass of
$\cl{S}$.
\begin{enumerate}[nosep]
\item We say $\cl{U}$ is \emph{preserved under substructures modulo
  $\lambda$-cruxes over $\cl{S}$},
\sindex[term]{preservation under!substructures modulo
  $\lambda$-cruxes} abbreviated \emph{$\cl{U}$ is
  $PSC(\lambda)$\sindex[symb]{$PSC(\lambda), PCE(\lambda)$} over
  $\cl{S}$}, if for each structure $\mf{A} \in \cl{U}$, there is a
subset $C$ of the universe of $\mf{A}$, of size less than $\lambda$,
such that, if $\mf{B} \subseteq \mf{A}$, $\mf{B}$ contains $C$ and
$\mf{B} \in \cl{S}$, then $\mf{B} \in \cl{U}$. The set $C$ is called
an \emph{$\lambda$-crux} of $\mf{A}$ w.r.t. $\cl{U}$ over $\cl{S}$.
\item We say $\cl{U}$ is \emph{preserved under $\lambda$-ary covered
  extensions over $\cl{S}$}, 
\sindex[term]{preservation under!$\lambda$-ary covered extensions}
abbreviated \emph{$\cl{U}$ is
  $PCE(\lambda)$\sindex[symb]{$PSC(\lambda), PCE(\lambda)$} over
  $\cl{S}$}, if for every collection $R$ of structures of $\cl{U}$, if
$\mf{A}$ is an $\lambda$-ary covered extension of $R$ and $\mf{A} \in
\cl{S}$, then $\mf{A} \in \cl{U}$.
\end{enumerate}
\end{defn}

It is easy to see that given classes $\cl{U}$ and $\cl{S}$, and
infinite cardinals $\lambda$ and $\mu$ such that $\lambda \leq \mu$,
if $\cl{U}$ is $PSC(\lambda)$ (resp. $PCE(\lambda)$) over $\cl{S}$,
then $\cl{U}$ is $PSC(\mu)$ (resp. $PCE(\mu)$) over $\cl{S}$.

If $\phi(\bar{x})$ and $T(\bar{x})$ are respectively a formula and a
theory with free variables $\bar{x}$, then given a theory $V$, the
notions of `$\phi(\bar{x})$ is $PSC(\lambda)$ ($PCE(\lambda)$) modulo
$V$' and `$T(\bar{x})$ is $PSC(\lambda)$ ($PCE(\lambda)$) modulo $V$'
are defined similarly as the corresponding notions for $PSC(k)$ and
$PCE(k)$.

Analogous to Lemma~\ref{lemma:PSC(k)-PCE(k)-duality},
Lemma~\ref{lemma:k-ary-covers-for-saturated-models},
Lemma~\ref{lemma:exis-of-k-ary-cover-in-an-elem-ext-satisfying-T} and
Theorem~\ref{theorem:glt(k)}, we have the following results for
$PSC(\lambda)$ and $PCE(\lambda)$. The proofs are similar to the
corresponding results for $PSC(k)$ and $PCE(k)$ and are hence skipped.

\begin{lemma}[$PSC(\lambda)$-$PCE(\lambda)$ duality]\label{lemma:PSC(lambda)-PCE(lambda)-duality}
Let $\cl{S}$ be a class of structures, $\cl{U}$ be a subclass of
$\cl{S}$ and $\overline{\cl{U}}$ be the complement of $\cl{U}$ in
$\cl{S}$. Then $\cl{U}$ is $PSC(\lambda)$ over $\cl{S}$ iff
$\overline{\cl{U}}$ is $PCE(\lambda)$ over $\cl{S}$. In particular, if
$\,\cl{S}$ is defined by a theory $V$, then a sentence $\phi$ is
$PSC(\lambda)$ modulo $V$ iff $\neg \phi$ is $PCE(\lambda)$ modulo
$V$.
\end{lemma}

\begin{lemma}\label{lemma:lambda-ary-covers-for-saturated-models}
Let $V$ and $T$ be consistent theories, and let $\Gamma$ be the set of
$\Pi^0_2$ consequences of $T$ modulo $V$.  Then for all infinite
cardinals $\lambda$ and $\mu$, and for every $\mu$-saturated structure
$\mf{A}$ that models $V$, we have that $\mf{A} \models \Gamma$ iff
there exists a $\lambda$-ary cover $R$ of $\mf{A}$ such that $\mf{B}
\models (V \cup T)$ for every $\mf{B} \in R$.
\end{lemma}

\begin{lemma}\label{lemma:exis-of-lambda-ary-cover-in-an-elem-ext-satisfying-T}
Let $V$ and $T$ be consistent theories, and let $\Gamma$ be the set of
$\Pi^0_2$ consequences of $T$ modulo $V$. Then for all infinite
cardinals $\lambda$, and for every structure $\mf{A}$ that models $V$,
we have that $\mf{A} \models \Gamma$ iff there exists an elementary
extension $\mf{A}^+$ of $\mf{A}$ and a $\lambda$-ary cover $R$ of
$\mf{A}$ in $\mf{A}^+$ such that $\mf{B} \models (V \cup T)$ for every
$\mf{B} \in R$.
\end{lemma}


\begin{theorem}\label{theorem:PCE(lambda)-and-PSC(lambda)-characterizations}
Given a theory $V$, the following hold for each infinite cardinal
$\lambda$.
\begin{enumerate}[nosep]
\item A formula $\phi(\bar{x})$ is $PSC(\lambda)$ modulo $V$ iff
  $\phi(\bar{x})$ is equivalent modulo $V$ to a $\Sigma^0_2$ formula
  having free variables $\bar{x}$.\label{theorem:PSC(lambda)-char}
\item A formula $\phi(\bar{x})$ is $PCE(\lambda)$ modulo $V$ iff
  $\phi(\bar{x})$ is equivalent modulo $V$ to a $\Pi^0_2$ formula
  having free variables $\bar{x}$.\label{theorem:PCE(lambda)-char}
\end{enumerate}
\end{theorem}

The above theorem implies the following result that is not obvious
from the definitions of the properties concerned.

\begin{corollary}\label{corollary:PSC(lambda)=PSC-and-PCE(lambda)=PCE}
For every infinite cardinal $\lambda$, a sentence is $PSC(\lambda)$
(resp. $PCE(\lambda)$) modulo a theory $V$ iff it is $PSC$
(resp. $PCE$) modulo $V$.
\end{corollary}

The above characterizations, along with the characterizations in
Section~\ref{section:glt-for-sentences}, are depicted pictorially
below.\\

\begin{figure}[H]
\centering 
\includegraphics[scale=0.6]{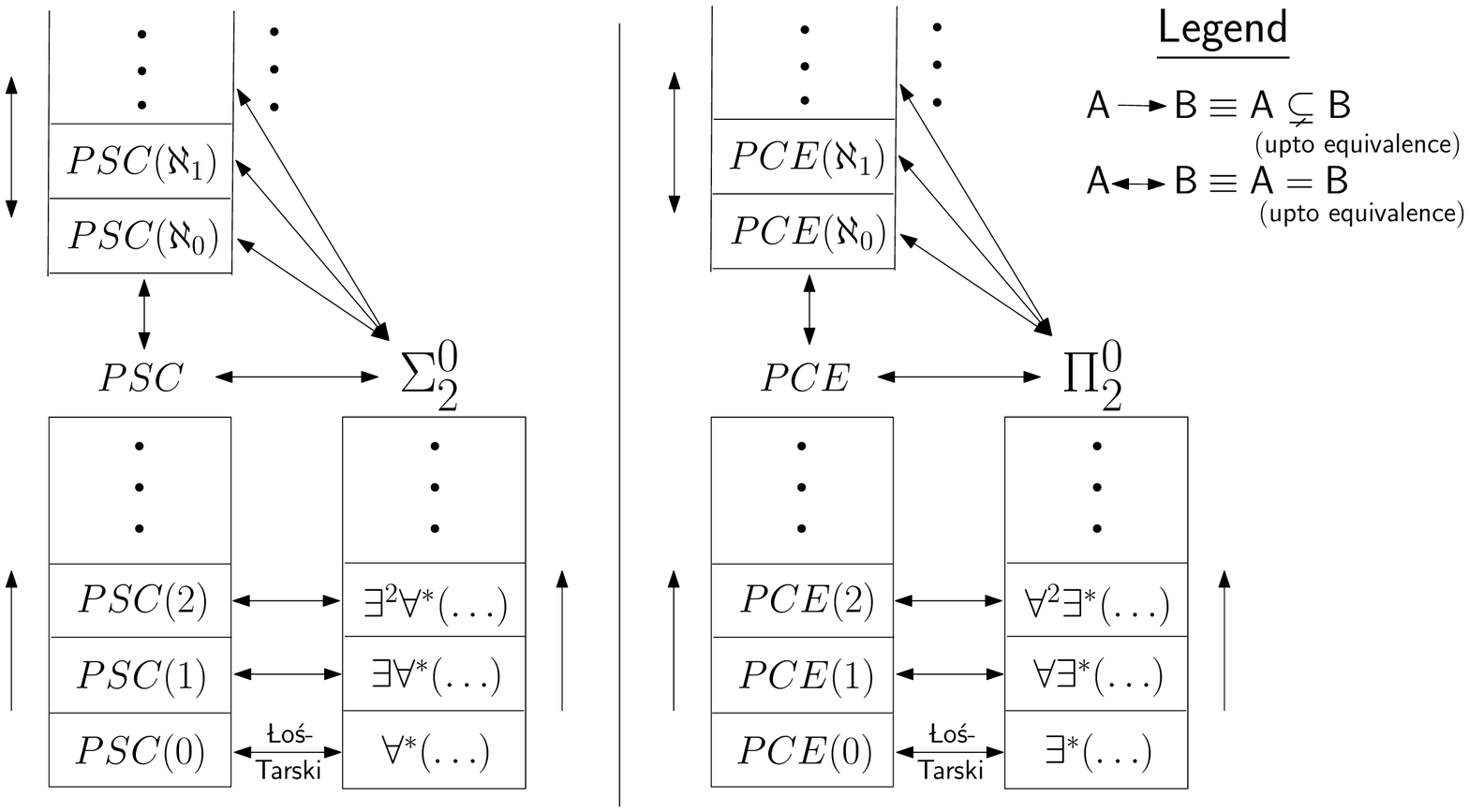}
\caption{Characterizations of $PSC(k), PCE(k), PSC(\lambda)$ and $PCE(\lambda)$ sentences}
\end{figure}

\input{inexp}

%% file: inexp.tex
\subsection{Applications: new proofs of inexpressibility results in FO}\label{subsection:inexp}

Typical proofs of inexpressibility results in FO are either via
compactness theorem, or Ehrenfeucht-Fr\"aiss\'e games or locality
arguments. We present a new approach to proving inexpressibility
results, using our results. We illustrate this approach via the
example presented below. Below, the underlying class $\cl{S}$ of
graphs is the class of all undirected graphs.

Consider the subclass $\cl{U}$ of $\cl{S}$ consisting of graphs that
contain a cycle as a subgraph.  It is easy to see that in any graph $G
\in \cl{U}$, the vertices of any cycle form an $\aleph_0$-crux of
$G$. Then $\cl{U}$ is $PSC(\aleph_0)$. If $\cl{U}$ were definable by
an FO sentence, say $\varphi$, then $\varphi$ is $PSC(\aleph_0)$. By
Corollary \ref{corollary:PSC(lambda)=PSC-and-PCE(lambda)=PCE}, it
follows that $\varphi$ is $PSC(k)$ for some $k \in \mathbb{N}$. Now
consider the cycle graph $G$ of length $k+1$; clearly $G$ models
$\varphi$. No proper induced subgraph of $G$ is a cycle, whence $G$
contains no $k$-crux at all. This contradicts the earlier inference
that $\varphi$ is $PSC(k)$. Thus $\cl{U}$ is not definable by any FO
sentence.

A short report containing more examples of inexpressibility results
proven using our preservation theorems can be found
at~\cite{inexp-TR}. These examples include connectedness,
bipartiteness, caterpillars, etc.  Note that the notion of `core'
in~\cite{inexp-TR} is exactly what we mean by a `crux' in this thesis.

%% file: uncomputability.tex
\section{An uncomputability result}\label{section:uncomputability}

Corollary~\ref{corollary:PSC(lambda)=PSC-and-PCE(lambda)=PCE} tells
that given a sentence $\phi$ that is $PSC(\lambda)$
(resp.\ $PCE(\lambda)$) modulo a theory $V$, there exists $k \in
\mathbb{N}$ such that $\phi$ is $PSC(k)$ (resp.\ $PCE(k)$) modulo $V$.
This raises the question: is $k$ computable? The following proposition
answers the aforesaid question in the negative for $PSC(\aleph_0)$
(resp.\ $PCE(\aleph_0)$), and hence for $PSC(\lambda)$
(resp. $PCE(\lambda)$) for each (infinite cardinal) $\lambda$. Below, a
\emph{relational} sentence is a sentence over a vocabulary that does
not contain any function symbols. Let the length of a sentence $\phi$
be denoted by $|\phi|$.

\begin{proposition}\label{prop:non-recursiveness-for-lambda=bounded}
Let $V$ be the empty theory. For every recursive function $\nu:
\mathbb{N} \rightarrow \mathbb{N}$,  the following are true:
\begin{enumerate}[nosep]
\item There is a relational $\Pi^0_2$ sentence $\phi$ that is
  $PSC(\aleph_0)$ modulo $V$ but that is not $PSC(k)$ modulo $V$ for
  any $k \leq
  \nu(|\phi|)$.\label{lemma:non-recursiveness-for-lambda=bounded-part-1}
\item There is a relational $\Sigma^0_2$ sentence $\phi$ that is
  $PCE(\aleph_0)$ modulo $V$ but that is not $PCE(k)$ modulo $V$ for
  any $k \leq
  \nu(|\phi|)$.\label{lemma:non-recursiveness-for-lambda=bounded-part-2}
\end{enumerate}
\end{proposition}

Towards the proof of the above proposition, we first present a recent
unpublished result of Rossman~\cite{rossman-los-tarski}.

\begin{theorem}[Rossman, 2012]\label{theorem:rossman-los-tarski}
Let $V$ be the empty theory. For every recursive function $\nu:
\mathbb{N} \rightarrow \mathbb{N}$, there exists a relational
$\Sigma^0_2$ sentence $\phi$ that is $PS$ modulo $V$, and for which
every equivalent $\Pi^0_1$ sentence has length at least $\nu(|\phi|) +
1$.
\end{theorem}
Theorem~\ref{theorem:rossman-los-tarski} gives a non-recursive lower
bound on the length of $\Pi^0_1$ sentences equivalent to sentences
that are $PS$ (in terms of the lengths of the latter sentences). This
strengthens the non-elementary lower bound proved
in~\cite{dawar-model-theory-large}.

\begin{corollary}\label{corollary:los-tarski-no-of-foralls}
Let $V$ be the empty theory. For every recursive function $\nu:
\mathbb{N} \rightarrow \mathbb{N}$, there exists a relational
$\Sigma^0_2$ sentence $\phi$ that is $PS$ modulo $V$, and for which
every equivalent $\Pi^0_1$ sentence has at least $\nu(|\phi|) + 1$
universal variables.
\end{corollary}
\begin{proof}

We show below that there is a monotone recursive function $\rho:
\mathbb{N} \rightarrow \mathbb{N}$ such that if $\xi$ is a $\Pi^0_1$
sentence with $n$ variables, then the shortest (in terms of length)
$\Pi^0_1$ sentence equivalent to $\xi$ has length at most
$\rho(n)$. That would prove this corollary as follows.  Suppose there
is a recursive function $\nu: \mathbb{N} \rightarrow \mathbb{N}$ such
that for each relational $\Sigma^0_2$ sentence $\psi$ that is $PS$
modulo $V$, there is an equivalent $\Pi^0_1$ sentence having at most
$\nu(|\psi|)$ universal variables. Then consider the recursive
function $\theta:\mathbb{N} \rightarrow \mathbb{N}$ given by
$\theta(n) = \rho(\nu(n))$ and let $\phi$ be the relational
$\Sigma^0_2$ sentence given by Theorem
\ref{theorem:rossman-los-tarski} for the function $\theta$. Then
$\phi$ is $PS$ modulo $V$ and the shortest $\Pi^0_1$ sentence
equivalent to $\phi$ has length $> \theta(|\phi|)$. By the assumption
about $\nu$ above, there is a $\Pi^0_1$ sentence equivalent to $\phi$
having at most $\nu(|\phi|)$ universal variables. Whence there is a
$\Pi^0_1$ sentence equivalent to $\phi$ whose length is at most
$\rho(\nu|\phi|) = \theta(|\phi|)$ -- a contradiction.

Let $\xi$ be a universal sentence given by $\xi = \forall^n \bar{z}
\beta(\bar{z})$. Let the vocabulary of $\xi$ be $\tau$ and the maximum
arity of any predicate of $\tau$ be $q$. Then the number $k$ of atomic
formulae of $\tau$ having variables from $\bar{z}$ is at most $|\tau|
\cdot n^q$. It follows that the length $r$ of the disjunctive normal
form, say $\alpha$, of $\beta$ satisfies $r \leq (d \cdot k \cdot
2^k)$ for some constant $d \ge 1$.  Then $\xi$ is equivalent to the
sentence $\gamma = \forall^n \bar{z} \alpha(\bar{z})$; the size of
$\gamma$ is at most $e \cdot (n + r)$ for some constant $e \ge
1$. Since $k$ and $r$ are bounded by monotone recursive functions of
$n$, so is the length of $\gamma$.
\end{proof}




\begin{proof}[Proof of Proposition~\ref{prop:non-recursiveness-for-lambda=bounded}]

We give the proof for part
(\ref{lemma:non-recursiveness-for-lambda=bounded-part-1}). The
negation of the sentence $\phi$ showing part
(\ref{lemma:non-recursiveness-for-lambda=bounded-part-1}) proves part
(\ref{lemma:non-recursiveness-for-lambda=bounded-part-2}). Also, we
omit the mention of $V$ for the sake of readability.

Suppose there is a recursive function $\nu: \mathbb{N} \rightarrow
\mathbb{N}$ such that if $\xi$ is a relational $\Pi^0_2$ sentence that
is $PSC(\aleph_0)$, then $\xi$ is $PSC(k)$ for some $k \leq
\nu(|\xi|)$. In other words, for $\xi$ as mentioned, every model of
$\xi$ has a crux of size at most $\nu(|\xi|)$. Consider the recursive
function $\rho: \mathbb{N} \rightarrow \mathbb{N}$ given by $\rho(n) =
\nu(n+1)$. Then, for the function $\rho$, consider the relational
$\Sigma^0_2$ sentence $\phi$ given by Corollary
\ref{corollary:los-tarski-no-of-foralls}. The sentence $\phi$ is $PS$
and every $\Pi^0_1$ sentence equivalent to it has $> \rho(|\phi|)$
number of universal variables. Now the $\Pi^0_2$ sentence $\psi$ given
by $\psi = \neg \phi$ is equivalent to a $\Sigma^0_1$ sentence. Since
$\Sigma^0_1$ sentences are $PSC$, and hence $PSC(\aleph_0)$, it
follows that $\psi$ is $PSC(\aleph_0)$. Now, by our assumption about
$\nu$ above, every model of $\psi$ has a crux of size at most
$\nu(|\psi|) = \nu(|\phi| + 1) = \rho(|\phi|)$. Then all minimal
models of $\psi$ have size at most $\rho(|\phi|) + q$, where $q$ is
the number of constant symbols in the vocabulary of $\phi$. Using the
fact that $\psi$ is preserved under extensions, it is easy to
construct a $\Sigma^0_1$ sentence having $\rho(|\phi|)$ number of
existential variables, that is equivalent to $\psi$. Whereby $\phi$ is
equivalent to a $\Pi^0_1$ sentence having $\rho(|\phi|)$ number of
universal variables -- a contradiction.
\end{proof}

%% file: the-case-of-theories.tex
\chapter{Characterizations: the case of theories}\label{chapter:the-case-of-theories}

\input{char-of-ext-prop-for-theories}
\input{char-of-subst-prop-for-theories}

%% file: char-of-ext-prop-for-theories.tex
\section{Characterizations of the extensional properties}\label{section:ext-char-for-theories}

The central result of this section is as below.

\begin{theorem}\label{theorem:ext-chars}
Given a theory $V$, the following hold for each $k \in \mathbb{N}$ and
each $\lambda \ge \aleph_0$:
\begin{enumerate}[nosep]
\item A theory $T(\bar{x})$ is $PCE(k)$ modulo $V$ iff $T(\bar{x})$ is
  equivalent modulo $V$ to a theory of $\Pi^0_2$ formulae, all of
  whose free variables are among $\bar{x}$ and all of which have $k$
  universal quantifiers.\label{theorem:char-of-PCE(k)-theories}
\item A theory $T(\bar{x})$ is $PCE(\lambda)$ modulo $V$ iff
  $T(\bar{x})$ is equivalent modulo $V$ to a theory of $\Pi^0_2$
  formulae, all of whose free variables are among
  $\bar{x}$.\label{theorem:char-of-PCE(lambda)-theories}
\end{enumerate}
\end{theorem}

The proofs of part (\ref{theorem:char-of-PCE(k)-theories}) and part
(\ref{theorem:char-of-PCE(lambda)-theories}) of the above result are
respectively, nearly identical to the proofs of
Theorem~\ref{theorem:glt(k)}(\ref{theorem:glt(k)-ext}) and
Theorem~\ref{theorem:PCE(lambda)-and-PSC(lambda)-characterizations}(\ref{theorem:PCE(lambda)-char})
-- we just consider theories instead of sentences in the latter proofs
and use the following lemma that is straightforward.
\begin{lemma}\label{lemma:PCE(k)-is-closed-under-intersections}
Let $\cl{S}$ be a class of structures, $k$ a natural number and
$\lambda$ an infinite cardinal. For an index set $I$, let $\{\cl{U}_i
\mid i \in I\}$ be a collection of subclasses of $\cl{S}$ such that
$\cl{U}_i$ is $PCE(k)$, resp. $PCE(\lambda)$, over $\cl{S}$, for each
$i \in I$. Then $\bigcap_{i \in I} \cl{U}_i$ is $PCE(k)$,
resp. $PCE(\lambda)$, over $\cl{S}$.
\end{lemma}

\begin{remark} 
By considering singleton theories in Theorem~\ref{theorem:ext-chars},
and using compactness theorem and the fact that a finite conjunction
of $\forall^k \exists^*$ sentences, respectively $\Pi^0_2$ sentences,
is also a $\forall^k \exists^*$ sentence, respectively a $\Pi^0_2$
sentence, we get
Theorem~\ref{theorem:glt(k)}(\ref{theorem:glt(k)-ext}) and
Theorem~\ref{theorem:PCE(lambda)-and-PSC(lambda)-characterizations}(\ref{theorem:PCE(lambda)-char}).
\end{remark}

The following proposition reveals an important difference between
considering the properties of $PCE(k)$ and $PCE(\lambda)$ in the
context of theories, vis-\'a-vis considering these properties in the
context of sentences. Specifically, in contrast to
Corollary~\ref{corollary:PSC(lambda)=PSC-and-PCE(lambda)=PCE}, it
turns out that $PCE(\lambda)$ theories are more general than $PCE$
theories.

\begin{proposition}\label{prop:PCE(lambdas)-and-PCE}
Let $\lambda$ be an infinite cardinal.
\begin{enumerate}[nosep]
\item A theory is $PCE(\lambda)$ modulo a theory $V$ iff it is
  $PCE(\aleph_0)$ modulo $V$.\label{prop:PCE(lambdas)-are-all-the-same}
\item There are theories $T$ and $V$ such that $T$ is $PCE(\aleph_0)$
  modulo $V$, and hence $PCE(\lambda)$ modulo $V$, but $T$ is not $PCE$
  modulo $V$.\label{prop:PCE(lambda)-more-general-than-PCE}
\end{enumerate}
\end{proposition}
\begin{proof}

Part (\ref{prop:PCE(lambdas)-are-all-the-same}) follows easily from
Theorem~\ref{theorem:ext-chars}(\ref{theorem:char-of-PCE(lambda)-theories}). We
prove part (\ref{prop:PCE(lambda)-more-general-than-PCE}) below.

Let $V$ be the theory defining the class of all undirected graphs.
Let $T$ be a $\Pi^0_1$ theory over graphs asserting that there is no
cycle of length $k$ for any $k \in \mathbb{N}$. Then $T$ defines the
class $\cl{U}$ of all acyclic graphs, and is $PCE(\aleph_0)$ modulo
$V$ by
Theorem~\ref{theorem:ext-chars}(\ref{theorem:char-of-PCE(lambda)-theories}).
Suppose $T$ is $PCE$ modulo $V$, whence $T$ is $PCE(k)$ modulo $V$ for
some $k \in \mathbb{N}$. Then $\cl{U}$ is $PCE(k)$ modulo the class of
models of $V$.  By Lemma~\ref{lemma:PSC(k)-PCE(k)-duality},
$\overline{\cl{U}}$ (the complement of $\cl{U}$) is $PSC(k)$ modulo
the class of models of $V$.  Now consider a cycle $G$ of length $k+1$.
Clearly, $G$ is in $\overline{\cl{U}}$ but every proper substructure
of $G$ is in $\cl{U}$. This contradicts our earlier inference that
$\overline{\cl{U}}$ is $PSC(k)$ modulo the class of models of $V$.
\end{proof}

The characterizations of this section are depicted pictorially below.

\begin{figure}[H]
\centering 
\includegraphics[scale=0.7]{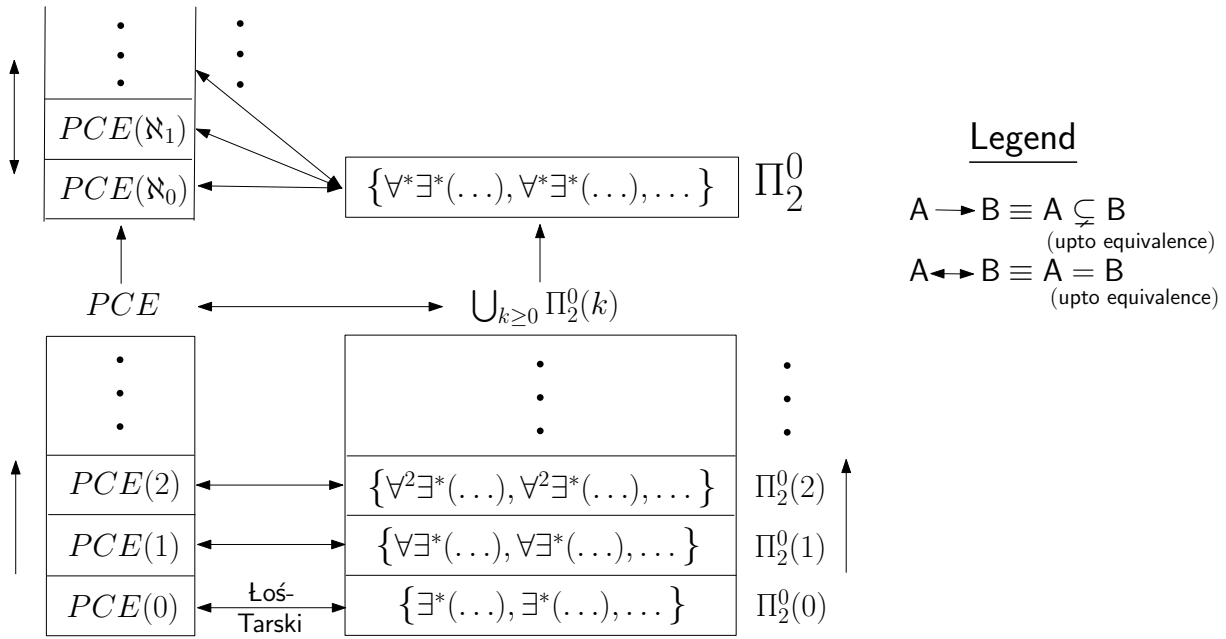}
\caption{Characterizations of $PCE(k)$ and $PCE(\lambda)$ theories}
\end{figure}

%% file: char-of-subst-prop-for-theories.tex
\section{Characterizations of the substructural properties}\label{section:subst-char-for-theories}
The central results of this section are as follows.

\begin{theorem}\label{theorem:subst-char-for-theories}
Let $V$ be a given theory, $k \in \mathbb{N}$ and $\lambda >
\aleph_0$.
\begin{enumerate}[nosep]
\item A theory $T(\bar{x})$ is $PSC(\lambda)$ modulo $V$ iff
  $T(\bar{x})$ is equivalent modulo $V$ to a theory of $\Sigma^0_2$
  formulae, all of whose free variables are among
  $\bar{x}$.\label{theorem:PSC(lambda)-char-for-theories}
\item If a theory $T(\bar{x})$ is $PSC(\aleph_0)$ modulo $V$, then
  $T(\bar{x})$ is equivalent modulo $V$ to a theory of $\Sigma^0_2$
  formulae, all of whose free variables are among $\bar{x}$. The same
  consequent (therefore) holds if $T(\bar{x})$ is $PSC(k)$ modulo $V$.
  The converses of these implications are not true. There exist
  theories $T$ and $V$ such that (i) each sentence of $T$ is a
  $\Sigma^0_2$ sentence having exactly one existential quantifier, and
  (ii) $T$ is not $PSC(\aleph_0)$ modulo $V$, and hence not $PSC(k)$
  modulo $V$.\label{theorem:PSC-and-PSC(aleph_0)-for-theories}
\end{enumerate}
\end{theorem}

Since Theorem~\ref{theorem:glt(k)}(\ref{theorem:glt(k)-subst}) shows
that a $PSC(k)$ sentence is always equivalent to an $\exists^k
\forall^*$ sentence, it is natural to ask if a $PSC(k)$ theory is
always equivalent to a theory of $\exists^k \forall^*$ sentences. We
give an affirmative answer to this question, conditioned on a
hypothesis that we present below, and thereby provide a (conditional)
refinement of
Theorem~\ref{theorem:subst-char-for-theories}(\ref{theorem:PSC-and-PSC(aleph_0)-for-theories}).

\begin{hypothesis}\label{hypothesis:H}
Given theories $V$ and $T(\bar{x})$, and $k \in \mathbb{N}$ such that
$T(\bar{x})$ is $PSC(k)$ modulo $V$, it is the case that for each
model $(\mf{A}, \bar{a})$ of $T(\bar{x})$, there exists a $k$-crux
$\bar{b}$ of $(\mf{A}, \bar{a})$ (w.r.t. $T(\bar{x})$ modulo $V$) such
that $\bar{b}$ is also a $k$-crux (w.r.t. $T(\bar{x})$ modulo $V$) of
a $\mu$-saturated elementary extension of $(\mf{A}, \bar{a})$, for
some $\mu \ge \omega$.
\end{hypothesis}


\begin{theorem}\label{theorem:conditional-refinement}
Given theories $V$ and $T(\bar{x})$, suppose $T(\bar{x})$ is $PSC(k)$
modulo $V$ for a given $k \in \mathbb{N}$. Then assuming
Hypothesis~\ref{hypothesis:H}, $T(\bar{x})$ is equivalent modulo $V$
to a theory of $\Sigma^0_2$ formulae, all of whose free variables are
among $\bar{x}$, and all of which have $k$ existential quantifiers.
\end{theorem}

The approach of `dualizing' adopted in proving
Theorem~\ref{theorem:glt(k)}(\ref{theorem:glt(k)-subst}) cannot work
for characterizing theories that are $PSC(k)$ or $PSC(\lambda)$ since
the negation of an FO theory might, in general, not be equivalent to
any FO theory. We therefore present in this section, altogether
different approaches to proving the above results. While we show that
$\Sigma^0_2$ theories characterize $PSC(\lambda)$ theories for
$\lambda > \aleph_0$, it is unclear at present what syntactic
fragments of FO theories serve to characterize $PSC(k)$ and
$PSC(\aleph_0)$ theories. However,
Theorem~\ref{theorem:subst-char-for-theories}(\ref{theorem:PSC-and-PSC(aleph_0)-for-theories})
shows that these syntactic fragments must be semantically contained
inside the class of $\Sigma^0_2$ theories, and
Theorem~\ref{theorem:conditional-refinement} shows that under
Hypothesis~\ref{hypothesis:H}, any syntactic fragment that
characterizes $PSC(k)$ theories must be semantically contained inside
the class of theories of $\exists^k \forall^*$ sentences.

The remainder of this section is entirely devoted to proving the
results above. In the next two sections, we present the proofs of
Theorem~\ref{theorem:subst-char-for-theories} and
Theorem~\ref{theorem:conditional-refinement}, and also show that
Hypothesis~\ref{hypothesis:H} is indeed well-motivated.

\input{char-of-PSC-lambda-theories}

\input{conditional-char-3.0}

%% file: char-of-PSC-lambda-theories.tex
\subsection*{Proof of Theorem~\ref{theorem:subst-char-for-theories}}\label{subsection:char-of-PSC-lambda-theories}

We first present the proof of part
(\ref{theorem:PSC-and-PSC(aleph_0)-for-theories}) of
Theorem~\ref{theorem:subst-char-for-theories}, assuming part
(\ref{theorem:PSC(lambda)-char-for-theories}) of
Theorem~\ref{theorem:subst-char-for-theories}.

\begin{proof}[Proof of
Theorem~\ref{theorem:subst-char-for-theories}(\ref{theorem:PSC-and-PSC(aleph_0)-for-theories})]
Since a theory $T(\bar{x})$ that is $PSC(k)$ modulo $V$ or
$PSC(\aleph_0)$ modulo $V$ is also $PSC(\lambda)$ modulo $V$ for
$\lambda > \aleph_0$, it follows from
Theorem~\ref{theorem:subst-char-for-theories}(\ref{theorem:PSC(lambda)-char-for-theories})
that $T(\bar{x})$ is equivalent modulo $V$ to a theory of $\Sigma^0_2$
formulae, all of whose free variables are among $\bar{x}$. We show
below that the converse is not true.

Let $V = \{\forall x \forall y (E(x, y) \rightarrow E(y, x))\}$ be the
theory that defines exactly all undirected graphs.  For $n \ge 1$, let
$\varphi_{n}(x)$ be a formula asserting that $x$ is not a part of a
cycle of length $n$. Explicitly, $\varphi_1(x) = \neg E(x, x)$ and for
$n \ge 1$, we have $\varphi_{n+1}(x) = \neg \exists z_1 \ldots \exists
z_n \big((\bigwedge_{1 \leq i < j \leq n} z_i \neq z_j) \wedge
(\bigwedge_{i = 1}^{i = n} (x \neq z_i)) \wedge E(x, z_1) \wedge E
(z_n, x) \wedge$ $ \bigwedge_{i = 1}^{i = n -1} E(z_i, z_{i+1})\big)$.
Consider $\chi_n(x) = \bigwedge_{i = 1}^{i = n} \varphi_i(x)$ which
asserts that $x$ is not a part of any cycle of length $\leq
n$. Observe that $\chi_n(x)$ is equivalent to a universal formula.
Also, if $m \leq n$, then $\chi_n(x) \rightarrow \chi_m(x)$.

Now consider the theory $T = \{ \psi_n \mid n \ge 1\}$, where $\psi_n
= \exists x \chi_n(x)$. Each sentence of $T$ is a $\Sigma^0_2$
sentence having only one existential quantifier. We show that $T$ is
not $PSC(\aleph_0)$ modulo $V$.

Consider the infinite graph $G$ given by $G = \bigsqcup_{i \ge 3} C_i$
where $C_i$ is the cycle graph of length $i$ and $\bigsqcup$ denotes
disjoint union. Any vertex $x$ of $C_i$ satisfies $\chi_j(x)$ in $G$,
for $j < i$. Then $G \models T$. Now consider any finite set $S$ of
vertices of $G$. Let $r$ be the highest index such that some vertex in
$S$ is in the cycle $C_r$. Consider the subgraph $G_1$ of $G$ induced
by the vertices of all the cycles in $G$ of length $\leq r$. Then no
vertex $x$ of $G_1$ satisfies $\chi_l(x)$ for $l > r$. Then $G_1
\not\models T$, whence $S$ cannot be a $k$-crux of $G$ w.r.t. $T$
modulo $V$, for any $k \ge |S|$. Since $S$ is an arbitrary finite
subset of $G$, we conclude that $G$ has no $k$-crux w.r.t. $T$ modulo
$V$, for any $k \in \mathbb{N}$; in other words, $G$ has no
$\aleph_0$-crux. Then $T$ is not $PSC(\aleph_0)$ modulo $V$.
\end{proof}

Towards the proof of part
(\ref{theorem:PSC(lambda)-char-for-theories}) of
Theorem~\ref{theorem:subst-char-for-theories}, we recall the notion of
\emph{sandwiches}\sindex[term]{sandwich} as defined by Keisler
in~\cite{keisler-sandwich}. We say that a triple $(\mf{A}, \mf{B},
\mf{C})$ of structures is a \emph{sandwich} if $\mf{A} \preceq \mf{C}$
and $\mf{A} \subseteq \mf{B} \subseteq \mf{C}$.  Given structures
$\mf{A}$ and $\mf{B}$, we say that \emph{$\mf{B}$ is sandwiched by
  $\mf{A}$} if there exist structures $\mf{A}'$ and $\mf{B}'$ such
that (i) $\mf{B} \preceq \mf{B}'$ and (ii) $(\mf{A}, \mf{B}',
\mf{A}')$ is a sandwich. Given theories $V$ and $T$, we say $T$ is
\emph{preserved under sandwiches by models of $T$ modulo $V$} if for
each model $\mf{A}$ of $V \cup T$, if $\mf{B}$ is sandwiched by
$\mf{A}$ and $\mf{B}$ models $V$, then $\mf{B}$ models $T$. The
following theorem of Keisler (Corollary 5.2
of~\cite{keisler-sandwich}) gives a syntactic characterization of the
aforesaid preservation property in terms of $\Sigma^0_2$ theories.

\begin{theorem}[Keisler, 1960]\label{theorem:keisler-sandwich}
Let $V$ and $T$ be theories. Then $T$ is preserved under sandwiches by
models of $T$ modulo $V$ iff $T$ is equivalent modulo $V$ to a theory
of $\Sigma^0_2$ sentences.
\end{theorem}

To prove the `Only if' direction of
Theorem~\ref{theorem:subst-char-for-theories}(\ref{theorem:PSC(lambda)-char-for-theories})
therefore, it just suffices to show that if $T$ is a theory that is
$PSC(\lambda)$ modulo $V$, then $T$ is preserved under sandwiches by
models of $T$ modulo $V$. To do this, we first prove the following
lemmas.

\begin{lemma}[Sandwich by saturated structures]\label{lemma:sandwich-by-saturated-structures}
Let $\mf{A}_1$ and $\mf{B}_1$ be structures such that $\mf{B}_1$ is
sandwiched by $\mf{A}_1$.  Then for each $\mu \ge \omega$, for every
$\mu$-saturated elementary extension $\mf{A}$ of $\mf{A}_1$, there
exists a structure $\mf{B}$ isomorphic to $\mf{B}_1$ such that
$\mf{B}$ is sandwiched by $\mf{A}$.
\end{lemma}

\begin{lemma}[Preservation under sandwich by saturated models]\label{lemma:closure-of-PSC(k)-under-sandwich-by-saturated-models}
Let $V$ and $T$ be theories such that $T$ is $PSC(\lambda)$ modulo
$V$, for some $\lambda \ge \aleph_0$. Let $\mf{A}$ be a
$\mu$-saturated model of $V \cup T$, for some $\mu \ge \lambda$, and
let $\mf{B}$ be a model of $V$. If $\mf{B}$ is sandwiched by $\mf{A}$,
then $\mf{B}$ is a model of $T$.
\end{lemma}

Using the above lemmas, we can prove
Theorem~\ref{theorem:subst-char-for-theories}(\ref{theorem:PSC(lambda)-char-for-theories})
as follows.

\begin{proof}[Proof of Theorem~\ref{theorem:subst-char-for-theories}(\ref{theorem:PSC(lambda)-char-for-theories})] 
We give the proof for theories without free variables. The proof for
theories with free variables follows from definitions.

\underline{If}:
Suppose $T$ is equivalent modulo $V$ to a $\Sigma^0_2$ theory $Y$. We
show that $Y$ is $PSC(\lambda)$ modulo $V$ for each $\lambda
> \aleph_0$, whereby the same is true of $T$. Towards this, observe
that $Y$ is a countable set. Let $\mf{A}$ be a model of $V \cup Y$.
Let $C \subseteq \mathsf{U}_{\mf{A}}$ be the (countable) set of
witnesses in $\mf{A}$, of the existential quantifiers of the sentences
of $Y$. In other words, $C$ is a countable subset of $\univ{\mf{A}}$
such that for each sentence $\phi$ of $Y$, there exist elements of $C$
that form a witness in $\mf{A}$, of the existential quantifiers of
$\phi$. It is easy to see that $C$ is an $\aleph_1$-crux of $\mf{A}$
w.r.t. $Y$ modulo $V$. Then $Y$ is $PSC(\aleph_1)$ modulo $V$, and
hence $PSC(\lambda)$ modulo $V$, for each $\lambda > \aleph_0$.

\underline{Only If}:
Suppose $T$ is $PSC(\lambda)$ modulo $V$ for $\lambda > \aleph_0$. To
complete the proof, it suffices to show, owing to
Theorem~\ref{theorem:keisler-sandwich}, that $T$ is preserved under
sandwiches by models of $T$ modulo $V$.  Suppose $\mf{A}_1$ and
$\mf{B}_1$ are given structures such that $\mf{B}_1$ is sandwiched by
$\mf{A}_1$, $\mf{B}_1 \models V$ and $\mf{A}_1 \models (V \cup T)$. 
Consider a $\mu$-saturated elementary extension $\mf{A}$
of $\mf{A}_1$, for some $\mu \ge \lambda$.  By
Lemma~\ref{lemma:sandwich-by-saturated-structures}, there exists a
structure $\mf{B}$ isomorphic to $\mf{B}_1$ such that $\mf{B}$ is
sandwiched by $\mf{A}$.  Then $\mf{A} \models (V \cup T)$ and $\mf{B}
\models V$, whence by
Lemma~\ref{lemma:closure-of-PSC(k)-under-sandwich-by-saturated-models},
we have $\mf{B} \models T$. Since $\mf{B}_1 \cong \mf{B}$, we have
$\mf{B}_1 \models T$, completing the proof.
\end{proof}

We now prove Lemmas~\ref{lemma:sandwich-by-saturated-structures} and
\ref{lemma:closure-of-PSC(k)-under-sandwich-by-saturated-models}. We 
refer the reader to Section~\ref{section:background:structures} for
the notions of
$\tau_{\mf{A}}, \mf{A}_{\mf{A}}, \mf{B}_{\mf{A}}, \diag{\mf{A}}$ and
$\eldiag{\mf{A}}$ for $\mf{A} \subseteq \mf{B}$, that we use in our
proofs below.  We make the simple yet important observation that each
of $\text{Diag}(\mf{A})$ and $\text{El-diag}(\mf{A})$ is closed under
finite conjunctions. We let $\mf{A} \preceq_1 \mf{B}$ denote that (i)
$\mf{A} \subseteq \mf{B}$ and (ii) every $\Sigma^0_1$ sentence of
$\text{FO}(\tau_{\mf{A}})$ that is true in $\mf{B}_{\mf{A}}$ is also
true in $\mf{A}_{\mf{A}}$.

\begin{lemma}\label{lemma:relating-e.c.-and-sandwiches}
$\mf{A} \preceq_1 \mf{B}$ iff there exists $\mf{A}'$ such that
  $(\mf{A}, \mf{B}, \mf{A}')$ is a sandwich.
\end{lemma}
\begin{proof}
The `If' direction follows easily from the definition of elementary
substructure and the fact that existential formulae are preserved
under extensions. For the converse, suppose that $\mf{A} \preceq_1
\mf{B}$. Let the vocabularies $\tau_{\mf{B}}$ and $\tau_{\mf{A}}$ be
such that for every element $a$ of $\mf{A}$, the constant in
$\tau_{\mf{B}}$ corresponding to $a$ is \emph{the same} as the
constant in $\tau_{\mf{A}}$ corresponding to $a$ (and hence the
constants in $\tau_{\mf{B}} \setminus \tau_{\mf{A}}$ correspond
exactly to the elements in $\mf{B}$ that are not in $\mf{A}$). Now
consider the theory $Y$ given by $Y = \text{Diag}(\mf{B}) \cup
\text{El-diag}(\mf{A})$. Any non-empty finite subset of
$\text{Diag}(\mf{B})$, resp.\ $\text{El-diag}(\mf{A})$, is satisfied
in $\mf{B}_{\mf{B}}$, resp.\ $\mf{A}_{\mf{A}}$. Let $Z$ be any finite
subset of $Y$, that has a non-empty intersection with both
$\text{Diag}(\mf{B})$ and $\text{El-diag}(\mf{A})$; we can consider
$Z$ as given by $Z = \{\xi, \psi\}$ where $\xi \in
\text{Diag}(\mf{B})$ and $\psi \in \text{El-diag}(\mf{A})$.  Let $c_1,
\ldots, c_r$ be the (distinct) constants of $\tau_{\mf{B}} \setminus
\tau_{\mf{A}}$ appearing in $\xi$, and let $x_1, \ldots, x_r$ be fresh
variables. Consider the sentence $\phi$ given by $\phi = \exists x_1
\ldots \exists x_r \xi \left[c_1 \mapsto x_1; \ldots; c_r \mapsto x_r
  \right]$, where $c_i \mapsto x_i$ denotes substitution of $x_i$ for
$c_i$, for $1 \leq i \leq r$. Observe that $\phi$ is a $\Sigma^0_1$
sentence of $\text{FO}(\tau_{\mf{A}})$ and that $\mf{B}_{\mf{A}}
\models \phi$. Since $\mf{A} \preceq_1 \mf{B}$, we have that
$\mf{A}_{\mf{A}} \models \phi$. Let $a_1, \ldots, a_r$ be the
witnesses in $\mf{A}_{\mf{A}}$, of the quantifiers of $\phi$
corresponding to variables $x_1, \ldots, x_r$. Interpreting the
constants $c_1, \ldots, c_r$ as $a_1, \ldots, a_r$ respectively, we
see that $(\mf{A}_{\mf{A}}, a_1, \ldots, a_r) \models Z$. Since $Z$ is
an arbitrary finite subset of $Y$, by the compactness theorem, $Y$ is
satisfied in a $\tau_{\mf{B}}$-structure $\mf{C}$. The $\tau$-reduct
of $\mf{C}$ is the desired structure $\mf{A}'$. 
\end{proof}

\begin{proof}[Proof of Lemma~\ref{lemma:sandwich-by-saturated-structures}] 
Let $\mf{A}$ be a $\mu$-saturated elementary extension of
$\mf{A}_1$, for some $\mu \ge \omega$.  We show below the
existence of a structure $\mf{B}_2$ such that (i) $\mf{A} \preceq_1
\mf{B}_2$ and (ii) $\mf{B}_1$ is elementarily embeddable in $\mf{B}_2$
via an embedding say $f$. Let $\mf{B}$ be the image of $\mf{B}_1$
under $f$; then $\mf{B} \cong \mf{B}_1$ and $\mf{B} \preceq \mf{B}_2$.
By Lemma~\ref{lemma:relating-e.c.-and-sandwiches}, there exists a
structure $\mf{A}_2$ such that $(\mf{A}, \mf{B}_2, \mf{A}_2)$ is a
sandwich, whence $\mf{B}$ is sandwiched by $\mf{A}$. Then $\mf{B}$ is
indeed as desired.  For our arguments below, we make the following
observation, call it (*): If $\mf{B}$ is sandwiched by $\mf{A}$, then
every $\Sigma^0_2$ sentence true in $\mf{A}$ is also true in
$\mf{B}$. This follows simply from
Theorem~\ref{theorem:keisler-sandwich} by taking $T$ to be the set of
all $\Sigma^0_2$ sentences that are true in $\mf{A}$, and taking $V$
to be the empty theory.

Let $\tau$ be the vocabulary of $\mf{A}$ and $\mf{B}_1$, and let
$\tau_{\mf{A}}$ and $\tau_{\mf{B}_1}$ be such that
$\tau_{\mf{A}} \cap \tau_{\mf{B}_1} = \tau$.  Consider the theory $Y$
given by $Y = S_\Pi(\mf{A}_{\mf{A}}) \cup
\text{El-diag}(\mf{B}_1)$, where $S_\Pi(\mf{A}_{\mf{A}})$ denotes the
set of all $\Pi^0_1$ sentences true in $\mf{A}_{\mf{A}}$. Observe that
$S_\Pi(\mf{A}_{\mf{A}})$ is closed under finite conjunctions.  Let $Z$
be any non-empty finite subset of $Y$. If $Z \subseteq
S_\Pi(\mf{A}_{\mf{A}})$ or $Z \subseteq \text{El-diag}(\mf{B}_1)$,
then $Z$ is clearly satisfiable. Else, $Z = \{\xi, \psi\}$ where $\xi
\in S_\Pi(\mf{A}_{\mf{A}})$ and $\psi \in
\text{El-diag}(\mf{B}_1)$. Let $c_1, \ldots, c_r$ be the (distinct)
constants of $\tau_{\mf{A}} \setminus \tau$ appearing in $\xi$, and
let $x_1, \ldots, x_r$ be fresh variables. Consider the sentence
$\phi$ given by $\phi = \exists x_1 \ldots \exists x_r \xi \left[c_1
  \mapsto x_1; \ldots; c_r \mapsto x_r\right]$, where $c_i \mapsto
x_i$ denotes substitution of $x_i$ for $c_i$, for $1 \leq i \leq r$.
Clearly $\mf{A} \models \phi$, whence $\mf{A}_1 \models \phi$. Since
$\mf{B}_1$ is sandwiched by $\mf{A}_1$ and $\phi$ is a $\Sigma^0_2$
sentence, it follows from observation (*) above, that $\mf{B}_1
\models \phi$. Let $b_1, \ldots, b_r$ be the witnesses in $\mf{B}_1$
of the quantifiers of $\phi$ associated with $x_1, \ldots, x_r$. One
can now check that if $\mf{R} = \mf{B}_1$, then $(\mf{R}_{\mf{R}},
b_1, \ldots, b_r) \models Z$. Since $Z$ is an arbitrary finite subset
of $Y$, by compactness theorem, $Y$ is satisfiable. Whereby, there
exists a $\tau$-structure $\mf{B}_2$ such that (i) $\mf{A} \preceq_1
\mf{B}_2$ and (ii) $\mf{B}_1$ is elementarily embeddable in
$\mf{B}_2$. 
\end{proof}

We now turn to proving
Lemma~\ref{lemma:closure-of-PSC(k)-under-sandwich-by-saturated-models}.
The notion of the FO-type of a $k$-tuple in a given structure for $k
\in \mathbb{N}$ can be naturally extended to the notion of the FO-type
of a tuple of length $< \lambda$ in a given structure for $\lambda \ge
\aleph_0$. Formally, given a structure $\mf{A}$ and a tuple $\bar{a} =
(a_1, a_2, \ldots)$ of $\mf{A}$, of length $< \lambda$, the FO-type of
$\bar{a}$ in $\mf{A}$, denoted $\mathsf{tp}_{\mf{A}, \bar{a}}(x_1,
x_2, \ldots)$, is the set of formulae given by $\mathsf{tp}_{\mf{A},
  \bar{a}}(x_1, x_2, \ldots) = \{ \varphi(x_{\eta_1}, \ldots,
x_{\eta_k}) \mid k \in \mathbb{N}, \,1 \leq \eta_1 < \ldots < \eta_k <
\lambda,\, \varphi(x_{\eta_1}, \ldots, x_{\eta_k})~\text{is an FO
  formula such that}$ $~ (\mf{A}, a_{\eta_1}, \ldots, a_{\eta_k})
\models \varphi(x_{\eta_1}, \ldots, x_{\eta_k})\}$.  The subset of
$\mathsf{tp}_{\mf{A}, \bar{a}}(x_1, x_2, \ldots)$ consisting of all
$\Pi^0_1$ formulae in $\mathsf{tp}_{\mf{A}, \bar{a}}(x_1, x_2,
\ldots)$ is denoted as $\mathsf{tp}_{\Pi, \mf{A}, \bar{a}}(x_1, x_2,
\ldots)$.  For theories $V$ and $T$ such that $T$ is $PSC(\lambda)$
modulo $V$, and for $\mf{A}$ and $\bar{a}$ as mentioned above, we say
that $\mathsf{tp}_{\Pi, \mf{A}, \bar{a}}(x_1, x_2, \ldots)$
\emph{determines a $\lambda$-crux w.r.t. $T$ modulo
  $V$}\sindex[term]{type determining a crux} if it is the case that
given a model $\mf{D}$ of $V$ and a tuple $\bar{d}$ of $\mf{D}$, of
length equal to that of $\bar{a}$, if $(\mf{D}, \bar{d}) \models
\mathsf{tp}_{\Pi, \mf{A}, \bar{a}}(x_1, x_2, \ldots)$, then $\mf{D}
\models T$. Since universal formulae are preserved under
substructures, it follows that for $\mf{D}$ as just mentioned, the
elements of $\bar{d}$ form a $\lambda$-crux of $\mf{D}$ w.r.t. $T$
modulo $V$. To prove
Lemma~\ref{lemma:closure-of-PSC(k)-under-sandwich-by-saturated-models},
we need the next result which characterizes when a $\Pi^0_1$-type
determines a $\lambda$-crux.

\begin{lemma}[Characterizing ``crux determination'']\label{lemma:crux-determination-using-saturated-models}
Let $V$ and $T$ be theories such that $T$ is $PSC(\lambda)$ modulo $V$
for some $\lambda \ge \aleph_0$. Let $\mf{A}$ be a model of $V$ and
$\bar{a}$ be a tuple of elements of $\mf{A}$, of length less than
$\lambda$.  Then $\mathsf{tp}_{\Pi, \mf{A}, \bar{a}}(x_1, x_2,
\ldots)$ determines a $\lambda$-crux w.r.t. $T$ modulo $V$ iff $\mf{A}
\models T$ and for some $\mu \ge \lambda$, there exists a
$\mu$-saturated elementary extension $\mf{B}$ of $\mf{A}$ (hence
$\mf{B} \models (V \cup T)$) such that $\bar{a}$ is a $\lambda$-crux
of $\mf{B}$ w.r.t. $T$ modulo $V$.

\end{lemma}
\begin{proof}
`Only If:' Since $\mf{A} \models V$ and
$(\mf{A}, \bar{a}) \models \mathsf{tp}_{\Pi, \mf{A}, \bar{a}}(x_1,
x_2, \ldots)$, we have that $\mf{A} \models T$. Let $\mf{B}$ be any
$\mu$-saturated elementary extension of $\mf{A}$ for
$\mu \ge \lambda$; then
$(\mf{B}, \bar{a}) \models \mathsf{tp}_{\Pi, \mf{A}, \bar{a}}(x_1,
x_2, \ldots)$ and $\mf{B} \models V$.  Since
$\mathsf{tp}_{\Pi, \mf{A}, \bar{a}}(x_1, x_2, \ldots)$ determines a
$\lambda$-crux w.r.t. $T$ modulo $V$, we have that $\bar{a}$ is a
$\lambda$-crux of $\mf{B}$ w.r.t. $T$ modulo $V$.

`If:' Let $\mf{A}$, $\mf{B}$ and $\bar{a}$ be as mentioned in the
statement. Consider a model $\mf{D}$ of $V$ and a tuple $\bar{d}$ of
$\mf{D}$, of length equal to that of $\bar{a}$, such that
$(\mf{D}, \bar{d}) \models \mathsf{tp}_{\Pi, \mf{A}, \bar{a}}(x_1,
x_2, \ldots)$. By the downward L\"owenheim-Skolem theorem, there
exists $\mf{D}_1 \preceq \mf{D}$ such that (i) $\mf{D}_1$ contains
$\bar{d}$ and (ii) $|\mf{D}_1| \leq \lambda$.  Then
$(\mf{D}_1, \bar{d}) \models \mathsf{tp}_{\Pi, \mf{A}, \bar{a}}(x_1,
x_2, \ldots)$.  Now since $\mf{A} \preceq \mf{B}$, we have that
$\mathsf{tp}_{\Pi, \mf{B}, \bar{a}}(x_1, x_2, \ldots)
= \mathsf{tp}_{\Pi, \mf{A}, \bar{a}}(x_1, x_2, \ldots)$.  Then every
existential sentence that is true in $(\mf{D}_1, \bar{d})$ is also
true in $(\mf{B}, \bar{a})$. Since $\mf{B}$ is $\mu$-saturated, and
the length of $\bar{a}$ is $< \lambda \leq \mu$, we have that
$(\mf{B}, \bar{a})$ is also $\mu$-saturated (by
Proposition~\ref{prop:saturation-results}(\ref{prop:saturation-results:expansion-preserves-saturation})). Further,
since $|\mf{D}_1| \leq \lambda$, we have
$|(\mf{D}_1, \bar{d})| \leq \lambda \leq \mu$. Then there exists an
embedding $f: (\mf{D}_1, \bar{d}) \rightarrow (\mf{B}, \bar{a})$ (by
Proposition~\ref{prop:saturation-results}(\ref{prop:saturation-results:isomorphic-embedding})).
The image of $(\mf{D}_1, \bar{d})$ under $f$ is a substructure
$(\mf{B}_1, \bar{a})$ of $(\mf{B}, \bar{a})$. Since
$\mf{D}_1 \preceq \mf{D}$ and $\mf{D} \models V$, we have
$\mf{B}_1 \models V$.  Further since $\bar{a}$ forms a $\lambda$-crux
of $\mf{B}$ w.r.t. $T$ modulo $V$ (by assumption), we have
$\mf{B}_1 \models T$. Then $\mf{D}_1$, and hence $\mf{D}$, models $T$,
completing the proof.
\end{proof}

\begin{proof}[Proof of Lemma~\ref{lemma:closure-of-PSC(k)-under-sandwich-by-saturated-models}]  
We assume the vocabulary to be $\tau$.  Since $\mf{B}$ is sandwiched
by $\mf{A}$, there exist structures $\mf{A}_1$ and $\mf{B}_1$ such
that (i) $\mf{B} \preceq \mf{B}_1$ and (ii) $(\mf{A},
\mf{B}_1, \mf{A}_1)$ is a sandwich. Let $\mf{D}$ be a
$\mu$-saturated elementary extension of $\mf{A}_1$ for some
$\mu \ge \lambda$. Then $\mf{A} \preceq \mf{D}$. Since $\mf{A}$ models
$V \cup T$, so does $\mf{D}$.

Now, given that $T$ is $PSC(\lambda)$ modulo $V$, there exists a
$\lambda$-crux of $\mf{D}$ w.r.t. $T$ modulo $V$; let $\bar{d}$ be any
tuple (of length $< \lambda$) formed from this
$\lambda$-crux. Consider $\mathsf{tp}_{\mf{D}, \bar{d}}(x_1,
x_2, \ldots)$, namely the FO-type of $\bar{d}$ in $\mf{D}$. Since
$\mf{A} \preceq \mf{D}$, we have $\mf{A} \equiv \mf{D}$ (see
Lemma~\ref{lemma:prop-of-structures}). Then since $\mf{A}$ is
$\mu$-saturated, there exists a tuple $\bar{a}$ of $\mf{A}$, of length
equal to that of $\bar{d}$, such that $(\mf{A}, \bar{a}) \equiv
(\mf{D}, \bar{d})$ (by
Proposition~\ref{prop:saturation-results}(\ref{prop:saturation-results:type-realization})). In
other words, $\mathsf{tp}_{\mf{A}, \bar{a}}(x_1, x_2, \ldots)
= \mathsf{tp}_{\mf{D}, \bar{d}}(x_1, x_2, \ldots)$.  Then since
$\mf{A} \preceq \mf{D}$, it follows that the FO-type of $\bar{a}$ in
$\mf{D}$, namely $\mathsf{tp}_{\mf{D}, \bar{a}}(x_1, x_2, \ldots)$, is
exactly $\mathsf{tp}_{\mf{A}, \bar{a}}(x_1, x_2, \ldots)$. Whence,
$\mathsf{tp}_{\Pi, \mf{D}, \bar{a}}(x_1, x_2, \ldots)
= \mathsf{tp}_{\Pi, \mf{A}, \bar{a}}(x_1, x_2, \ldots)
= \mathsf{tp}_{\Pi, \mf{D}, \bar{d}}(x_1, x_2, \ldots)$. Now since (i)
$\mf{D} \models (V \cup T)$, (ii) $\mf{D}$ is itself $\mu$-saturated
and (iii) $\bar{d}$ is a $\lambda$-crux of $\mf{D}$ w.r.t. $T$ modulo
$V$, we have by
Lemma~\ref{lemma:crux-determination-using-saturated-models}, that
$\mathsf{tp}_{\Pi, \mf{D}, \bar{d}}(x_1, x_2, \ldots)$, and hence
$\mathsf{tp}_{\Pi, \mf{D}, \bar{a}}(x_1, x_2, \ldots)$, determines a
$\lambda$-crux w.r.t. $T$ modulo $V$.  Whence the elements of
$\bar{a}$ form a $\lambda$-crux of $\mf{D}$ w.r.t. $T$ modulo
$V$. Since (i) $\mf{B}_1
\subseteq \mf{A}_1 \preceq \mf{D}$ (ii) $\mf{B}_1$ contains $\bar{a}$ and (iii)
$\mf{B}_1 \models V$ (since $\mf{B} \models V$ and $\mf{B} \preceq
\mf{B}_1$), we have by definition of a $\lambda$-crux (w.r.t. $T$ modulo $V$), 
that $\mf{B}_1 \models T$, whence $\mf{B} \models T$.
\end{proof}

The characterizations of this section are depicted pictorially below.

\begin{figure}[H]\centering 
\includegraphics[scale=0.7]{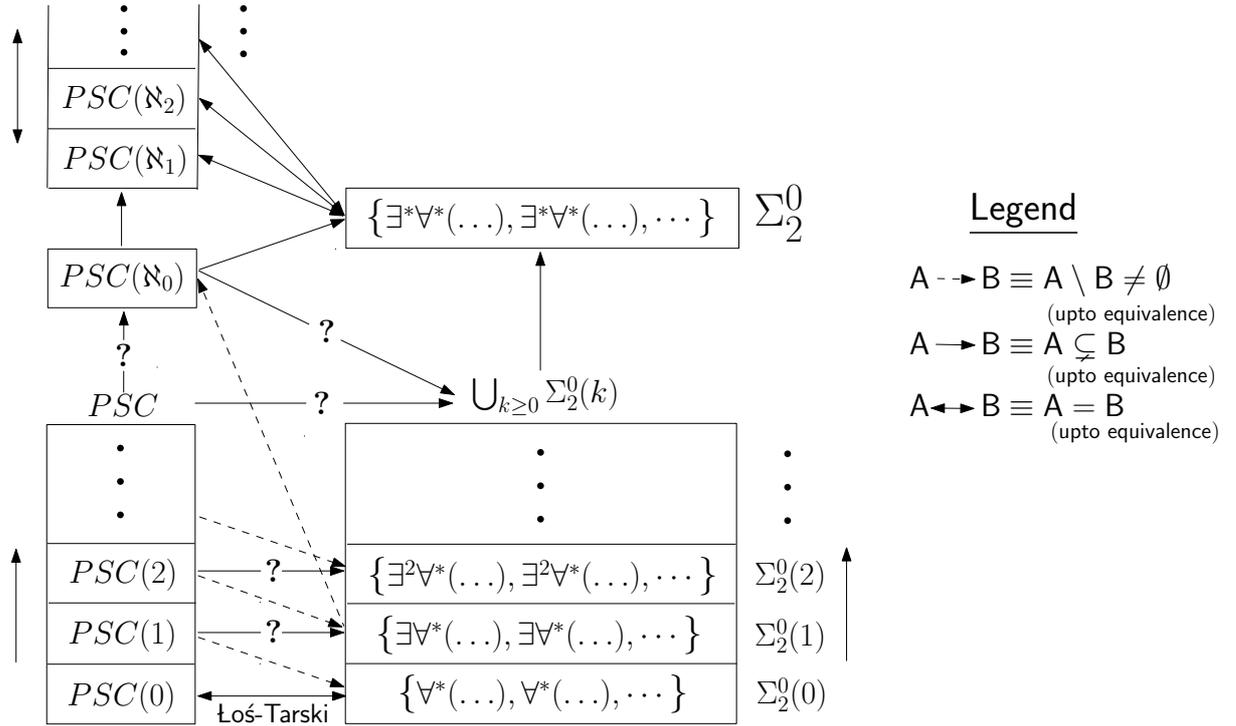}
\caption{(Partial) characterizations of $PSC(k)$ and $PSC(\lambda)$ theories}
\label{figure:subst-chars}
\end{figure}







%% file: conditional-char-3.0.tex
\subsection*{Proof of Theorem~\ref{theorem:conditional-refinement}}\label{subsection:conditional-refinement}

The technique of our proof is as presented below.
\vspace{2pt}
\begin{enumerate}[leftmargin=*, nosep]
\item We first define a variant of $PSC(k)$, that we call $PSC_{var}(k)$,
into whose definition we build Hypothesis~\ref{hypothesis:H}.
\item We then show that $PSC_{var}(k)$ theories are equivalent to
  theories of $\exists^k \forall^*$ sentences. This is done in the
  following two steps:
\begin{itemize}[leftmargin=13pt, nosep]
\item ``Going up'': We give a characterization of $PSC_{var}(k)$
  theories in terms of sentences of a special infinitary logic
  (Lemma~\ref{lemma:inf-characterization-of-PSC-var(k)}).
\item ``Coming down'': We provide a translation of sentences of the
  aforesaid infinitary logic, into their equivalent FO theories,
  whenever these sentences define \emph{elementary} (i.e. definable
  using FO theories) classes of structures
  (Proposition~\ref{prop:char-of-inf-formulae-defining-elem-classes}).
  The FO theories are obtained from suitable \emph{finite
  approximations} of the infinitary sentences, and turn out to be
  theories of $\exists^k \forall^*$ sentences.
\end{itemize}
\item We hypothesize that $PSC_{var}(k)$ theories are no different
  from $PSC(k)$ theories, as an equivalent reformulation of
  Hypothesis~\ref{hypothesis:H}, to obtain
  Theorem~\ref{theorem:conditional-refinement}. To show that this
  hypothesis is well-motivated, we define a variant of $PCE(k)$,
  denoted $PCE_{var}(k)$, that is dual to $PSC_{var}(k)$. We show that
  $PCE_{var}(k)$ coincides with $PCE(k)$ for theories, and use this to
  conclude that $PSC_{var}(k)$ coincides with $PSC(k)$ for sentences
  (Lemma~\ref{lemma:relations-between-old-and-new-notions}).
\end{enumerate}
\vspace{2pt}

Throughout the section, whenever $V$ and $T$ are clear from the
context, we skip mentioning the qualifier `w.r.t. $T$ modulo $V$' for
a $k$-crux, if $T$ is $PSC(k)$ modulo $V$. Before we present the
definitions of $PSC_{var}(k)$ and $PCE_{var}(k)$, we first define the
notion of `distinguished $k$-crux'.

\begin{defn}\label{defn:distinguished-cruxes}
Suppose $T$ is $PSC(k)$ modulo $V$ for theories $T$ and $V$. Given a
model $\mf{A}$ of $V \cup T$, we call a $k$-tuple $\bar{a}$ of
$\mf{A}$ a \emph{distinguished $k$-crux} of $\mf{A}$, if for some $\mu
\ge \omega$, there is a $\mu$-saturated elementary extension
$\mf{A}^+$ of $\mf{A}$ (whence $\mf{A}^+ \models V \cup T$) such that
$\bar{a}$ is a $k$-crux of $\mf{A}^+$ (whence $\bar{a}$ is also a
$k$-crux of $\mf{A}$).
\end{defn}

We now define $PSC_{var}(k)$ and $PCE_{var}(k)$. We refer the reader
to Definition~\ref{defn:k-ary-cover-inside-an-elem-ext} for the
meaning of the phrase `$k$-ary cover of $\mf{A}$ in $\mf{A}^+$'
appearing in the definition below.


\begin{defn}\label{defn:variants}
Let $V$ and $T$ be theories. 
\begin{enumerate}[nosep]
\item We say $T$ is $PSC_{var}(k)$ modulo $V$ if $T$ is $PSC(k)$
  modulo $V$ and every model of $V \cup T$ contains a distinguished
  $k$-crux.
\item We say $T$ is $PCE_{var}(k)$ modulo $V$ if for every model
  $\mf{A}$ of $V$, there exists a $\mu$-saturated elementary extension
  $\mf{A}^+$ of $\mf{A}$ for some $\mu \ge \omega$, such that for
  every collection $R$ of models of $V \cup T$, if $R$ is a $k$-ary
  cover of $\mf{A}$ in $\mf{A}^+$, then $\mf{A} \models T$.
\end{enumerate}
\end{defn}

If $\phi(\bar{x})$ and $T(\bar{x})$ are respectively a formula and a
theory, each of whose free variables are among $\bar{x}$, then for a
theory $V$, the notions of `$\phi(\bar{x})$ is $PSC_{var}(k)$ (resp.\
$PCE_{var}(k)$) modulo $V$' and `$T(\bar{x})$ is $PSC_{var}(k)$
(resp.\ $PCE_{var}(k)$) modulo $V$' are defined similar to
corresponding notions for $PSC(k)$ (resp. $PCE(k)$).  The following
duality is easy to see.
\begin{lemma}[$PSC_{var}(k)$-$PCE_{var}(k)$ duality]\label{lemma:PSC-var(k)-PCE-var(k)-duality}
Given a theory $V$, a formula $\phi(\bar{x})$ is $PSC_{var}(k)$ modulo
$V$ iff $\neg \phi(\bar{x})$ is $PCE_{var}(k)$ modulo $V$.
\end{lemma}

Towards the proof of Theorem~\ref{theorem:conditional-refinement}, we
first show the following.
\begin{lemma}\label{lemma:relations-between-old-and-new-notions}
Given a theory $V$, each of the following holds.
\begin{enumerate}[nosep]
\item A formula $\phi(\bar{x})$ is $PSC(k)$ modulo $V$ iff
  $\phi(\bar{x})$ is $PSC_{var}(k)$ modulo
  $V$.\label{lemma:relations-between-old-and-new-notions-part-1}
\item A theory $T(\bar{x})$ is $PCE(k)$ modulo $V$ iff $T(\bar{x})$ is
  $PCE_{var}(k)$ modulo
  $V$.\label{lemma:relations-between-old-and-new-notions-part-2}
\end{enumerate}
\end{lemma}

\begin{proof}
We show below the following equivalence, call it ($\dagger$): A theory
$T(\bar{x})$ is $PCE_{var}(k)$ modulo $V$ iff $T(\bar{x})$ is
equivalent modulo $V$ to a theory of $\forall^k\exists^*$ formulae,
all of whose free variables are among $\bar{x}$.  Then part
(\ref{lemma:relations-between-old-and-new-notions-part-2}) of this
lemma follows from ($\dagger$) and
Theorem~\ref{theorem:ext-chars}(\ref{theorem:char-of-PCE(k)-theories}).
Part (\ref{lemma:relations-between-old-and-new-notions-part-1}) of the
lemma in turn follows from part
(\ref{lemma:relations-between-old-and-new-notions-part-2}) and the
dualities given by Lemma~\ref{lemma:PSC(k)-PCE(k)-duality} and
Lemma~\ref{lemma:PSC-var(k)-PCE-var(k)-duality}.

Given that the notion of `$k$-ary cover of $\mf{A}$ in $\mf{A}$' is
the same as the notion of `$k$-ary cover' as defined in
Definition~\ref{defn:k-ary-covered-ext}, we can prove the `Only if'
direction of ($\dagger$) in a manner identical to the proof of the
`Only if' direction of
Theorem~\ref{theorem:ext-chars}(\ref{theorem:char-of-PCE(k)-theories}). The
proof of the `If' direction of ($\dagger$) is also nearly the same as
that of the `If' direction of
Theorem~\ref{theorem:ext-chars}(\ref{theorem:char-of-PCE(k)-theories});
we present this proof below for completeness.  It suffices to give the
proof for theories without free variables.

Let $T$ be equivalent modulo $V$ to a theory of $\forall^k \exists^*$
sentences. Given a model $\mf{A}$ of $V$, let $\mf{A}^+$ be a
$\mu$-saturated elementary extension of $\mf{A}$, for some
$\mu \ge \omega$. Let $R$ be a collection of models of $V \cup T$ that
forms a $k$-ary cover of $\mf{A}$ in $\mf{A}^+$. We show that
$\mf{A} \models T$.  Consider $\varphi \in T$; let $\varphi =
\forall^k \bar{x} \psi(\bar{x})$ for a $\Sigma^0_1$ formula
$\psi(\bar{x})$, and let $\bar{a}$ be a $k$-tuple of $\mf{A}$. Since
$R$ is a $k$-ary cover of $\mf{A}$ in $\mf{A}^+$, there exists
$\mf{B}_{\bar{a}} \in R$ such that $\mf{B}_{\bar{a}}$ contains
$\bar{a}$. Since $\mf{B}_{\bar{a}} \models (V \cup T)$, we have
$\mf{B}_{\bar{a}} \models \varphi$ and hence $(\mf{B}_{\bar{a}},
\bar{a}) \models \psi(\bar{x})$. Since $\psi(\bar{x})$ is a
$\Sigma^0_1$ formula and $\mf{B}_{\bar{a}} \subseteq \mf{A}^+$, we
have $(\mf{A}^+, \bar{a}) \models \psi(\bar{x})$, whence $(\mf{A},
\bar{a}) \models \psi(\bar{x})$ since $\mf{A} \preceq \mf{A}^+$. Since
$\bar{a}$ is arbitrary, $\mf{A} \models \varphi$, and since $\varphi$
is an arbitrary sentence of $T$, we have $\mf{A} \models T$. 
\end{proof}

Motivated by Lemma~\ref{lemma:relations-between-old-and-new-notions},
we put forth the hypothesis below.

\begin{hypothesis}\label{hypothesis:H-reformulated}
If $V$ and $T(\bar{x})$ are theories, then $T(\bar{x})$ is $PSC(k)$
modulo $V$ iff $T(\bar{x})$ is $PSC_{var}(k)$ modulo $V$. 
\end{hypothesis}

It is easy to see that Hypothesis~\ref{hypothesis:H-reformulated} is
an equivalent reformulation of Hypothesis~\ref{hypothesis:H}.  

The following result is the essence of
Theorem~\ref{theorem:conditional-refinement}.

\begin{theorem}\label{theorem:char-of-PSC-var(k)}
Given theories $V$ and $T(\bar{x})$, suppose $T(\bar{x})$ is
$PSC_{var}(k)$ modulo $V$. Then $T(\bar{x})$ is equivalent modulo $V$
to a theory of $\Sigma^0_2$ formulae, all of whose free variables are
among $\bar{x}$, and all of which have $k$ existential quantifiers.
\end{theorem}

\begin{proof}[Proof of Theorem~\ref{theorem:conditional-refinement}]
Follows from the equivalence of Hypotheses~\ref{hypothesis:H}
and~\ref{hypothesis:H-reformulated}, and
Theorem~\ref{theorem:char-of-PSC-var(k)}.
\end{proof}

We devote the rest of this section to proving
Theorem~\ref{theorem:char-of-PSC-var(k)}.  We first introduce some
notation and terminology. These are adapted versions of similar
notation and terminology introduced in~\cite{keisler-sandwich}
and~\cite{keisler:finite-approx}.  Given a class $\mc{F}$ of formulae
and $k \ge 0$, denote by $\left[\exists^k \bigwedge\right]
\mc{F}$\sindex[symb]{$\left[\exists^k \bigwedge\right] \mc{F}$} the
class of infinitary formulae\sindex[term]{infinitary formulae}
$\Phi(\bar{x})$ of the form $\exists y_1 \ldots \exists y_k
\bigwedge_{i \in I} \psi_i(y_1, \ldots, y_k, \bar{x})$ where $I$ is an
index set (of arbitrary cardinality) and for each $i \in I$, $\psi_i$
is a formula of $\mc{F}$, whose free variables are among $y_1, \ldots,
y_k, \bar{x}$.  Let $\left[\exists^* \bigwedge\right] \mc{F} =
\bigcup_{k \ge 0} \left[\exists^k \bigwedge\right] \mc{F}$. Observe
that $\mc{F} \subseteq \left[\exists^* \bigwedge\right] \mc{F}$.  For
each $j \in \mathbb{N}$, let $\left[\exists^* \bigwedge\right]^j
\mc{F} =$ $\left[\exists^* \bigwedge\right]$ $\left[\exists^*
  \bigwedge\right]^{j-1} \mc{F}$, where $\left[\exists^*
  \bigwedge\right]^0 \mc{F} = \mc{F}$. Let $\left[\exists^*
  \bigwedge\right]^* \mc{F} = \bigcup_{j \ge 0} \left[\exists^*
  \bigwedge\right]^{j} \mc{F}$.  Finally, let $\left[\bigvee\right]
\mc{F}$\sindex[symb]{$\left[\bigvee\right] \mc{F}$} denote arbitrary
disjunctions of formulae of $\mc{F}$.  Observe that $\mc{F} \subseteq
\left[\bigvee\right] \mc{F}$.

Let $\Phi(\bar{x})$ be a formula of
$\left[\bigvee\right]\left[\exists^* \bigwedge\right]^* \text{FO}$,
where FO denotes as usual, the class of all first order formulae.  We
define below, the set $\mc{A}(\Phi)(\bar{x})$ of \emph{finite
  approximations}\sindex[term]{infinitary formulae!finite
  approximations of} of $\Phi(\bar{x})$. Let $\subseteq_f$ denote
`finite subset of'.

\begin{enumerate}[nosep]
\item If $\Phi(\bar{x}) \in \text{FO}$, then $\mc{A}(\Phi)(\bar{x}) =
  \{\Phi(\bar{x})\}$.
\item If $\Phi(\bar{x}) = \exists^k \bar{y} \bigwedge_{i \in I}
  \Psi_i(\bar{x}, \bar{y})$ for $k \ge 0$ and some index set $I$, then
  $\mc{A}(\Phi)(\bar{x}) =$ \linebreak $ \{ \exists^k \bar{y}$
  $\bigwedge_{i \in I_1} \gamma_i(\bar{x}, \bar{y})
  ~\mid~ \gamma_i(\bar{x}, \bar{y}) \in \mc{A}(\Psi_i)(\bar{x}, \bar{y}),~
  I_1 \subseteq_f I \}$.
\item If $\Phi(\bar{x}) = \bigvee_{i \in I} \Psi_i(\bar{x})$, then
  $\mc{A}(\Phi)(\bar{x}) = \{\bigvee_{i \in I_1} \gamma_i(\bar{x})
  \mid \gamma_i(\bar{x}) \in \mc{A}(\Psi_i)(\bar{x}),$ $I_1
  \subseteq_f I \}$.
\end{enumerate}

\vspace{5pt} Our proof of Theorem~\ref{theorem:char-of-PSC-var(k)}
is in two parts. The first part, namely the ``going up'' part as
alluded to in the beginning of this subsection, gives a
characterization of $PSC_{var}(k)$ theories in terms of the formulae
of $\left[\bigvee\right]\left[\exists^k \bigwedge\right]\Pi^0_1$,
where $\Pi^0_1$ is the usual class of all prenex FO formulae having
only universal quantifiers.

\begin{lemma}\label{lemma:inf-characterization-of-PSC-var(k)}
Let $V$ and $T(\bar{x})$ be given theories. Then $T(\bar{x})$ is
$PSC_{var}(k)$ modulo $V$ iff $T(\bar{x})$ is equivalent modulo $V$ to
a formula of $\left[\bigvee\right]\left[\exists^k
  \bigwedge\right]\Pi^0_1$, whose free variables are among $\bar{x}$.
\end{lemma}

The second part of the proof of
Theorem~\ref{theorem:char-of-PSC-var(k)}, namely the ``coming down''
part, consists of getting FO theories equivalent to the formulae of
$\left[\bigvee\right]\left[\exists^k \bigwedge\right]\Pi^0_1$,
whenever the latter define elementary classes of structures. In fact,
we show a more general result as we now describe. Given a theory $V$,
we say that a formula $\Phi(x_1, \ldots, x_k)$ of
$\left[\bigvee\right]\left[\exists^* \bigwedge\right]^* \text{FO}$
(over a vocabulary say $\tau$) \emph{defines an elementary class
  modulo $V$} if the sentence (over the vocabulary $\tau_k$) obtained
by substituting fresh and distinct constants $c_1, \ldots, c_k$ for
the free occurrences of $x_1, \ldots, x_k$ in $\Phi(x_1, \ldots,
x_k)$, defines an elementary class (of $\tau_k$-structures) modulo
$V$.  The result below characterizes formulae of
$\left[\bigvee\right]\left[\exists^* \bigwedge \right]^*\text{FO}$
that define elementary classes, in terms of the finite approximations
of these formulae.

\begin{proposition}\label{prop:char-of-inf-formulae-defining-elem-classes}
Let $\Phi(\bar{x})$ be a formula of
$\left[\bigvee\right]\left[\exists^* \bigwedge \right]^*\text{FO}$ and
$V$ be a given theory. Then $\Phi(\bar{x})$ defines an elementary
class modulo $V$ iff $\Phi(\bar{x})$ is equivalent modulo $V$ to a
countable subset of $\mc{A}(\Phi)(\bar{x})$.
\end{proposition}

The above results prove Theorem~\ref{theorem:char-of-PSC-var(k)}
as follows.

\begin{proof}[Proof of Theorem~\ref{theorem:char-of-PSC-var(k)}]
For any formula $\Phi(\bar{x})$ of
$\left[\bigvee\right]\left[\exists^k \bigwedge \right]\Pi^0_1$, each
formula of the set $\mc{A}(\Phi)(\bar{x})$ can be seen to be
equivalent to an $\exists^k \forall^*$ formula whose free variables
are among $\bar{x}$. The result then follows from
Lemma~\ref{lemma:inf-characterization-of-PSC-var(k)} and
Proposition~\ref{prop:char-of-inf-formulae-defining-elem-classes}.
\end{proof}

We now prove Lemma~\ref{lemma:inf-characterization-of-PSC-var(k)} and
Proposition~\ref{prop:char-of-inf-formulae-defining-elem-classes}. We
observe that it suffices to prove these results only for
theories/formulae without free variables.

The proof of Lemma~\ref{lemma:inf-characterization-of-PSC-var(k)}
requires the following result that characterizes when a $k$-crux of a
model of a $PSC(k)$ theory is a distinguished $k$-crux of the
model. To state this result, we define the notion of `the
$\Pi^0_1$-type of a $k$-tuple determining a $k$-crux', analogously to
the notion of the $\Pi^0_1$-type of a tuple of length $< \lambda$
determining a $\lambda$-crux, that was introduced just before
Lemma~\ref{lemma:crux-determination-using-saturated-models}. Formally,
given theories $V$ and $T$ such that $T$ is $PSC(k)$ modulo $V$, and
given a structure $\mf{A}$ and a $k$-tuple $\bar{a}$ of $\mf{A}$, we
say $\mathsf{tp}_{\Pi, \mf{A}, \bar{a}}(\bar{x})$ \emph{determines a
  $k$-crux w.r.t. $T$ modulo $V$} 
\sindex[term]{type determining a crux} if it is the case that given a
model $\mf{D}$ of $V$ and a $k$-tuple $\bar{d}$ of $\mf{D}$, if
$(\mf{D}, \bar{d}) \models \mathsf{tp}_{\Pi, \mf{A},
  \bar{a}}(\bar{x})$, then $\mf{D} \models T$. Since universal
formulae are preserved under substructures, it follows that for
$\mf{D}$ as just mentioned, the elements of $\bar{d}$ form a $k$-crux
of $\mf{D}$ w.r.t. $T$ modulo $V$.

\begin{lemma}[Characterizing distinguished $k$-cruxes]\label{lemma:char-of-distinguished-k-cruxes}
Let $V$ and $T$ be theories such that $T$ is $PSC(k)$ modulo $V$. Let
$\mf{A}$ be a model of $V \cup T$ and $\bar{a}$ be a $k$-tuple of
elements of $\mf{A}$.\linebreak Then $\bar{a}$ is a distinguished $k$-crux of
$\mf{A}$ w.r.t. $T$ modulo $V$ iff
$\mathsf{tp}_{\Pi, \mf{A}, \bar{a}}(\bar{x})$ determines a $k$-crux\linebreak
w.r.t. $T$ modulo $V$.
\end{lemma}

\begin{proof}
Similar to the proof of Lemma~\ref{lemma:crux-determination-using-saturated-models}
\end{proof}

\begin{proof}[Proof of Lemma~\ref{lemma:inf-characterization-of-PSC-var(k)}]
\underline{If:} Let $T$ be equivalent modulo $V$ to the sentence $\Phi
=$ $ \bigvee_{i \in I} \exists^k \bar{y}_i \bigwedge Y_i(\bar{y}_i)$,
where $I$ is an index set and for each $i \in I$, $Y_i$ is a set of
$\Pi^0_1$ formulae, all of whose free variables are among
$\bar{y}_i$. Then given a model $\mf{A}$ of $V \cup T$, there exist $i
\in I$ and $\bar{a}$ in $\mf{A}$ such that $(\mf{A}, \bar{a}) \models
\bigwedge Y_i(\bar{y}_i)$.  Let $\mf{A}^+$ be a $\mu$-saturated
elementary extension of $\mf{A}$, for some $\mu \ge \omega$. Then
$(\mf{A}^+, \bar{a}) \models \bigwedge Y_i(\bar{y}_i)$. Whence for
each $\mf{B} \subseteq \mf{A}^+$ such that $\mf{B}$ contains
$\bar{a}$, $(\mf{B}, \bar{a}) \models \bigwedge Y_i(\bar{y}_i)$, and
hence $\mf{B} \models \Phi$. Since $\Phi$ is equivalent to $T$ modulo
$V$, we have $\bar{a}$ as a distinguished $k$-crux of $\mf{A}$.

\underline{Only If}: Suppose $T$ is $PSC_{var}(k)$ modulo $V$.  Given
a model $\mf{A}$ of $V \cup T$, let\linebreak
$\mathsf{Dist\text{-}k\text{-}cruxes}(\mf{A})$ be the (non-empty) set
of all distinguished $k$-cruxes of $\mf{A}$.  Consider the sentence
$\Phi = \bigvee_{\mf{A} \models V\cup T,\, \bar{a} \in
  \mathsf{Dist\text{-}k\text{-}cruxes}(\mf{A})} \exists^k \bar{x}
\bigwedge \mathsf{tp}_{\Pi, \mf{A}, \bar{a}}(\bar{x})$. 
We show that $T$ is equivalent to $\Phi$ modulo $V$. That $T$ implies
$\Phi$ modulo $V$ is obvious from the definition of $\Phi$. Towards
the converse, suppose $\mf{B} \models \{\Phi\} \cup V$. Then for some
model $\mf{A}$ of $V \cup T$, some distinguished $k$-crux $\bar{a}$ of
$\mf{A}$, and for some $k$-tuple $\bar{b}$ of $\mf{B}$, we have
$(\mf{B}, \bar{b}) \models$
$\mathsf{tp}_{\Pi, \mf{A}, \bar{a}}(\bar{x})$. By
Lemma~\ref{lemma:char-of-distinguished-k-cruxes},
$\mathsf{tp}_{\Pi, \mf{A}, \bar{a}}(\bar{x})$ determines a $k$-crux
w.r.t. $T$ modulo $V$, whence $\mf{B} \models T$.
\end{proof}

We now turn to proving
Proposition~\ref{prop:char-of-inf-formulae-defining-elem-classes}. Our
proof of
Proposition~\ref{prop:char-of-inf-formulae-defining-elem-classes}
crucially uses our compactness result for formulae of
$\left[\exists^* \bigwedge\right]^* \text{FO}$, that we state now.

\begin{lemma}\label{lemma:compactness-generalization-in-terms-of-approximations}
Let $\Phi(\bar{x})$ be a formula of
$\left[\exists^* \bigwedge\right]^* \text{FO}$.  If every formula of
$\mc{A}(\Phi)(\bar{x})$ is satisfiable modulo a theory $V$, then
$\Phi(\bar{x})$ is satisfiable modulo $V$.
\end{lemma}

Observe that the standard compactness theorem for FO is a special case
of the above result: Given an FO theory $T(\bar{x})$, let
$\Phi(\bar{x})$ be the formula of $\left[\exists^* \bigwedge\right]^*
\text{FO}$ given by $\Phi(\bar{x}) = \bigwedge T(\bar{x})$. Then every
formula of $\mc{A}(\Phi)(\bar{x})$ is equivalent to a finite subset of
$T(\bar{x})$ and vice-versa.

\begin{remark}
The formulas of $\left[\exists^* \bigwedge \right]^* \text{FO}$ are
special kinds of ``conjunctive formulas'', where the latter are as
defined in~\cite{keisler:finite-approx}. The
paper~\cite{keisler:finite-approx} gives a generalization of the
compactness theorem by proving a compactness result for conjunctive
formulas, whose statement is similar to that of
Lemma~\ref{lemma:compactness-generalization-in-terms-of-approximations}.
However,
Lemma~\ref{lemma:compactness-generalization-in-terms-of-approximations}
does not follow from this result of~\cite{keisler:finite-approx}
because the set of finite approximations of sentences $\Phi(\bar{x})$
of $\left[\exists^* \bigwedge \right]^* \text{FO}$, as defined
in~\cite{keisler:finite-approx}, is semantically strictly larger than
the set $\mc{A}(\Phi)(\bar{x})$ that we have defined. Further, the
techniques that we use in proving
Lemma~\ref{lemma:compactness-generalization-in-terms-of-approximations}
are much different from those used in~\cite{keisler:finite-approx} for
proving the compactness result for conjunctive formulas.
\end{remark}

We finally require the following two auxiliary lemmas for the proofs
of Proposition~\ref{prop:char-of-inf-formulae-defining-elem-classes}
and
Lemma~\ref{lemma:compactness-generalization-in-terms-of-approximations}.

\begin{lemma}\label{lemma:compactness-direct-generalization}
For $j \in \mathbb{N}$, let $T(\bar{x})$ be a set of formulae of
$\left[\exists^* \bigwedge\right]^j \text{FO}$, all of whose free
variables are among $\bar{x}$. If every finite subset of $T(\bar{x})$
is satisfiable modulo a theory $V$, then $T(\bar{x})$ is satisfiable
modulo $V$.
\end{lemma}

\begin{proof}
We prove the statement by induction on $j$. The base case of $j = 0$
is the standard compactness theorem. As induction hypothesis, suppose
the statement is true for $j$. For the inductive step, consider a set
$T(\bar{x}) = \{\Phi_i(\bar{x}) \mid i \in I\}$ of $\left[\exists^*
  \bigwedge\right]^{j+1} \text{FO}$ formulae, all of whose free
variables are among $\bar{x}$, and suppose every finite subset of
$T(\bar{x})$ is satisfiable modulo $V$. Let $\Phi_i(\bar{x}) = \exists
\bar{y}_i \bigwedge T_i(\bar{x}, \bar{y}_i)$ where $T_i(\bar{x},
\bar{y}_i)$ is a set of formulae of $\left[\exists^*
  \bigwedge\right]^j \text{FO}$. Assume for $i, j \in I$ and $i \neq
j$, that $\bar{y}_i$ and $\bar{y}_j$ have no common variables.
We show that the set $Y$ of $\left[\exists^* \bigwedge\right]^j
\text{FO}$ formulae given by $Y = \bigcup_{i \in I} T_i$ is
satisfiable modulo $V$; then so is $T(\bar{x})$.

By the induction hypothesis, it suffices to show that every finite
subset $Z$ of $Y$ is satisfiable modulo $V$. Let $Z (\bar{x},
\bar{y}_{i_1}, \ldots, \bar{y}_{i_n}) = \bigcup_{r = 1}^{r = n}
Z_r(\bar{x}, \bar{y}_{i_r})$, where $n > 0$, $Z_r(\bar{x},
\bar{y}_{i_r}) \subseteq_f T_{i_r}(\bar{x}, \bar{y}_{i_r})$ and $i_r
\in I$, for each $r \in \{1, \ldots, n\}$.  The subset
$\{\Phi_{i_r}(\bar{x}) \mid r \in \{1, \ldots, n\}\}$ of $T(\bar{x})$
is satisfiable modulo $V$ by assumption, whence for some model
$\mf{A}$ of $V$, and interpretations $\bar{a}$ of $\bar{x}$ and
$\bar{b}_{i_r}$ of $\bar{y}_{i_r}$, we have that $\bigcup_{r = 1}^{r =
n} T_{i_r}(\bar{x}, \bar{y}_{i_r})$ is satisfied in
$(\mf{A}, \bar{a}, \bar{b}_{i_1}, \ldots, \bar{b}_{i_n})$; then
$(\mf{A},\bar{a}, \bar{b}_{i_1}, \ldots, \bar{b}_{i_n}) \models
Z(\bar{x}, \bar{y}_{i_1}, \ldots, \bar{y}_{i_n})$.
\end{proof}

\begin{lemma}\label{lemma:formula-implies-its-approximations}
Let $\Phi(\bar{x})$ be a formula of
  $\left[\exists^* \bigwedge\right]^* \text{FO}$. If
  $(\mf{A}, \bar{a}) \models
\Phi(\bar{x})$, then $(\mf{A}, \bar{a}) \models \xi(\bar{x})$ \linebreak 
for every formula $\xi(\bar{x})$ of $\mc{A}(\Phi)(\bar{x})$.
\end{lemma}
\begin{proof}
We prove the lemma by induction. The statement is trivial for formulae
of $\text{FO} = \left[\exists^* \bigwedge\right]^0 \text{FO}$. Assume
the statement for $\left[\exists^* \bigwedge\right]^j \text{FO}$
formulae. Consider a formula $\Phi(\bar{x})$ of
$\left[\exists^* \bigwedge\right]^{j+1} \text{FO}$ given by
$\Phi(\bar{x}) = \exists^n \bar{y} \bigwedge_{i \in
I} \Psi_i(\bar{x}, \bar{y})$, where $\Psi_i(\bar{x}, \bar{y}) \in
\left[\exists^* \bigwedge\right]^{j} \text{FO}$ for each $i \in
I$. Consider a formula $\xi(\bar{x})$ of $\mc{A}(\Phi)(\bar{x})$; then
$\xi(\bar{x}) = \exists^n \bar{y} \bigwedge_{i \in I_1}
\gamma_{i}(\bar{x}, \bar{y})$, for some $I_1 \subseteq_f I$ and
$\gamma_{i}(\bar{x}, \bar{y}) \in $ $ \mc{A}(\Psi_{i})(\bar{x},
\bar{y})$ for each $i \in I_1$. Since $(\mf{A}, \bar{a}) \models
\Phi(\bar{x})$, there is an $n$-tuple $\bar{b}$ from $\mf{A}$ such
that $(\mf{A}, \bar{a}, \bar{b}) \models \Psi_{i}(\bar{x}, \bar{y})$
for each $i \in I_1$.  By induction hypothesis, $(\mf{A}, \bar{a},
\bar{b}) \models \gamma_{i}(\bar{x}, \bar{y})$ for each $i \in I_1$;
then $(\mf{A}, \bar{a}) \models \xi(\bar{x})$.
\end{proof}

\begin{proof}[Proof of Lemma~\ref{lemma:compactness-generalization-in-terms-of-approximations}]
The proof proceeds by induction. The statement trivially holds for
formulae of $\text{FO} = \left[\exists^* \bigwedge\right]^0
\text{FO}$.  Assume the statement is true for formulae of
$\left[\exists^* \bigwedge\right]^j \text{FO}$.  Consider a formula
$\Phi(\bar{x})$ of $\left[\exists^* \bigwedge\right]^{j+1} \text{FO}$
given by $\Phi(\bar{x})$ = $\exists \bar{y} \bigwedge_{i \in I}
\Psi_i(\bar{x}, \bar{y})$, where $\Psi_i(\bar{x}, \bar{y})$ is a
formula of $\left[\exists^* \bigwedge\right]^{j} \text{FO}$ for each
$i \in I$.  We show that every finite subset of $T(\bar{x}, \bar{y}) =
\{\Psi_i(\bar{x}, \bar{y}) \mid i \in I \}$ is satisfiable modulo
$V$. Then by Lemma~\ref{lemma:compactness-direct-generalization},
$T(\bar{x}, \bar{y})$ is satisfiable modulo $V$; then
$\Phi(\bar{x})$ is also satisfiable modulo $V$.

Let $I_1$ be a finite subset of $I$. For $i \in I_1$, consider the
formula $\Psi_i(\bar{x}, \bar{y})$ of $T(\bar{x}, \bar{y})$; it is
given by $\Psi_i(\bar{x}, \bar{y}) = \exists \bar{z}_i \bigwedge
Z_i(\bar{x}, \bar{y}, \bar{z}_i)$ where $Z_i(\bar{x}, \bar{y},
\bar{z}_i)$ is a set of formulas of $\left[\exists^*
  \bigwedge\right]^{j-1} \text{FO}$. Let $\bar{z} = (\bar{z}_i)_{i \in
  I_1}$ be the tuple of all the variables of the $\bar{z}_i$s, for $i$
ranging over $I_1$. Assume without loss of generality that for $i_1,
i_2 \in I$ such that $i_1 \neq i_2$, none of the variables of
$\bar{z}_{i_1}$ appear in $\Psi_{i_2}$. Consider the formula
$\Psi(\bar{x}, \bar{y})$ of $\left[\exists^* \bigwedge\right]^{j}
\text{FO}$ given by $\Psi(\bar{x}, \bar{y}) = \exists \bar{z}
\bigwedge \big(\bigcup_{i \in I_1} Z_i(\bar{x}, \bar{y}, \bar{z}_i)
\big)$. It is easy to verify that $\Psi(\bar{x}, \bar{y})$ is
equivalent (over all structures) to $\{\Psi_i(\bar{x}, \bar{y}) \mid i
\in I_1 \}$. We now show that the latter is satisfiable modulo $V$ by
showing that the former is satisfiable modulo $V$ -- this in turn is
done by showing that every formula in $\mc{A}(\Psi)(\bar{x}, \bar{y})$
is satisfiable modulo $V$, and then applying the induction hypothesis
mentioned at the outset.

Let $\gamma(\bar{x}, \bar{y})$ be an arbitrary formula of
$\mc{A}(\Psi)(\bar{x}, \bar{y})$. Then $\gamma(\bar{x}, \bar{y}) $ is
of the form $\exists \bar{z} \bigwedge_{i \in I_2} $ $\bigwedge_{~l
  \in \{1, \ldots, n_i\}}$ $\alpha_{i, l}(\bar{x}, \bar{y},
\bar{z}_i)$, where $I_2 \subseteq I_1$, and for each $i \in I_2$, we
have $n_i \ge 1$, $\alpha_{i, l}(\bar{x}, \bar{y}, \bar{z}_i) \in $ $
\mc{A}(\beta_{i, l})(\bar{x}, \bar{y}, \bar{z}_i)$, and $\{\beta_{i,
  1}(\bar{x}, \bar{y}, \bar{z}_i), \ldots,$ $ \beta_{i, n_i}(\bar{x},
\bar{y}, \bar{z}_i)\} \subseteq_f Z_i(\bar{x}, \bar{y},
\bar{z}_i)$. It is easy to see that $\gamma(\bar{x}, \bar{y})$ is
equivalent to the formula $\bigwedge_{i \in I_2} \gamma_i(\bar{x},
\bar{y})$ where $\gamma_i(\bar{x}, \bar{y}) = \exists \bar{z}_i
\bigwedge_{~l \in \{1, \ldots, n_i\}}$ $\alpha_{i, l}(\bar{x},
\bar{y}, \bar{z}_i)$. Observe now that $\gamma_i(\bar{x}, \bar{y}) \in
\mc{A}(\Psi_i)(\bar{x}, \bar{y})$, whence $\exists \bar{y}
\bigwedge_{i \in I_2} \gamma_i(\bar{x}, \bar{y})$ $ \in
\mc{A}(\Phi)(\bar{x})$. By assumption, every formula of
$\mc{A}(\Phi)(\bar{x})$ is satisfiable modulo $V$; then so are
$\exists \bar{y} \bigwedge_{i \in I_2} \gamma_i(\bar{x}, \bar{y})$ and
$\gamma(\bar{x}, \bar{y})$.
\end{proof}

\begin{proof}[Proof of Proposition~\ref{prop:char-of-inf-formulae-defining-elem-classes}]
It suffices to show just the `Only if' direction of the result. Hence,
consider a sentence $\Phi$ of $\left[\bigvee\right]\left[\exists^*
  \bigwedge \right]^*\text{FO}$ given by $\Phi = \bigvee_{i \in I}
\Psi_i$ where $\Psi_i \in \left[\exists^* \bigwedge
  \right]^*\text{FO}$.  Let $\mc{B} = \prod_{i \in I} \mc{A}(\Psi_i)$
where $\prod$ denotes Cartesian product. We now show the following
equivalences modulo $V$:
\begin{eqnarray}
\Phi & \leftrightarrow & \bigvee_{i \in
I} \bigwedge_{\gamma \in \mc{A}(\Psi_i)} \gamma \label{eq:1}\\ 
& \leftrightarrow & \bigwedge_{(\gamma_i) \,\in \,\mc{B}} \,\bigvee_{i \in
I} \gamma_i \label{eq:2}
\end{eqnarray}
In equivalence Eq.~\ref{eq:2} above, $(\gamma_i)$ denotes a sequence in
$\mc{B}$. Let $\mc{P}_{fin}(I)$ be the set of all finite subsets of
$I$. We finally show the existence of a function $g: \mc{B}
\rightarrow \mc{P}_{fin}(I)$ that gives the following
equivalence
\begin{eqnarray}
\Phi & \leftrightarrow &
\bigwedge_{(\gamma_i) \in \mc{B}} \,\bigvee_{j \in
  g((\gamma_i))} ~\gamma_j \label{eq:3}
\end{eqnarray}
Observe that each disjunction in the RHS of Eq.~\ref{eq:3} is a
sentence of $\mc{A}(\Phi)$. Observe also that instead of ranging over
all of $\mc{B}$ in the RHS of Eq.~\ref{eq:3} above, we can range over
only a countable subset of $\mc{B}$, since the number of FO sentences
over a finite vocabulary is countable. We now show the above
equivalences to complete the proof. The equivalence Eq.~\ref{eq:2} is
obtained by applying the standard distributivity laws for conjunctions
and disjunctions, to the sentence in the RHS of Eq.~\ref{eq:1}.

\underline{Proof of Eq.~\ref{eq:1}:} Let $\Gamma = \bigvee_{i \in I} $
$\bigwedge_{\gamma \in \mc{A}(\Psi_i)} \gamma$.  Let $\mf{A}$ be a
model of $V$ such that $\mf{A} \models \Phi$. Then $\mf{A} \models
\Psi_i$ for some $i \in I$. By
Lemma~\ref{lemma:formula-implies-its-approximations}, we have $\mf{A}
\models \mc{A}(\Psi_i)$, whence $\mf{A} \models \Gamma$. Thus $\Phi$
implies $\Gamma$ modulo $V$.  Towards the converse, let $\mf{A}$ be a
model of $V$ such that $\mf{A} \models \Gamma$. Then $\mf{A} \models
\mc{A}(\Psi_i)$ for some $i \in I$.  Let $\Psi = \bigwedge
\big(\foth(\mf{A}) \,\cup $ $ \{\Psi_i\} \big)$, where
$\foth(\mf{A})$ denotes the theory of $\mf{A}$.  It is easy to
see that $\mf{A} \models \mc{A}(\Psi)$ because any sentence $\gamma$
in $\mc{A}(\Psi)$ is given by either $\gamma = \bigwedge Z$ or $\gamma
= \gamma_i \wedge \bigwedge Z$, where $Z \subseteq_f
\foth(\mf{A})$ and $\gamma_i \in \mc{A}(\Psi_i)$.  Also observe
that $\Psi \in \left[\exists^* \bigwedge \right]^*\text{FO} $; then
since every sentence of $\mc{A}(\Psi)$ is satisfiable modulo $V$, it
follows from
Lemma~\ref{lemma:compactness-generalization-in-terms-of-approximations}
that $\Psi$ is satisfied in a model of $V$, say $\mf{B}$. Then (i)
$\mf{B} \equiv \mf{A}$ and (ii) $\mf{B} \models \Psi_i$ whence $\mf{B}
\models \Phi$.  Since $\Phi$ defines an elementary class modulo $V$,
we have $\mf{A} \models \Phi$.

\underline{Proof of Eq.~\ref{eq:3}:} We show the following result, call it
($\ddagger$): If $T, S$ and $V$ are FO theories such that $T$ implies
$\bigvee S$ modulo $V$, then $T$ implies $\bigvee S'$ modulo $V$ for
some finite subset $S'$ of $S$.  Then Eq.~\ref{eq:3} follows from
Eq.~\ref{eq:2} as follows. By Eq.~\ref{eq:2}, we have $\Phi$ implies
$ \bigvee_{i \in I} \gamma_i$ modulo $V$ for each sequence
$(\gamma_i)$ of $\mc{B}$ (recall that $\mc{B} = \prod_{i \in
I} \mc{A}(\Psi_i)$). Then by ($\ddagger$), $\Phi$ implies
$ \bigvee_{i \in I_1} \gamma_i$ modulo $V$ for some $I_1 \subseteq_f
I$. Defining $g((\gamma_i)) = I_1$, we get the forward direction of
Eq.~\ref{eq:3}. The backward direction of Eq.~\ref{eq:3} is trivial
from Eq.~\ref{eq:2} and the fact that $\bigvee_{i \in I_1} \gamma_i$
implies $\bigvee_{i \in I} \gamma_i$ (over all structures). We now
show ($\ddagger$).

Since $T$ implies $\bigvee S$ modulo $V$, we have that $T \cup
\{\neg \xi \mid \xi \in S\}$ is unsatisfiable modulo $V$. Then
by compactness theorem, $T \cup \{\neg \xi \mid \xi \in S'\}$ is
unsatisfiable modulo $V$, for some finite subset $S'$ of $S$. Whereby,
$T$ implies $ \bigvee S'$ modulo $V$.
\end{proof} 

The following figure gives the picture of the (partial) substructural
characterizations, under Hypothesis~\ref{hypothesis:H}, or equivalently,
Hypothesis~\ref{hypothesis:H-reformulated}. (cf.
Figure~\ref{figure:subst-chars}).

\begin{figure}[H]
\includegraphics[scale=0.7]{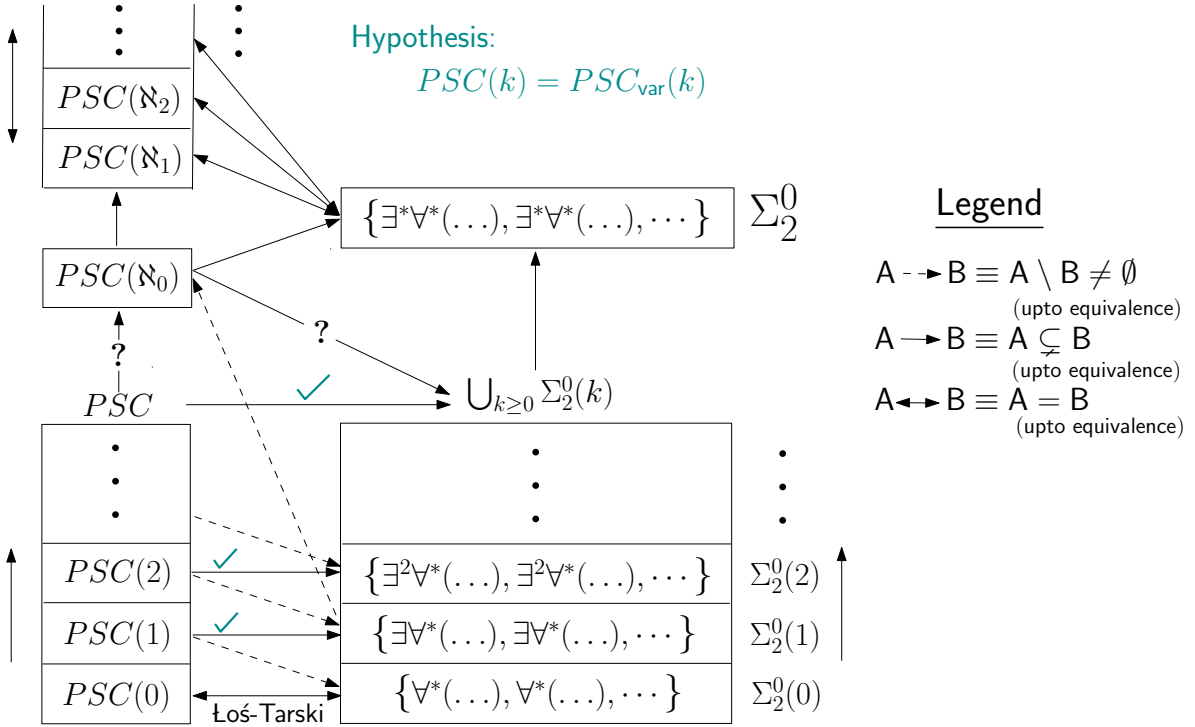}
\caption{(Partial) conditional characterizations of $PSC(k)$ and $PSC(\lambda)$ theories}
\label{figure:conditional-char}
\end{figure}

%% file: open-questions-classical-model-theory.tex
\chapter{Directions for future work}\label{chapter:open-questions-classical-model-theory}

We propose as a part of future work, various directions that naturally
arise from our results presented thus far.

\begin{enumerate}[leftmargin=*, nosep]
\item We would like to investigate what syntactic subclasses of $\fo$
  theories correspond exactly to $PSC(k)$ and $PSC(\aleph_0)$
  theories. As Theorem~\ref{theorem:subst-char-for-theories} shows,
  these syntactic classes must semantically be subclasses of
  $\Sigma^0_2$ theories. For $PSC(k)$ theories, in addition to
  verifying whether Hypothesis~\ref{hypothesis:H} is true, we would
  further like to investigate what syntactic subclass of theories of
  $\exists^k \forall^*$ sentences characterizes $PSC(k)$ theories,
  assuming Hypothesis~\ref{hypothesis:H} holds.  A technique to
  identify the latter syntactic subclass is to examine the syntactic
  properties of the $\fo$ theories given by
  Proposition~\ref{prop:char-of-inf-formulae-defining-elem-classes},
  and exploit the fact that these theories are obtained from the
  finite approximations of the infinitary sentences of
  $\left[\bigvee\right]\left[\exists^k \bigwedge\right] \Pi^0_1$.

\item As ``converses'' to the investigations above, and as analogues
  of the semantic characterizations of $\Pi^0_2$ theories and theories
  of $\forall^k \exists^*$ sentences by $PCE(\lambda)$ and $PCE(k)$
  respectively (cf. Theorem~\ref{theorem:ext-chars}), we would like to
  semantically characterize $\Sigma^0_2$ theories and theories of
  $\exists^k \forall^*$ sentences, in terms of properties akin to
  (though not the same as) $PSC(\lambda)$ and $PSC(k)$. Likewise, as
  an analogue of
  Proposition~\ref{prop:PCE(lambdas)-and-PCE}(\ref{prop:PCE(lambda)-more-general-than-PCE}),
  we would like to investigate if there are $PSC(\aleph_0)$ theories
  that are not equivalent to (i) any $PSC$ theory (ii) any theory of
  $\exists^k \forall^*$ sentences, for any $k \ge 0$.

\item It is conceivable that many semantic properties of FO theories
  have natural and intuitive descriptions/characterizations in
  infinitary logics
  (Lemma~\ref{lemma:inf-characterization-of-PSC-var(k)} gives one such
  example). Then, results like
  Proposition~\ref{prop:char-of-inf-formulae-defining-elem-classes}
  can be seen as ``compilers'' (in the sense of compilers used in
  computer science), in that they give a means of translating a ``high
  level'' description -- via infinitary sentences that are known to be
  equivalent to $\fo$ theories -- to an equivalent ``low level''
  description -- via $\fo$ theories. The latter $\fo$ theories are
  obtained from appropriately defined finite approximations of the
  infinitary sentences.  It would therefore be useful to investigate
  other infinitary logics and their fragments for which such
  ``compiler-results'' can be established.  An interesting logic to
  investigate in this regard would be $\mc{L}_{\omega_1, \omega}$,
  which is well-known to enjoy excellent model-theoretic properties
  despite compactness theorem not holding of
  it~\cite{keisler:infinitary-logic-book}.

\item Our results give characterizations of $\Sigma^0_2$ and $\Pi^0_2$
  sentences in which the number of quantifiers in the leading block is
  given. As natural generalizations of these results, we can ask for
  characterizations of $\Sigma^0_n$ and $\Pi^0_n$ sentences for each
  $n \ge 2$, where the numbers of quantifiers in all the $n$ blocks
  are given, and further extend these characterizations to
  theories. It may be noted that the results in the literature
  characterize $\Sigma^0_n$ and $\Pi^0_n$ theories as a whole and do
  not provide the finer characterizations suggested here.
\end{enumerate}

We conclude this part of the thesis by presenting our ideas (in
progress) on the last future work mentioned above, and suggesting
concrete directions for pursuing the latter.

\vspace{3pt}
\tbf{Directions for finer characterizations of $\Sigma^0_n$
  and $\Pi^0_n$}

For $n \ge 2$, let $\Sigma^0_n(k, l_1, *, l_2, *, \ldots)$ be the
class of all $\Sigma^0_n$ formulae in which the quantifier prefix is
such that the leading block of quantifiers has $k$ quantifiers, the
$(2i)^{\text{th}}$ block has $l_i$ quantifiers for $i \ge 1$, and the
$(2i+1)^{\text{th}}$ block has zero or more quantifiers for $i \ge
1$. Analogously, define the subclass $\Pi^0_n(k, l_1, *, l_2, *,
\ldots)$ of $\Pi^0_n$. For $n \ge 1$, let $\Sigma^0_n(l_1, *, l_2, *,
\ldots)$ denote the class of all formulae of $\Pi^0_{n+1}(0, l_1, *,
l_2, *, \ldots)$; likewise, let $\Pi^0_n(l_1, *, l_2, *, \ldots)$
denote the class of all formulae of $\Sigma^0_{n+1}(0, l_1, *, l_2, *,
\ldots)$.

Given a structure $\mf{A}$ and a $k$-tuple $\bar{a}$ of $\mf{A}$, the
\emph{$\Pi^0_n(l_1, *, l_2, *, \ldots)$-type} of $\bar{a}$ in $\mf{A}$
is the set of all $\Pi^0_n(l_1, *, l_2, *, \ldots)$ formulae having
free variables among $x_1, \ldots, x_k$, that are true of $\bar{a}$ in
$\mf{A}$.  We say a structure $\mf{B}$ \emph{realizes} the
$\Pi^0_n(l_1, *, l_2, *, \ldots)$-type of $\bar{a}$ in $\mf{A}$, if
there exists a $k$-tuple $\bar{b}$ of $\mf{B}$ such that the
$\Pi^0_n(l_1, *, l_2, *, \ldots)$-type of $\bar{b}$ in $\mf{B}$
contains (as a subset) the $\Pi^0_n(l_1, *, l_2, *, \ldots)$-type of
$\bar{a}$ in $\mf{A}$. We now present generalizations of the notions
of $k$-ary cover, $PSC(k)$ and $PCE(k)$.

\begin{futureworkdefn}[$(n; k, l_1, *, l_2, *, \ldots)$-ary cover] 
A collection $R$ of structures is said to be a $(n; k, l_1, *, l_2, *,
\ldots)$-ary cover of a structure $\mf{A}$ if for every $k$-tuple
$\bar{a}$ from $\mf{A}$, there exists a structure in $R$ that realizes
the $\Pi^0_{n-1}(l_1, *, l_2, *, \ldots)$-type of $\bar{a}$ in
$\mf{A}$.
\end{futureworkdefn}

\begin{futureworkremark}
Note that in the definition above, no structure in $R$ need be a
substructure of $\mf{A}$. This is in contrast with the notion of
$k$-ary cover as presented in Definition~\ref{defn:k-ary-covered-ext},
where if $R_1$ is a $k$-ary cover of a structure $\mf{A}_1$, then each
structure of $R_1$ is necessarily a substructure of $\mf{A}_1$.
\end{futureworkremark}

\begin{futureworkdefn}
Let $T$ and $V$ be given theories.
\begin{enumerate}[nosep]
\item 
We say $T$ is $PSC(n; k, l_1, *, l_2, *, \ldots)$ modulo $V$ if for
every model $\mf{A}$ of $T \cup V$, there exists a $k$-tuple $\bar{a}$
from $\mf{A}$ such that any model of $V$ that realizes the
$\Pi^0_{n-1}(l_1, *, l_2, *, \ldots)$-type of $\bar{a}$ in $\mf{A}$,
is also a model of $T$.
\item
We say $T$ is $PCE(n; k, l_1, *, l_2, *, \ldots)$ modulo $V$ if for
every model $\mf{A}$ of $V$ and every $(n; k, l_1, *, l_2, *,
\ldots)$-ary cover $R$ of $\mf{A}$, if each structure of $R$ is a
model of $T \cup V$, then $\mf{A}$ is a model of $T$.
\end{enumerate}
We say a sentence $\phi$ is $PSC(n; k, l_1, *, l_2, *, \ldots)$,
resp.\ $PCE(n; k, l_1, *, l_2, *, \ldots)$, modulo $V$, if the theory
$\{\phi\}$ is $PSC(n; k, l_1, *, l_2, *, \ldots)$, resp.\ $PCE(n; k,
l_1, *, l_2, *, \ldots)$, modulo $V$.
\end{futureworkdefn}

If $\phi(\bar{x})$ and $T(\bar{x})$ are respectively a formula and a
theory, each of whose free variables are among $\bar{x}$, then for a
theory $V$, the notions of `$\phi(\bar{x})$ is $PSC(n; k, l_1, *, l_2,
*, \ldots)$ modulo $V$', `$\phi(\bar{x})$ is $PCE(n; k, l_1, *, l_2,
*, \ldots)$ modulo $V$', `$T(\bar{x})$ is $PSC(n; k, l_1, *, l_2, *,
\ldots)$ modulo $V$' and `$T(\bar{x})$ is $PCE(n; k, l_1, *, l_2, *,
\ldots)$ modulo $V$' are defined similar to corresponding notions for
$PSC(k)$ and $PCE(k)$.

We can now show the following results analogous to
Lemma~\ref{lemma:PSC(k)-PCE(k)-duality} and
Theorem~\ref{theorem:glt(k)}.

\begin{futureworklemma}
Let $V$ be a given theory. A formula $\phi(\bar{x})$ is $PSC(n; k,
l_1, *, l_2, *, \ldots)$ modulo $V$ iff $\neg \phi(\bar{x})$ is
$PCE(n; k, l_1, *, l_2, *, \ldots)$ modulo $V$.
\end{futureworklemma}

\begin{futureworktheorem}\label{theorem:glt-for-sigma-n-and-pi-n}
Given a theory $V$, each of the following holds.
\begin{enumerate}[leftmargin=*,nosep]
\item A formula $\phi(\bar{x})$ is $PSC(n; k, l_1, *, l_2, *, \ldots)$
  modulo $V$ iff $\phi(\bar{x})$ is equivalent modulo $V$ to a finite
  disjunction of $\Sigma^0_n(k, l_1, *, l_2, *, \ldots)$ formulae, all
  of whose free variables are among $\bar{x}$.
\item A formula $\phi(\bar{x})$ is $PCE(n; k, l_1, *, l_2, *, \ldots)$
  modulo $V$ iff $\phi(\bar{x})$ is equivalent modulo $V$ to a finite
  conjunction of $\Pi^0_n(k, l_1, *, l_2, *, \ldots)$ formulae, all of
  whose free variables are among $\bar{x}$.
\item A theory $T(\bar{x})$ is $PCE(n; k, l_1, *, l_2, *, \ldots)$
  modulo $V$ iff $T(\bar{x})$ is equivalent modulo $V$ to a theory of
  $\Pi^0_n(k, l_1, *, l_2, *, \ldots)$ formulae, all of whose free
  variables are among $\bar{x}$.
\end{enumerate}
\end{futureworktheorem}

The above theorem answers in part, the question raised in the last
future work mentioned above.  It also gives new semantic
characterizations of $\Sigma^0_n$ and $\Pi^0_n$ sentences, via the
properties $PSC_n = \bigcup_{k, l_1, l_2, \ldots \in \mathbb{N}}
PSC(n; k, l_1, *, l_2, *, \ldots)$ and $PCE_n = \bigcup_{k, l_1, l_2,
  \ldots \in \mathbb{N}} PCE(n; k, l_1, *, l_2, *, \ldots)$ \linebreak
respectively.

A natural direction for future work that is suggested by
Theorem~\ref{theorem:glt-for-sigma-n-and-pi-n} is the investigation of
suitable variants of $PSC(n; k, l_1, *, l_2, *, \ldots)$ and $PCE(n;
k, l_1, *, l_2, *, \ldots)$ that respectively characterize
$\Sigma^0_n(k, l_1, *, l_2, *, \ldots)$ and $\Pi^0_n(k, l_1, *, l_2,
*, \ldots)$ formulae exactly; likewise, an investigation of whether
$PSC(n; k, l_1, *, l_2, *, \ldots)$ or some suitable variant of it,
characterizes theories of $\Sigma^0_n(k, l_1, *, l_2, *, \ldots)$
formulae. If these characterizations are not obtained in general, then
we would like to get them at least under plausible hypotheses
(cf. Hypothesis~\ref{hypothesis:H}).  Of course, the question of
characterizing $\Sigma^0_n$ and $\Pi^0_n$ sentences and theories in
which the numbers of quantifiers in all blocks are given, still
remains largely. However, observe that for the case of $n = 2$, the
properties $PSC(n; k, l)$, resp.\ $PCE(n; k, l)$, do give semantic
characterizations of finite disjunctions of $\Sigma^0_2$ sentences,
resp.\ finite conjunctions of $\Pi^0_2$ sentences, in which the
numbers of quantifiers in the first and second blocks are given
numbers $k$ and $l$ respectively. This, in a sense, gives a finer
characterization of $\Sigma^0_2$ and $\Pi^0_2$ than the one given by
Theorem~\ref{theorem:glt(k)}. But, we note that $PSC(n; k, l)$ and
$PCE(n; k, l)$ are not combinatorial in nature.  We would therefore
like to investigate whether these notions, and more generally, the
notions of $PSC(n; k, l_1, *, l_2, *, \ldots)$ and $PCE(n; k, l_1, *,
l_2, *, \ldots)$ have combinatorial, or even purely algebraic
equivalents, so that the ``syntactical flavour'' currently in the
definitions of these notions, is eliminated.

%% file: finite-model-theory-results.tex
\input{background-FMT}

\input{need-to-investigate-new-classes}

\input{lebsp}

\input{classes-satisfying-lebsp}

\input{additional-properties-of-lebsp}

\input{open-questions-finite-model-theory}

%% file: background-FMT.tex
\chapter{Background and preliminaries}\label{chapter:background-FMT}

In this part of the thesis, we consider only finite structures over
finite vocabularies $\tau$ that are relational, i.e. vocabularies that
do not contain any constant or function symbols, unless explicitly
stated otherwise. All the classes of structures that we consider are
thus classes of finite relational structures. We denote classes of
structures by $\cl{S}$ possibly with numbers as subscripts, and assume
these to be \emph{closed under isomorphisms}.  We consider two logics
in this part of the thesis, one $\fo$, and the other, an extension of
it called \emph{monadic second order} logic, denoted $\mso$. We use
the notation $\mc{L}$ to mean either $\fo$ or $\mso$. Any \emph{notion
  or result stated for $\mc{L}$ means that the notion or result is
  stated for both $\fo$ and $\mso$}.  The classic references for the
background from finite model theory presented in this chapter
are~\cite{libkin,gradel-vardi-et-al,ebbinghaus-flum}.

\section{Syntax and semantics of $\mso$}

\tbf{Syntax}: The syntax of $\mso$\sindex[term]{monadic second order
  logic} extends that of $\fo$ by using $\mso$ variables that range
over subsets of the universes of structures, and using quantification
(existential and universal) over these variables. We denote $\mso$
variables using the capital letter $X$, possibly with numbers as
subscripts. A sequence of $\mso$ variables is denoted as
$\bar{X}$. For a vocabulary $\tau$, the notions of $\mso$ terms and
$\mso$ formulae, and their free variables, are defined as follows.
\begin{enumerate}[leftmargin=*,nosep]
\item An $\mso$ term\sindex[term]{term} over $\tau$ is an $\fo$ term
  over $\tau$, i.e. either a constant or an $\fo$ variable. A term
  that is a variable $x$ has only one free variable, namely $x$. A
  constant has no free variables.\sindex[term]{free variable}
\item An atomic $\mso$ formula\sindex[term]{atomic formula} over
  $\tau$ is either of the following.
\begin{itemize}[nosep]
\item An atomic $\fo$ formula over $\tau$, i.e. either $t_1 = t_2$ or
  $R(t_1, \ldots, t_n)$, where $t_1, \ldots, t_n$ are $\mso$ terms
  over $\tau$, and $R$ is a relation symbol of $\tau$ of arity
  $n$. The free variables of these are all $\fo$ variables, and are as
  defined in
  Section~\ref{section:background:FO-syntax-semantics}. These formulae
  have no free $\mso$ variables.
\item $X(t)$ where $X$ is an $\mso$ variable and $t$ is an $\mso$ term
  over $\tau$. This formula has at most one free $\fo$ variable,
  namely the free variable of $t$ (if any), and has exactly one free
  $\mso$ variable, namely $X$.
\end{itemize} 
\item Boolean combinations of $\mso$ formulae using the boolean
  connectives $\wedge, \vee$ and $\neg$ are $\mso$ formulae. The free
  variables of such formulae are defined analogously to those of
  boolean combinations of $\fo$ formulae (see
  Section~\ref{section:background:FO-syntax-semantics}).
\item Given an $\mso$ formula $\varphi$, the formulae $\exists x
  \varphi, \forall x \varphi, \exists X \varphi$ and $\forall X
  \varphi$ are all $\mso$ formulae. The free $\fo$ variables of
  $\exists x \varphi$ and $\forall x \varphi$ are the free $\fo$
  variables of $\varphi$, except for $x$. The free $\mso$ variables of
  these formulae are exactly those of $\varphi$. The free $\fo$
  variables of $\exists X \varphi$ and $\forall X \varphi$ are exactly
  the free $\fo$ variables of $\varphi$, while the free $\mso$
  variables of these formulae are exactly those of $\varphi$, except
  for $X$.\sindex[term]{free variable}
\end{enumerate} 
We denote an $\mso$ formula $\varphi$ with free $\fo$ variables among
$\bar{x}$, and free $\mso$ variables among $\bar{X}$ as
$\varphi(\bar{x}, \bar{X})$. For $\varphi$ as just mentioned, if
$\bar{X}$ is empty, i.e. if $\varphi$ has no free $\mso$ variables,
then we denote $\varphi$ as $\varphi(\bar{x})$.  Like $\fo$ formulae,
an $\mso$ formula with no free variables is called a
\emph{sentence}\sindex[term]{sentence}, and an $\mso$ formula with no
quantifiers is called
\emph{quantifier-free}\sindex[term]{quantifier-free}. Again, like for
$\fo$ formulae, we denote $\mso$ formulae using the Greek letters
$\phi, \varphi, \psi, \chi, \xi, \gamma, \alpha$ or
$\beta$.\sindex[symb]{$\alpha, \beta, \gamma, \xi, \phi, \varphi,
  \chi, \psi$}

Before looking at the semantics, we define the important notion of
\emph{quantifier-rank}, or simply \emph{rank}\sindex[term]{rank}, of
an $\mso$ formula $\varphi$, denoted $\rank{\varphi}$. The definition
is by structural induction.
\begin{enumerate}[nosep]
\item If $\varphi$ is quantifier-free, then $\rank{\varphi} = 0$.
\item If $\varphi = \varphi_1 \wedge \varphi_2$, then $\rank{\varphi}
  = \text{max}(\rank{\varphi_1}, \rank{\varphi_2})$. The same holds if
  $\varphi = \varphi_1 \vee \varphi_2$.
\item If $\varphi = \neg \varphi_1$, then $\rank{\varphi} =
  \rank{\varphi_1}$.
\item If $\varphi$ is any of $\exists x \varphi_1$, $\forall x
  \varphi_1$, $\exists X \varphi_1$ or $\forall X \varphi_1$, then
  $\rank{\varphi} = 1 + \rank{\varphi_1}$.
\end{enumerate}

The above definition also defines the rank of an $\fo$ formula, since
every $\fo$ formula is also an $\mso$ formula.

\begin{remark}
The $\mso$ formulae $\varphi$ that we consider in this part of the
thesis always have \emph{only $\fo$ free variables}, and \emph{no
  $\mso$ free variables} (although of course, the $\mso$ formulae that
build up $\varphi$ surely would contain free $\mso$
variables). 
\end{remark}

Since the semantics of $\mso$ is defined inductively, we consider,
only for the purposes of defining the semantics, $\mso$ formulae with
free $\mso$ variables, in addition to free $\fo$ variables.

\tbf{Semantics}: Given a $\tau$-structure $\mf{A}$ and an $\mso$
formula $\varphi(\bar{x}, \bar{X})$, we define the notion of truth of
$\varphi(\bar{x}, \bar{X})$ for a given assignment $\bar{a}$ of
elements of $\mf{A}$, to $\bar{x}$ and a given assignment $\bar{A}$ of
subsets of elements of $\mf{A}$, to $\bar{X}$. We denote by $(\mf{A},
\bar{a}, \bar{A}) \models \varphi(\bar{x},
\bar{X})$,\sindex[symb]{$\models$}\sindex[term]{truth in a structure}
that $\varphi(\bar{x}, \bar{X})$ is true in $\mf{A}$ for the
assignments $\bar{a}$ to $\bar{x}$ and $\bar{A}$ to $\bar{X}$, and
call $(\mf{A}, \bar{a}, \bar{A})$ a model of $\varphi(\bar{x},
\bar{X})$\sindex[term]{model}. We give the semantics only for the
syntactic features of $\mso$ that are different from those of
$\fo$. Below, $\bar{X} = (X_1, \ldots, X_n)$ and $\bar{A} = (A_1,
\ldots, A_n)$.

\begin{itemize}[nosep]
\item If $\varphi(\bar{x}, \bar{X})$ is the formula $X_i(t)$, then
  $(\mf{A}, \bar{a}, \bar{A}) \models \varphi(\bar{x}, \bar{X})$ iff
  $t^{\mf{A}}(\bar{a}) \in A_i$.
\item If $\varphi(\bar{x}, \bar{X})$ is the formula $\exists X_{n+1}
  \varphi_1(\bar{x}, \bar{X}, X_{n+1})$, then $(\mf{A}, \bar{a},
  \bar{A}) \models \varphi(\bar{x}, \bar{X})$ iff there exists
  $A_{n+1} \subseteq \univ{\mf{A}}$ such that $(\mf{A}, \bar{a},
  \bar{A}, A_{n+1}) \models \varphi(\bar{x}, \bar{X}, X_{n+1})$.
\item If $\varphi(\bar{x}, \bar{X})$ is the formula $\forall X_{n+1}
  \varphi_1(\bar{x}, \bar{X}, X_{n+1})$, then $(\mf{A}, \bar{a},
  \bar{A}) \models \varphi(\bar{x}, \bar{X})$ iff for all $A_{n+1}
  \subseteq \univ{\mf{A}}$, it is the case that $(\mf{A}, \bar{a},
  \bar{A}, A_{n+1}) \models \varphi(\bar{x}, \bar{X}, X_{n+1})$.
\end{itemize}

If $\varphi(\bar{a})$ is an $\mso$ formula with no $\mso$ variables,
then we denote the truth of $\varphi(\bar{x})$ in $\mf{A}$ for an
assignment $\bar{a}$ to $\bar{x}$ as $(\mf{A}, \bar{a}) \models
\varphi(\bar{x})$. For a sentence $\varphi$, we denote the truth of
$\varphi$ in $\mf{A}$ as $\mf{A} \models \varphi$.

\section[Adaptations of classical model theory concepts to finite model theory]{Adaptations of classical model theory concepts to the \\finite model theory setting}\label{section:background-FMT-adaptation-from-CMT-to-FMT}

We assume familiarity with the notions introduced in
Chapter~\ref{chapter:background-CMT}.  We adapt some of these notions
to versions of these \emph{over} classes of structures. For other
notions, they are exactly as defined in
Chapter~\ref{chapter:background-CMT}.

\vspace{3pt}
\tbf{1. Consistency, validity, entailment, and equivalence \emph{over}
  classes of structures}

Given a non-empty class $\cl{S}$ of structures, we say a formula
$\varphi(\bar{x})$ is \emph{satisfiable over
  $\cl{S}$}\sindex[term]{satisfiable} if there exists structure
$\mf{A} \in \cl{S}$ and a tuple $\bar{a}$ of $\mf{A}$ such that
$|\bar{a}| = |\bar{x}|$ and $(\mf{A}, \bar{a}) \models
\varphi(\bar{x})$. We say $\varphi(\bar{x})$ is \emph{unsatisfiable
  over $\cl{S}$} if $\varphi(\bar{x})$ is not satisfiable over
$\cl{S}$.  We say $\psi(\bar{x})$ \emph{entails $\varphi(\bar{x})$
  over $\cl{S}$}\sindex[term]{entails} if any model $(\mf{A},
\bar{a})$ of $\psi(\bar{x})$ such that $\mf{A} \in \cl{S}$ is also a
model of $\varphi(\bar{x})$.  We say $\psi(\bar{x})$ and
$\varphi(\bar{x})$ are equivalent over $\cl{S}$, or simply
\emph{$\cl{S}$-equivalent},\sindex[term]{$\cl{S}$-equivalent} if
$\psi(\bar{x})$ entails $\varphi(\bar{x})$ over $\cl{S}$, and
vice-versa.

\vspace{3pt} \tbf{2. $(m, \mc{L})$-types}: ~ Given $m \in \mathbb{N}$,
a $\tau$-structure $\mf{A}$ and a $k$-tuple $\bar{a}$ from $\mf{A}$,
the \emph{$(m, \mc{L})$-type\sindex[term]{$(m, \mc{L})$-type} of
  $\bar{a}$ in $\mf{A}$}, denoted $\ltp{\mf{A}}{\bar{a}}{m}(x_1,
\ldots, x_k)$,\sindex[symb]{$\ltp{\mf{A}}{\bar{a}}{m}(x_1, \ldots,
  x_k)$} is the set of all $\mc{L}(\tau)$ formulae of rank at most
$m$, whose free variables are among $x_1, \ldots, x_k$, and that are
true of $\bar{a}$ in $\mf{A}$. We denote by
$\lth{m}{\mf{A}}$,\sindex[symb]{$\lth{m}{\mf{A}}$} the set of all
$\mc{L}(\tau)$ sentences of rank at most $m$ that are true in
$\mf{A}$.  Given a $\tau$-structure $\mathfrak{B}$, we say that
$\mf{A}$ and $\mf{B}$ are \emph{$(m, \mc{L})$-equivalent}, denoted
$\mathfrak{A} \lequiv{m} \mathfrak{B}$
\sindex[symb]{$\lequiv{m}, \equiv_m$} if $\lth{m}{\mathfrak{A}} =
\lth{m}{\mf{B}}$. Observe that if $\bar{a}$ and $\bar{b}$ are
$k$-tuples from $\mf{A}$ and $\mf{B}$ respectively, then
$\ltp{\mf{A}}{\bar{a}}{m}(x_1, \ldots, x_k) =
\ltp{\mf{B}}{\bar{b}}{m}(x_1, \ldots, x_k)$ iff $(\mathfrak{A},
\bar{a}) \lequiv{m} (\mf{B}, \bar{b})$. If $\mc{L} = \fo$, then we
also denote $\fequiv{m}$ simply as $\equiv_m$,
\sindex[symb]{$\lequiv{m}, \equiv_m$} following standard notation in
the literature. It is easy to see that $\lequiv{m}$ is an equivalence
relation over all structures.  Following is an important result
concerning the $\lequiv{m}$ relation.

\begin{proposition}[Proposition 7.5, ref.~\cite{libkin}]\label{prop:finite-index-of-lequiv(m)-relation}
There exists a computable function $f: \mathbb{N} \rightarrow
\mathbb{N}$ such that for each $m \in \mathbb{N}$, the index of the
$\lequiv{m}$ relation is at most $f(m)$.
\end{proposition}

Given a class $\cl{S}$ of structures and $m \in \mathbb{N}$, we let
$\Delta_{\mc{L}}(m, \cl{S})$\sindex[symb]{$\Delta_{\mc{L}}(m,
  \cl{S})$} denote the set of all equivalence classes of the
$\lequiv{m}$ relation restricted to the structures in $\cl{S}$.  We
denote by $\Lambda_{\cl{S}, \mc{L}}: \mathbb{N} \rightarrow
\mathbb{N}$ \sindex[symb]{$\Lambda_{\cl{S}, \mc{L}}$} a fixed
computable function with the property that $\Lambda_{\cl{S},
  \mc{L}}(m) \ge |\Delta_{\mc{L}}(m, \cl{S})|$. The existence of
$\Lambda_{\cl{S}, \mc{L}}$ is guaranteed by
Proposition~\ref{prop:finite-index-of-lequiv(m)-relation}.

The notion of $\lequiv{m}$ has a characterization in terms of
\emph{Ehrenfeucht-Fr{\"a}iss{\'e}} games for $\mc{L}$. We describe
these in the next section.


\section{$\mc{L}$-Ehrenfeucht-Fr\"aiss\'e games}

We first define the notion of \emph{partial isomorphism} that is
crucially used in the definition of the games we present below. We
assume for this section, that the vocabulary $\tau$ possibly contains
constant symbols, in addition to relation symbols. Let $\mf{A},
\mf{B}$ be two $\tau$-structures, and $\bar{a}=(a_1, \ldots, a_m)$ and
$\bar{b} = (b_1, \ldots, b_m)$ be two $m$-tuples from $\mf{A}$ and
$\mf{B}$ respectively. Then $(\bar{a}, \bar{b})$ defines a partial
isomorphism between $\mf{A}$ and $\mf{B}$ if the following are true.

\begin{enumerate}[nosep]
\item For $1\leq i, j \leq m$, we have $a_i = a_j$ iff $b_i = b_j$.
\item For every constant symbol $c \in \tau$ and $1 \leq i \leq m$, we
  have $a_i = c^{\mf{A}}$ iff $b_i = c^{\mf{B}}$.
\item For every $r$-ary relation symbol $R \in \tau$ and every
  sequence $(i_1, \ldots, i_r)$ of numbers from $\{1, \ldots, m\}$, we
  have $(a_{i_1}, \ldots a_{i_r}) \in R^{\mf{A}}$ iff $(b_{i_1},
  \ldots b_{i_r}) \in R^{\mf{B}}$.
\end{enumerate}

If $\tau$ contains no constant symbols, then the map $a_i \mapsto b_i$
for $1 \leq i \leq n$ is an isomorphism from the substructure of
$\mf{A}$ induced by $\{a_1, \ldots, a_r\}$ to the substructure of
$\mf{B}$ induced by $\{b_1, \ldots, b_r\}$.

\tbf{The $\fo$-Ehrenfeucht-Fr\"aiss\'e
  game}\sindex[term]{Ehrenfeucht-Fr\"aiss\'e game}

The $\fo$-Ehrenfeucht-Fr\"aiss\'e game, or simply the $\fef$ game, is
played on two given structures $\mf{A}$ and $\mf{B}$, and by two
players, called \emph{spoiler} and \emph{duplicator}.  The spoiler
tries to show that the two structures are non-isomorphic, while the
duplicator tries to show otherwise. The $\fef$ game of $m$ rounds
between $\mf{A}$ and $\mf{B}$ is as defined below. Each round consists
of the following steps.

\begin{enumerate}[nosep]
\item The spoiler chooses one of the structures and picks an element
  from it.
\item The duplicator responds by picking an element of the other structure.
\end{enumerate}

At the end of $m$ rounds, let $\bar{a} = (a_1, \ldots, a_m)$ be the
elements chosen from $\mf{A}$ and let $\bar{b} = (b_1, \ldots, b_m)$
be the elements chosen from $\mf{B}$. Let $c_1, \ldots, c_p$ be the
constant symbols of $\tau$, and let $\bar{c}^{\mf{A}} = (c_1^{\mf{A}},
\ldots, c_p^{\mf{A}})$ and $\bar{c}^{\mf{B}} = (c_1^{\mf{B}}, \ldots,
c_p^{\mf{B}})$. Call $(\bar{a}, \bar{b})$ as a \emph{play} of $m$
rounds of the $\fef$ game on $\mf{A}$ and $\mf{B}$. The duplicator is
said to \emph{win} the play $(\bar{a}, \bar{b})$ iff $((\bar{a},
\bar{c}^{\mf{A}}), (\bar{b}, \bar{c}^{\mf{B}}))$ is a partial
isomorphism between $\mf{A}$ and $\mf{B}$. An \emph{$m$-round strategy
  for the duplicator in the $\fef$ game} on $\mf{A}$ and $\mf{B}$ is a
function $\tbf{S}: \bigcup_{i = 0}^{i = m-1} \big( (\univ{\mf{A}}
\times \univ{\mf{B}})^i \times (\univ{\mf{A}} \cup \univ{\mf{B}})
\big) \rightarrow (\univ{\mf{A}} \cup \univ{\mf{B}})$ such that for
all $i \in \{0, \ldots, m-1\}$, for all $(a_1, b_1), \ldots, (a_i,b_i)
\in (\univ{\mf{A}} \times \univ{\mf{B}})$ and all $d \in
(\univ{\mf{A}} \cup \univ{\mf{B}})$, it is the case that
$\tbf{S}((a_1, b_1), \ldots, (a_i, b_i), d) \in \univ{\mf{A}}$ iff $d
\in \univ{\mf{B}}$.  The duplicator is said to have a \emph{winning
  strategy} in the $m$-round $\fef$ game on $\mf{A}$ and $\mf{B}$ if
there exists an $m$-round strategy $\tbf{S}$ for the duplicator in the
$\fef$ game on $\mf{A}$ and $\mf{B}$, such that the duplicator wins
every play of $m$ rounds of the $\fef$ game on $\mf{A}$ and $\mf{B}$,
in which the duplicator responds in accordance with $\tbf{S}$.

The following theorem shows that the existence of a winning strategy
for the duplicator in the $\fef$ game of $m$-rounds on $\mf{A}$ and
$\mf{B}$ characterizes $(m, \fo)$-equivalence of $\mf{A}$ and $\mf{B}$
(see Theorem 3.9 of~\cite{libkin}).

\begin{theorem}[Ehrenfeucht-Fr\"aiss\'e]\label{theorem:foef-theorem}
Let $\mf{A}$ and $\mf{B}$ be two structures over a vocabulary that
possibly contains constant symbols. Let $\bar{a}$ and $\bar{b}$ be
given tuples of elements from $\mf{A}$ and $\mf{B}$ respectively. Then
$(\mf{A}, \bar{a}) \fequiv{m} (\mf{B}, \bar{b})$ iff the duplicator
has a winning strategy in the $m$-round $\fef$ game on $(\mf{A},
\bar{a})$ and $(\mf{B}, \bar{b})$.
\end{theorem}

\tbf{The $\mso$-Ehrenfeucht-Fr\"aiss\'e
  game}\sindex[term]{Ehrenfeucht-Fr\"aiss\'e game}

The $\mso$-Ehrenfeucht-Fr\"aiss\'e game, or simply the $\mef$ game, is
similar to the $\fef$ game. It is played on two given structures
$\mf{A}$ and $\mf{B}$, and by two players, namely the spoiler and
duplicator. The difference with the $\fef$ game is that in an $\mef$
game of $m$ rounds on $\mf{A}$ and $\mf{B}$, in each round there are
two kinds of moves.

\begin{enumerate}[nosep]
\item Point move: This is like in the $\fef$ game on $\mf{A}$ and
  $\mf{B}$. The spoiler chooses one of the structures and picks an
  element from it. The duplicator responds by picking an element of
  the other structure.
\item Set move: The spoiler chooses a subset of elements from one of
  the structures.  The duplicator responds by picking a subset of
  elements of the other structure.
\end{enumerate}

At the end of $m$ rounds, let $\bar{a} = (a_1, \ldots, a_p)$ be the
elements chosen from $\mf{A}$ and let $\bar{b} = (b_1, \ldots, b_p)$
be the elements chosen from $\mf{B}$. Likewise, let $\bar{A} = (A_1,
\ldots, A_r)$ be the sets chosen from $\mf{A}$ and $\bar{B} = (B_1,
\ldots, B_r)$ be the sets chosen from $\mf{B}$ such that $p + r = m$.
Let $c_1, \ldots, c_s$ be the constant symbols of $\tau$, and let
$\bar{c}^{\mf{A}} = (c_1^{\mf{A}}, \ldots, c_s^{\mf{A}})$ and
$\bar{c}^{\mf{B}} = (c_1^{\mf{B}}, \ldots, c_s^{\mf{B}})$. Call
$((\bar{a}, \bar{A}), (\bar{b}, \bar{B}))$ as a play of $m$ rounds of
the $\mef$ game on $\mf{A}$ and $\mf{B}$.  The duplicator wins the
play $((\bar{a}, \bar{A}), (\bar{b}, \bar{B}))$ if $((\bar{a},
\bar{c}^{\mf{A}}), (\bar{b}, \bar{c}^{\mf{B}})$ is a partial
isomorphism between $(\mf{A}, \bar{A})$ and $(\mf{B}, \bar{B})$.  Note
that the latter implies $a_i \in A_j$ iff $b_i \in B_j$, for $i \in
\{1, \ldots, p\}$ and $j \in \{1, \ldots, r\}$.  We now define the
notion of a strategy for the duplicator in an $\mef$ game analogous to
that for an $\fef$ game. For a set $X$, let $2^X$ denote the powerset
of $X$. An \emph{$m$-round strategy for the duplicator in the $\mef$
  game} on $\mf{A}$ and $\mf{B}$ is a function $\tbf{S}: \bigcup_{i =
  0}^{i = m-1} \big( \big( (\univ{\mf{A}} \cup 2^{\univ{\mf{A}}})
\times (\univ{\mf{B}} \cup 2^{\univ{\mf{B}}}) \big)^i \times
(\univ{\mf{A}} \cup 2^{\univ{\mf{A}}} \cup \univ{\mf{B}} \cup
2^{\univ{\mf{B}}}) \big) \rightarrow (\univ{\mf{A}} \cup
2^{\univ{\mf{A}}} \cup \univ{\mf{B}} \cup 2^{\univ{\mf{B}}})$ such
that for all $i \in \{0, \ldots, m-1\}$, for all $(d_1, e_1), \ldots,
(d_i, e_i) \in \big( (\univ{\mf{A}} \cup 2^{\univ{\mf{A}}}) \times
(\univ{\mf{B}} \cup 2^{\univ{\mf{B}}}) \big)$ and all $d \in
(\univ{\mf{A}} \cup 2^{\univ{\mf{A}}} \cup \univ{\mf{B}} \cup
2^{\univ{\mf{B}}})$, it is the case that (i) $\tbf{S}((d_1, e_1),
\ldots, (d_i, e_i), d) \in (\univ{\mf{A}} \cup \univ{\mf{B}})$ iff $d
\in (\univ{\mf{A}} \cup \univ{\mf{B}})$, (ii) $\tbf{S}((d_1, e_1),
\ldots, (d_i, e_i), d) \in \univ{\mf{A}}$ iff $d \in \univ{\mf{B}}$,
and (ii) $\tbf{S}((d_1, e_1), \ldots, (d_i, e_i), d) \in
2^{\univ{\mf{A}}}$ iff $d \in 2^{\univ{\mf{B}}}$.  The duplicator is
said to have a \emph{winning strategy} in the $m$-round $\mef$ game on
$\mf{A}$ and $\mf{B}$ if there exists an $m$-round strategy $\tbf{S}$
for the duplicator in the $\mef$ game on $\mf{A}$ and $\mf{B}$ such
that the duplicator wins every play of $m$ rounds of the $\mef$ game
on $\mf{A}$ and $\mf{B}$, when the duplicator responds in accordance
with $\tbf{S}$.

Like Theorem~\ref{theorem:foef-theorem}, the following theorem shows
that the existence of a winning strategy for the duplicator in the
$\mef$ game of $m$-rounds on $\mf{A}$ and $\mf{B}$ characterizes $(m,
\mso)$-equivalence of $\mf{A}$ and $\mf{B}$. This theorem is stated
more generally in Theorem 7.7 of~\cite{libkin}. However, we need the
theorem only as stated below.

\begin{theorem}\label{theorem:msoef-theorem}
Let $\mf{A}$ and $\mf{B}$ be two structures over a vocabulary that
possibly contains constant symbols. Let $\bar{a}$ and $\bar{b}$ be
given tuples of elements from $\mf{A}$ and $\mf{B}$ respectively.
Then $(\mf{A}, \bar{a}) \mequiv{m} (\mf{B}, \bar{b})$ iff the
duplicator has a winning strategy in the $m$-round $\mef$ game on
$(\mf{A}, \bar{a})$ and $(\mf{B}, \bar{b})$.
\end{theorem}

\section{Translation schemes}\label{section:background-FMT-translation-schemes}

We recall the notion of translation schemes from the
literature~\cite{makowsky}. These were first introduced in the context
of classical model theory, and are known in the literature by
different names, like \emph{$\fo$ interpretations},
\emph{transductions}, etc. We define these below and look at some of
their properties subsequently.

Let $\tau$ and $\sigma$ be given vocabularies, and $t \ge 1$ be a
natural number. Let $\bar{x}_0$ be a fixed $t$-tuple of first order
variables, and for each relation $R \in \sigma$ of arity $\#R$, let
$\bar{x}_R$ be a fixed $(t \times \#R)$-tuple of first order
variables.  A \emph{$(t, \tau, \sigma, \mc{L})$-translation
  scheme}\sindex[term]{translation scheme} $\Xi = (\xi, (\xi_R)_{R \in
  \sigma})$\sindex[symb]{$\Xi$} is a sequence of formulas of
$\mc{L}(\tau)$ such that the free variables of $\xi$ are among those
in $\bar{x}_0$, and for $R \in \sigma$, the free variables of $\xi_R$
are among those in $\bar{x}_R$.  When $t, \sigma$ and $\tau$ are clear
from context, we call $\Xi$ simply as a translation scheme. We call
$t$ as the \emph{dimension}\sindex[term]{translation scheme!dimension}
of $\Xi$.  If $t = 1$, we say $\Xi$ is a \emph{scalar} translation
scheme, and if $t \ge 2$, we say $\Xi$ is a \emph{vectorized}
translation scheme.  In our results in the subsequent chapters, we
consider vectorized translation schemes when $\mc{L} = \fo$, and
scalar translation schemes when $\mc{L} = \mso$.

One can associate with a translation scheme $\Xi$, two partial maps:
(i) $\Xi^*$ from $\tau$-structures to $\sigma$-structures (ii)
$\Xi^\sharp$ from $\mc{L}(\sigma)$ formulae to $\mc{L}(\tau)$
formulae, each of which we define below.  For the ease of readability,
we abuse notation slightly and use $\Xi$\sindex[symb]{$\Xi$} to denote
both $\Xi^*$ and $\Xi^\sharp$.

1. Given a $\tau$-structure $\mf{A}$, the $\sigma$-structure
$\Xi(\mf{A})$ is defined as follows. 


\begin{enumerate}[nosep]
\item $\univ{\Xi(\mf{A})} = \{ \bar{a} \in \univ{\mf{A}}^t \mid
  (\mf{A}, \bar{a}) \models \xi(\bar{x}_0)\}$.
\item For each $R \in \sigma$ of arity $n$, the interpretation of $R$
  in $\Xi(\mf{A})$ is the set $\{ \bar{a} \in \univ{\mf{A}}^{t \times
    n} \mid (\mf{A}, \bar{a}) \models \xi_R(\bar{x}_R) \}$.
\end{enumerate}

2. We define the map $\Xi$ from $\mc{L}(\sigma)$ formulae to
$\mc{L}(\tau)$ formulae. We first define this map for the case of a
$(t, \tau, \sigma, \fo)$-translation scheme.  Given a $\fo(\sigma)$
formula $\varphi(\bar{x})$ where $\bar{x} = (x_1, \ldots, x_n)$, the
$\fo(\tau)$ formula $\Xi(\varphi)(\bar{x}_1, \ldots, \bar{x}_n)$ where
$\bar{x}_i = (x_{i, 1}, \ldots, x_{i, t})$ for $1 \leq i \leq n$, is
as defined below.

\begin{enumerate}[nosep]
\item If $\varphi(\bar{x})$ is the formula $R(x_1, \ldots, x_r)$ for
  an $r$-ary relation symbol $R \in \sigma$, then
\[
\Xi(\varphi)(\bar{x}_1, \ldots, \bar{x}_n) = \xi_R(\bar{x}_1,
\ldots, \bar{x}_n) \wedge \bigwedge_{i = 1}^{i = r} \xi(x_{i, 1},
\ldots, x_{i, t})
\]
\item If $\varphi(\bar{x})$ is the formula $x_1 = x_2$, then 
\[
\Xi(\varphi)(\bar{x}_1, \ldots, \bar{x}_n) = \bigwedge_{j =
  1}^{j = t} (x_{1, j} = x_{2, j}) \wedge \bigwedge_{i = 1}^{i = 2}
\xi(x_{i, 1}, \ldots, x_{i, t})
\]
\item If $\varphi(\bar{x}) = \varphi_1(\bar{x}) \wedge
  \varphi_2(\bar{x})$,  then 
\[
\Xi(\varphi)(\bar{x}_1, \ldots, \bar{x}_n) =
\Xi(\varphi_1)(\bar{x}_1, \ldots, \bar{x}_n) \wedge
\Xi(\varphi_2)(\bar{x}_1, \ldots, \bar{x}_n)
\]
The same holds with $\wedge$ replaced with $\vee$.
\item If $\varphi(\bar{x}) = \neg \varphi_1(\bar{x})$, then
\[
\Xi(\varphi)(\bar{x}_1, \ldots, \bar{x}_n) = \neg 
\Xi(\varphi_1)(\bar{x}_1, \ldots, \bar{x}_n) 
\]
\item If $\varphi(\bar{x}) = \exists y \varphi_1(\bar{x}, y)$, then for
  $\bar{y} = (y_1, \ldots, y_t)$
\[
\Xi(\varphi)(\bar{x}_1, \ldots, \bar{x}_n) = \exists \bar{y}
\big(\Xi(\varphi_1)(\bar{x}_1, \ldots, \bar{x}_n, \bar{y}) \wedge \xi(y_1,
\ldots, y_t)\big)
\]
\item If $\varphi(\bar{x}) = \forall y \varphi_1(\bar{x}, y)$, then
  for $\bar{y} = (y_1, \ldots, y_t)$
\[
\Xi(\varphi)(\bar{x}_1, \ldots, \bar{x}_n) = \forall \bar{y}
\big( \xi(y_1, \ldots, y_t) \rightarrow \Xi(\varphi_1)(\bar{x}_1,
\ldots, \bar{x}_n, \bar{y}) \big)
\]
\end{enumerate}

We now define $\Xi$ for a $(1, \tau, \sigma, \mso)$-translation
scheme.  Given an $\mso(\sigma)$ formula $\varphi(\bar{x})$ where
$\bar{x} = (x_1, \ldots, x_n)$, the $\mso(\tau)$ formula
$\Xi(\varphi)(\bar{x})$, is as defined below.
\begin{enumerate}[nosep]
\item If $\varphi(\bar{x})$ is an $\fo$ atomic formula, then
  $\Xi(\varphi)(\bar{x})$ is as defined in the case of $\fo$
  above.
\item If $\varphi(\bar{x})$ is the formula $X(x_1)$ for an $\mso$
  variable $X$,
\[
\Xi(\varphi)(\bar{x}) = X(x_1) \wedge \xi(x_1)
\]
\item For boolean combinations, and quantification over $\fo$
  variables, $\Xi(\varphi)(\bar{x})$ is as defined in the case
  of $\fo$ above.
\item If $\varphi(\bar{x}) = \exists Y \varphi_1(\bar{x}, Y)$, then
\[
\Xi(\varphi)(\bar{x}) = \exists Y  \Xi(\varphi_1)(\bar{x}, Y) 
\]
\item If $\varphi(\bar{x}) = \forall Y \varphi_1(\bar{x}, Y)$, then
\[
\Xi(\varphi)(\bar{x}) = \forall Y  \Xi(\varphi_1)(\bar{x}, Y)
\]
\end{enumerate}

\tbf{Properties of translation schemes}:

For a $(t, \tau, \sigma, \mc{L})$-translation scheme $\Xi = (\xi,
(\xi_R)_{R \in \sigma})$, let \emph{$\text{rank}(\Xi)$} denote the
maximum of the quantifier ranks of the formulae $\xi$ and $\xi_R$ for
each $R \in \sigma$.

\begin{lemma}\label{lemma:rank-of-Xi(varphi)-in-terms-of-rank-of-varphi}
Let $\Xi$ be a $(t, \tau, \sigma, \mc{L})$-translation scheme, and let
$\varphi$ be an $\mc{L}(\sigma)$ formula of quantifier rank $m$.
\begin{enumerate}[nosep]
\item If $\mc{L} = $FO, then $\Xi(\varphi)$ is an FO$(\tau)$ formula
  having quantifier rank at most $t \cdot m + \text{rank}(\Xi)$.
\item If $\mc{L} = $MSO and $\Xi$ is scalar, then $\Xi(\varphi)$ is an
  MSO$(\tau)$ formula having quantifier rank at most $m +
  \text{rank}(\Xi)$.
\end{enumerate}
\end{lemma}


The following proposition relates the application of transductions to
structures and formulas to each other.

\begin{proposition}\label{prop:relating-transductions-applications-to-structures-and-formulae}
Let $\Xi$ be either a $(t, \tau, \sigma, \fo)$-translation scheme, or
a $(1, \tau, \sigma, \mso)$-translation scheme.  Then for every
$\mc{L}(\sigma)$ formula $\varphi(x_1, \ldots, x_n)$ where $n \ge 0$,
for every $\tau$-structure $\mf{A}$ and for every $n$-tuple
$(\bar{a}_1, \ldots, \bar{a}_n)$ from $\Xi(\mf{A})$, the following
holds.
\[
\begin{array}{lrll}
& (\Xi(\mf{A}), \bar{a}_1, \ldots, \bar{a}_n) & \models & \varphi(x_1,
  \ldots, x_n)\\ 
\mbox{iff} & (\mf{A}, \bar{a}_1, \ldots, \bar{a}_n) & \models &
\Xi(\varphi)(\bar{x}_1, \ldots, \bar{x}_n)\\
\end{array}
\]
where $\bar{x}_i = (x_{i, 1}, \ldots, x_{i, t})$ for each $i \in \{1,
\ldots, n\}$.
\end{proposition}

An immediate consequence of
Lemma~\ref{lemma:rank-of-Xi(varphi)-in-terms-of-rank-of-varphi} and
Proposition~\ref{prop:relating-transductions-applications-to-structures-and-formulae}
is the following.

\begin{corollary}\label{corollary:transferring-m-equivalence-across-transductions}
Let $\Xi$ be a $(t, \tau, \sigma, \mc{L})$-translation scheme.  Let
$m, r \in \mathbb{N}$ be such that $r = t \cdot m + \text{rank}(\Xi)$.
Suppose $\mf{A}$ and $\mf{B}$ are $\tau$-structures, and suppose
$\bar{a}_1, \ldots, \bar{a}_n$, resp. $\bar{b}_1, \ldots, \bar{b}_n$,
are $n$ elements from $\Xi(\mf{A})$, resp. $\Xi(\mf{B})$.
\begin{enumerate}[nosep]
\item If $(\mf{A}, \bar{a}_1, \ldots, \bar{a}_n) \fequiv{r} (\mf{B},
  \bar{b}_1, \ldots, \bar{b}_n)$, then $(\Xi(\mf{A}), \bar{a}_1,
  \ldots, \bar{a}_n) \fequiv{m} (\Xi(\mf{B}), \bar{b}_1, \ldots,
  \bar{b}_n)$.
\item If $(\mf{A}, \bar{a}_1, \ldots, \bar{a}_n) \mequiv{r} (\mf{B},
  \bar{b}_1, \ldots, \bar{b}_n)$, then $(\Xi(\mf{A}), \bar{a}_1,
  \ldots, \bar{a}_n) \mequiv{m} (\Xi(\mf{B}), \bar{b}_1, \ldots,
  \bar{b}_n)$, when $\Xi$ is scalar.
\end{enumerate}
\end{corollary}

We use the above properties of translation schemes in our results in
the forthcoming chapters.

Before we move to the next chapter, we define three notions that will
frequently appear in our discussions. These are the notions of
hereditariness, disjoint union and cartesian product.  Given a class
$\cl{S}_1$ of structures and a subclass $\cl{S}_2$ of $\cl{S}_1$, we
say $\cl{S}_2$ is \emph{hereditary over
  $\cl{S}_1$}\sindex[term]{hereditary}, if $\cl{S}_2$ is $PS$ over
$\cl{S}_1$ (see Section~\ref{section:background-classical-properties}
for the definition of $PS$). A class $\cl{S}$ is \emph{hereditary} if
it is hereditary over the class of all (finite) structures. Given two
$\tau$-structures $\mf{A}$ and $\mf{B}$, the \emph{disjoint
  union}\sindex[term]{disjoint union} of $\mf{A}$ and $\mf{B}$,
denoted $\mf{A} \sqcup \mf{B}$, is defined as follows. Let $\mf{B}'$
be an isomorphic copy of $\mf{B}$ whose universe is disjoint with that
of $\mf{A}$. Then $\mf{A} \sqcup \mf{B}$ is defined upto isomorphism
as the structure $\mc{C}$ such that (i) $\univ{\mc{C}} = \univ{\mf{A}}
\cup \univ{\mf{B}'}$ and (ii) $R^{\mf{C}} = R^{\mf{A}} \cup
R^{\mf{B}'}$ for each relation symbol $R \in \tau$. Finally, the
\emph{cartesian product}\sindex[term]{cartesian product} of $\mf{A}$
and $\mf{B}$, denoted $\mf{A} \times \mf{B}$, is the structure
$\mc{C}$ such that (i) $\univ{\mc{C}} = \univ{\mf{A}} \times
\univ{\mf{B}}$ and (ii) for each $n$-ary relation symbol $R \in \tau$,
for each $n$-tuple $((a_1, b_1), \ldots, (a_n, b_n))$ from $\mf{C}$,
where $(a_1, \ldots, a_n)$ is an $n$-tuple from $\mf{A}$ and $(b_1,
\ldots, b_n)$ is an $n$-tuple from $\mf{B}$, we have $((a_1, b_1),
\ldots, (a_n, b_n)) \in R^{\mf{C}}$ iff $(a_1, \ldots, a_n) \in
R^{\mf{A}}$ and $(b_1, \ldots, b_n) \in R^{\mf{B}}$.

%% file: need-to-investigate-new-classes.tex
\chapter{The need to investigate new classes of finite structures for $\glt{k}$}\label{chapter:need-for-new-classes-for-GLT}

\vspace{15pt} As already mentioned in the introduction, the class of
finite structures behaves very differently compared to the class of
arbitrary structures. The failure of the compactness theorem in the
finite causes a collapse of the proofs of most of the classical
preservation theorems in the finite. Worse still, the statements of
these theorems also fail in the finite. The {\lt} theorem is an
instance of these failures~\cite{tait,gurevich84,rosen}.

\begin{chaptertheorem}[Tait 1959, Gurevich-Shelah 1984]
There is a sentence that is preserved under substructures over the
class of all finite structures, but that is not equivalent, over all
finite structures, to any $\Pi^0_1$ sentence.
\end{chaptertheorem}

In the last 15 years, a lot of research in the area of preservation
theorems in finite model theory has focussed on identifying classes of
finite structures over which classical preservation theorems hold.
The following theorem from~\cite{dawar-pres-under-ext} identifies
classes of finite structures that satisfy structural restrictions that
are interesting from a computational standpoint and also from the
standpoint of modern graph structure theory, and shows that these
classes are ``well-behaved'' with respect some classical preservation
theorems, indeed in particular the {\lt} theorem.

\begin{chaptertheorem}[Atserias-Dawar-Grohe, 2008]\label{theorem:lt-positive-cases-finite}
The {\lt} theorem holds over each of the following classes of finite
structures:
\begin{enumerate}[nosep]
\item Any class of acyclic structures that is closed under
  substructures and disjoint unions.
\item Any class of bounded degree structures that is closed under
  substructures and disjoint unions.
\item The class of all structures of tree-width at most $k$, for each
  $k \in \mathbb{N}$.
\end{enumerate}
\end{chaptertheorem}

It is natural to ask what happens to $\mathsf{GLT}(k)$ over the
classes of finite structures referred to above. Unfortunately,
$\glt{k}$ fails in general over each of these classes. We show this
failure first for the class of all finite structures, and then for the
special classes of finite structures considered in
Theorem~\ref{theorem:lt-positive-cases-finite}.

\section{Failure of $\glt{k}$ over all finite structures}
The failure of the {\lt} theorem over the class of all finite
structures already implies the failure of $\mathsf{GLT}(0)$ over this
class. We show below that this failure happens for $\mathsf{GLT}(k)$
for each $k \ge 0$. In fact, we show something stronger.


\begin{proposition}[Failure of $\glt{k}$ in the finite]\label{prop:failure-of-glt(k)-in-the-finite}
There exists a vocabulary $\tau$ such that if $\cl{S}$ is the class of
all finite $\tau$-structures, then for each $k \ge 0$, there exists an
FO($\tau$) sentence $\psi_k$ that is preserved under substructures
over $\cl{S}$, but that is not $\cl{S}$-equivalent to any $\exists^k
\forall^*$ sentence. It follows that there is a sentence that is
$PSC(k)$ over $\cl{S}$ ($\psi_k$ being one such sentence) but that is
not $\cl{S}$-equivalent to any $\exists^k \forall^*$ sentence.
\end{proposition}

\begin{proof}
The second part of the proposition follows from the first part since a
sentence that is preserved under substructures over $\cl{S}$ is also
$PSC(k)$ over $\cl{S}$ for each $k \ge 0$. We now prove the first
part of the proposition.

Consider the vocabulary $\tau = \{\leq, S, P, c, d\}$
where $\leq$ and $S$ are both binary relation symbols, $P$ is a unary
relation symbol, and $c$ and $d$ are constant symbols. The sentence
$\psi_k$ is constructed along the lines of the known counterxample to
the {\lt} theorem in the finite. Following are the details.



\[
\begin{array}{lll}
\psi_k & = & (\xi_1 \wedge \xi_2 \wedge \xi_3) \wedge \neg (\xi_4 \wedge \xi_5) \\
\xi_1 & = & \text{``}\leq~\text{is a linear order~''} \\
\xi_2 & = & \text{``}~c~\text{is minimum under}~\leq~\text{and}~d~\text{is maximum under}~\leq~\text{"}\\
\xi_3 & = & \forall x \forall y ~S(x, y) \rightarrow \text{``}~y~\text{is the successor of}~x~\text{under}~\leq~\text{"}\\
\xi_4 & = & \forall x ~(x \neq d) \rightarrow \exists y S(x, y)\\
\xi_5 & = & \text{``~There exist at most}~k~\text{elements in (the set interpreting)}~P~\text{"}\\
\end{array}
\]

It is easy to see that each of $\xi_1, \xi_2, \xi_3$ and $\xi_5$ can
be expressed using a universal sentence. In particular, $\xi_1$ and
$\xi_3$ can be expressed using a $\forall^3$ sentence each, $\xi_2$
can be expressed using a $\forall$ sentence, and $\xi_5$ can be
expressed using a $\forall^{k+1}$ sentence.

\underline{The sentence $\psi_k$ is preserved under substructures over
  $\cl{S}$}

We show that $\varphi_k = \neg \psi_k$ is preserved under
extensions over $\cl{S}$.

Let $\mf{A} \models \varphi_k$ and $\mf{A} \subseteq \mf{B}$. If
$\alpha = (\xi_1 \,\wedge\, \xi_2 \wedge\, \xi_3)$ is such that
$\mf{A} \models \neg \alpha$, then since $\neg \alpha$ is equivalent
to an existential sentence, we have $\mf{B} \models \neg \alpha$;
whence $\mf{B} \models \varphi_k$. Else, $\mf{A} \models \alpha \wedge
\xi_4$. Let $b$ be an element of $\mf{B}$ that is not in
$\mf{A}$. Then there are two possibilities: 
\begin{enumerate}[nosep]
\item $(\mf{B}, a_1, b, a_2) \models ((x \leq y) \wedge (y \leq z))$
  for two elements $a_1, a_2$ of $\mf{A}$ such that $(\mf{A}, a_1,
  a_2) \models S(x, z)$; then $\mf{B} \models \neg \xi_3$ and hence
  $\mf{B} \models \varphi_k$.
\item $(\mf{B}, b) \models ((d \leq x) \vee (x \leq c))$. Since the
  interpretations of $c, d$ in $\mf{B}$ are the same as those of $c,
  d$ in $\mf{A}$ respectively, we have $\mf{B} \models \neg \xi_2$ and
  hence $\mf{B} \models \varphi_k$.
\end{enumerate}
In all cases, we have $\mf{B} \models \varphi_k$.

\underline{The sentence $\psi_k$ is not equivalent over $\cl{S}$ to
  any $\exists^k \forall^*$ sentence}

Suppose $\psi_k$ is equivalent to the sentence $\gamma = \exists x_1
\ldots \exists x_k \forall^n \bar{y} \beta(x_1, \ldots, x_k,
\bar{y})$, where $\beta$ is a quantifier-free formula.  We show below
that this leads to a contradiction, showing that the sentence $\gamma$
cannot exist.

Consider the structure $\mf{A} = (\mathsf{U}_{\mf{A}}, \leq^{\mf{A}},
S^{\mf{A}}, P^{\mf{A}}, c^{\mf{A}}, d^{\mf{A}})$, where the universe
$\mathsf{U}_{\mf{A}} = \{1, \ldots, (8n+1)\times (k+1)\}$,
$\leq^{\mf{A}}$ and $S^{\mf{A}}$ are respectively the usual linear
order and successor relation on $\mathsf{U}_{\mf{A}}$, $c^{\mf{A}} =
1, d^{\mf{A}} = (8n+1) \times (k+1)$ and $P^{\mf{A}} = \{(4n+1) + i
\times (8n +1) \mid i \in \{0, \ldots, k\} \}$. We see that $\mf{A}
\models (\xi_1 \wedge \xi_2 \wedge \xi_3 \wedge \xi_4 \wedge \neg
\xi_5)$ and hence $\mf{A} \models \psi_k$. Then $\mf{A} \models
\gamma$. Let $a_1, \ldots, a_k$ be the witnesses in $\mf{A}$ to the
$k$ existential quantifiers of $\gamma$.

It is clear that there exists $i^* \in \{0, \ldots, k\}$ such that
$a_j$ does not belong to $\{(8n+1)\times i^* + 1, \ldots, (8n+1)\times
(i^*+1) \}$ for each $j \in \{1, \ldots, k\}$.  Then consider the
structure $\mf{B}$ that is identical to $\mf{A}$ except that
$P^{\mf{B}} = P^{\mf{A}} \setminus \{(4n+1) + i^* \times (8n+1)\}$.
It is clear from the definition of $\mf{B}$ that $\mf{B} \models
(\xi_1 \wedge \xi_2 \wedge \xi_3 \wedge \xi_4 \wedge \xi_5)$ and hence
$\mf{B} \models \neg \psi_k$. We now show a contradiction by showing
that $\mf{B} \models \gamma$.

We show that $\mf{B} \models \gamma$ by showing that $(\mf{B}, a_1,
\ldots, a_k) \models \forall^n \bar{y} \beta(x_1, \ldots, x_k,
\bar{y})$.  This is in turn done by showing that for any $n$-tuple
$\bar{e} = (e_1, \ldots, e_n)$ from $\mf{B}$, there exists an
$n$-tuple $\bar{f} = (f_1, \ldots, f_n)$ from $\mf{A}$ such that the
(partial) map $\rho: \mf{B} \rightarrow \mf{A}$ given by $\rho(1)
= 1$, $\rho((8n+1)\times (k+1)) = (8n+1)\times (k+1)$, $\rho(a_j)
= a_j$ for $j \in \{1, \ldots, k\}$ and $\rho(e_j) = f_j$ for $j \in
\{1, \ldots, n\}$ is such that $\rho$ is a partial isomorphism from
$\mf{B}$ to $\mf{A}$. Then since $(\mf{A}, a_1, \ldots, a_k) \models
\forall^n \bar{y} \beta(x_1, \ldots, x_k, \bar{y})$, we have $(\mf{A},
a_1, \ldots, a_k, \bar{f}) \models \beta(x_1, \ldots, x_k, \bar{y})$
whereby $(\mf{B}, a_1, \ldots, a_k, \bar{e}) \models \beta(x_1,
\ldots, x_k, \bar{y})$. Since $\bar{e}$ is an arbitrary $n$-tuple from
$\mf{B}$, we have $(\mf{B}, a_1, \ldots, a_k) \models \forall^n
\bar{y} \beta(x_1, \ldots, x_k, \bar{y})$.

Define a \emph{contiguous segment in $\mf{B}$} to be a set of $l$
distinct elements of $\mf{B}$, for some $l \ge 1$, that are contiguous
w.r.t. the linear ordering in $\mf{B}$. That is, if $b_1, \ldots, b_l$
are the distinct elements of the aforesaid contiguous segment such
that $(b_j, b_{j+1}) \in \leq^{\mf{B}}$ for $1 \leq j \leq l-1$, then
$(b_j, b_{j+1}) \in S^{\mf{B}}$. We represent such a contiguous
segment as $\left[b_1, b_l\right]$, and view it as an interval in
$\mf{B}$.  Given an $n$-tuple $\bar{e}$ from $\mf{B}$, a
\emph{contiguous segment of $\bar{e}$ in $\mf{B}$} is a contiguous
segment in $\mf{B}$, all of whose elements belong to (the set
underlying) $\bar{e}$.  A \emph{maximal contiguous segment of
  $\bar{e}$ in $\mf{B}$} is a contiguous segment of $\bar{e}$ in
$\mf{B}$ that is not strictly contained in another contiguous segment
of $\bar{e}$ in $\mf{B}$. Let $\mathsf{CS}$ be the set of all maximal
contiguous segments of $\bar{e}$ in $\mf{B}$. Let $\mathsf{CS}_1
\subseteq \mathsf{CS}$ be the set of all those segments of
$\mathsf{CS}$ that have an intersection with the set $\{1, \ldots,
(8n+1)\times i^*\} \cup \{ (8n+1)\times (i^*+1) + 1, \ldots,
(8n+1)\times(k+1) \}$. Let $\mathsf{CS}_2 = \mathsf{CS} \setminus
\mathsf{CS}_1$.  Then all intervals in $\mathsf{CS}_2$ are strictly
contained in the interval $\left[(8n+1) \times i^* + 1, (8n+1) \times
  (i^* + 1) \right]$.  Let $\mathsf{CS}_2 = \{\left[i_1, j_1\right],
\left[i_2, j_2\right] \ldots, \left[i_r, j_r\right]\}$.  Observe that
$r \leq n$. Without loss of generality, assume that $i_1 \leq j_1 <
i_2 \leq j_2 < \ldots < i_r \leq j_r$.  Let $\mathsf{CS}_3$ be the set
of contiguous segments in $\mf{A}$ defined as $\mathsf{CS}_3 =
\{\left[i_1', j_1'\right], \left[i_2', j_2'\right], \ldots,
\left[i_r', j_r'\right]\}$ where $i_1' = (8n+1)\times i^* + n + 1,
j_1' = i_1' + (j_1 - i_1)$, and for $2 \leq l \leq r$, we have $i_l' =
j_{l-1}' + 2$ and $j_l' = i_l' + (j_l - i_l)$. Observe that the sum of
the lengths of the segments of $\mathsf{CS}_2$ is at most $n$, whereby
$j_r' \leq (8n+1)\times i^* + 3n + 1$.

Now consider the tuple $\bar{f} = (f_1, \ldots, f_n)$ defined using
$\bar{e} = (e_1, \ldots, e_n)$ as follows. Let
$\text{Elements}(\mathsf{CS}_1)$,
resp. $\text{Elements}(\mathsf{CS}_2)$, denote the elements contained
in the segments of $\mathsf{CS}_1$, resp. $\mathsf{CS}_2$.  For $1
\leq l \leq n$, if $e_l \in \text{Elements}(\mathsf{CS}_1)$, then $f_l
= e_l$. Else suppose $e_l$ belongs to the segment $\left[i_s,
  j_s\right]$ of $\mathsf{CS}_2$ where $1 \leq s \leq r$, and suppose
that $e_l = i_s + t$ for some $t \in \{0, (j_s - i_s)\}$. Then choose
$f_l = i_s' + t$.

It is easy to see that the (partial) map $\rho: \mf{B} \rightarrow
\mf{A}$ given by $\rho(1) = 1$, $\rho((8n+1)\times (k+1)) =
(8n+1)\times (k+1)$, $\rho(a_j) = a_j$ for $j \in \{1, \ldots, k\}$
and $\rho(e_j) = f_j$ for $j \in \{1, \ldots, n\}$ is such that
$\rho$ is a partial isomorphism from $\mf{B}$ to $\mf{A}$.
\end{proof}

\begin{remark}\label{remark:suggestions-for-PSC=Sigma_2-in-the-finite}
Proposition~\ref{prop:failure-of-glt(k)-in-the-finite} is a stronger
statement than the failure of the {\lt} theorem in the finite. While
the latter only shows that the class of sentences that are preserved
under substructures in the finite, i.e. the class of sentences that
are $PSC(0)$ in the finite, cannot be characterized by the class of
$\forall^*$ sentences,
Proposition~\ref{prop:failure-of-glt(k)-in-the-finite} shows that for
each $l \ge 0$, the class of sentences that is $PSC(l)$ in the finite
cannot be characterized by, or even semantically subsumed by, the
class of $\exists^k \forall^*$ sentences, for any $k \ge
0$. Interestingly, the sentence $\psi_k$ in
Proposition~\ref{prop:failure-of-glt(k)-in-the-finite}, which is not
equivalent in the finite to any $\exists^k \forall^*$ sentence, is
actually equivalent (in the finite) to an $\exists^{k+1} \forall^*$
sentence. We dwell on this observation towards the end of
Chapter~\ref{chapter:open-questions-finite-model-theory}.
\end{remark}

\section[Failure of $\glt{k}$ over classes well-behaved w.r.t. the {\lt} theorem]{Failure of $\glt{k}$ over classes that are well-behaved w.r.t. the {\lt} theorem}

Towards the central result of this section, we first show the following.

\begin{lemma}\label{lemma:failure-of-glt(k)-over-disj-unions-of-paths}
Let $\cl{U}$ be the class of all undirected graphs that are (finite)
disjoint unions of (finite) undirected paths.  Let $\cl{S}$ be a class
of undirected graphs of degree at most 2, that contains $\cl{U}$ as
a subclass. Then for each $k \ge 2$, there is a sentence $\phi_k$ that
is $PSC(k)$ over $\cl{S}$, but that is not $\cl{S}$-equivalent to any
$\exists^k \forall^*$ sentence.
\end{lemma}
\begin{proof} Given $k \ge 2$, consider $\phi_k$ that
asserts that either (i) there are at least $k$ nodes of degree exactly
0 or (ii) there are at least $k+1$ nodes of degree atmost 1. We claim
that the sentence $\phi_k$ is the desired sentence for the given
$k$. We give the reasoning for the case of $k = 2$; an analogous
reasoning can be done for $k > 2$. In our arguments below, $\phi =
\phi_2$. We observe that any graph in $\cl{S}$ is a disjoint union of
undirected paths and undirected cycles.

Any graph in $\cl{S}$ that contains a single connected component that
is a path (whereby every other connected component is a cycle), cannot
be a model of $\phi$. Then every model $\mf{D}$ of $\phi$ in $\cl{S}$
has at least two connected components, each of which is a path (of
length $\ge 0$). Consider a set $C$ formed by an end point of one of
these paths and an end point of the other of these paths. It is easy
to check that $C$ is a $2$-crux of $\mf{D}$ w.r.t. $\phi$ over
$\cl{S}$, whereby $\phi$ is $PSC(2)$ over $\cl{S}$. Suppose $\phi$ is
$\cl{S}$-equivalent to $\psi = \exists x_1 \exists x_2~ \forall^n
\bar{y} ~ \beta(x_1, x_2, \bar{y})$ where $\beta$ is a quantifier-free
formula. Consider a graph $\mf{A} \in \cl{U}$ that has exactly two
connected components each of which is a path of length $\geq 2n$.
Clearly $\mf{A}$ is a model of $\phi$. Further since $\mf{A} \in
\cl{S}$, we have $\mf{A} \models \psi$. Let $a_1, a_2$ be witnesses in
$\mf{A}$, to the existential quantifiers of $\psi$. It cannot be that
$a_1, a_2$ are both from the same path of $\mf{A}$ else the path by
itself would be a model for $\psi$, and hence $\phi$. Now consider a
structure $\mf{B} \in \cl{U} \subseteq \cl{S}$, which is a single path
that has length $\ge 4n$, and let $p_1, p_2$ be the end points of this
path. If $a_1$ (resp. $a_2$) is at a distance of $r \leq n$ from the
end point of any path in $\mf{A}$, then choose a point $b_1$
(resp. $b_2$) at the same distance, namely $r$, from $p_1$
(resp. $p_2$) in $\mf{B}$. Else choose $b_1$ (resp. $b_2$) at a
distance of $n+1$ from $p_1$ (resp. $p_2$).  Now consider any
$n$-tuple $\bar{e}$ from $\mf{B}$. By a similar kind of reasoning as
done in Proposition~\ref{prop:failure-of-glt(k)-in-the-finite}, one
can show that it is possible to choose an $n$-tuple $\bar{f}$ from
$\mf{A}$ such that the (partial) map $\rho: \mf{B} \rightarrow \mf{A}$
given by $\rho(b_i) = a_i, ~\rho(e_j) = f_j$ for $1 \leq i \leq 2$ and
$1 \leq j \leq n$, is a partial isomorphism from $\mf{B}$ to $\mf{A}$.
Since $(\mf{A}, a_1, a_2) \models \forall^n \bar{y} \beta(x_1, x_2,
\bar{y})$, we have $(\mf{A}, a_1, a_2, \bar{f}) \models \beta(x_1,
x_2, \bar{y})$ whereby $(\mf{B}, b_1, b_2, \bar{e}) \models \beta(x_1,
x_2, \bar{y})$. Since $\bar{e}$ is arbitrary, we have $\mf{B}$ models
$\psi$, and hence $\phi$ -- a contradiction.
\end{proof}

To state the central result of this section, we first introduce some
terminology.  Let $\tau$ be a vocabulary consisting of unary and
binary relation symbols only. Given a $\tau$-structure $\mf{A}$, let
$\mc{G}(\mf{A}) = (V, E)$ be the graph such that (i) $V$ is exactly
the universe of $\mf{A}$, and (ii) $(a, b) \in E$ iff for some binary
relation symbol $B \in \tau$, we have $(\mf{A}, a, b) \models (B(x, y)
\vee B(y, x))$.  In the language of translation schemes (see
Section~\ref{section:background-FMT-translation-schemes}), the
structure $\mc{G}(\mf{A})$ is indeed $\Xi(\mf{A})$ where $\Xi = (\xi,
\xi_E)$ is the $(1, \tau, \{E\}, \fo)$-translation scheme such that
$\xi$ is the formula $x = x$, and $\xi_E$ is the formula $\bigvee_{D
  \in \tau_{bin}} (D(x, y) \vee D(y, x))$, where $\tau_{bin}$ is the
set of all binary relation symbols of $\tau$. Given a class $\cl{S}$
of $\tau$-structures, let $\mc{G}(\cl{S}) = \{ \mc{G}(\mf{A}) \mid
\mf{A} \in \cl{S}\}$.

We say a class of undirected graphs has \emph{unbounded induced path
  lengths} if for every $n \in \mathbb{N}$, there exists a graph $G$
in the class such that $G$ contains an induced path of length $\ge n$.
We say a class $\cl{S}$ of $\tau$-structures has unbounded induced
path lengths if the class $\mc{G}(\cl{S})$ of undirected graphs has
unbounded induced path lengths. A class of $\tau$-structures is said
to have \emph{bounded induced path lengths} if it does not have
unbounded induced path lengths. 

The central result of this section is as follows.

\begin{theorem}\label{theorem:glt(k)-implies-bounded-ind-path-lengths}
Let $\cl{V}$ be a hereditary class of undirected graphs.  Let $\tau$
be a vocabulary containing unary and binary relation symbols only, and
$\cl{S}$ be the class of all $\tau$-structures $\mf{A}$ such that
$\mc{G}(\mf{A})$ belongs to $\cl{V}$. If $\glt{k}$ holds over
$\cl{S}$ for any $k \ge 2$, then $\cl{S}$ has bounded induced path
lengths.
\end{theorem}
\begin{proof}
If $\tau$ contains only unary relation symbols, then trivially
$\cl{S}$ has bounded induced path lengths; the induced path lengths in
all structures in $\cl{S}$ is bounded by 0. Therefore, assume $\tau$
contains at least one binary relation symbol.  Let $\tau_{bin}$ be the
set of all binary relation symbols of $\tau$, and let $B$ be one such
relation symbol of $\tau_{bin}$.

We prove the result by contradiction. Suppose $\cl{S}$ has unbounded
induced path lengths. Then for every $n \in \mathbb{N}$, there exists
$\mf{A} \in \cl{S}$ such that the graph $\mc{G}(\mf{A})$ contains an
induced path of length $r \ge n$. Since $\mc{G}(\mf{A}) \in \cl{V}$
and $\cl{V}$ is hereditary, it follows that the undirected path graph
of length $n$ belongs to $\cl{V}$, for each $n \in \mathbb{N}$. Then,
again by the hereditariness of $\cl{V}$, the class $\cl{U}$ of all
(finite) disjoint unions of (finite) undirected paths is a subclass of
$\cl{V}$.  Let $\chi$ be a universal sentence in the vocabulary
$\{E\}$ of graphs such that any model of $\chi$ is an undirected graph
whose degree is at most 2, and let $\cl{V}_1$ be the class of all
models of $\chi$ in $\cl{V}$. Clearly, $\cl{U}$ is a subclass of
$\cl{V}_1$.  For $k \ge 2$, let $\phi_k$ be the sentence given by
Lemma~\ref{lemma:failure-of-glt(k)-over-disj-unions-of-paths} such
that $\phi_k$ is $PSC(k)$ over $\cl{V}_1$ but $\phi_k$ is not
$\cl{V}_1$-equivalent to any $\exists^k \forall^*$ sentence.

Before we proceed, we present two observations.  Let $\xi_E(x, y) =
\bigvee_{D \in \tau_{bin}} (D(x, y) \vee D(y, x))$ as seen earlier.
Given an $\fo(\{E\})$ sentence $\beta$, let $\beta\left[E \mapsto
  \xi_E \right]$ be the $\fo(\tau)$ sentence obtained from $\beta$ by
replacing each occurence of ``$E(x, y)$'' in $\beta$, with the formula
$\xi_E(x, y)$.  Following are two observations that are easy to
verify. Let $\mf{A}, \mf{B}$ be given $\tau$-structures.

\begin{enumerate}[nosep]
\item[O.1] If $\mf{B} \subseteq \mf{A}$, then $\mc{G}(\mf{B})
  \subseteq \mc{G}(\mf{A})$. 
\item[O.2] $\mf{A} \models \beta\left[E \mapsto \xi_E \right]$ ~iff~
  $\mc{G}(\mf{A}) \models \beta$.  
\end{enumerate}

(Indeed, Observation O.2 can also be verified using
Proposition~\ref{prop:relating-transductions-applications-to-structures-and-formulae},
while Observation O.1 can also be verified using
Lemma~\ref{lemma:qt-free-transductions-preserve-substructure-prop},
that we present later.)

Consider now the $\fo(\tau)$ sentence $\alpha = (\phi_k \wedge
\chi)\left[E \mapsto \xi_E\right]$. We have the following.

\ul{$\alpha$ is $PSC(k)$ over $\cl{S}$}: Suppose $\mf{A} \in \cl{S}$
is such that $\mf{A} \models \alpha$. Then $\mc{G}(\mf{A}) \in \cl{V}$
and $\mc{G}(\mf{A}) \models \phi_k \wedge \chi$ by Observation O.2;
whereby $\mc{G}(\mf{A}) \in \cl{V}_1$. Then since $\phi_k$ is $PSC(k)$
over $\cl{V}_1$, there exists a $k$-crux $C$ of $\mc{G}(\mf{A})$
w.r.t. $\phi_k$ over $\cl{V}_1$. Consider a substructure $\mf{B}$ of
$\mf{A}$ such that $\mf{B} \in \cl{S}$ and $\mf{B}$ contains $C$. Then
$\mc{G}(\mf{B}) \in \cl{V}$. By Observation O.1 above, we have
$\mc{G}(\mf{B}) \subseteq \mc{G}(\mf{A})$. Since $\chi$ is a universal
sentence that is true in $\mc{G}(\mf{A})$, we have that
$\mc{G}(\mf{B})$ models $\chi$, whereby $\mc{G}(\mf{B}) \in \cl{V}_1$.
Then since $\mc{G}(\mf{B})$ contains $C$ and $C$ is a $k$-crux of
$\mc{G}(\mf{A})$ w.r.t. $\phi_k$ over $\cl{V}_1$, it follows that
$\mc{G}(\mf{B}) \models (\phi_k \wedge \chi)$. By Observation O.2
above, $\mf{B} \models \alpha$.  Whereby $C$ is a $k$-crux of $\mf{A}$
w.r.t. $\alpha$ over $\cl{S}$.  Then $\alpha$ is $PSC(k)$ over
$\cl{S}$.

\ul{$\alpha$ is not $\cl{S}$-equivalent to any $\exists^k \forall^*$
  sentence}: Suppose $\alpha$ is $\cl{S}$-equivalent to an $\exists^k
\forall^*$ $\fo(\tau)$ sentence $\gamma_1$. Let $\gamma_2 =
\gamma_1\left[B \mapsto E; D \mapsto \false, D \in \tau_{bin}, D \neq
  B\right]$ be the $\fo(\{E\})$ sentence obtained from $\gamma_1$ by
replacing (i) each occurence of $B$ in $\gamma_1$ with $E$, and (ii)
each occurence of $D$ in $\gamma_1$ with the constant formula
$\false$, for each $D \in \tau_{bin}, D \neq B$. Observe that
$\gamma_2$ is an $\exists^k \forall^*$ sentence.  We show that
$\phi_k$ is $\cl{V}_1$-equivalent to $\gamma_2$, contradicting
Lemma~\ref{lemma:failure-of-glt(k)-over-disj-unions-of-paths}.

Let $\psi$ be the $\fo(\tau)$ sentence given by $\psi = \forall x
\forall y \big( (B(x, y) \leftrightarrow B(y, x)) \wedge \bigwedge_{D
  \in \tau_{bin}, D \neq B} \neg D(x, y) \big)$. Given a
$\{E\}$-structure $\mf{C}$, let $\mf{C}^\psi$ be the $\tau$-structure
such that (i) $\univ{\mf{C}^\psi} = \univ{\mf{C}}$, (ii) for any two
elements $a, b \in \univ{\mf{C}^\psi}$, $(\mf{C}^\psi, a, b) \models
B(x, y)$ iff $(\mf{C}, a, b) \models E(x, y)$, and (iii) for any two
elements $a, b \in \univ{\mf{C}^\psi}$, $(\mf{C}^\psi, a, b) \models
\neg D(x, y)$ for each $D \in \tau_{bin}, D \neq B$. It is easy to see
that $\mc{C}^\psi \models \psi$, and that $\mc{G}(\mf{C}^\psi) =
\mf{C}$.  

Before proceeding to show that $\phi_k$ is $\cl{V}_1$-equivalent to
$\gamma_2$, we make the simple yet important observation, call it
$(\dagger)$, that any model of $\psi$ in $\cl{S}$ models the sentence
$(\gamma_1 \leftrightarrow \gamma_2\left[E \mapsto \xi_E\right])$.

\begin{itemize}[leftmargin=*,nosep]
\item $\phi_k$ entails $\gamma_2$ over $\cl{V}_1$: Suppose $\mf{C} \in
  \cl{V}_1$ is such that $\mf{C} \models \phi_k$. Then $\mf{C} \models
  (\phi_k \wedge \chi)$. Let $\mf{A} = \mf{C}^\psi$.  Then $\mf{A}
  \models \psi$ and $\mc{G}(\mf{A}) = \mf{C}$, whereby $\mf{A} \in
  \cl{S}$. Since $\mf{C} \models (\phi_k \wedge \chi)$, we have
  $\mf{A} \models \alpha$ by Observation O.2 above. Now since $\alpha$
  is $\cl{S}$-equivalent to $\gamma_1$ (by assumption), we have
  $\mf{A} \models \gamma_1$. Since $\mf{A}$ models $\psi$, we have by
  $(\dagger)$ that $\mf{A} \models \gamma_2\left[E \mapsto
    \xi_E\right]$. By Observation O.2 above, $\mf{C} \models
  \gamma_2$.
\item $\gamma_2$ entails $\phi_k$ over $\cl{V}_1$: Suppose $\mf{C} \in
  \cl{V}_1$ is such that $\mf{C} \models \gamma_2$. Let $\mf{A} =
  \mf{C}^\psi$. Then $\mf{A} \models \psi$ and $\mc{G}(\mf{A}) =
  \mf{C}$. By Observation O.2 above, $\mf{A} \models \gamma_2\left[E
    \mapsto \xi_E\right]$. By $(\dagger)$, we have $\mf{A} \models
  \gamma_1$. Since $\gamma_1$ is $\cl{S}$-equivalent to $\alpha$, we
  have $\mf{A} \models \alpha$. By Observation O.2, $\mf{C} \models
  (\phi_k \wedge \chi)$, and hence $\mf{C} \models \phi_k$.
\end{itemize}
\end{proof}












\begin{proposition}\label{prop:glt(k)-failure-for-special-classes}
Let $\tau$ be a vocabulary containing unary and binary relation
symbols only, and let there be at least one binary relation symbol in
$\tau$. Then there exist classes $\cl{S}_1$ and $\cl{S}_2$ of
$\tau$-structures such that,
\begin{itemize}[nosep]
\item $\cl{S}_1$ is acyclic, of degree at most 2, and is closed under
  substructures and disjoint
  unions\label{corollary:glt(k)-failure-for-special-classes:acyclic-and-bounded-degree}
\item $\cl{S}_2$ is the class of all $\tau$-structures of treewidth
  1\label{corollary:glt(k)-failure-for-special-classes:treewidth-1}
\end{itemize}
and $\glt{k}$ fails over each of $\cl{S}_1$ and $\cl{S}_2$ for each $k
\ge 2$.
\end{proposition}
\begin{proof}
Let $\cl{V}_2$ be the class of all undirected graphs that are acyclic,
and let $\cl{V}_1$ be the class of all the graphs in $\cl{V}_1$ that
have degree at most 2. Clearly $\cl{V}_1$ and $\cl{V}_2$ are
hereditary.  Let $\cl{S}_i$ be the class of all $\tau$-structures
$\mf{A}$ such that $\mc{G}(\mf{A}) \in \cl{V}_i$, for $i \in \{1,
2\}$. It is easy to see that $\cl{S}_1$ is acyclic, of degree at most
2, and is closed under substructures and disjoint unions. That
$\cl{S}_2$ is the class of all $\tau$-structures of tree-width 1
follows from definitions. Observe now that each of $\cl{S}_1$ and
$\cl{S}_2$ has unbounded induced path lengths. It then follows from
Theorem~\ref{theorem:glt(k)-implies-bounded-ind-path-lengths}, that
$\glt{k}$ cannot hold over $\cl{S}_i$ for any $k \ge 2$, for each $i
\in \{1, 2\}$.
\end{proof}

The above result motivates us to ask the following question: \emph{Can
  we identify structural properties (possibly abstract) of classes of
  finite structures that are satisfied by interesting classes of
  finite structures, and that entail $\mathsf{GLT}(k)$? And further,
  entail $\mathsf{GLT}(k)$ in effective form?}.  We identify one such
property in this thesis.  Notably, the classes of structures that
satisfy our property are incomparable to those studied
in~\cite{dawar-pres-under-ext,nicole-lmcs-15,duris-ext}.

%% file: lebsp.tex
\chapter[The $\mc{L}$-Equivalent Bounded Substructure Property -- $\lebsp{\cl{S}}{k}$]{The $\mc{L}$-Equivalent Bounded Substructure Property -- $\lebsp{\cl{S}}{k}$}\label{chapter:lebsp}


We define a new logic-based combinatorial property of classes of
finite structures.

\begin{chapterdefn}[$\lebsp{\cl{S}}{k}$]\label{defn:lebsp}
Let $\cl{S}$ be a class of structures and $k \in \mathbb{N}$.  We say
that $\cl{S}$ satisfies the \emph{$\mc{L}$-equivalent bounded
  substructure property} for parameter $k$, abbreviated
\emph{$\lebsp{\cl{S}}{k}$ is true} (alternatively,
\emph{$\lebsp{\cl{S}}{k}$ holds}),\sindex[symb]{$\lebsp{\cl{S}}{k}$}
if there exists a function $\witfn{\cl{S}}{k}{\mc{L}}: \mathbb{N}
\rightarrow \mathbb{N}$\sindex[symb]{$\witfn{\cl{S}}{k}{\mc{L}}$} such
that for each $m \in \mathbb{N}$, for each structure $\mf{A}$ of
$\cl{S}$ and for each $k$-tuple $\bar{a}$ from $\mf{A}$, there exists
a structure $\mf{B}$ such that (i) $\mf{B} \in \cl{S}$, (ii) $\mf{B}
\subseteq \mf{A}$, (iii) the elements of $\bar{a}$ are contained in
$\mf{B}$, (iv) $|\mf{B}| \leq \witfn{\cl{S}}{k}{\mc{L}}(m)$, and (v)
$\ltp{\mf{B}}{\bar{a}}{m}(\bar{x}) =
\ltp{\mf{A}}{\bar{a}}{m}(\bar{x})$.  The conjunction of these five
conditions is denoted as \emph{$\lebspcond(\cl{S}, \mf{A}, \mf{B}, k,
  m, \bar{a}, \witfn{\cl{S}}{k}{\mc{L}})$}\sindex[symb]{$\lebspcond$}.
We call $\witfn{\cl{S}}{k}{\mc{L}}$ a \emph{witness
  function}\sindex[term]{witness function} for $\lebsp{\cl{S}}{k}$.
\end{chapterdefn}

\begin{chapterremark}\label{remark:lebsp-witness-functions}
Observe that there can be several witness functions for
$\lebsp{\cl{S}}{k}$. For instance, if $\witfn{\cl{S}}{k}{\mc{L}}$ is a
witness function, then any function $\nu:\mathbb{N} \rightarrow
\mathbb{N}$ such that $\witfn{\cl{S}}{k}{\mc{L}}(m) \leq \nu(m)$ for
all $m \in \mathbb{N}$ is also a witness function. Observe also that
there always exists a monotonic witness function. This is easily seen
as follows. For a function $\nu:\mathbb{N} \rightarrow \mathbb{N}$,
let $\nu':\mathbb{N} \rightarrow \mathbb{N}$ be the function defined
as $\nu'(m) = \Sigma_{i = 0}^{i = m} \,\nu(i)$. Then $\nu'$ is
monotonic and $\nu(m) \leq \nu'(m)$ for all $m \in \mathbb{N}$;
whereby if $\nu$ is a witness function, then so is $\nu'$.  Therefore,
we assume henceforth that all witness functions are monotonic.
\end{chapterremark}

We list below two simple examples of classes $\cl{S}$ that satisfy
$\lebsp{\cl{S}}{k}$ for every $k \in \mathbb{N}$.  Many more such
classes are presented in Chapter~\ref{chapter:classes-satisfying-lebsp}.

\begin{enumerate}
\item Let $\cl{S}$ be a finite class of structures. Clearly,
  $\lebsp{\cl{S}}{k}$ holds for all $k \in \mathbb{N}$, with
  $\witfn{\cl{S}}{k}{\mc{L}}(m)$ giving the size of the largest
  structure in $\cl{S}$.
\item Let $\cl{S}$ be the class of all $\tau$-structures, where all
  relation symbols in $\tau$ are unary. By a simple $\fef$ game
  argument, one can see that $\febsp{\cl{S}}{k}$ holds for all $k \in
  \mathbb{N}$, with $\witfn{\cl{S}}{k}{\mc{L}}(m) = m \cdot 2^{|\tau|}
  + k$. In more detail: given $\mf{A} \in \cl{S}$, one can associate
  exactly one of $2^{|\tau|}$ colors with each element $a$ of
  $\mf{A}$, where the colour of $a$ gives the valuation in $\mf{A}$,
  of all predicates of $\tau$ for $a$. Then given a $k$-tuple
  $\bar{a}$ from $\mf{A}$, consider $\mf{B} \subseteq \mf{A}$
  satisfying the following: (i) the set $W$ of elements of $\bar{a}$,
  is contained in $\mathsf{U}_{\mf{B}}$, and (ii) for each colour $c$,
  if $A_c = \{a \mid a \in \mathsf{U}_{\mf{A}} \setminus W,
  ~a~\text{has colour}~c~\text{in}~\mf{A}\}$, then $A_c \subseteq
  \mathsf{U}_{\mf{B}}$ if $|A_c| < m$, else $|A_c \cap
  \mathsf{U}_{\mf{B}}| = m$. It is then easy to see that
  $\febspcond(\cl{S}, \mf{A}, \mf{B}, k, m, \bar{a},
  \witfn{\cl{S}}{k}{\mc{L}})$ is true.

  By a similar $\mef$ game argument, one can show that
  $\mebsp{\cl{S}}{k}$ holds for all $k \in \mathbb{N}$, with the same
  witness function $\witfn{\cl{S}}{k}{\mc{L}}$ as above.
\end{enumerate}



\input{lebsp-and-glt}

\input{lebsp-and-dlsp}

\input{sufficient-condition-for-lebsp}

%% file: lebsp-and-glt.tex
\section{$\lebsp{\cdot}{k}$ entails $\glt{k}$}\label{section:lebsp-and-glt}


In this section, we show that $\lebsp{\cl{S}}{k}$ indeed entails
$\glt{k}$. Towards this result, we first observe that given a class
$\cl{S}$ of structures and a natural number $n$, there exists an $\fo$
sentence that defines the subclass of all structures in $\cl{S}$ of
size at most $n$ (it is easy to construct such a sentence that is in
$\exists^n \forall^*$).  We fix such a sentence and denote it as
$\xi_{\cl{S}, n}$. Secondly, we observe the following.

\begin{lemma}\label{lemma:finite-relativizations} 
Given a sentence $\psi$ over a vocabulary $\tau$ and variables
$\bar{x}= (x_1, \ldots, x_n)$ not appearing in $\psi$, there exists a
{\em quantifier-free} formula $\psi|_{\bar{x}}(\bar{x})$ over $\tau$,
whose free variables are among $\bar{x}$, such that the following
holds: Let $\mf{A}$ be a structure and $\bar{a}=( a_1, \ldots, a_n)$ a
sequence of elements of $\mf{A}$. Then
$$(\mf{A}, a_1, \ldots, a_n) \models \psi|_{\bar{x}}(\bar{x})
\mbox{~iff~} \M{\mf{A}}{\aset} \models \psi$$
where $\M{\mf{A}}{\aset}$ denotes the substructure of $\mf{A}$ induced
by $\aset$. Further, $\psi|_{\bar{x}}(\bar{x})$ is computable from
$\psi$.
\end{lemma}
\begin{proof}
Let $X = \{x_1, \ldots, x_n\}$.  Replace every subformula of $\psi$ of
the form $\exists z \chi(z, y_1, \ldots, y_k)$ with $\bigvee_{z \in X}
\chi(z, y_1, \ldots, y_k)$, and every subformula of $\psi$ of the form
$\forall z \chi(z, y_1, \ldots, y_k)$ with $\bigwedge_{z \in X}
\chi(z, y_1, \ldots, y_k)$. It is clear that the resulting formula can
be taken to be $\psi|_{\bar{x}}(\bar{x})$. It is also clear that
$\psi_{\bar{x}}(\bar{x})$ is computable from $\psi$.
\end{proof}

The formula $\psi|_{\bar{x}}(\bar{x})$ is read as \emph{$\psi$
  relativized to $\bar{x}$}. 

Given a class $\cl{S}$ of structures and an $\mc{L}$ sentence
$\varphi$, we say that $\varphi$ is $PSC(k)$ over $\cl{S}$ if the
class of models of $\varphi$ in $\cl{S}$ is $PSC(k)$ over $\cl{S}$
(see Definition~\ref{defn:PSC(k)} for the definition of $PSC(k)$).  We
say $\mc{L}$-$\glt{k}$ holds over $\cl{S}$ if for all $\mc{L}$
sentences $\varphi$, it is the case that $\varphi$ is $PSC(k)$ over
$\cl{S}$ iff $\varphi$ is $\cl{S}$-equivalent to an
$\exists^k\forall^*$ FO sentence. Observe that over any class
$\cl{S}$, $\text{FO}$-$\glt{k}$ holds iff $\glt{k}$ holds, and that if
$\text{MSO}$-$\glt{k}$ holds, then so does $\text{FO}$-$\glt{k}$.


\begin{theorem}\label{theorem:lebsp-implies-glt(k)}
Let $\cl{S}$ be a class of finite structures and $k \in \mathbb{N}$ be
such that $\lebsp{\cl{S}}{k}$ holds. Then $\mc{L}$-$\glt{k}$, and
hence $\glt{k}$, holds over $\cl{S}$.  Further, if there exists a
computable witness function for $\lebsp{\cl{S}}{k}$, then the
translation from an $\mc{L}$ sentence that is $PSC(k)$ over $\cl{S}$,
to an $\cl{S}$-equivalent $\exists^k\forall^*$ sentence, is effective.
\end{theorem}
\begin{proof} 

Let $\witfn{\cl{S}}{k}{\mc{L}}$ be a witness function for
$\lebsp{\cl{S}}{k}$. We show below that an $\mc{L}$ sentence $\varphi$
of quantifier rank $m$ that is $PSC(k)$ over $\cl{S}$, is
$\cl{S}$-equivalent to the sentence $\chi$ given by $\chi = \exists^k
\bar{x} \forall^p \bar{y} ~\psi|_{\bar{x}\bar{y}}(\bar{x}, \bar{y})$,
where $p = \witfn{\cl{S}}{k}{\mc{L}}(m)$ and $\psi = (\xi_{\cl{S}, p}
\rightarrow \varphi)$.  It is easy to see that if
$\witfn{\cl{S}}{k}{\mc{L}}$ is computable, then since $m$ is
effectively computable from $\varphi$, so are $p$, $\xi_{\cl{S}, p}$
and $\chi$.

From the discussion in Section~\ref{section:PSC(k)}, we see that a
$\mc{L}$ sentence that is $\cl{S}$-equivalent to an
$\exists^k\forall^*$ sentence, is $PSC(k)$ over $\cl{S}$.  Towards the
converse, consider an $\mc{L}$ sentence $\varphi$, of quantifier rank
$m$, that is $PSC(k)$ over $\cl{S}$.  Let $\chi$ be the sentence as
given in the previous paragraph.  Since $\varphi$ is $PSC(k)$ over
$\cl{S}$, every model $\mf{A}$ of $\varphi$ in $\cl{S}$ also satisfies
$\chi$.  This is because the elements of any $k$-crux of $\mf{A}$ can
serve as witnesses to the existential quantifiers of $\chi$.  Thus
$\varphi$ entails $\chi$ over $\cl{S}$. To show $\chi$ entails
$\varphi$ over $\cl{S}$, suppose $\mf{A}$ is a model of $\chi$ in
$\cl{S}$.  Let $\bar{a}$ be a $k$-tuple that is a witness in $\mf{A}$
to the $k$ existential quantifiers of $\chi$.  Since
$\lebsp{\cl{S}}{k}$ holds, there exists a structure $\mf{B}$ such that
$\lebspcond(\cl{S}, \mf{A}, \mf{B}, k, m, \bar{a},
\witfn{\cl{S}}{k}{\mc{L}})$ is true. In other words, we have (i)
$\mf{B} \in \cl{S}$, (ii) $\mf{B} \subseteq \mf{A}$, (iii) the
elements of $\bar{a}$ are contained in $\mf{B}$, (iv) $|\mf{B}| \leq
\witfn{\cl{S}}{k}{\mc{L}}(m) = p$, and (v) $(\mf{B}, \bar{a})
\lequiv{m} (\mf{A}, \bar{a})$; then $\mf{B} \lequiv{m} \mf{A}$.  Since
$(\mf{A}, \bar{a}) \models \forall^p \bar{y}
~\psi|_{\bar{x}\bar{y}}(\bar{x}, \bar{y})$, by instantiating the
universal variables $\bar{y}$ with the elements of
$\mathsf{U}_{\mf{B}}$, and by using
Lemma~\ref{lemma:finite-relativizations}, we get $\mf{B} \models
(\xi_{\cl{S}, p} \rightarrow \varphi)$. Since $\mf{B} \in \cl{S}$ and
$|\mf{B}| \leq p$, we have $\mf{B} \models \xi_{\cl{S}, p}$ whereby
$\mf{B} \models \varphi$.  Since $\varphi$ is an $\mc{L}$ sentence of
quantifier rank $m$ and $\mf{B} \lequiv{m} \mf{A}$, it follows that
$\mf{A} \models \varphi$. Thus $\chi$ entails $\varphi$ over $\cl{S}$
whereby $\varphi$ is $\cl{S}$-equivalent to $\chi$.
\end{proof}

%% file: lebsp-and-dlsp.tex
\section[$\lebsp{\cdot}{k}$ and the downward L\"owenheim-Skolem property]{$\lebsp{\cl{S}}{k}$ -- a finitary analogue of the downward\\ L\"owenheim-Skolem property}\label{section:lebsp-and-dlsp}

The downward L\"owenheim-Skolem theorem, as already seen in
Section~\ref{section:background:structures}, is one of the first
results of classical model theory.  Well before the central tool of
classical model theory, namely the compactness theorem, was
discovered, L\"owenheim and Skolem~\cite{lowenheim, skolem}, showed
the following result.

\begin{theorem}[L\"owenheim 1915, Skolem 1920]\label{theorem:LS-original}
If an FO theory over a countable vocabulary has an infinite model,
then it has a countable model.
\end{theorem}

Historically, the proof of Theorem~\ref{theorem:LS-original},
initially due to L\"owenheim, assumed K\"onig's lemma, though the
latter lemma was proven only in 1927. Skolem in 1920 gave the first
fully self-contained proof of the theorem and hence the theorem is
jointly attributed to Skolem~\cite{badesa}. In subsequent years,
Skolem came up with a more general statement.

\begin{theorem}[Skolem's revised version of the downward {\ls} theorem]\label{theorem:skolem-LS}
Let $\tau$ be a countable vocabulary. For every $\tau$-structure
$\mf{A}$ and every countable set $W$ of elements of $\mf{A}$, there is
a countable substructure $\mf{B}$ of $\mf{A}$ such that $\mf{B}$
contains the elements of $W$, and $\mf{B}$ is elementarily equivalent
to $\mf{A}$.
\end{theorem}

Finally, Mal'tsev~\cite{maltsev} proved the most general version of
\dlsfull, which is also considered as the modern statement of the
theorem.  This version by Mal'tsev, stated as
Theorem~\ref{theorem:DLS} in
Section~\ref{section:background:structures}, is restated
below. Recall, that $\mf{B} \preceq \mf{A}$ denotes that $\mf{B}$ is
an elementary substructure of $\mf{A}$.

\begin{theorem}[Modern statement of the {\ls} theorem, Mal'tsev 1936]\label{theorem:modern-LS}
Let $\tau$ be a countable vocabulary. For every $\tau$-structure
$\mf{A}$ and every infinite cardinal $\kappa$, if $W$ is a set of at
most $\kappa$ elements of $\mf{A}$, then there exists a structure
$\mf{B}$ such that (i) $\mf{B} \subseteq \mf{A}$, (ii) $\mf{B}$
contains the elements of $W$, (iii) $|\mf{B}| \leq \kappa$, and (iv)
$\mf{B} \preceq \mf{A}$.
\end{theorem}

We now define a model-theoretic property of arbitrary structures, that
is closely related to the model-theoretic properties contained in the
versions of the {\dlsfull} by Skolem and Mal'tsev.  Given structures
$\mf{A}$ and $\mf{B}$, we say $\mf{A}$ and $\mf{B}$ are
\emph{$\mc{L}$-equivalent}, denoted $\mf{A} \equiv_{\mc{L}} \mf{B}$,
if $\mf{A}$ and $\mf{B}$ agree on all $\mc{L}$ sentences. If $\mf{B}
\subseteq \mf{A}$ and $\bar{b}$ is a (potentially infinite) tuple from
$\mf{B}$, that contains exactly the elements of $\mf{B}$, then we say
$\mf{B}$ is an \emph{$\mc{L}$-substructure} of $\mf{A}$, denoted
$\mf{B} \preceq_{\mc{L}} \mf{A}$, if $(\mf{B}, \bar{b})
\equiv_{\mc{L}} (\mf{A}, \bar{b})$. The reader can recognize that when
$\mc{L} = \fo$, then $\equiv_{\mc{L}}$ and $\preceq_{\mc{L}}$ are
exactly the literature notions of elementary equivalence and
elementary substructure (see
Section~\ref{section:background:structures}).  One easily sees that if
$\mf{B} \subseteq \mf{A}$ and $\bar{a}$ is any tuple (of any length)
from $\mf{B}$, then
\begin{center}
\begin{tabular}{ccccc}
$\mf{B} \preceq_{\mc{L}} \mf{A}$ & $\rightarrow$ & $(\mf{B}, \bar{a})
  \equiv_{\mc{L}} (\mf{A}, \bar{a})$ & $\rightarrow$ & $\mf{B}
  \equiv_{\mc{L}} \mf{A}$\\
\end{tabular}
\end{center}
where $\rightarrow$ denotes the usual implication.  Consider now the
following model-theoretic property of a class $\cl{S}$ of
\emph{arbitrary} structures. Below, a $\kappa$-tuple is a tuple of
length $\kappa$.

\begin{defn}\label{defn:dlsp}
Let $\cl{S}$ be a class of arbitrary structures over a countable
vocabulary, and let $\kappa$ be an infinite cardinal.  We say that
\emph{$\ldlsp{\cl{S}}{\kappa}$ is true}, if for each structure $\mf{A}
\in \cl{S}$ and each $\kappa$-tuple $\bar{a}$ from $\mf{A}$, there
exists a structure $\mf{B}$ such that (i) $\mf{B} \in \cl{S}$, (ii)
$\mf{B} \subseteq \mf{A}$, (iii) the elements of $\bar{a}$ are
contained in $\mf{B}$, (iv) $|\mf{B}| \leq \kappa$, and (v) $(\mf{B},
\bar{a}) \equiv_{\mc{L}} (\mf{A}, \bar{a})$.
\end{defn}

Let $\ldlspm{\cl{S}}{\kappa}$, resp. $\ldlsps{\cl{S}}{\kappa}$, be the
properties obtained from Definition~\ref{defn:dlsp} by simply
replacing the last condition in the definition with ``$\mf{B}
\preceq_{\mc{L}} \mf{A}$'', resp. ``$\mf{B} \equiv_{\mc{L}}
\mf{A}$''. The implication above then shows that
\begin{center}
\begin{tabular}{ccccc}
$\ldlspm{\cl{S}}{\kappa}$ & $\rightarrow$ & $\ldlsp{\cl{S}}{\kappa}$ & $\rightarrow$ & $\ldlsps{\cl{S}}{\kappa}$\\
\end{tabular}
\end{center}
Observe now that by taking $\mc{L}$ as $\fo$ and $\cl{S}$ as the class
of all (i.e. finite and infinite) structures, both
$\ldlspm{\cl{S}}{\kappa}$ and $\ldlsps{\cl{S}}{\omega}$ are true,
since indeed, these are respectively, the versions of the {\dlsfull}
by Mal'tsev and Skolem, given by Theorem~\ref{theorem:modern-LS} and
Theorem~\ref{theorem:skolem-LS}.  Whereby, we can see
$\ldlsp{\cl{S}}{\kappa}$ as a version of the {\dls} property that is
``intermediate'' between the versions of the {\dls} property by
Mal'tsev and Skolem.  And now, as the figure below shows,
$\lebsp{\cl{S}}{k}$ reads very much like
$\ldlsp{\cl{S}}{\kappa}$. Indeed then, $\lebsp{\cl{S}}{k}$ can very
well be regarded as a finitary analogue of the {\dls} property.

\begin{figure}[H]\centering \includegraphics[scale=0.7]{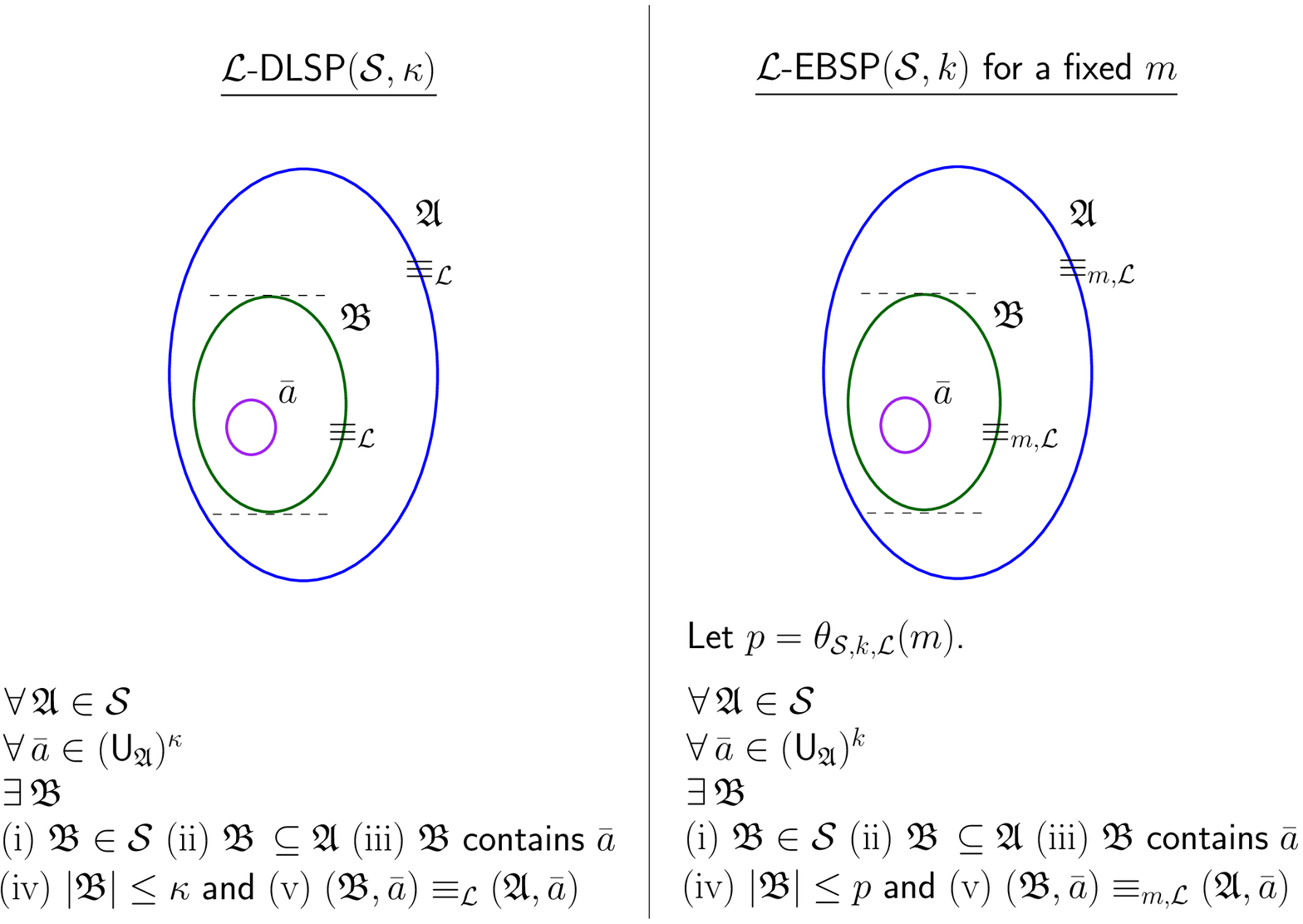}
\caption{\label{figure:lebsp-and-dlsp}$\lebsp{\cl{S}}{k}$ as a finitary analogue of $\ldlsp{\cl{S}}{\kappa}$}

\end{figure}

%% file: sufficient-condition-for-lebsp.tex
\section{A sufficient condition for $\lebsp{\cdot}{k}$}\label{section:sufficient-condition-for-lebsp}

For a class $\cl{S}$ of $\tau$-structures, let $\cl{S}_p = \{ (\mf{A},
\nu) \mid \mf{A} \in \cl{S},\, \nu: \mathsf{U}_{\mf{A}} \rightarrow
\{0, \ldots, p-1\} \}$ be the class of all structures obtained by
labeling the elements of the structures of $\cl{S}$, with elements
from $\{0, \ldots, p-1\}$. Formally, each structure from $\cl{S}_p$
can be seen as a $\tau'$-structure where $\tau' = \tau \cup \{Q_i \mid
i \in \{0, \ldots, p-1\}\}$ and $Q_i$ is a unary relation symbol that
does not appear in $\tau$, for each $i \in \{0, \ldots, p-1\}$. We
have the following lemma.


\begin{lemma}\label{lemma:gebsp-and-ebsp}
Let $\cl{S}$ be a class of finite $\tau$-structures, and $p, k \in
\mathbb{N}$ be such that $p \ge k$. Then the following are true.
\begin{enumerate}[nosep]
\item $\lebsp{\cl{S}_p}{0}$ implies
  $\lebsp{\cl{S}_k}{0}$\label{lemma:gebsp-and-ebsp:1}
\item $\lebsp{\cl{S}_{p\cdot k}}{0}$ implies
  $\lebsp{(\cl{S}_p)_k}{0}$\label{lemma:gebsp-and-ebsp:2}
\item $\lebsp{\cl{S}_{k+1}}{0}$ implies
  $\lebsp{\cl{S}}{k}$\label{lemma:gebsp-and-ebsp:3}
\item $\lebsp{\cl{S}}{k}$ implies
  $\febsp{\cl{S}}{k}$\label{lemma:gebsp-and-ebsp:4}
\end{enumerate}
Further, in each of the implications above, any witness function for
the antecedent is also a witness function for the consequent.
\end{lemma}

\begin{proof}
\ul{Part \ref{lemma:gebsp-and-ebsp:1}}: Obvious since $\cl{S}_k
\subseteq \cl{S}_p$ and $\cl{S}_k$ is closed under substructures that
are in $\cl{S}_p$.

\ul{Part \ref{lemma:gebsp-and-ebsp:2}}: Consider $\mf{A} \in
(\cl{S}_p)_k$. Each element of $\mf{A}$ can be seen to be labelled
with a pair $(i, j)$ of labels where $i \in \{0, \ldots, p-1\}$ and $j
\in \{0, \ldots, k-1\}$. Then $\mf{A}$ can naturally be represented as
a structure $\mf{A}'$ of $\cl{S}_{p \cdot k}$ as follows: (i) the
$\tau$-reducts of $\mf{A}$ and $\mf{A}'$ are the same, and (ii) an
element $a$ of $\mf{A}'$ having label $(i, j)$ in $\mf{A}$ is labelled
with label $i\times k + j$ in $\mf{A}'$. Since $\lebsp{\cl{S}_{p \cdot
    k}}{0}$ is true, there exists a witness function $\witfn{\cl{S}_{p
    \cdot k}}{0}{\mc{L}}:\mathbb{N} \rightarrow \mathbb{N}$ such that
for any $m \in \mathbb{N}$, there exists $\mf{B}' \in \cl{S}_{p \cdot
  k}$ such that (i) $\mf{B}' \subseteq \mf{A}'$ (ii) $|\mf{B}'| \leq
\witfn{\cl{S}_{p \cdot k}}{0}{\mc{L}}(m)$ and (iii) $\mf{B}'
\lequiv{m} \mf{A}'$. Consider the structure $\mf{B} \in (\cl{S}_p)_k$
such that (i) the $\tau$-reducts of $\mf{B}'$ and $\mf{B}$ are the
same, and (ii) if the label of an element $a$ of $\mf{B}$ is $i \times
k + j$ in $\mf{B}'$, then its label in $\mf{B}$ is $(i, j)$.  It is
easy to see that (i) $\mf{B} \subseteq \mf{A}$ (ii) $|\mf{B}| \leq
\witfn{\cl{S}_{p \cdot k}}{0}{\mc{L}}(m)$, and (iii) using the same
strategy as of the duplicator in the $m$-round $\mc{L}$-EF game
between $\mf{B}'$ and $\mf{A}'$, the duplicator always wins in the
$m$-round $\mc{L}$-EF game between $\mf{B}$ and $\mf{A}$; in other
words, $\mf{B} \lequiv{m} \mf{A}$. Then $\lebspcond((\cl{S}_p)_k,
\mf{A}, \mf{B}, 0, m, \mathsf{null}, \witfn{(\cl{S}_p)_k}{0}{\mc{L}})$
is true, where $\mathsf{null}$ denotes the empty tuple and
$\witfn{(\cl{S}_p)_k}{0}{\mc{L}} = \witfn{\cl{S}_{p \cdot
    k}}{0}{\mc{L}}$.

\ul{Part \ref{lemma:gebsp-and-ebsp:3}}: Let $\mf{A} \in \cl{S}$ and
$\bar{a} = (a_1, \ldots, a_k)$ be a $k$-tuple from $\mf{A}$. Consider
the structure $\mf{A}'$ of $\cl{S}_{k+1}$ whose $\tau$-reduct is
$\mf{A}$ and in which the element $a_i$ has been labelled with label
$i-1$ for $1 \leq i \leq k$ and all elements other than the $a_i$s
have been labelled with label $k$. Since $\lebsp{\cl{S}_{k+1}}{0}$
holds, there exists a witness function
$\witfn{\cl{S}_{k+1}}{0}{\mc{L}}:\mathbb{N} \rightarrow \mathbb{N}$
such that given any $m \in \mathbb{N}$, there exists $\mf{B}' \in
\cl{S}_{k+1}$ satisfying (i) $\mf{B}' \subseteq \mf{A}'$ (ii)
$|\mf{B}'| \leq \witfn{\cl{S}_{k+1}}{0}{\mc{L}}(m)$ and (iii) $\mf{B}'
\lequiv{m} \mf{A}'$. Let $\mf{B}$ be the $\tau$-reduct of
$\mf{B}'$. It is clear that (i) $\mf{B} \in \cl{S}$ since $\mf{B}' \in
\cl{S}_{k+1}$ (ii) $\mf{B} \subseteq \mf{A}$ (iii) $\mf{B}$ must
contain the elements of $\bar{a}$ since $a_i$ is the unique element of
$\mf{A}'$ that is labeled with label $i-1$, for each $i \in \{1,
\ldots, k\}$ (iv) $|\mf{B}| \leq \witfn{\cl{S}_{k+1}}{0}{\mc{L}}(m)$,
and (v) using the same strategy as of the duplicator in the $m$-round
$\mc{L}$-EF game between $\mf{B}'$ and $\mf{A}'$, the duplicator
always wins in the $m$-round $\mc{L}$-EF game between $(\mf{B},
\bar{a})$ and $(\mf{A}, \bar{a})$; in other words,
$\ltp{\mf{B}}{\bar{a}}{m}(\bar{x}) =
\ltp{\mf{A}}{\bar{a}}{m}(\bar{x})$. Then $\lebspcond(\cl{S}, \mf{A},
\mf{B}, k, m, \bar{a}, \witfn{\cl{S}}{k}{\mc{L}})$ is true, where
$\witfn{\cl{S}}{k}{\mc{L}} = \witfn{\cl{S}_{k+1}}{0}{\mc{L}}$.

\ul{Part \ref{lemma:gebsp-and-ebsp:4}}: Obvious; since FO $\subseteq
  \mc{L}$, it follows that for structures $(\mf{A}, \bar{a})$ and
  $(\mf{B}, \bar{b})$, if $\ltp{\mf{A}}{\bar{a}}{m}(\bar{x}) =
  \ltp{\mf{B}}{\bar{b}}{m}(\bar{x})$, then
  $\ftp{\mf{A}}{\bar{a}}{m}(\bar{x}) =
  \ftp{\mf{B}}{\bar{b}}{m}(\bar{x})$. \vspace{2pt}

We observe that in all the cases above, any witness function for the
antecdent is also a witness function for the consequent.
\end{proof}

In the next chapter, we prove that $\lebsp{\cl{S}}{k}$ holds of
several classes $\cl{S}$ which are of interest in computer science.
In doing so, we use Lemma~\ref{lemma:gebsp-and-ebsp} in an important
way to simplify our proofs since as this lemma shows, to prove
$\reflebsp$ for a class $\cl{S}$ and parameter $k$, it suffices to
prove $\reflebsp$ for the class $\cl{S}_p$ and parameter 0, where $p$
is suitably chosen.

%% file: classes-satisfying-lebsp.tex
\chapter{Classes satisfying $\lebsp{\cdot}{k}$}\label{chapter:classes-satisfying-lebsp} 

In this chapter, we show that various interesting classes of
structures satisfy $\lebsp{\cdot}{k}$, and also give methods to
construct new classes of structures that satisfy $\lebsp{\cdot}{k}$
from classes known to satisfy the latter property. Broadly speaking,
the specific classes that we consider are of two kinds -- one that are
special kinds of posets, and the other that are special kinds of
graphs. In Section~\ref{section:words-trees-nestedwords}, we consider
the former kind of classes and prove our results for words, trees
(unordered, ordered, or ranked) and nested words over a given finite
alphabet $\Sigma$. In Section~\ref{section:n-partite-cographs}, we
consider the latter kind of classes and prove our results for
$n$-partite cographs, and hence various subclasses of these including
cographs, graph classes of bounded tree-depth, graph classes of
bounded shrub-depth and graph classes of bounded $\mc{SC}$-depth. In
Section~\ref{section:closure-prop-of-lebsp}, we show that classes that
satisfy $\lebsp{\cdot}{\cdot}$ are, under suitable assumptions, closed
under set-theoretic operations, under operations that are
implementable using quantifier-free translation schemes, and under
transformations that are defined using regular operation-tree
languages, where an operation-tree is a finite composition of the
aforementioned operations on classes of structures. These closure
properties give us means to construct a wide array of classes that
satisfy $\lebsp{\cdot}{\cdot}$.

All of the above results derive from an abstract result concerning
tree representations that we now describe in
Section~\ref{section:abstract-tree-theorem}.


\input{abstract-tree-result-2.0}
\input{words-and-trees-and-nested-words}

\input{n-partite-cographs}

\input{closure-prop-of-lebsp}

%% file: abstract-tree-result-2.0.tex
\section{An abstract result concerning tree representations}\label{section:abstract-tree-theorem}

\newcommand{\troot}[1]{\ensuremath{\mathsf{root}(#1)}}
\newcommand{\sigmaint}{\ensuremath{\Sigma_{\text{int}}}}
\newcommand{\sigmaleaf}{\ensuremath{{\Sigma}_{\text{leaf}}}}
\newcommand{\str}[1]{\ensuremath{\mathsf{Str}(#1)}}
\newcommand{\stri}[2]{\ensuremath{\mathsf{Str}_{#1}(#2)}}
\newcommand{\rftrees}{\ensuremath{\mathsf{RF}\text{-}\mathsf{trees}}}
\newcommand{\mylabel}[1]{\ensuremath{\mathsf{#1}}}
\newcommand{\disjun}{\ensuremath{\dot\cup}}
\newcommand{\liequiv}[2]{\ensuremath{\equiv_{#1, #2}}}
\newcommand{\rfse}{\ensuremath{\mc{L}\text{-RFSE}}}
\newcommand{\mrfse}{\ensuremath{\mso\text{-RFSE}}}

An \emph{unlabeled unordered tree}\sindex[term]{tree!unordered} is a
finite poset $P = (A, \leq)$ with a unique minimal element (called
``root''), and such that for each $c \in A$, the set $\{b \mid b \leq
c\}$ is totally ordered by $\leq$.  Informally speaking, the Hasse
diagram of $P$ is an inverted (graph-theoretic) tree.  We call $A$ as
the set of \emph{nodes} of $P$.  We use the standard notions of leaf,
internal node, ancestor, descendent, parent, child, degree, height and
subtree in connection with trees. Explicitly, these are as defined
below. Let $P = (A, \leq)$ be a (unlabeled unordered) tree. Let $a, b$
be distinct nodes of $A$.
\begin{enumerate}[nosep]
\item We say $a$ is a \emph{leaf} of $P$ if for any node $c \in A$, we
  have $(a \leq c) \rightarrow (c = a)$. A node of $P$ that is not a
  leaf of $P$ is called an \emph{internal node} of $P$.
\item 
We say $a$ is an \emph{ancestor} of $b$ in $P$, or equivalently that
$b$ is a \emph{descendent} of $a$ in $P$, if $a \leq b$ and $a \neq
b$.  A \emph{common ancestor} of $a$ and $b$ in $P$ is a node $c$ of
$P$ such that $c \leq a$ and $c \leq b$. The \emph{greatest common
  ancestor}\sindex[term]{greatest common ancestor} of $a$ and $b$ in
$P$, denoted $a \wedge_{P} b$, is a common ancestor of $a$ and $b$ in
$P$ such that for every common ancestor $c$ of $a$ and $b$ in $P$, we
have $c \leq (a \wedge_{P} b)$.
\item 
We say $a$ is the \emph{parent} of $b$ in $P$, or equivalently, that
$b$ is a \emph{child} of $a$ in $P$, if $a$ is an ancestor of $b$ in
$P$ and any ancestor of $b$ in $P$ is either $a$ itself or an ancestor
of $a$ in $P$.  We let $\mathsf{Children}_{P}(a)$ denote the set $\{b
\mid b~\text{is a child of}~a~\text{in}~P\}$.
\item 
The \emph{degree} of $a$ in $P$ is the size of
$\mathsf{Children}_{P}(a)$.  The degree of $P$ is the maximum of the
degrees of the nodes of $P$.
\item 
The \emph{height} of $P$ is one less than the size of the longest
chain in $P$.
\item 
The \emph{subtree of $P$ induced by a subset $A'$ of $A$} is the tree
$(A', \leq')$ where $\leq' \,= \,\leq \,\cap\, (A' \times A')$. A
\emph{subtree of $P$} is a subtree of $P$ induced by some subset of
$A$.


\end{enumerate}

An \emph{unlabeled ordered}\sindex[term]{tree!ordered} tree is a pair
$O = (P, \lesssim)$ where $P$ is an unlabeled unordered tree and
$\lesssim$ is a binary relation that imposes a linear order on the
children of any internal node of $P$. In this section, by trees we
always mean \emph{ordered trees}. It is clear that the notions
introduced above for unordered trees can be adapted for ordered
trees. We define some additional notions for ordered trees below. It
is clear that these notions can be adapted for unordered trees.
\begin{enumerate}[nosep]
\item[7.] Given a countable alphabet $\Sigma$, a \emph{tree over
  $\Sigma$}, also called a \emph{$\Sigma$-tree}, or simply \emph{tree}
  when $\Sigma$ is clear from context, is a pair $(O, \lambda)$ where
  $O$ is an unlabeled tree and $\lambda: A \rightarrow \Sigma$ is a
  labeling function, where $A$ is the set of nodes of $O$. We denote
  $\Sigma$-trees by $\tree{s}, \tree{t}, \tree{x}, \tree{y} $ or
  $\tree{z}$, possibly with numbers as subscripts.
\item[8.] Given a tree $\tree{t}$, we denote the root of $\tree{t}$ as
  $\troot{\tree{t}}$\sindex[symb]{$\troot{\tree{t}}$}. For a node $a$
  of $\tree{t}$, we denote the subtree of $\tree{t}$ rooted at $a$ as
  $\tree{t}_{\ge a}$,\sindex[symb]{$\tree{t}_{\ge a}$} and the subtree
  of $\tree{t}$ obtained by deleting $\tree{t}_{\ge a}$ from
  $\tree{t}$, as $\tree{t} - \tree{t}_{\ge a}$.\sindex[symb]{$\tree{t}
    - \tree{t}_{\ge a}$}
\item[9.] Given a tree $\tree{s}$ and a non-root node $a$ of
  $\tree{t}$, the \emph{replacement of $\tree{t}_{\ge a}$ with
  $\tree{s}$ in $\tree{t}$}, denoted $\tree{t} \left[ \tree{t}_{\ge a}
  \mapsto \tree{s} \right]$,\sindex[symb]{$\tree{t} \left[
    \tree{t}_{\ge a} \mapsto \tree{s} \right]$} is a tree defined as
  follows: let $c$ be the parent of $a$ in $\tree{t}$ and $\tree{s}'$
  be an isomorphic copy of $\tree{s}$ whose nodes are disjoint with
  those of $\tree{t}$. Then $\tree{t} \left[ \tree{t}_{\ge a} \mapsto
    \tree{s} \right]$ is defined upto isomorphism as the tree obtained
  by deleting $\tree{t}_{\ge a}$ from $\tree{t}$ to get a tree
  $\tree{t}'$, and inserting (the root of) $\tree{s}'$ at the same
  position among the children of $c$ in $\tree{t}'$, as the position
  of $a$ among the children of $c$ in $\tree{t}$.
\item[10.] For $\tree{t}, \tree{s}$ and $\tree{s}'$ as in the previous
  point, suppose the roots of each of these trees have the same
  label. Then the \emph{merge of $\tree{s}$ with $\tree{t}$}, denoted
  $\tree{t} \odot \tree{s}$,\sindex[symb]{$\tree{t} \odot \tree{s}$}
  is defined upto isomorphism as the tree obtained by deleting
  $\troot{\tree{s}'}$ from $\tree{s}'$ and concatenating the sequence
  of subtrees hanging at $\troot{\tree{s}'}$ in $\tree{s}'$, to the
  sequence of subtrees hanging at $\troot{\tree{t}}$ in
  $\tree{t}$. Thus the children of $\troot{\tree{s}'}$ in $\tree{s}'$
  are the ``new'' children of $\troot{\tree{t}}$, and appear ``after''
  the ``old'' children of $\troot{\tree{t}}$, and in the order they
  appear in $\tree{s}'$.
\end{enumerate}



Fix a finite alphabet $\sigmaint$ and a countable alphabet
$\sigmaleaf$ (where the two alphabets are allowed to be overlapping).
We say a class $\mc{T}$ of $(\sigmaint \cup \sigmaleaf)$-trees is
\emph{representation-feasible}\sindex[term]{representation!-feasible}
if it is closed under (label-preserving) isomorphisms, and if every
tree $\tree{t} = (O, \lambda)$ in the class has the property that for
every leaf, resp. internal, node $a$ of $\tree{t}$, the label
$\lambda(a)$ belongs to $\sigmaleaf$, resp. $\sigmaint$.  Given a
class $\cl{S}$ of $\tau$-structures, let $\mathsf{Str}: \mc{T}
\rightarrow \cl{S}$ be a map that associates with each tree in
$\mc{T}$, a structure in $\cl{S}$. For a tree $\tree{t} \in \mc{T}$,
if $\mf{A} = \str{\tree{t}}$, then we say $\tree{t}$ is a \emph{tree
  representation} of $\mf{A}$ under $\mathsf{Str}$, or simply a
\emph{tree representation} of $\mf{A}$.  We call $\mathsf{Str}$ as a
\emph{representation map},\sindex[term]{representation!map} and
$\mc{T}$ as a class of representation trees. For the purposes of our
result, we consider maps $\mathsf{Str}$ that map isomorphic
(preserving labels) trees to isomorphic structures, and that have
additional properties among those mentioned below.

\begin{enumerate}[label=\Alph*., ref=\Alph*, nosep]
\item  Transfer properties:  
\begin{enumerate}[label=\arabic*., ref=\theenumi.\arabic*, nosep]
\item Let $\tree{t}, \tree{s}_1 \in \mc{T}$, $\tree{t}$ be of size
  $\ge 2$, and $a$ be a child of $\troot{\tree{t}}$. Suppose
  $\tree{s}_2 = \tree{t}_{\ge a}$ and $\tree{z} = \tree{t}
  \left[\tree{s}_{2} \mapsto \tree{s}_1\right] \in \mc{T}$.
\begin{enumerate}[label=\alph*., ref=\theenumii.\alph*, nosep]
\item \label{A.1.a} If $\str{\tree{s}_1} \hookrightarrow \str{\tree{s}_2}$, then
  $\str{\tree{z}} \hookrightarrow \str{\tree{t}}$.
\item \label{A.1.b} If $\str{\tree{s}_1} \lequiv{m} \str{\tree{s}_2}$, then
  $\str{\tree{z}} \lequiv{m} \str{\tree{t}}$.
\end{enumerate}
\item \label{A.2} Let $\tree{t}, \tree{s}_1, \tree{s}_2 \in \mc{T}$ be  trees
  of size $\ge 2$ such that the roots of all these trees have the same
  label. For $i \in \{1, 2\}$, suppose $\tree{z}_i
  = \tree{t} \odot \tree{s}_i \in \mc{T}$.  If
  $\str{\tree{s}_1} \lequiv{m} \str{\tree{s}_2}$, then
  $\str{\tree{z}_1} \lequiv{m} \str{\tree{z}_2}$.
\end{enumerate}
\item  Monotonicity: Let $\tree{t} \in \mc{T}$ and  $a$ be a child of
  $\troot{\tree{t}}$.
\begin{enumerate}[label=\arabic*., ref=\theenumi.\arabic*, nosep]
\item \label{B.1} If $\tree{s} = \tree{t}_{\ge a} \in \mc{T}$, then $\str{\tree{s}} \hookrightarrow \str{\tree{t}}$

\item \label{B.2} If $\tree{s} = (\tree{t} - \tree{t}_{\ge a}) \in \mc{T}$, then $\str{\tree{s}} \hookrightarrow \str{\tree{t}}$.
\end{enumerate}
\end{enumerate}
We say a representation map is \emph{$\mc{L}$-height-reduction
  favourable}\sindex[term]{reduction favourable!for height} if it
satisfies conditions \ref{A.1.a}, \ref{A.1.b} and \ref{B.1} above, for
all $m \ge m_0$, for some $m_0 \in \mathbb{N}$. A representation map
is said to be \emph{$\mc{L}$-degree-reduction
  favourable}\sindex[term]{reduction favourable!for degree} if it
satisfies all the conditions above, except possibly \ref{B.1}, for all
$m \ge m_0$, for some $m_0 \in \mathbb{N}$.  The following result
justifies why the two kinds of representation maps just defined, are
called so. This result contains the core argument of most results in
the subsequent sections of this chapter.  Below, $\mc{T}$ is said to
be \emph{closed under rooted subtrees and under replacements with
  rooted subtrees} if for all $\tree{t} \in \mc{T}$ and non-root nodes
$a$ and $b$ of $\tree{t}$ such that $a$ is an ancestor of $b$ in
$\tree{t}$, we have that each of the subtrees $\tree{t}_{\ge a}$ and
$\tree{t} \left[ \tree{t}_{\ge a} \mapsto \tree{t}_{\ge b} \right]$
are in $\mc{T}$.







\begin{theorem}\label{theorem:abstract-tree-theorem}
For $I = \{1, \ldots, n\}$ and $i \in I$, let $\cl{S}_i$ be a class of
$\tau_i$-structures, $\mathsf{Str}_i: \mc{T} \rightarrow \cl{S}_i$ be
a representation map, and $\mc{L}_i$ be either FO or MSO. Then there
exist computable functions
$\eta_1, \eta_2: \mathbb{N}^n \rightarrow \mathbb{N}$, such that for
each $\tree{t} \in \mc{T}$ and $m_1, \ldots, m_n \in \mathbb{N}$, we
have the following.
\begin{enumerate}[nosep]
\item If $\mc{T}$ is closed under subtrees, and $\mathsf{Str}_i$ is
  $\mc{L}_i$-degree-reduction favourable for all $i \in I$, then there
  exists a subtree $\tree{s}_1$ of $\tree{t}$ in $\mc{T}$, of degree
  at most $\eta_1(m_1, \ldots, m_n)$, such that for all $i \in I$ (i)
  $\stri{i}{\tree{s}_1} \hookrightarrow \stri{i}{\tree{t}}$, and (ii)
  $\stri{i}{\tree{s}_1} \liequiv{m_i}{\mc{L}_i}
  \stri{i}{\tree{t}}$.\label{theorem:abstract-tree-theorem-degree-reduction}
\item 
If $\mc{T}$ is closed under rooted subtrees and under replacements
with rooted subtrees, and $\mathsf{Str}_i$ is
$\mc{L}_i$-height-reduction favourable for all $i \in I$, then there
exists a subtree $\tree{s}_2$ of $\tree{t}$ in $\mc{T}$, of height at
most $\eta_2(m_1, \ldots, m_n)$, such that for all $i \in I$ (i)
$\stri{i}{\tree{s}_2} \hookrightarrow \stri{i}{\tree{t}}$, and (ii)
$\stri{i}{\tree{s}_2} \liequiv{m_i}{\mc{L}_i}
\stri{i}{\tree{t}}$.\label{theorem:abstract-tree-theorem-height-reduction}
\end{enumerate}
\end{theorem}

\begin{proof}
We recall from
Section~\ref{section:background-FMT-adaptation-from-CMT-to-FMT} of
Chapter~\ref{chapter:background-FMT} that for a class $\cl{S}$ of
structures, $\Delta_{\mc{L}}(m, \cl{S})$ denotes the set of all
equivalence classes of the $\lequiv{m}$ relation restricted to the
structures in $\cl{S}$, and $\Lambda_{\cl{S}, \mc{L}}: \mathbb{N}
\rightarrow \mathbb{N}$ is a fixed computable function with the
property that $\Lambda_{\cl{S}, \mc{L}}(m) \ge |\Delta_{\mc{L}}(m,
\cl{S})|$.

(Part \ref{theorem:abstract-tree-theorem-degree-reduction}): For
$i \in I$, let $k_i$ be such that $\mathsf{Str}_i$ satisfies all the
transfer and monotonicity properties, except possibly \ref{B.1}, for
all $m \ge k_i$. Define $\eta_1:\mathbb{N}^n \rightarrow \mathbb{N}$
as follows: for $l_1, \ldots, l_n \in \mathbb{N}$,
$\eta_1(l_1, \ldots, l_n) = \Pi_{i \in
I} \,\Lambda_{\cl{S}_i, \mc{L}_i}(\text{max}(k_i, l_i))$. Then
$\eta_1$ is computable. Now given $m_1, \ldots, m_n \in \mathbb{N}$,
let $p = \eta_1(m_1, \ldots, m_n)$.  If $\tree{t}$ has degree $\leq
p$, then putting $\tree{s}_1 = \tree{t}$ we are done. Else, some node
of $\tree{t}$, say $a$, has degree $r > p$. Let $\tree{z}
= \tree{t}_{\ge a}$. Let $a_1, \ldots, a_r$ be the sequence of
children of $\troot{\tree{z}}$ in $\tree{z}$. For $j \in \{1, \ldots,
r\}$, let $\tree{x}_j$, resp. $\tree{y}_j$, be the subtree of
$\tree{z}$ obtained from $\tree{z}$ by deleting the subtrees rooted at
$a_j, a_{j+1}, \ldots, a_r$, resp.  deleting the subtrees rooted at
$a_1, a_2, \ldots, a_{j-1}$. Then $\tree{z} = \tree{y}_1
= \tree{x}_j \odot \tree{y}_j$ for all $j \in \{2, \ldots, r\}$.  Let
$q_i = \text{max}(k_i, m_i)$ for $i \in I$ and let $g: \{1, \ldots,
r\} \rightarrow \Pi_{i \in I} \,\Delta_{\mc{L}_i}(q_i, \cl{S}_i)$ be
such that for $j \in \{1, \ldots, r\}$, $g(j)$ is the sequence
$(\delta_i)_{i \in I}$ where $\delta_i$ is the
$\liequiv{q_i}{\mc{L}_i}$ class of $\stri{i}{\tree{y}_j}$. Verify that
$|\Pi_{i \in I}\, \Delta_{\mc{L}_i}(q_i, \cl{S}_i)| \leq p$. Then
since $r > p$, there exist $j, k \in \{1, \ldots, r\}$ such that $j <
k$ and $g(j) = g(k)$, i.e. for all $i \in I$,
$\stri{i}{\tree{y}_j} \liequiv{q_i}{\mc{L}_i} \stri{i}{\tree{y}_k}$.
If $\tree{z}_1 = \tree{x}_j \odot \tree{y}_k$, then since $\mc{T}$ is
closed under subtrees, we have $\tree{z}_1 \in \mc{T}$.  By the
properties \ref{B.2} and \ref{A.2} above, we have
$\stri{i}{\tree{z}_1} \hookrightarrow \stri{i}{\tree{z}}$ and
$\stri{i}{\tree{z}_1} \liequiv{q_i}{\mc{L}_i} \stri{i}{\tree{z}}$, for
all $i \in I$.  By iteratively applying properties \ref{A.1.a}
and \ref{A.1.b} to the nodes along the path from $a$ to
$\troot{\tree{t}}$, we see that if $\tree{t}_1
= \tree{t} \left[\tree{z} \mapsto \tree{z}_1 \right]$, then
$\stri{i}{\tree{t}_1} \hookrightarrow \stri{i}{\tree{t}}$ and
$\stri{i}{\tree{t}_1} \liequiv{q_i}{\mc{L}_i} \stri{i}{\tree{t}}$, for
all $i \in I$. Observe that $\tree{t}_1$ has strictly lesser size than
$\tree{t}$. Recursing on $\tree{t}_1$, we eventually get a subtree
$\tree{s}_1$ of $\tree{t}$ of degree at most $p$ such that for all
$i \in I$, (i)
$\stri{i}{\tree{s}_1} \hookrightarrow \stri{i}{\tree{t}}$ and (ii)
$\stri{i}{\tree{s}_1} \liequiv{q_i}{\mc{L}_i} \stri{i}{\tree{t}}$. Since
$q_i = \text{max}(k_i, m_i)$, we have
$\stri{i}{\tree{s}_1} \liequiv{m_i}{\mc{L}_i} \stri{i}{\tree{t}}$ for
all $i \in I$.

(Part \ref{theorem:abstract-tree-theorem-height-reduction}): For
$i \in I$, let $r_i$ be such that $\mathsf{Str}_i$ satisfies conditions 
\ref{A.1.a}, \ref{A.1.b} and \ref{B.1}, for
all $m \ge r_i$. Define $\eta_2:\mathbb{N}^n \rightarrow \mathbb{N}$
as follows: for $l_1, \ldots, l_n \in \mathbb{N}$,
$\eta_1(l_1, \ldots, l_n) = \Pi_{i \in
I} \,\Lambda_{\cl{S}_i, \mc{L}_i}(\text{max}(r_i, l_i))$. Then
$\eta_2$ is computable. Now given $m_1, \ldots, m_n \in \mathbb{N}$,
let $p = \eta_2(m_1, \ldots, m_n)$.  If $\tree{t}$ has height $\leq
p$, then putting $\tree{s}_2 = \tree{t}$ we are done. Else there is a
path from the root of $\tree{t}$ to some leaf of $\tree{t}$, whose
length is $> p$. Let $A$ be the set of nodes appearing along this
path. For $i \in I$, let $q_i = \text{max}(r_i, m_i)$.  Consider the
function $h : A \rightarrow \Pi_{i \in
I}\, \Delta_{\mc{L}_i}(q_i, \cl{S}_i)$ such that for each $a \in A$,
$h(a) = (\delta_i)_{i \in I}$ where $\delta_i$ is the
$\liequiv{q_i}{\mc{L}_i}$ class of $\stri{i}{\tree{t}_{\ge a}}$.
Verify that $|\Pi_{i \in I} \Delta_{\mc{L}_i}(q_i, \cl{S}_i)| \leq
p$. Since $|A| > p$, there exist distinct nodes $a, b \in A$ such that
$a$ is an ancestor of $b$ in $\tree{t}$ and $h(a) = h(b)$. We have two
cases: (i) Node $a$ is the root of $\tree{t}$; then let $\tree{t}_2
= \tree{t}_{\ge b}$. (ii) Node $a$ is not the root of $\tree{t}$; then
let $\tree{t}_2 = \tree{t}\left[\tree{t}_{\ge a} \mapsto \tree{t}_{\ge
b} \right]$.  Since $\mc{T}$ is closed under rooted subtrees and under
replacements with rooted subtrees, we have $\tree{t}_{\ge
a}, \tree{t}_{\ge b}$ and $\tree{t}_2$ are all in $\mc{T}$.  By
property \ref{B.1}, $\stri{i}{\tree{t}_{\ge
b}} \hookrightarrow \stri{i}{\tree{t}_{\ge a}}$. Also since $h(a) =
h(b)$, we have $\stri{i}{\tree{t}_{\ge
b}} \liequiv{q_i}{\mc{L}_i} \stri{i}{\tree{t}_{\ge a}}$. Then by
iteratively applying properties \ref{A.1.a} and \ref{A.1.b} to the
nodes along the path from $a$ to $\troot{\tree{t}}$, we get in either
of the cases above, that
$\stri{i}{\tree{t}_2} \hookrightarrow \stri{i}{\tree{t}}$ and
$\stri{i}{\tree{t}_2} \liequiv{q_i}{\mc{L}_i} \stri{i}{\tree{t}}$, for
all $i \in I$. Observe that $\tree{t}_2$ has strictly less size than
$\tree{t}$. Recursing on $\tree{t}_2$, we eventually get a subtree
$\tree{s}_2$ of $\tree{t}$ of height at most $p$ such that for all
$i \in I$, (i)
$\stri{i}{\tree{s}_2} \hookrightarrow \stri{i}{\tree{t}}$ and (ii)
$\stri{i}{\tree{s}_2} \liequiv{q_i}{\mc{L}_i} \stri{i}{\tree{t}}$. Since
$q_i = \text{max}(r_i, m_i)$, we have
$\stri{i}{\tree{s}_2} \liequiv{m_i}{\mc{L}_i} \stri{i}{\tree{t}}$ for
all $i \in I$.
\end{proof}

Call a representation map $\mathsf{Str}: \mc{T} \rightarrow \cl{S}_1$
as \emph{size effective} if there is a computable function $f:
\mathbb{N} \rightarrow \mathbb{N}$ such that $|\str{\tree{t}}| \leq
f(|\tree{t}|)$ for all $\tree{t} \in \mc{T}$. Call $\mathsf{Str}$
\emph{onto upto isomorphism} if for every structure in $\cl{S}_1$,
there is an isomorphic structure that is in the range of
$\mathsf{Str}$.  For a given class $\cl{S}$, we say $\cl{S}$ admits an
\emph{$\mc{L}$-reduction favourable size effective ($\rfse$)} tree
representation\sindex[term]{reduction favourable!size effective
  representation}\sindex[symb]{$\rfse$} if there exist finite
alphabets $\sigmaint$ and $\sigmaleaf$, a class $\mc{T}$ of
representation-feasible trees over $\sigmaint \cup \sigmaleaf$, and a
size effective representation map $\mathsf{Str}: \mc{T} \rightarrow
\cl{S}$ that is onto upto isomorphism, such that $\mc{T}$ and
$\mathsf{Str}$ are of one of the following types:
\begin{itemize}[nosep]
\item $\rfse$-type I: $\mc{T}$ has bounded degree and is closed under
  rooted subtrees and under replacements with rooted subtrees, and
  $\mathsf{Str}$ is $\mc{L}$-height-reduction favourable.
\item $\rfse$-type II: $\mc{T}$ is closed under subtrees, and
  $\mathsf{Str}$ is both $\mc{L}$-height-reduction favourable and
  $\mc{L}$-degree-reduction favourable.
\end{itemize}
We say $\cl{S}$ admits an \emph{$\rfse$ tree representation
  schema}\sindex[term]{representation!schema} if for each $p \in
\mathbb{N}$, the class $\cl{S}_p$ (defined in
Section~\ref{section:sufficient-condition-for-lebsp}) admits an
$\rfse$ tree representation. We now have the following result.

\begin{lemma}\label{lemma:rfse-tree-rep-implies-lebsp}
Let $\cl{S}$ be a class of structures that admits an $\rfse$ tree
representation schema. Then $\lebsp{\cl{S}_p}{k}$ holds with a
computable witness function, for all $p, k \in \mathbb{N}$,
\end{lemma}
\begin{proof}
By Lemma~\ref{lemma:gebsp-and-ebsp}, it suffices to show that
$\lebsp{\cl{S}_r}{0}$ holds for $r = p \cdot (k+1)$, with a computable
witness function. Since $\cl{S}$ admits an $\rfse$ tree representation
schema, the class $\cl{S}_r$ admits an $\rfse$ tree
representation. Let the latter fact be witnessed by a
representation-feasible class $\mc{T}$ of trees, and a representation
map $\mathsf{Str}: \mc{T} \rightarrow \cl{S}_r$. We have two cases from
the definition of $\rfse$ tree representation:
\begin{enumerate}[leftmargin=*,nosep]
\item $\mc{T}$ and $\mathsf{Str}$ are of $\rfse$-type I: ~Given $\mf{A}
  \in \cl{S}_r$, let $\tree{t} \in \mc{T}$ be such that
  $\str{\tree{t}} \cong \mf{A}$ (this is guaranteed since
  $\mathsf{Str}$ is onto upto isomorphism).  Let
  $m \in \mathbb{N}$. By Theorem~\ref{theorem:abstract-tree-theorem},
  there is a computable function
  $\eta_2: \mathbb{N} \rightarrow \mathbb{N}$ and a subtree $\tree{s}$
  of $\tree{t}$ in $\mc{T}$ such that (i) the height of $\tree{s}$ is
  at most $h = \eta_2(m)$, (ii)
  $\str{\tree{s}} \hookrightarrow \str{\tree{t}}$ and (iii)
  $\str{\tree{s}} \lequiv{m} \str{\tree{t}}$. Since the degree of
  $\tree{s}$ is bounded, by say $d$, the size of $\tree{s}$ is at most
  $d^{h+1}$. Whereby if $f$ is the computable function witnessing the
  size effectiveness of $\mathsf{Str}$, then $|\str{\tree{s}}| \leq
  f(d^h)$. Let $\mc{B}$ be the substructure of $\mf{A}$ such that
  $\mf{B} \cong \str{\tree{s}}$. Since $\cl{S}$ is closed under
  isomorphisms, so is $\cl{S}_r$, whereby $\mf{B} \in \cl{S}_r$. We
  can now see that $\lebspcond(\cl{S}_r, \mf{A}, \mf{B}, 0,
  m, \mathsf{null}, \witfn{\cl{S}_r}{0}{\mc{L}})$ is true, where
  $\mathsf{null}$ is the empty tuple and
  $\witfn{\cl{S}_r}{0}{\mc{L}}(m) = f(d^h)$. Clearly
  $\witfn{\cl{S}_r}{0}{\mc{L}}$ is computable.

\item $\mc{T}$ and $\mathsf{Str}$ are of $\rfse$-type II: ~By using both
  parts of Theorem~\ref{theorem:abstract-tree-theorem} and reasoning
  similarly as in the previous case, we can show that
  $\lebsp{\cl{S}_r}{0}$ holds with a computable witness function given
  by $\witfn{\cl{S}}{0}{\mc{L}}(m) = f(d^{h+1})$ where $d = \eta_1(m)$
  and $h = \eta_2(m)$.  Here, $f$ is the computable function
  witnessing the size-effectiveness of $\mathsf{Str}$, and
  $\eta_1, \eta_2$ are the computable functions given by
  Theorem~\ref{theorem:abstract-tree-theorem}.
\end{enumerate}
\end{proof}


%% file: words-and-trees-and-nested-words.tex
\section{Words, trees and nested words}\label{section:words-trees-nestedwords}


\newcommand{\unorderedtrees}{\mathsf{Unordered}\text{-}\mathsf{trees}}
\newcommand{\orderedtrees}{\mathsf{Ordered}\text{-}\mathsf{trees}}
\newcommand{\orderedrankedtrees}{\mathsf{Ordered}\text{-}\mathsf{ranked}\text{-}\mathsf{trees}}

~Let $\Sigma$ be a finite alphabet. The notion of unordered and
ordered $\Sigma$-trees was already introduced in the previous
section. A $\Sigma$-tree whose underlying poset is a linear order is
called a \emph{$\Sigma$-word}. An ordered $\Sigma$-tree $\tree{t} =
(((A, \leq), \lesssim), \lambda)$ is said to be
\emph{ranked}\sindex[term]{tree!ranked} by a function $\rho: \Sigma
\rightarrow \mathbb{N}$ if the number of children of any internal node
$a$ of $\tree{t}$ is exactly $\rho(\lambda(a))$.

Nested words\sindex[term]{nested word} were introduced by Alur and
Madhusudan in~\cite{alur-madhu}. Intuitively speaking, a nested word
is a $\Sigma$-word equipped with a binary relation that is interpreted
as a matching. Formally, given a finite alphabet $\Sigma$, a
\emph{nested word over $\Sigma$}, henceforth also called a
\emph{nested $\Sigma$-word}, is a 4-tuple $(A, \leq, \lambda,
\leadsto)$, where $A$ is a finite set (called the \emph{set of
  positions}), $\leq$ is a total linear order on $A$, $\lambda : A
\rightarrow \Sigma$ is a labeling function, and $\leadsto$ is a binary
\emph{matching relation} on $A$. Each pair $(i, j) \in\, \leadsto$ is
called a \emph{nesting edge}, the position $i$ is called a \emph{call
  position}, and the position $j$ is called a \emph{return successor}.
The relation $\leadsto$ satisfies the following properties. Below, $i
< j$ denotes $\big( (i \leq j) \wedge (i \neq j) \big)$.
\begin{enumerate}[nosep]
\item Nesting edges go only forward: For $i, j \in A$, if $i \leadsto
  j$, then $i < j$.
\item No two nesting edges share a position: For $i \in A$, each of
  the sets $\{ j \in A \mid i \leadsto j \}$ and $\{ j \in A \mid j
  \leadsto i \}$ has cardinality at most 1.
\item Nesting edges do not cross: For $i_1, i_2, j_1, j_2 \in A$, if
  $i_1 \leadsto j_1$ and $i_2 \leadsto j_2$, then it is not the case
  that $i_1 < i_2 \leq j_1 < j_2$.
\end{enumerate}

\begin{remark}
The definition of nested words presented above corresponds to the
definition of nested words in~\cite{alur-madhu}, in which the nested
words \emph{do not have any pending calls or pending returns}. Hence,
the elements $+\infty$ and $-\infty$ present in the definition
in~\cite{alur-madhu} are not necessary here, and have been dropped in
the definition above without any loss of generality.
\end{remark}

For each of the classes of words, trees (unordered, ordered and
ranked) and nested words introduced above, the notion of
\emph{regularity} of a subclass is well-studied in the literature. For
words and nested words, this notion is defined in terms of finite
state (word) automata and finite state nested word
automata~\cite{alur-madhu}. For each of the aforementioned classes of
trees, regularity is defined in terms of variants of finite state tree
automata~\cite{tata}. All of these notions of regularity have been
shown to be equivalent to definability via $\mso$
sentences~\cite{alur-madhu, tata}. Therefore, in our result below, a
subclass $\cl{S}$ of any of the classes $\cl{S}'$ of words, trees and
nested words, is said to be \emph{regular} if it is definable over
$\cl{S}'$ using an $\mso$ sentence (in other words, $\cl{S}$ is the
class of models in $\cl{S}'$, of an $\mso$ sentence). The central
theorem of this section is now stated as follows.





\begin{theorem}\label{theorem:words-and-trees-and-nested-words-satisfy-lebsp}
Given a finite alphabet $\Sigma$ and a function $\rho: \Sigma
\rightarrow \mathbb{N}$, let ${\words}(\Sigma)$,\linebreak
${\unorderedtrees}(\Sigma)$, ${\orderedtrees}(\Sigma)$,
$\orderedrankedtrees(\Sigma, \rho)$ and $\nestedwords{\Sigma}$ denote
respectively, the classes of all $\Sigma$-words, all unordered
$\Sigma$-trees, all ordered $\Sigma$-trees, all ordered $\Sigma$-trees
ranked by $\rho$, and all nested $\Sigma$-words.  Let $\cl{S}$ be a
regular subclass of any of these classes.  Then $\lebsp{\cl{S}}{k}$
holds with a computable witness function for each $k \in \mathbb{N}$.
\end{theorem}

We devote the rest of this section to proving
Theorem~\ref{theorem:words-and-trees-and-nested-words-satisfy-lebsp}. 


\input{proof-of-lebsp-for-words-and-trees}

\input{proof-of-lebsp-for-nested-words}

%% file: proof-of-lebsp-for-words-and-trees.tex
\vspace{-5pt}
\subsection*{Proof of
  Theorem~\ref{theorem:words-and-trees-and-nested-words-satisfy-lebsp}
  for words and trees}\label{subsection:proof-for-words-and-trees}

We show $\mebsp{\cl{S}}{k}$ holds with a computable witness function
when $\cl{S}$ is exactly one of the classes ${\words}(\Sigma)$,
${\unorderedtrees}(\Sigma)$, ${\orderedtrees}(\Sigma)$, and
$\orderedrankedtrees(\Sigma, \rho)$.  That $\lebsp{\cdot}{k}$ holds
with a computable witness function for a regular subclass follows,
because (i) a regular subclass of any of the above classes is $\mso$
definable over the class, (ii) $\mebsp{\cdot}{k}$ and the
computability of witness function are preserved under $\mso$ definable
subclasses
(Lemma~\ref{lemma:lebsp-closure-under-set-theoretic-ops}(\ref{lemma:lebsp-closure-under-set-theoretic-ops:part-4})),
and (iii) $\mebsp{\cdot}{k}$ implies $\febsp{\cdot}{k}$, and any
witness function for the former is also a witness function for the
latter (Lemma~\ref{lemma:gebsp-and-ebsp}(\ref{lemma:gebsp-and-ebsp:4})).

Of the classes mentioned above, we show $\mebsp{\cl{S}}{k}$ holds with
a computable witness function for the case when $\cl{S}$ is either
${\unorderedtrees}(\Sigma)$ or
$\orderedrankedtrees(\Sigma, \rho)$. The proof for
${\orderedtrees}(\Sigma)$ can be done similarly.  That the result
holds for $\words(\Sigma)$ follows from the fact that $\words(\Sigma)$
is a subclass of ${\unorderedtrees}(\Sigma)$ that is hereditary over
the latter, and then by using
Lemma~\ref{lemma:lebsp-closure-under-set-theoretic-ops}(\ref{lemma:lebsp-closure-under-set-theoretic-ops:part-1}).

For our proofs, we need $\mso$ \emph{composition
lemmas} \sindex[term]{composition lemma} for unordered trees and
ordered trees. Composition results were first studied by Feferman and
Vaught, and subsequently by many others (see~\cite{makowsky}).  We
state the $\mso$ composition lemma first for ordered trees, towards
which we define some terminology. For a finite alphabet $\Omega$,
given ordered $\Omega$-trees $\tree{t}, \tree{s}$ and a non-root node
$a$ of $\tree{t}$, the \emph{join of $\tree{s}$ to $\tree{t}$ to the
right of $a$}, denoted $\tree{t} \cdot^{\rightarrow}_a \tree{s}$, is
defined as follows: Let $\tree{s}'$ be an isomorphic copy of
$\tree{s}$ whose set of nodes is disjoint with the set of nodes of
$\tree{t}$. Then $\tree{t} \cdot^{\rightarrow}_a \tree{s}$ is defined
upto isomorphism as the tree obtained by making $\tree{s}'$ as a new
child subtree of the parent of $a$ in $\tree{t}$, at the successor
position of the position of $a$ among the children of $\tree{t}$.  We
can similarly define the \emph{join of $\tree{s}$ to $\tree{t}$ to the
left of $a$}, denoted $\tree{t} \cdot^{\leftarrow}_a
\tree{s}$. Likewise, for $\tree{t}$ and $\tree{s}$ as above, if $a$ is
a leaf node of $\tree{t}$, we can define the \emph{join of $\tree{s}$
to $\tree{t}$ below $a$}, denoted
$\tree{t} \cdot^{\uparrow}_a \tree{s}$, as the tree obtained upto
isomorphism by making the root of $\tree{s}$ as a child of $a$. The
$\mso$ composition lemma for ordered trees can now be stated as
follows. The proof of this lemma is provided towards the end of this
section.

\begin{lemma}[Composition lemma for ordered trees]\label{lemma:mso-composition-lemma-for-ordered-trees}\sindex[term]{composition lemma!for ordered trees}
For a finite alphabet $\Omega$, let ${\tree{t}}_i, \tree{s}_i$ be
non-empty ordered $\Omega$-trees, and let $a_i$ be a non-root node of
$\tree{t}_i$, for each $i \in \{1, 2\}$. Let $m \ge 2$ and suppose
that $({\tree{t}}_1, a_1) \mequiv{m} ({\tree{t}}_2, a_2)$ and
${\tree{s}}_1
\mequiv{m} {\tree{s}}_2$. Then each of the following hold.
\begin{enumerate}[nosep]
\item
$(({\tree{t}}_1 \cdot^{\rightarrow}_{a_1} {\tree{s}}_1), a_1) \mequiv{m}
(({\tree{t}}_2 \cdot^{\rightarrow}_{a_2} {\tree{s}}_2), a_2)$\label{lemma:mso-composition-lemma-for-ordered-trees:part-1}
\item
$(({\tree{t}}_1 \cdot^{\leftarrow}_{a_1} {\tree{s}}_1), a_1) \mequiv{m}
(({\tree{t}}_2 \cdot^{\leftarrow}_{a_2} {\tree{s}}_2), a_2)$
\item
$(({\tree{t}}_1 \cdot^{\uparrow}_{a_1} {\tree{s}}_1), a_1) \mequiv{m}
  (({\tree{t}}_2 \cdot^{\uparrow}_{a_2} {\tree{s}}_2), a_2)$ if $a_1,
  a_2$ are leaf nodes of $\tree{t}_1, \tree{t}_2$ resp.
\end{enumerate}
\end{lemma}

We now state the $\mso$ composition lemma for unordered trees, towards
which we introduce terminology akin to that introduced above for
ordered trees. Given unordered trees $\tree{t}$ and $\tree{s}$, and a
node $a$ of $\tree{t}$, define the \emph{join of $\tree{s}$ to
$\tree{t}$ to $a$}, denoted $\tree{t} \cdot_a \tree{s}$, as follows:
Let $\tree{s}'$ be an isomorphic copy of $\tree{s}$ whose set of nodes
is disjoint with the set of nodes of $\tree{t}$. Then
$\tree{t} \cdot_a \tree{s}$ is defined upto isomorphism as the tree
obtained by making $\tree{s}'$ as a new child subtree of $a$ in
$\tree{t}$. The $\mso$ composition lemma for unordered trees is now as
stated below. The proof is similar to that of
Lemma~\ref{lemma:mso-composition-lemma-for-ordered-trees}, and is hence
skipped.

\begin{lemma}[Composition lemma for unordered trees]\label{lemma:mso-composition-lemma-for-unordered-trees}\sindex[term]{composition lemma!for unordered trees}
For a finite alphabet $\Omega$, let ${\tree{t}}_i, \tree{s}_i$ be
non-empty unordered $\Omega$-trees, and let $a_i$ be a node of
$\tree{t}_i$, for each $i \in \{1, 2\}$. For $m \in \mathbb{N}$,
suppose that $({\tree{t}}_1, a_1) \mequiv{m} ({\tree{t}}_2, a_2)$ and
${\tree{s}}_1 \mequiv{m} {\tree{s}}_2$. Then
$(({\tree{t}}_1 \cdot_{a_1} {\tree{s}}_1), a_1) \mequiv{m}
(({\tree{t}}_2 \cdot_{a_2} {\tree{s}}_2), a_2)$.
\end{lemma}

We now show $\lebsp{\cdot}{k}$ holds with a computable witness
function for each of the classes ${\unorderedtrees}(\Sigma)$ and
$\orderedrankedtrees(\Sigma, \rho)$.

\begin{enumerate}[leftmargin=*, nosep]
\item 
Let $\cl{S}$ be the class of all unlabeled unordered trees. We show
that $\cl{S}$ admits an $\mrfse$ tree representation schema. Then
using Lemma~\ref{lemma:rfse-tree-rep-implies-lebsp}, we get that
$\mebsp{\cl{S}_p}{k}$ holds with a computable witness function for
each $p \in \mathbb{N}$. Since there is a 1-1 correspondence between
$\cl{S}_p$ and ${\unorderedtrees}(\Sigma)$ when $|\Sigma| = p$, it
follows that 
\linebreak $\mebsp{{\unorderedtrees}(\Sigma)}{k}$ holds with a
computable witness function.

Consider the class $\cl{S}_p$ where $p \ge 1$. Let $\mc{T}$ be the
class of all representation-feasible trees over
$\sigmaint \cup \sigmaleaf$ where $\sigmaint = \sigmaleaf
= \{0, \ldots, p-1\}$.  There is a natural map
$\mathsf{Str}: \mc{T} \rightarrow \cl{S}_p$ that simply ``forgets''
the ordering among the children of any node of its input tree. More
precisely, for an ordered tree $(O, \lambda)$ over $\{0, \ldots,
p-1\}$ where $O = ((A, \leq), \lesssim)$, we have
$\mathsf{Str}((O, \lambda)) = ((A, \leq), \lambda)$. It is easy to see
that $\mathsf{Str}$ satisfies for $\mc{L} = \mso$, the
conditions \ref{A.1.a}, \ref{B.1} and \ref{B.2} stated in
Section~\ref{section:abstract-tree-theorem}. That $\mathsf{Str}$
satisfies for $\mc{L} = \mso$, the conditions \ref{A.1.b}
and \ref{A.2} for $m \ge 0$ follows from
Lemma~\ref{lemma:mso-composition-lemma-for-unordered-trees}
above. Then $\mathsf{Str}$ is $\mso$-height-reduction favourable and
$\mso$-degree-reduction favourable. That $\mc{T}$ is closed under
subtrees, and that $\mathsf{Str}$ is size effective and onto upto
isomorphism, are obvious. Then $\mc{T}$ and $\mathsf{Str}$ are of
$\mrfse$-type II. Since $p$ is arbitrary, we get that $\cl{S}$ admits
an $\mrfse$ tree representation schema.

\item 
Let $\cl{S}$ be the class of all unlabeled ordered trees ranked by
$\rho$. We show that $\cl{S}$ admits an $\mrfse$ tree representation
schema. Then by Lemma~\ref{lemma:rfse-tree-rep-implies-lebsp}, we get
that\linebreak $\mebsp{\cl{S}_p}{k}$ holds with a computable witness
function for each $p \in \mathbb{N}$. Since there is a 1-1
correspondence between $\cl{S}_p$ and
${\orderedrankedtrees}(\Sigma, \rho)$ when $|\Sigma| = p$, it follows
that $\mebsp{{\orderedrankedtrees}(\Sigma, \rho)}{k}$ holds with a
computable witness function.

Consider the class $\cl{S}_p$ where $p \ge 1$.  Let $\sigmaint
= \sigmaleaf = \{0, \ldots, p-1\}$, and let $\mc{T}$ be the class of
all representation-feasible trees over $\sigmaint \cup \sigmaleaf$,
that are ranked by $\rho$. Indeed, then $\mc{T} = \cl{S}_p$. That
$\mc{T}$ is of bounded degree, and is closed under rooted subtrees and
under replacements with rooted subtrees is clear.  Let
$\mathsf{Str}: \mc{T} \rightarrow \cl{S}_p$ be the identity map.  That
$\mathsf{Str}$ satisfies for $\mc{L} = \mso$, the
conditions \ref{A.1.a} and \ref{B.1}, and that $\mathsf{Str}$ is size
effective and onto upto isomorphism, are clear.  That $\mathsf{Str}$
satisfies \ref{A.1.b} for $m \ge 2$ follows from
Lemma~\ref{lemma:mso-composition-lemma-for-ordered-trees}. Whence
$\mathsf{Str}$ is $\mso$-height-reduction favourable, whereby $\mc{T}$
and $\mathsf{Str}$ are of \linebreak$\mrfse$-type I. Since $p$ is
arbitrary, we get that $\cl{S}$ admits an $\mrfse$ tree representation
schema.
\end{enumerate}

We now prove the $\mso$ composition lemma for ordered trees. 
\begin{proof}[Proof of Lemma~\ref{lemma:mso-composition-lemma-for-ordered-trees}]
Without loss of generality, we assume $\tree{t}_i$ and $\tree{s}_i$
have disjoint sets of nodes for $i \in \{1. 2\}$. We show the result
for part (\ref{lemma:mso-composition-lemma-for-ordered-trees:part-1})
above. The others are similar. Let $\tree{z}_i =
({\tree{t}}_i \cdot^{\rightarrow}_{a_i} {\tree{s}}_i)$ for $i \in \{1,
2\}$.

Let $\beta_1$ be the winning strategy of the duplicator in the $m$
round $\mef$ game between $(\tree{t}_1, a_1)$ and $(\tree{t}_2,
a_2)$. Let $\beta_2$ be the winning strategy of the duplicator in the
$m$ round $\mef$ game between $\tree{s}_1$ and $\tree{s}_2$. Observe
that since $m \ge 2$, we have $\beta_2$ is such that if the spoiler
picks $\troot{\tree{s}_1}$ (resp. $\troot{\tree{s}_2}$), then
$\beta_2$ will require the duplicator to pick $\troot{\tree{s}_2}$
(resp. $\troot{\tree{s}_1}$). We use this observation later on.  The
strategy $\alpha$ of the duplicator in the $m$-round $\mef$ game
between $(\tree{z}_1, a_1)$ and $(\tree{z}_2, a_2)$ is defined as
follows:
\begin{enumerate}[nosep]
\item Point move: (i) If the spoiler picks an element of $\tree{t}_1$
  (resp. $\tree{t}_2$), the duplicator picks the element of
  $\tree{t}_2$ (resp. $\tree{t}_1$) given by $\beta_1$.  (ii) If the
  spoiler picks an element of $\tree{s}_1$ (resp. $\tree{s}_2$), the
  duplicator picks the element of $\tree{s}_2$ (resp. $\tree{s}_1$)
  given by $\beta_2$.
\item Set move: If the spoiler picks a set $X$ from $\tree{z}_1$, then
  let $X = Y_1 \sqcup Y_2$ where $Y_1$ is a set of elements of
  $\tree{t}_1$ and $Y_2$ is a set of elements of $\tree{s}_1$. Let
  $Y_1'$ and $Y_2'$ be the sets of elements of $\tree{t}_2$ and
  $\tree{s}_2$ respectively, chosen according to strategies $\beta_1$
  and $\beta_2$. Then in the game between $(\tree{z}_1, a_1)$ and
  $(\tree{z}_2, a_2)$, the duplicator responds with the set $X' = Y_1'
  \cup Y_2'$. A similar choice of set is made by the duplicator from
  $\tree{z}_1$ when the spoiler chooses a set from $\tree{z}_2$.
\end{enumerate}

We now show that the strategy $\alpha$ is winning for the duplicator
in the $m$-round $\mef$ game between $(\tree{z}_1, a_1)$ and 
$(\tree{z}_2, a_2)$.

Let at the end of $m$ rounds, the vertices and sets chosen from
$\tree{z}_1$, resp. $\tree{z}_2$, be $e_1, \ldots, e_p$ and $E_1,
\ldots, E_r$, resp. $f_1, \ldots, f_p$ and $F_1, \ldots, F_r$, where
$p + r = m$. For $l \in \{1, \ldots, r\}$, let $E_l^t$, resp. $E_l^s$
be the intersection of $E_l$ with the nodes of $\tree{t}_1$,
resp. nodes of $\tree{s}_1$, and likewise, let $F_l^t$, resp. $F_l^s$
be the intersection of $F_l$ with the nodes of $\tree{t}_2$,
resp. nodes of $\tree{s}_2$.

Firstly, it is straightforward to verify that the labels of $e_i$ and
$f_i$ are the same for all $i \in \{1, \ldots, p\}$, and that for
$l \in \{1, \ldots, r\}$, $e_i$ is in $E_l^s$, resp. $E_l^t$, iff
$f_i$ is in $F_l^s$, resp. $F_l^t$, whereby $e_i \in E_l$ iff $f_i \in
F_l$.  For $1 \leq i, j \leq p$, if $e_i$ and $e_j$ both belong to
$\tree{t}_1$ or both belong to $\tree{s}_1$, then it is clear from the
strategy $\alpha$ described above, that $f_i$ and $f_j$ both belong
resp. to $\tree{t}_2$ or both belong to $\tree{s}_2$. It is easy to
verify from the description of $\alpha$ that for every binary relation
(namely, the ancestor-descendent-order $\leq$, and the
ordering-on-the-children-order $\lesssim$), the pair $(e_i, e_j)$ is
in the binary relation in $\tree{z}_1$ iff $(f_i, f_j)$ is in that
binary relation in $\tree{z}_2$.  Consider the case when without loss
of generality, $e_1 \in \tree{t}_1$ and $e_2 \in \tree{s}_1$. Then
$f_1 \in \tree{t}_2$ and $f_2 \in \tree{s}_2$. We have the following
cases. Assume that the ordered tree underlying $\tree{z}_i$ is
$((A_i, \leq_i), \lesssim_i)$ for $i \in \{1, 2\}$.
\begin{enumerate}[nosep]
\item $e_1 \lesssim_1 a_1$ and $e_2 = \troot{\tree{s}_1}$: Then we see
  that $f_1 \lesssim_2 a_2$ and $f_2 = \troot{\tree{s}_2$}. Observe that
  $f_2$ must be $\troot{\tree{s}_2}$ by the property of $\beta_2$
  stated at the outset. Whereby $e_1 \lesssim_1 e_2$ and $f_1
  \lesssim_2 f_2$. Likewise $e_1 \not\leq_1 e_2$, $e_2 \not\leq_1 e_1$ and
  $f_1 \not\leq_2 f_2$, $f_2 \not\leq_2 f_1$.
\item $e_1 \lesssim_1 a_1$ and $e_2 \neq \troot{\tree{s}_1}$: Then we
  see that $f_1 \lesssim_2 a_2$ and $f_2 \neq \troot{\tree{s}_2$}
  (again by the property of $\beta_2$ stated at the outset). Whereby
  $e_1 \not\lesssim_1 e_2$, $e_2 \not\lesssim_1 e_1$ and $f_1
  \not\lesssim_2 f_2$, $f_2 \not\lesssim_2 f_1$. Likewise, $e_1
  \not\leq_1 e_2$, $e_2 \not\leq_1 e_1$ and $f_1 \not\leq_2 f_2$, $f_2
  \not\leq_2 f_1$.
\item $a_1 \lesssim_1 e_1$, $a_1 \neq e_1$ and $e_2 =
  \troot{\tree{s}_1}$: Then we see that $a_2 \lesssim_2 f_1$, $a_2
  \neq f_1$ and $f_2 = \troot{\tree{s}_2$}. Observe that $f_2$ must be
  $\troot{\tree{s}_2}$ by the property of $\beta_2$ stated at the
  outset. Whereby $e_2 \lesssim_1 e_1$ and $f_2 \lesssim_2
  f_1$. Likewise $e_1 \not\leq_1 e_2$, $e_2 \not\leq_1 e_1$ and $f_1
  \not\leq_2 f_2$, $f_2 \not\leq_2 f_1$.
\item $a_1 \lesssim_1 e_1$, $a_1 \neq e_1$ and $e_2 \neq
  \troot{\tree{s}_1}$: Then we see that $a_2 \lesssim_2 f_1$, $a_2
  \neq f_1$ and $f_2 \neq \troot{\tree{s}_2$} (again by the property
  of $\beta_2$ stated at the outset). Whereby $e_1 \not\lesssim_1
  e_2$, $e_2 \not\lesssim_1 e_1$ and $f_1 \not\lesssim_2 f_2$, $f_2
  \not\lesssim_2 f_1$. Likewise, $e_1 \not\leq_1 e_2$, $e_2 \not\leq_1
  e_1$ and $f_1 \not\leq_2 f_2$, $f_2 \not\leq_2 f_1$.
\item $e_1 \neq a_1, e_1 \leq_1 a_1$: Then $f_1 \neq a_1, f_1 \leq_2
  a_2$. Whereby $e_1 \leq_1 e_2$ and $f_1 \leq_2 f_2$. This is because
  $e_1 \leq_1 c_1$ and $f_1 \leq_2 c_2$ where $c_1$ and $c_2$ are
  resp. the parents of $a_1$ and $a_2$ in $\tree{z}_1$ and
  $\tree{z}_2$. Also $e_1 \not\lesssim_1 e_2$, $e_2 \not\lesssim_1
  e_1$ and $f_1 \not\lesssim_2 f_2$, $f_2 \not\lesssim_2 f_1$.
\item $e_1$ and $e_2$ are not related by $\leq_1$ or $\lesssim_1$: Then
  $f_1$ and $f_2$ are also not related by $\leq_2$ or $\lesssim_2$.
\end{enumerate}
In all cases, we have that the pair $(e_i, e_j)$ is in $\leq_1$
(resp. $\lesssim_1$) iff $(f_i, f_j)$ is in $\leq_2$
(resp. $\lesssim_2$).
\end{proof}

%% file: proof-of-lebsp-for-nested-words.tex
\subsection*{Proof of
  Theorem~\ref{theorem:words-and-trees-and-nested-words-satisfy-lebsp}
  for nested words}\label{subsection:proof-for-nested-words}


We first prove a composition lemma for nested words. Towards the
statement of this lemma, we define the notion of \emph{insert of a
nested word $\nesword{v}$ in a nested word $\nesword{u}$ at a given
position $e$ of $\nesword{u}$}.

\begin{defn}[Insert]
Let $\nesword{u} =
(A_{\nesword{u}}, \leq_{\nesword{u}}, \lambda_{\nesword{u}}, \leadsto_{\nesword{u}})$
and $\nesword{v} =
(A_{\nesword{v}}, \leq_{\nesword{v}}, \lambda_{\nesword{v}}, \leadsto_{\nesword{v}})$
be given nested $\Sigma$-words, and let $e$ be a position in
$\nesword{u}$. The \emph{insert of $\nesword{v}$ in $\nesword{u}$ at
$e$}, denoted $\nesword{u}\uparrow_{e}\nesword{v}$, is a nested
$\Sigma$-word defined as below.
\begin{enumerate}[nosep]
\item 
If $\nesword{u}$ and $\nesword{v}$ have disjoint sets of positions,
then $\nesword{u}\uparrow_{e}\nesword{v} =
(A, \leq, \lambda, \leadsto)$ where
\begin{itemize}[nosep]
\item $A = A_{\nesword{u}} \sqcup A_{\nesword{v}}$
\item $\leq \,=\, \leq_{\nesword{u}} \cup \leq_{\nesword{v}}$ $\cup
  \,\{(i, j) \mid i \in A_{\nesword{u}}, j \in A_{\nesword{v}}, i
  \leq_{\nesword{u}} e\}$ $\cup\, \{(j, i) \mid i \in A_{\nesword{u}},
  j \in A_{\nesword{v}}, e \leq_{\nesword{u}} i, e \neq i\}$
\item $\lambda(a) = \lambda_{\nesword{u}}(a)$ if $a \in
  A_{\nesword{u}}$, else $\lambda(a) = \lambda_{\nesword{v}}(a)$
\item $\leadsto \,=\, \leadsto_{\nesword{u}} \cup \leadsto_{\nesword{v}}$
\end{itemize}
\item 
If $\nesword{u}$ and $\nesword{v}$ have overlapping sets of positions,
then let $\nesword{v}_1$ be an isomorphic copy of $\nesword{v}$ whose
set of positions is disjoint with that of $\nesword{u}$. Then
$\nesword{u}\uparrow_e\nesword{v}$ is defined upto isomorphism as
$\nesword{u}\uparrow_e\nesword{v}_1$.
\end{enumerate}
In the special case that $e$ is the last (under $\leq_{\nesword{u}}$)
position of $\nesword{u}$, we denote
$\nesword{u} \uparrow_e \nesword{v}$ as
$\nesword{u} \cdot \nesword{v}$, and call the latter as
the \emph{concatenation of $\nesword{v}$ with $\nesword{u}$}.
\end{defn}

\begin{lemma}[Composition lemma for nested words]\label{lemma:nested-words-composition-lemma}\sindex[term]{composition lemma!for nested words}
For a finite alphabet $\Sigma$, let
$\nesword{u}_i, \nesword{v}_i \in \nestedwords{\Sigma}$, and let $e_i$
be a position in $\nesword{u}_i$ for $i \in \{1, 2\}$. Then the
following hold for each $m \in \mathbb{N}$.
\begin{enumerate}[nosep]
\item If $(\nesword{u}_1, e_1) \lequiv{m} (\nesword{u}_2, e_2)$ and
  $\nesword{v}_1 \lequiv{m} \nesword{v}_2$, then
  $(\nesword{u}_1\uparrow_{e_1}\nesword{v}_1) \lequiv{m}
  (\nesword{u}_2\uparrow_{e_2}\nesword{v}_2)$.
\item 
$\nesword{u}_1 \lequiv{m} \nesword{u}_2$ and
  $\nesword{v}_1 \lequiv{m} \nesword{v}_2$, then
  $\nesword{u}_1\cdot\nesword{v}_1 \lequiv{m} \nesword{u}_2\cdot\nesword{v}_2$.
\end{enumerate}
\end{lemma}

\begin{proof}
We give the proof for $\mc{L} = $MSO. The proof for $\mc{L} = $FO is
similar.  

The winning strategy $S$ for the duplicator in the $m$-round $\mef$
game between $\nesword{u}_1\uparrow_{e_1}\nesword{v}_1$ and
$\nesword{u}_2\uparrow_{e_2}\nesword{v}_2$ is simply the composition
of the winning strategies $S_1$, resp. $S_2$, of the duplicator in the
$m$-round $\mef$ game between $(\nesword{u}_1, e_1)$ and
$(\nesword{u}_2, e_2)$, resp.  $\nesword{v}_1$ and
$\nesword{v}_2$. Formally, $S$ is defined as follows. 
\begin{enumerate}[nosep]
\item Point move: If the spoiler picks an element of $\nesword{u}_1$,
resp. $\nesword{v}_1$, from
$\nesword{u}_1\uparrow_{e_1}\nesword{v}_1$, then the duplicator picks
the element of $\nesword{u}_2$, resp. $\nesword{v}_2$, from
$\nesword{u}_2\uparrow_{e_2}\nesword{v}_2$, that is given by the
strategy $S_1$, resp. $S_2$.  A similar choice of an element from
$\nesword{u}_1\uparrow_{e_1}\nesword{v}_1$ is made by the duplicator
if the spoiler picks an element from
$\nesword{u}_2\uparrow_{e_2}\nesword{v}_2$.
\item Set move: If the spoiler picks a set $Z$ from
$\nesword{u}_1\uparrow_{e_1}\nesword{v}_1$, then let $Z = X \sqcup Y$
where $X$ is a subset of positions of $\nesword{u}_1$ and $Y$ is a
subset of positions of $\nesword{v}_1$. Then the duplicator picks the
set $Z'$ from $\nesword{u}_2\uparrow_{e_2}\nesword{v}_2$ where $Z' =
X' \sqcup Y'$, $X'$ is the subset of positions of $\nesword{u}_2$ that
is chosen by the duplicator in response to $X$ according to strategy
$S_1$, and $Y'$ is the subset of positions of $\nesword{v}_2$ that is
chosen by the duplicator in response to $Y$ according to strategy
$S_2$.  A similar choice of a set from
$\nesword{u}_1\uparrow_{e_1}\nesword{v}_1$ is made by the duplicator
if the spoiler picks a set from
$\nesword{u}_2\uparrow_{e_2}\nesword{v}_2$.  
\end{enumerate}
It is easy to see that $S$ is a winning strategy in the $\mef$ game
between $\nesword{u}_1\uparrow_{e_1}\nesword{v}_1$ and
$\nesword{u}_2\uparrow_{e_2}\nesword{v}_2$.
\end{proof}

Towards the proof of
Theorem~\ref{theorem:words-and-trees-and-nested-words-satisfy-lebsp}
for nested words, we first observe that each nested $\Sigma$-word has
a natural representation using a representation-feasible tree over
$\sigmaint \cup \sigmaleaf$, where $\sigmaleaf = \Sigma \cup
(\Sigma \times \Sigma)$, and $\sigmaint = \sigmaleaf \cup \{\circ\}$,
We demonstate this for the example of the nested $\Sigma$-word
$\nesword{w} = (abaabba, \{(2, 6), (4, 5)\})$, where $\Sigma = \{a,
b\}$. See Figure~\ref{figure:nested-word-as-tree}.
\vspace{-15pt}
\begin{figure}[h]
\begin{minipage}{0.47\textwidth}
\centering
\includegraphics[scale=1.3]{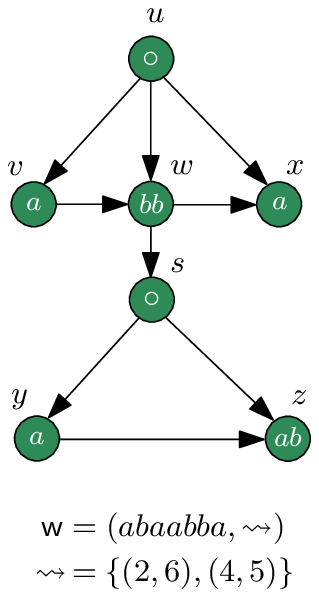}
\caption{\label{figure:nested-word-as-tree}Nested $\Sigma$-word as a
  tree over $\Sigma \cup (\Sigma \times \Sigma) \cup \{\circ\}$}
\end{minipage}\hfill
\begin{minipage}{0.53\textwidth}
\centering
\includegraphics[scale=1.3]{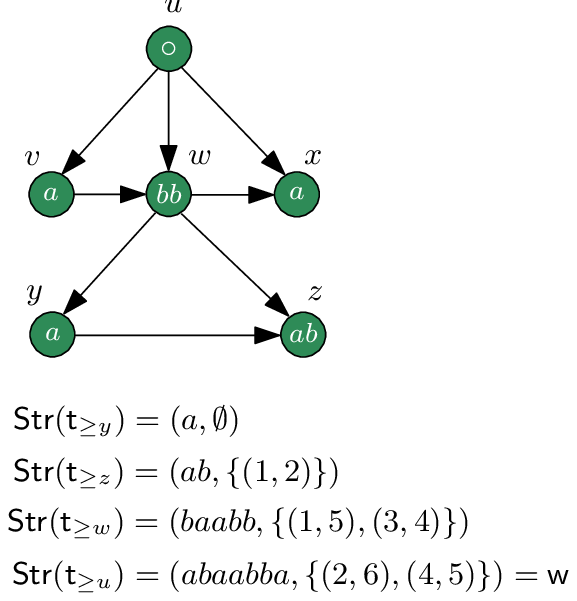}
\caption{\label{figure:tree-as-nested-word}Tree over $\Sigma \cup (\Sigma \times \Sigma) \cup \{\circ\}$ as a nested $\Sigma$-word}
\end{minipage}
\end{figure}

Formally, each non-empty nested $\Sigma$-word can be seen to be of one
of two types. A non-empty nested word $\nesword{u} = (A, \leq,
\lambda, \leadsto)$ is said to be of \emph{type A} if either
$\leadsto$ is empty and $|A| = 1$ (i.e. $\nesword{u}$ is really a
$\Sigma$-word of length 1), or for the minimum and maximum (under
$\leq$) positions $i$ and $j$ respectively of $\nesword{u}$, it is the
case that $i \leadsto j$. A non-empty nested word $\nesword{u}$ is
said to be of \emph{type B} if it is not of type A. It is easy to see
that a type B nested word can be written as a concatenation of type A
nested words. We describe inductively, the tree-representation of
$\nesword{u} = (A, \leq, \lambda, \leadsto)$ below. We have three
cases.

\begin{enumerate}[nosep]
\item $\nesword{u}$ is empty:
Then the tree $\tree{t}$ over
$\sigmaint \cup \sigmaleaf$ representing $\nesword{u}$ is the empty tree.
\item $\nesword{u}$ is of type A: 
If $\leadsto$ is empty and $|A| = 1$, then let the only element of $A$
be labeled (by $\lambda$) with the letter $a$, where $a \in
\Sigma$. Then the tree $\tree{t}$ over $\sigmaint \cup \sigmaleaf$
representing $\nesword{u}$ is a singleton whose only node is labeled
with $a$.

Else, let $\nesword{u}_1$ be the nested sub-$\Sigma$-word of
$\nesword{u}$ induced by the positions $l \in A$ such that $i \leq l
\leq j, l \neq i$ and $l \neq j$, where $i$ and $j$ are respectively
the minimum and maximum (under $\leq$) positions of $\nesword{u}$. Let
$\tree{t}_1$ be the tree over $\sigmaint \cup \sigmaleaf$ representing
$\nesword{u}_1$, if the latter is not empty. Then the tree $\tree{t}$
over $\sigmaint \cup \sigmaleaf$ representing $\nesword{u}$ is defined
as follows. If $\nesword{u}_1$ is empty, then $\tree{t}$ is a
singleton whose only node is labeled with the label $(\lambda(i),
\lambda(j))$. Else, $\tree{t}$ is such that (i) the label of the
$\troot{\tree{t}}$ is $(\lambda(i), \lambda(j))$ and (ii) the only
child of $\troot{\tree{t}}$ is $\troot{\tree{t}_1}$, i.e. $\tree{t}_1$
is the only child subtree of $\troot{\tree{t}}$.
\item $\nesword{u}$ is of type B:
Then $\nesword{u}$ can be written as $\nesword{u}
= \nesword{u}_1 \cdots \nesword{u}_n$ where $\nesword{u}_i$ is a type
A nested $\Sigma$-word for $i \in \{1, \ldots, n\}$. Let $\tree{t}_i$
be the tree over $\sigmaint \cup \sigmaleaf$ representing
$\nesword{u}_i$. Then the tree $\tree{t}$ over
$\sigmaint \cup \sigmaleaf$ representing $\nesword{u}$ is such that
(i) the label of $\troot{\tree{t}}$ is $\circ$, and (ii) the children
of $\troot{\tree{t}}$ in ``increasing order'' are
$\troot{\tree{t}_1}, \ldots, \troot{\tree{t}_n}$.
\end{enumerate}

Conversely, each representation-feasible tree $\tree{t}$ over
$\sigmaint \cup \sigmaleaf$ represents a nested $\Sigma$-word
$\nesword{u}_{\tree{t}}$. We demonstrate this for the example of the
nested $\Sigma$-word $\nesword{w} = (abaabba, \{(2, 6), (4, 5)\})$
where $\Sigma = \{a, b\}$, in
Figure~\ref{figure:tree-as-nested-word}. Formally, we see this
inductively as follows.
\begin{enumerate}[nosep]
\item If $\tree{t}$ is empty, then $\nesword{u}_{\tree{t}}$ is the
  empty nested $\Sigma$-word.
\item If
$\tree{t} = (O, \lambda)$ contains only a single node, say $e$, then
there are two cases. If $\lambda(e) \in \Sigma$, then
$\nesword{u}_{\tree{t}} =
(A_{\tree{t}}, \leq_{\tree{t}}, \lambda_{\tree{t}}, \leadsto_{\tree{t}})$
where $A_{\tree{t}} = \{e_1\}, \leq_{\tree{t}} = \{(e_1,
e_1)\}, \lambda_{\tree{t}}(e_1) = \lambda(e)$ and $\leadsto_{\tree{t}}
= \emptyset$. Else, i.e. if $\lambda(e) = (a,
b) \in \Sigma \times \Sigma$, then $\nesword{u}_{\tree{t}} =
(A_{\tree{t}}, \leq_{\tree{t}}, \lambda_{\tree{t}}, \leadsto_{\tree{t}})$
where $A_{\tree{t}} = \{e_1, e_2\}, \leq_{\tree{t}} = \{(e_1, e_1),
(e_1, e_2), (e_2, e_2)\}, \lambda_{\tree{t}}(e_1) =
a, \lambda_{\tree{t}}(e_2) = b$ and $\leadsto_{\tree{t}} = \{(e_1,
e_2)\}$.
\item If $\tree{t} = (O, \lambda)$ contains more than one node, then let
$\tree{t}_1, \ldots, \tree{t}_n$ be, in ``increasing order'', the
subtrees of $\tree{t}$ rooted at the children of
$\troot{\tree{t}}$. Let $\nesword{v}
= \nesword{u}_{\tree{t}_1} \cdots \nesword{u}_{\tree{t}_n}$, where
$\nesword{u}_{\tree{t}_i}$ is the nested $\Sigma$-word represented by
$\tree{t}_i$, for $i \in \{1, \ldots, n\}$. If
$\lambda(\troot{\tree{t}}) = \circ$, then $\nesword{u}_{\tree{t}}
= \nesword{v}$. Else, suppose that $\lambda(\troot{\tree{t}}) = (a,
b)$ where $a, b \in \Sigma$. Let $\nesword{w}$ be the 2-letter
$\Sigma$-word given by $\nesword{w} =
(A_{\nesword{w}}, \leq_{\nesword{w}}, \lambda_{\nesword{w}}, \leadsto_{\nesword{w}})$
where $A_{\nesword{w}} = \{e_1, e_2\}, \leq_{\nesword{w}} = \{(e_1,
e_1), (e_1, e_2), (e_2, e_2)\}, \lambda_{\nesword{w}}(e_1) =
a, \lambda_{\nesword{w}}(e_2) = b$ and $ \leadsto_{\nesword{w}}
= \{(e_1, e_2)\}$. Then $\nesword{u}_{\tree{t}}
= \nesword{w} \uparrow_{e_1} \nesword{v}$.
\end{enumerate}

We now prove
Theorem~\ref{theorem:words-and-trees-and-nested-words-satisfy-lebsp}
for nested words. It suffices to show that $\lebsp{\cdot}{k}$ holds
with a computable witness function for $\nestedwords{\Sigma}$. That
any regular subclass of \linebreak $\nestedwords{\Sigma}$ satisfies
$\lebsp{\cdot}{k}$ follows, because (i) a regular subclass of
\linebreak$\nestedwords{\Sigma}$ is $\mso$ definable, (ii)
$\mebsp{\cdot}{k}$ and the computability of witness function are
preserved under $\mso$ definable subclasses
(Lemma~\ref{lemma:lebsp-closure-under-set-theoretic-ops}(\ref{lemma:lebsp-closure-under-set-theoretic-ops:part-4})),
and \linebreak(iii) $\mebsp{\cdot}{k}$ implies $\febsp{\cdot}{k}$, and
any witness function for the former is also a witness function for the
latter
(Lemma~\ref{lemma:gebsp-and-ebsp}(\ref{lemma:gebsp-and-ebsp:4})).

Let $\cl{S}=\nestedwords{\Sigma}$. Consider the class $\cl{S}_p$ where
$p \ge 1$.  Let $\mc{T}$ be the class of all representation-feasible
trees over $\sigmaint' \cup \sigmaleaf'$ where $\sigmaint'
= \sigmaleaf' \cup \{\circ\}$, $\sigmaleaf' = \Sigma' \cup
(\Sigma' \times \Sigma')$ and $\Sigma' = \Sigma \times \{0, \ldots,
p-1\}$, Let $\mathsf{Str}: \mc{T} \rightarrow \cl{S}_p$ be the map
given by $\mathsf{Str}(\tree{t})$ is the nested $\Sigma'$-word
represented by $\tree{t}$ as described above. That $\mathsf{Str}$ is
size effective and onto upto isomorphism, is clear. That $\mc{T}$ is
closed under subtrees, and that $\mathsf{Str}$ satisfies the
conditions \ref{A.1.a}, \ref{B.1} and \ref{B.2} of
Section~\ref{section:abstract-tree-theorem} are easy to see. That
$\mathsf{Str}$ satisfies \ref{A.1.b} for $m \ge 0$ follows directly
from Lemma~\ref{lemma:nested-words-composition-lemma}. We show below
that $\mathsf{Str}$ satisfies \ref{A.2} for all $m \ge 2$. Then
$\mathsf{Str}$ is $\mc{L}$-height-reduction favourable and
$\mc{L}$-degree-reduction favourable, whereby $\mc{T}$ and
$\mathsf{Str}$ are of $\rfse$ type II. Then $\cl{S}$ admits an $\rfse$
tree representation schema, whereby using
Lemma~\ref{lemma:rfse-tree-rep-implies-lebsp}, it follows that
$\lebsp{\cl{S}}{k}$ holds with a computable witness function.

Let $\tree{t}, \tree{s}_1, \tree{s}_2 \in \mc{T}$ be trees of size
$\ge 2$ such that the roots of all these trees have the same label,
and suppose $\tree{z}_i = \tree{t} \odot \tree{s}_i \in \mc{T}$ for
$i \in \{1, 2\}$. Assume
$\mathsf{Str}(\tree{s}_1) \lequiv{m} \mathsf{Str}(\tree{s}_2)$ for
$m \ge 2$. If the label of $\troot{\tree{t}}$ is $\circ$, then
$\mathsf{Str}(\tree{z}_i)
= \mathsf{Str}(\tree{t}) \cdot \mathsf{Str}(\tree{s}_i)$ for
$i \in \{1, 2\}$ whereby from
Lemma~\ref{lemma:nested-words-composition-lemma},
$\mathsf{Str}(\tree{z}_1) \lequiv{m} \mathsf{Str}(\tree{z}_2)$. Else
suppose the label of $\troot{\tree{t}}$ is $(a, b)$ where $a,
b \in \Sigma$. Let $\nesword{u}$ be the concatenation, in ``increasing
order'', of the nested $\Sigma$-words represented by the subtrees of
$\tree{t}$ rooted at the children of the root of $\tree{t}$. Likewise,
let $\nesword{v}_i$ be the concatenation, in ``increasing order'', of
the nested $\Sigma$-words represented by the subtrees of $\tree{s}_i$
rooted at the children of the root of $\tree{s}_i$, for $i \in \{1,
2\}$.  Let $\nesword{w}$ be the 2-letter $\Sigma$-word given by
$\nesword{w} =
(A_{\nesword{w}}, \leq_{\nesword{w}}, \lambda_{\nesword{w}}, \leadsto_{\nesword{w}})$
where $A_{\nesword{w}} = \{e_1, e_2\}, \leq_{\nesword{w}} = \{(e_1,
e_1), (e_1, e_2), (e_2, e_2)\}, \lambda_{\nesword{w}}(e_1) =
a, \lambda_{\nesword{w}}(e_2) = b$ and $\leadsto_{\nesword{w}}
= \{(e_1, e_2)\}$. Then $\mathsf{Str}(\tree{z}_i)
= \nesword{w} \uparrow_{e_1} (\nesword{u} \cdot \nesword{v}_i)$. Now
observe that since
$\mathsf{Str}(\tree{s}_1) \lequiv{m} \mathsf{Str}(\tree{s}_2)$, we
have $\nesword{v}_1 \lequiv{m} \nesword{v}_2$ for each $m \ge 2$. Then
by Lemma~\ref{lemma:nested-words-composition-lemma}, we have
$\mathsf{Str}(\tree{z}_1) \lequiv{m} \mathsf{Str}(\tree{z}_2)$,
completing the proof.

%% file: n-partite-cographs.tex
\section{$n$-partite cographs}\label{section:n-partite-cographs}

\newcommand{\labelednpartitecographs}{\ensuremath{\mathsf{Labeled}\text{-}\mathsf{n}\text{-}\mathsf{partite}\text{-}\mathsf{cographs}}}

The class of $n$-partite cographs was introduced by Ganian et. al.
in~\cite{shrub-depth}.  An \emph{$n$-partite
  cograph}\sindex[term]{$n$-partite cograph} $G$ is a graph that
admits an $n$-partite \emph{cotree representation} $\mathsf{t}$. Here
$\mathsf{t}$ is an unordered tree whose leaves are exactly the
vertices of $G$, and are labeled with labels from $\left[n\right] =
\{1, \ldots, n\}$. Each internal node $v$ of $\mathsf{t}$ is labeled
with a binary symmetric function $f_v: \left[n\right] \times
\left[n\right] \rightarrow \{0, 1\}$ such that two vertices $a$ and
$b$ of $G$ with respective labels $i$ and $j$, are adjacent in $G$ iff
the greatest common ancestor of $a$ and $b$ in $\mathsf{t}$, call it
$c$, is such that $f_c(i, j) = 1$. Given below is an example of an
$n$-partite cograph $G$ and a cotree representation $\mathsf{t}$ of
it.

\begin{figure}[h]
\centering
\includegraphics[scale=0.8]{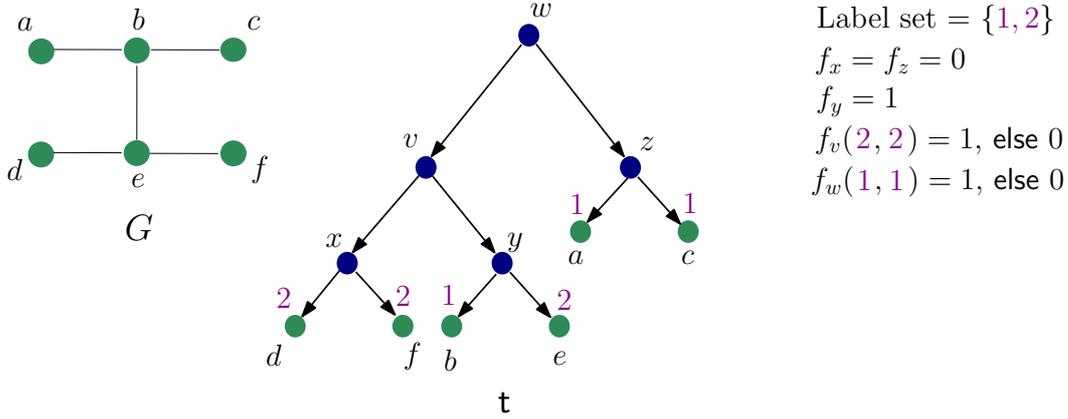}
\caption{\label{figure:n-partite-cograph}$n$-partite cograph $G$ and an $n$-partite cotree representation $\mathsf{t}$ of $G$}
\end{figure}

Given a finite alphabet $\Sigma$, a \emph{$\Sigma$-labeled $n$-partite
  cograph} is a pair $(G, \nu)$ where $G$ is an $n$-partite cograph
and $\nu: V \rightarrow \Sigma$ is a labeling function.  Recall that
given a class $\cl{S}_1$ of structures and a subclass $\cl{S}_2$ of
$\cl{S}_1$, we say $\cl{S}_2$ is \emph{hereditary over $\cl{S}_1$}, if
$\cl{S}_2$ is $PS$ over $\cl{S}_1$ (see the last paragraph of
Chapter~\ref{chapter:background-FMT}). A class
$\cl{S}$ is \emph{hereditary} if it is hereditary over the class of
all (finite) structures.

The central result of this section can now be stated as follows.

\begin{theorem}\label{theorem:n-partite-cographs-satisfy-lebsp}
Given $n, k \in \mathbb{N}$, let $\labelednpartitecographs(\Sigma)$ be
the class of all $\Sigma$-labeled $n$-partite cographs. Let $\cl{S}$
be any subclass of $\labelednpartitecographs(\Sigma)$, that is
hereditary over the latter. Then $\lebsp{\cl{S}}{k}$ holds with a
computable witness function.  Consequently, each of the following
classes of graphs satisfies $\lebsp{\cdot}{k}$ with a computable
witness function for each $k \ge 0$.
\begin{enumerate}[nosep]
\item Any hereditary class of $n$-partite cographs, for each $n \in
  \mathbb{N}$.
\item Any hereditary class of graphs of bounded shrub-depth.
\item Any hereditary class of graphs of bounded $\mc{SC}$-depth.
\item Any hereditary class of graphs of bounded tree-depth.
\item Any hereditary class of cographs.
\end{enumerate}
\end{theorem}

\begin{proof}

We first show that the class $\cl{S} = \ncographs$, where $\ncographs$
is the class of all $n$-partite cographs, admits an $\mrfse$ tree
representation schema. Then by
Lemma~\ref{lemma:rfse-tree-rep-implies-lebsp}, we have
$\mebsp{\cl{S}_p}{k}$ holds with a computable witness function for
each $p, k \in \mathbb{N}$. Then $\lebsp{\cl{S}_p}{k}$ holds with a
computable witness function for each $p, k \in \mathbb{N}$ by
Lemma~\ref{lemma:gebsp-and-ebsp}(\ref{lemma:gebsp-and-ebsp:4}). Since
there is a 1-1 correspondence between
$\labelednpartitecographs(\Sigma)$ and $\cl{S}_p$ if $|\Sigma| = p$,
it follows that $\lebsp{\cdot}{k}$ holds with a computable witness
function for $\labelednpartitecographs(\Sigma)$. Whereby, the same
holds of any subclass of \linebreak $\labelednpartitecographs(\Sigma)$
that is hereditary over the latter by
Lemma~\ref{lemma:lebsp-closure-under-set-theoretic-ops}(\ref{lemma:lebsp-closure-under-set-theoretic-ops:part-1}). That
\linebreak $\lebsp{\cdot}{k}$ holds with a computable witness function
for the various specific classes mentioned in the statement of this
result, follows from the fact that these classes are hereditary
subclasses of $\ncographs$, and the fact that $\ncographs$ is itself
hereditary~\cite{shrub-depth}.

Consider $\cl{S}_p$ for $p \ge 1$. Let $\sigmaleaf = \left[n\right]
\times \{0, \ldots, p-1\}$ and $\sigmaint = \{ f \mid f:
\left[n\right] \times \left[n\right] \rightarrow \{0, 1\}\}$.  Let
$\mc{T}$ be the class of all representation-feasible $(\sigmaint \cup
\sigmaleaf)$-trees.  Then $\mc{T}$ is closed under subtrees.  Let
$\mathsf{Str}: \mc{T} \rightarrow \cl{S}_p$ be such that for $\tree{t}
= (O, \lambda) \in \mc{T}$, we have $\str{\tree{t}} = (G, \nu)$ where
(i) $G$ is the $n$-partite cograph represented by the unordered tree
obtained from $\tree{t}$ by ``forgetting'' the ordering among the
children of $\tree{t}$ and by dropping the second component of the
labels of the leaves of $\tree{t}$, and (ii) $\nu$ is such that for
any vertex $a$ of $G$ (which is a leaf node of $\tree{t}$), it is the
case that $\nu(a)$ is the second component of $\lambda(a)$.  It is
easily seen that $\mathsf{Str}$ is size effective and onto upto
isomorphism. It is also easy to see that $\mathsf{Str}$ satisfies
conditions \ref{A.1.a}, \ref{B.1} and \ref{B.2} of
Section~\ref{section:abstract-tree-theorem}. We now show below that
$\mathsf{Str}$ satisfies for $\mc{L} = \mso$, the conditions
\ref{A.1.b} and \ref{A.2} for $m \ge 0$; then $\mc{T}$ and
$\mathsf{Str}$ are of $\rfse$-type II. Whereby, $\cl{S}$ admits an
$\rfse$ tree representation schema, completing the proof.


For our proof, we need the following composition lemma. We prove this
lemma towards the end of this section.

\begin{lemma}[Composition lemma for $n$-partite cographs]\label{lemma:composition-lemma-for-n-partite-cographs}\sindex[term]{composition lemma!for $n$-partite cographs}
For $i \in \{1, 2\}$, let $(G_i, \nu_{i, 1})$ and $(H_i, \nu_{i, 2})$
be graphs in $\cl{S}_p$. Suppose $\tree{t}_i$ and $\tree{s}_i$ are
trees of $\mc{T}$ such that $\str{\tree{t}_i} = (G_i, \nu_{i, 1})$,
$\str{\tree{s}_i} = (H_i, \nu_{i, 2})$, and the labels of
$\troot{\tree{t}_i}$ and $\troot{\tree{s}_i}$ are the same. Let
$\tree{z}_i = \tree{t}_i \odot \tree{s}_i$ and $\str{\tree{z}_i} =
(Z_i, \nu_{i})$ for $i \in \{1, 2\}$. For each $m \in \mathbb{N}$, if
$(G_1, \nu_{1, 1}) \mequiv{m} (G_2, \nu_{2, 1})$ and $(H_1, \nu_{1,
  2}) \mequiv{m} (H_2, \nu_{2, 2})$, then $(Z_1, \nu_{1}) \mequiv{m}
(Z_2, \nu_{2}) $.
\end{lemma}

We now show that
Lemma~\ref{lemma:composition-lemma-for-n-partite-cographs} implies
that $\mathsf{Str}$ satisfies \ref{A.1.b} and \ref{A.2} for $m \ge 0$
and $\mc{L} = \mso$.

That $\mathsf{Str}$ satisfies \ref{A.2} for $m \ge 0$ and $\mc{L} =
\mso$ follows easily from
Lemma~\ref{lemma:composition-lemma-for-n-partite-cographs}. Let
$\tree{t}, \tree{s}_1, \tree{s}_2 \in \mc{T}$ be trees of size $\ge 2$
such that the roots of all these trees have the same label. For $i \in
\{1, 2\}$, suppose $\tree{z}_i = \tree{t} \odot \tree{s}_i \in \mc{T}$
and that $\str{\tree{s}_1} \mequiv{m} \str{\tree{s}_2}$. Let
$\tree{t}_i = \tree{t}$, $(G_i, \nu_{i, 1}) = \str{\tree{t}_i}$ and
$(H_i, \nu_{i, 2}) = \str{\tree{s}_i}$ for $i \in \{1, 2\}$. It now
follows directly from
Lemma~\ref{lemma:composition-lemma-for-n-partite-cographs}, that
$\str{\tree{z}_1} \mequiv{m} \str{\tree{z}_2}$.

To see that $\mathsf{Str}$ satisfies \ref{A.1.b} for $m \ge 0$ and
$\mc{L} = \mso$, let $\tree{t}, \tree{s}_1 \in \mc{T}$ and $a$ be a
child of $\troot{\tree{t}}$. Suppose that $\tree{s}_2 = \tree{t}_{\ge
  a}$ and $\tree{z} = \tree{t} \left[\tree{s}_2 \mapsto
  \tree{s}_1\right] \in \mc{T}$, and that $\str{\tree{s}_1} \mequiv{m}
\str{\tree{s}_2}$.  For $i \in \{1, 2\}$, let $\tree{s}_i'$ be the
tree in $\mc{T}$ obtained by making the root of $\tree{s}_i$, the sole
child of a new node whose label is the same as the label of
$\troot{\tree{t}}$ in $\tree{t}$. Using the notation introduced in
Section~\ref{subsection:proof-for-words-and-trees}, if $\tree{s}_3$ is
the singleton tree whose sole node, say $b$, is labeled with the same
label as that of $\troot{\tree{t}}$ in $\tree{t}$, then $\tree{s}_i' =
\tree{s}_3 \cdot^{\uparrow}_b \tree{s}_i$. It is easy to verify that
$\str{\tree{s}_i} = \str{\tree{s}_i'}$, whereby $\str{\tree{s}_1'}
\mequiv{m} \str{\tree{s}_2'}$.

Let $\tree{y}_1$, resp. $\tree{y}_2$, be the subtree of $\tree{t}$
obtained by deleting the subtrees of $\tree{t}$ rooted at the children
of $\troot{\tree{t}}$ that are ``greater than or equal to'' $a$,
resp. ``less than or equal to'' $a$, under the ordering of the
children of $\troot{\tree{t}}$ in $\tree{t}$. Then $\tree{t} =
(\tree{y}_1 \odot \tree{s}_2') \odot \tree{y}_2$ and $\tree{z} =
(\tree{y}_1 \odot \tree{s}_1') \odot \tree{y}_2$.  Since
$\str{\tree{s}_1'} \mequiv{m} \str{\tree{s}_2'}$, we have by
Lemma~\ref{lemma:composition-lemma-for-n-partite-cographs} that
$\str{\tree{y}_1 \odot \tree{s}_1'} \mequiv{m} \str{\tree{y}_1 \odot
  \tree{s}_2'}$, whereby $\str{\tree{t}} \mequiv{m} \str{\tree{z}}$,
showing that $\mathsf{Str}$ satisfies \ref{A.1.b} for $m \ge 0$ and
$\mc{L} = \mso$, completing the proof.
\end{proof}

\newcommand{\V}[1]{\ensuremath{\mathsf{V}\text{-}\mathsf{Str}(#1)}}

\begin{proof}[Proof of  Lemma~\ref{lemma:composition-lemma-for-n-partite-cographs}]

We can assume w.l.o.g. that $\tree{t}_i$ and $\tree{s}_i$ have
disjoint sets of nodes for $i \in \{1, 2\}$.  Let the set of vertices
of $\mathsf{Str}(\tree{t}_i)$ and $\mathsf{Str}(\tree{s}_i)$ be
$\V{\tree{t}_i}$ and $\V{\tree{s}_i}$ respectively. Then the vertex
set $\V{\tree{z}_i}$ of $\mathsf{Str}(\tree{z}_i)$ is $\V{\tree{t}_i}
\sqcup \V{\tree{s}_i}$ for $i \in \{1, 2\}$.

Let $\tbf{S}_{\tree{t}}$, resp. $\tbf{S}_{\tree{s}}$, be the strategy
of the duplicator in the $m$-round $\mef$ game between
$\mathsf{Str}(\tree{t}_1)$ and $\mathsf{Str}(\tree{t}_2)$,
resp. between $\mathsf{Str}(\tree{s}_1)$ and $\mathsf{Str}(\tree{s}_2)$. For
the $m$-round $\mef$ game between $\mathsf{Str}({\tree{z}}_1)$ and
$\mathsf{Str}({\tree{z}}_2)$, the duplicator follows the following
strategy, call it $\tbf{R}$.
\begin{itemize}[nosep]
\item Point move: If the spoiler chooses a vertex from
  $\V{\tree{t}_1}$ (resp. $\V{\tree{t}_2}$), then the duplicator
  chooses a vertex from $\V{\tree{t}_2}$ (resp. $\V{\tree{t}_1}$)
  according to $\tbf{S}_{\tree{t}}$. Else, if the spoiler chooses a
  vertex from $\V{\tree{s}_1}$ (resp. $\V{\tree{s}_2}$), then the
  duplicator chooses a vertex from $\V{\tree{s}_2}$
  (resp. $\V{\tree{s}_1}$) according to $\tbf{S}_{\tree{s}}$.

\item Set move: If the spoiler chooses a set, say $U$, from
  $\V{\tree{z}_1}$ (resp. $\V{\tree{z}_2}$), then let $X = U \cap
  \V{\tree{t}_1}$ (resp. $X = U \cap \V{\tree{t}_2}$) and $Y = U \cap
  \V{\tree{s}_1}$ (resp. $Y = U \cap \V{\tree{s}_2}$). Let $X'$ be the
  subset of $\V{\tree{t}_2}$ (resp. $\V{\tree{t}_1}$) that is picked
  according to the strategy $\tbf{S}_{\tree{t}}$ in response to the
  choice of $X$ in $\V{\tree{t}_1}$
  (resp. $\V{\tree{t}_2}$). Likewise, let $Y'$ be the subset of
  $\V{\tree{s}_2}$ (resp. $\V{\tree{s}_1}$) that is picked according
  to $\tbf{S}_{\tree{s}}$ in response to the choice of $Y$ in
  $\V{\tree{s}_1}$ (resp. $\V{\tree{s}_2}$).  Then the set $U'$ picked
  by the duplicator from $\V{\tree{z}_2}$ according to strategy
  $\tbf{R}$ is given by $U' = X' \sqcup Y'$.
\end{itemize}
We now show that $\tbf{R}$ is a winning strategy for the duplicator.

Let at the end of $m$ rounds, the vertices and sets chosen from
$\mathsf{Str}({\tree{z}}_1)$, resp. $\mathsf{Str}({\tree{z}}_2)$, be
$a_1, \ldots, a_p$ and $A_1, \ldots, A_r$, resp.  $b_1, \ldots, b_p$
and $B_1, \ldots, B_r$, where $p + r = m$. Let $A_l^1 = A_l\cap
\V{\tree{t}_1}, A_l^2 = A_l \cap \V{\tree{s}_1}, B_l^1 = B_l \cap
\V{\tree{t}_2}$ and $B_l^2 = B_l \cap \V{\tree{s}_2}$ for $l \in \{1,
\ldots, r\}$.

It is easy to see that the labels of $a_i$ and $b_i$ are the same for
all $i \in \{1, \ldots, p\}$.  Also by the description of $\tbf{R}$
given above it is easy to check for all $i \in \{1, \ldots, p\}$ that
$a_i \in \V{\tree{t}_1}$ iff $b_i \in \V{\tree{t}_2}$ and $a_i \in
\V{\tree{s}_1}$ iff $b_i \in \V{\tree{s}_2}$. Likewise, for all $l \in
\{1, \ldots, r\}$ and $i \in \{1, \ldots, p\}$, we have $a_i \in
A_l^1$ iff $b_i \in B_l^1$ and $a_i \in A_l^2$ iff $b_i \in B_l^2$,
whereby $a_i \in A_l$ iff $b_i \in B_l$. 


Consider $a_i, a_j$ for $i \neq j$ and $i, j \in \{1, \ldots, p\}$.
We show below that $a_i, a_j$ are adjacent in $\str{\tree{z}_1}$ iff
$b_i, b_j$ are adjacent in $\str{\tree{z}_2}$. This would show that
$a_i \mapsto b_i$ is a partial isomorphism between $(\str{\tree{z}_1},
A_1, \ldots, A_r)$ and $(\str{\tree{z}_2}, B_1, \ldots, B_r)$
completing the proof. We have the following three cases:

\begin{enumerate}[nosep]
\item Each of $a_i$ and $a_j$ is from $\V{\tree{t}_1}$: Then by the
  description of $\tbf{R}$ above, we have that (i) $b_i$ and $b_j$ are
  both from $\V{\tree{t}_2}$ and (ii) $a_i, a_j$ are adjacent in
  $\mathsf{Str}(\tree{t}_1)$ iff $b_i, b_j$ are adjacent in
  $\mathsf{Str}(\tree{t}_2)$. Observe that $a_i, a_j$ are adjacent in
  $\mathsf{Str}(\tree{t}_1)$ iff $a_i, a_j$ are adjacent in
  $\mathsf{Str}(\tree{z}_1)$. Likewise, $b_i, b_j$ are adjacent in
  $\mathsf{Str}(\tree{t}_2)$ iff $b_i, b_j$ are adjacent in
  $\mathsf{Str}(\tree{z}_2)$. Then $a_i, a_j$ are adjacent in
  $\mathsf{Str}(\tree{z}_1)$ iff $b_i, b_j$ are adjacent in
  $\mathsf{Str}(\tree{z}_2)$.

\item Each of $a_i$ and $a_j$ is from $\V{\tree{s}_1}$: Reasoning
  similarly as in the previous case, we can show that $a_i, a_j$ are
  adjacent in $\mathsf{Str}(\tree{z}_1)$ iff $b_i, b_j$ are adjacent
  in $\mathsf{Str}(\tree{z}_2)$.


\item W.l.o.g.  $a_i \in \V{\tree{t}_1}$ and $a_j \in
  \V{\tree{s}_1}$: Then $b_i \in \V{\tree{t}_2}$ and $b_j \in
  \V{\tree{s}_2}$. Observe now that the greatest common ancestor of
  $a_i$ and $a_j$ in $\tree{z}_1$ is $\troot{\tree{z}_1}$, and the
  greatest common ancestor of $b_i$ and $b_j$ in $\tree{z}_2$ is
  $\troot{\tree{z}_2}$. Since (i) the labels of $\troot{\tree{z}_1}$
  and $\troot{\tree{z}_2}$ are the same (by assumption) and (ii) the
  label of $a_i$ (resp. $a_j$) in $\tree{z}_1$ = label of $a_i$
  (resp. $a_j$) in $\str{\tree{z}_1}$ = label of $b_i$ (resp. $b_j$)
  in $\str{\tree{z}_2}$ = label of $b_i$ (resp. $b_j$) in
  $\tree{z}_2$, it follows by the definition of an $n$-partite cograph
  that $a_i, a_j$ are adjacent in $\mathsf{Str}(\tree{z}_1)$ iff $b_i,
  b_j$ are adjacent in $\mathsf{Str}(\tree{z}_2)$.
\end{enumerate}

\end{proof}

%% file: closure-prop-of-lebsp.tex
\section{Closure properties of $\lebsp{\cdot}{\cdot}$}\label{section:closure-prop-of-lebsp}

\input{lebsp-closure-under-set-theoretic-ops}

\input{lebsp-closure-under-translation-schemes}

\input{lebsp-closure-under-regular-op-tree-languages}

%% file: lebsp-closure-under-set-theoretic-ops.tex
\subsection[Closure under set-theoretic operations]{Closure of $\lebsp{\cdot}{\cdot}$ under set-theoretic operations}\label{subsection:closure-under set-theoretic-ops}

\begin{lemma}\label{lemma:lebsp-closure-under-set-theoretic-ops}
Given classes $\cl{S}_1$ and $\cl{S}_2$ of finite structures, and
$k_1, k_2 \in \mathbb{N}$, suppose that $\lebsp{\cl{S}_i}{k_i}$ is
true for each $i \in \{1, 2\}$, and suppose
$\witfn{\cl{S}_i}{k_i}{\mc{L}}$ is a witness function for
$\lebsp{\cl{S}_i}{k_i}$. Then the following hold.
\begin{enumerate}[nosep]
\item If $\cl{S}$ is any subclass of $\cl{S}_i$ that is hereditary
  over $\cl{S}_i$, where $i \in \{1, 2\}$, then $\lebsp{\cl{S}}{k}$ is
  true for $k = k_i$, with witness function
  $\witfn{\cl{S}}{k}{\mc{L}}$ given by $\witfn{\cl{S}}{k}{\mc{L}} =
  \witfn{\cl{S}_i}{k_i}{\mc{L}}$.\label{lemma:lebsp-closure-under-set-theoretic-ops:part-1}
\item If $\cl{S} = \cl{S}_1 \cup \cl{S}_2$, then $\lebsp{\cl{S}}{k}$
  is true for $k = \text{min}(k_1, k_2)$, with witness function
  $\witfn{\cl{S}}{k}{\mc{L}}$ given by $\witfn{\cl{S}}{k}{\mc{L}} =
  \text{max}(\witfn{\cl{S}_1}{k_1}{\mc{L}},
  \witfn{\cl{S}_2}{k_2}{\mc{L}})$.\label{lemma:lebsp-closure-under-set-theoretic-ops:part-2}
\item If $\cl{S} = \cl{S}_1 \cap \cl{S}_2$ and $\cl{S}_1$ is
  hereditary, then $\lebsp{\cl{S}}{k}$ is true for $k = k_2$, with
  witness function $\witfn{\cl{S}}{k}{\mc{L}}$ given by
  $\witfn{\cl{S}}{k}{\mc{L}} = \witfn{\cl{S}_2}{k_2}{\mc{L}}$. If
  $\cl{S}_2$ is also hereditary, then $\lebsp{\cl{S}}{k}$ is true for
  $k = \text{max}(k_1, k_2)$, with witness function
  $\witfn{\cl{S}}{k}{\mc{L}}$ given by $\witfn{\cl{S}}{k}{\mc{L}} =$
  \linebreak$ \text{max}(\witfn{\cl{S}_1}{k_1}{\mc{L}},
  \witfn{\cl{S}_2}{k_2}{\mc{L}})$.\label{lemma:lebsp-closure-under-set-theoretic-ops:part-3}
\item If $\cl{S}$ is a subclass of $\cl{S}_i$ that is definable over
  $\cl{S}_i$ by an $\mc{L}$ sentence of rank $r$, then
  $\lebsp{\cl{S}}{k}$ is true for $k = k_i$, with witness function
  $\witfn{\cl{S}}{k}{\mc{L}}$ given by $\witfn{\cl{S}}{k}{\mc{L}}(m) =
  \witfn{\cl{S}_i}{k_i}{\mc{L}}(r)$ if $m \leq r$, else
  $\witfn{\cl{S}}{k}{\mc{L}}(m) =
  \witfn{\cl{S}_i}{k_i}{\mc{L}}(m)$. It follows that for $\cl{S}$ as
  aforementioned, if $\overline{\cl{S}}$ is the complement of $\cl{S}$
  in $\cl{S}_i$, then $\lebsp{\overline{\cl{S}}}{k}$ is also true for
  $k = k_i$, with witness function
  $\witfn{\overline{\cl{S}}}{k}{\mc{L}}$ given by
  $\witfn{\overline{\cl{S}}}{k}{\mc{L}} = \witfn{\cl{S}}{k}{\mc{L}}$
  where $\witfn{\cl{S}}{k}{\mc{L}}$ is as aforementioned.
   
\label{lemma:lebsp-closure-under-set-theoretic-ops:part-4}
\end{enumerate}
\end{lemma}
\begin{proof}
Let $m \in \mathbb{N}$ be given.

\ul{Part \ref{lemma:lebsp-closure-under-set-theoretic-ops:part-1}}:
Consider $\mf{A} \in \cl{S}$ and let $\bar{a}$ be a $k$-tuple from
$\mf{A}$ where $k = k_i$.  Since $\lebsp{\cl{S}_i}{k_i}$ is true,
there exists $\mf{B} \in \cl{S}_i$ such that $\lebspcond(\cl{S}_i,
\mf{A}, \mf{B}, k_i, m,$ $ \bar{a}, \witfn{\cl{S}_i}{k_i}{\mc{L}})$ is
true. Then since $\mf{B} \subseteq \mf{A}$ and $\cl{S}$ is hereditary
over $\cl{S}_i$, we have $\mf{B} \in \cl{S}$; whence
$\lebspcond(\cl{S}, \mf{A}, \mf{B}, $ $ k, m,$ $ \bar{a},
\witfn{\cl{S}}{k}{\mc{L}})$ is true, where $\witfn{\cl{S}}{k}{\mc{L}}
= \witfn{\cl{S}_i}{k_i}{\mc{L}}$.

\ul{Part \ref{lemma:lebsp-closure-under-set-theoretic-ops:part-2}}:
Consider $\mf{A} \in \cl{S}$ and let $\bar{a}$ be a $k$-tuple from
$\mf{A}$ where $k = \text{min}(k_1, k_2)$. Since $\cl{S} = \cl{S}_1
\cup \cl{S}_2$, assume w.l.o.g. that $\mf{A} \in \cl{S}_1$. Let $b$ be
an element of $\bar{a}$ and let $\bar{a}_1$ be a $k_1$-tuple whose
first $k$ components form exactly the tuple $\bar{a}$ and in which $b$
is the element at all the indices $k+1, \ldots, k_1$.  Since
$\lebsp{\cl{S}_1}{k_1}$ is true, there exists $\mf{B} \in \cl{S}_1$
such that $\lebspcond(\cl{S}_1, \mf{A}, \mf{B}, k_1, $ $ m, \bar{a}_1,
\witfn{\cl{S}_1}{k_1})$ is true.  Then $\mf{B} \in \cl{S}$. Further
since $\ltp{\mf{B}}{\bar{a}_1}{m}(\bar{x}) =
\ltp{\mf{A}}{\bar{a}_1}{m}(\bar{x})$, it follows that
$\ltp{\mf{B}}{\bar{a}}{m}(\bar{x}) =
\ltp{\mf{A}}{\bar{a}}{m}(\bar{x})$; whence $\lebspcond(\cl{S},
\mf{A},$ $ \mf{B}, k, m, \bar{a}, \witfn{\cl{S}}{k}{\mc{L}})$ is true,
where $\witfn{\cl{S}}{k}{\mc{L}} =
\text{max}(\witfn{\cl{S}_1}{k_1}{\mc{L}},
\witfn{\cl{S}_2}{k_2}{\mc{L}})$.

\ul{Part \ref{lemma:lebsp-closure-under-set-theoretic-ops:part-3}}:
Consider $\mf{A} \in \cl{S}$ and let $\bar{a}$ be a $k$-tuple from
$\mf{A}$ where $k = k_2$. Since $\lebsp{\cl{S}_2}{k_2}$ is true, there
exists $\mf{B} \in \cl{S}_2$ such that $\lebspcond(\cl{S}_2, \mf{A},
\mf{B}, k_2, m, \bar{a}, \witfn{\cl{S}_2}{k_2}{\mc{L}})$ is
true. Since $\mf{B} \subseteq \mf{A}, \mf{A} \in \cl{S}_1$ and
$\cl{S}_1$ is hereditary, we have $\mf{B} \in \cl{S}_1$, and hence
$\mf{B} \in \cl{S}$. Then $\lebspcond(\cl{S}, \mf{A}, \mf{B}, k, m,
\bar{a}, \witfn{\cl{S}}{k}{\mc{L}})$ is true, where
$\witfn{\cl{S}}{k}{\mc{L}} = \witfn{\cl{S}_2}{k_2}{\mc{L}}$. If
$\cl{S}_2$ is also hereditary, then let $\bar{a}$ be a $k$-tuple from
$\mf{A}$ where $k = \text{max}(k_1, k_2)$. W.l.o.g., suppose $k_1 \ge
k_2$, so that $k = k_1$. Since $\lebsp{\cl{S}_1}{k_1}$ is true, there
exists $\mf{B} \in \cl{S}_1$ such that $\lebspcond(\cl{S}_1,$ $
\mf{A}, $ $ \mf{B}, k_1, m, \bar{a}, \witfn{\cl{S}_1}{k_1}{\mc{L}})$
is true. Since $\mf{B} \subseteq \mf{A}, \mf{A} \in \cl{S}_2$ and
$\cl{S}_2$ is hereditary, we have $\mf{B} \in \cl{S}_2$, and hence
$\mf{B} \in \cl{S}$. Then $\lebspcond(\cl{S}, \mf{A}, \mf{B}, k, m,
\bar{a}, \witfn{\cl{S}}{k}{\mc{L}})$ is true, where
$\witfn{\cl{S}}{k}{\mc{L}} = \text{max}(\witfn{\cl{S}_1}{k_1}{\mc{L}},
\witfn{\cl{S}_2}{k_2}{\mc{L}})$.

\ul{Part \ref{lemma:lebsp-closure-under-set-theoretic-ops:part-4}}:
Consider $\mf{A} \in \cl{S}$ and let $\bar{a}$ be a $k$-tuple from
$\mf{A}$ where $k = k_i$. Let $\witfn{\cl{S}}{k}{\mc{L}}$ and
$\witfn{\overline{\cl{S}}}{k}{\mc{L}}$ be functions as defined in the
statement of this part. Since $\witfn{\cl{S}_i}{k_i}{\mc{L}}$ is
monotonic (see Remark~\ref{remark:lebsp-witness-functions}), so are
$\witfn{\cl{S}}{k}{\mc{L}}$ and
$\witfn{\overline{\cl{S}}}{k}{\mc{L}}$.

Suppose $m \leq r$. Since $\lebsp{\cl{S}_i}{k}$ is true, there exists
$\mf{B} \in \cl{S}_i$ such that \linebreak $\lebspcond(\cl{S}_i,$ $
\mf{A}, \mf{B}, k, r, \bar{a}, \witfn{\cl{S}_i}{k}{\mc{L}})$ is true.
Then since $\ltp{\mf{B}}{\bar{a}}{r}(\bar{x}) =
\ltp{\mf{A}}{\bar{a}}{r}(\bar{x})$, we have (i)
$\ltp{\mf{B}}{\bar{a}}{m}(\bar{x}) =
\ltp{\mf{A}}{\bar{a}}{m}(\bar{x})$ since $m \leq r$, and (ii) $\mf{B}
\lequiv{r} \mf{A}$. Since $\mf{A} \in \cl{S}$ and $\cl{S}$ is defined
over $\cl{S}_i$ by an $\mc{L}$ sentence of rank $r$, we have $\mf{B}
\in \cl{S}$.  Then $\lebspcond(\cl{S}, \mf{A}, \mf{B}, $ $ k, m,
\bar{a}, \witfn{\cl{S}}{k}{\mc{L}})$ is true.

Suppose $m > r$. Then there exists $\mf{B} \in \cl{S}_i$ such that
$\lebspcond(\cl{S}_i, \mf{A}, \mf{B}, k, m, \bar{a},
\witfn{\cl{S}_i}{k}{\mc{L}})$ is true.  Then since
$\ltp{\mf{B}}{\bar{a}}{m}(\bar{x}) =
\ltp{\mf{A}}{\bar{a}}{m}(\bar{x})$ and $m > r$, we have $\mf{B}
\lequiv{r} \mf{A}$ whereby reasoning as before, we have $\mf{B} \in
\cl{S}$.  Then $\lebspcond(\cl{S}, \mf{A}, \mf{B}, k, m, \bar{a},
\witfn{\cl{S}}{k}{\mc{L}})$ is true.

Finally, since $\overline{\cl{S}}$ is also definable over $\cl{S}_i$
by an $\mc{L}$ sentence of rank $r$, namely by the negation of the
sentence defining $\cl{S}$ over $\cl{S}_i$, we have that
$\lebsp{\overline{\cl{S}}}{k}$ is true, with the witness function
$\witfn{\overline{\cl{S}}}{k}{\mc{L}}$ as in the statement of this
part.
\end{proof}

%% file: lebsp-closure-under-translation-schemes.tex
\subsection[Closure under operations implemented using translation schemes]{Closure of $\lebsp{\cdot}{\cdot}$ under operations implemented using translation schemes}\label{subsection:closure-under-operations-implemented-using-translation-schemes}

\newcommand{\ndisjointsum}{\ensuremath{n\text{-}\mathsf{disjoint}\text{-}\mathsf{sum}}}
\newcommand{\ncopy}{\ensuremath{n\text{-}\mathsf{copy}}}
\newcommand{\twocopy}{\ensuremath{2\text{-}\mathsf{copy}}}

We look at ways of generating new classes of structures that satisfy
$\lebsp{\cdot}{\cdot}$ from those known to satisfy this
property. Examples of well-known operations that produce new
structures from given ones include
``sum-like''\sindex[term]{operation!sum-like}
operations~\cite{makowsky} like disjoint union and join~\cite{diestel}
and ``product-like'' operations\sindex[term]{operation!product-like}
like the cartesian and tensor products. All of these are examples of
operations that are ``implementable'' using quantifier-free
translation schemes. Let us look at the cartesian product as an
example. For a vocabulary $\tau$, let $\tau_{\text{disj-un}, 2}$ be
the vocabulary obtained by expanding $\tau$ with $2$ fresh unary
predicates $P_1$ and $P_2$. Given structures $\mf{A}_1$ and $\mf{A}_2$
whose cartesian product we intend to take, we first construct the
\emph{2-disjoint sum}~\cite{nicole-lmcs-15} of $\mf{A}_1$ and
$\mf{A}_2$, denoted $\mf{A}_1 \oplus \mf{A}_2$, which is the
$\tau_{\text{disj-un}, 2}$-structure obtained upto isomorphism, by
expanding the disjoint union $\mf{A}_1 \sqcup \mf{A}_2$ with $P_1$ and
$P_2$ interpreted respectively as the universes of the isomorphic
copies of $\mf{A}_1$ and $\mf{A}_2$ that are used in constructing
$\mf{A}_1 \sqcup \mf{A}_2$.  The cartesian product $\mf{A}_1 \otimes
\mf{A}_2$ is then the structure $\Xi(\mf{A}_1 \oplus \mf{A}_2)$ where
$\Xi$ is the $(2, \tau_{\text{disj-un}, 2}, \tau, \fo)$-translation
scheme given by $\Xi = (\xi, (\xi_R)_{R \in \tau})$ where $\xi(x, y) =
(P_1(x) \wedge P_2(y))$ and for $R \in \tau$ of arity $r$, we have
$\xi_R(x_1, y_1, \ldots, x_r, y_r) = R(x_1, \ldots, x_r) \wedge R(y_1,
\ldots, y_r)$. As a second example, consider the \emph{across-connect}
operation which takes two copies of a graph $\mathsf{G}$ and connects
corresponding nodes across. To implement this operation, we first
construct the \emph{2-copy} of
$\mathsf{G}$~\cite{blumensath}. Specifically, we take isomorphic
copies $\mathsf{G}_1$ and $\mathsf{G}_2$ of $\mathsf{G}$, where the
universe of $\mathsf{G}_i$ is $\{(i, a) \mid a \in
\mathsf{U}_{\mathsf{G}}\}$ for $i \in \{1, 2\}$. We then expand
$\mathsf{G}_1 \oplus \mathsf{G}_2$ with the relation $\sim$
interpreted as the set $\{((1, a), (2, a)) \mid a \in
\mathsf{U}_{\mathsf{G}}\}$, to get a $\tau_{\text{copy}, 2}$-structure
$\twocopy(\mathsf{G})$, where $\tau_{\text{copy}, 2} =
\tau_{\text{disj-un}, 2} \cup \{\sim\}$. Then the \emph{across-connect
  of $\mathsf{G}$} is the structure $\Phi(\twocopy(\mathsf{G}))$ where
$\Phi = (\phi, \phi_E)$, $\phi$ is the formula $(x = x)$ and
$\phi_E(x, y) = E(x, y) \vee \big( P_1(x) \wedge P_2(y) \wedge (x \sim
y)\big)$. Observe that both of the translation schemes above are
quantifier-free. The above operations are two instances of several
useful and well-studied operations on structures that are
implementable using quantifier-free translation schemes.  We consider
such operations in this section. Specifically, the quantifier-free
translation schemes implementing the operations are those that ``act''
on the \emph{$n$-disjoint sums} or \emph{$n$-copies} of structures of
a given class. We define these notions formally below.

\begin{defn}[$n$-disjoint sum]
Given a vocabulary $\tau$, let $\tau_{k_i}$ be the vocabulary obtained
by expanding $\tau$ with $k_i$ fresh constant symbols, for $k_i \ge 0$
and $i \in \{1, \ldots, n\}$. Let $\tau_{\text{disj-un}, k_1, \ldots,
  k_n}$ be the vocabulary obtained by expanding $\tau$ with $k_1 +
\cdots + k_n$ fresh constant symbols, and $n$ fresh unary relation
symbols $P_1, \ldots, P_n$.  For $i \in \{1, \ldots, n\}$, let
$(\mf{A}_i, \bar{a}_i)$ be a $\tau_{k_i}$-structure, where $\mf{A}_i$
is a $\tau$-structure. Then the \emph{$n$-disjoint sum of $(\mf{A}_1,
  \bar{a}_1), \ldots, (\mf{A}_n, \bar{a}_n)$}, denoted $\bigoplus_{i =
  1}^{i = n} (\mf{A}_i, \bar{a}_i)$, is the $\tau_{\text{disj-un},
  k_1, \ldots, k_n}$-structure $\mf{A}$ defined as follows.

\begin{enumerate}[nosep]
\item If $\mf{A}_1, \ldots, \mf{A}_n$ have disjoint universes, then
  $\mf{A}$ is such that (i) the $\tau$-reduct of $\mf{A}$ is the
  disjoint union of $\mf{A}_1, \ldots, \mf{A}_n$, (ii) $P_i^{\mf{A}} =
  \mathsf{U}_{\mf{A}_i}$ for each $i \in \{1, \ldots, n\}$ (thus the
  $P_i^{\mf{A}}$s form a partition of the universe of $\mf{A}$), and
  (iii) for $i \in \{1, \ldots, n\}$, if $l_i = k_1 + \cdots +
  k_{i-1}$ and $l_1 = 0$, then $(c^{\mf{A}}_{l_i + 1}, \ldots,
  c^{\mf{A}}_{l_i + k_i}) = \bar{a}_i$, where $c_1, \ldots, c_{k_1 +
    \cdots + k_n}$ are the fresh constant symbols of
  $\tau_{\text{disj-un}, k_1, \ldots, k_n}$.
\item In case, $\mf{A}_1, \ldots, \mf{A}_n$ do not have disjoint
  universes, then let $(\mf{A}_i', \bar{a}_i')$ be an isomorphic copy
  of $(\mf{A}_i, \bar{a}_i)$ for $i \in \{1, \ldots, n\}$, such that
  $\mf{A}'_1, \ldots, \mf{A}'_n$ have disjoint universes. Then
  $\mf{A}$ is defined upto isomorphism as the $\tau$-structure
  $\bigoplus_{i = 1}^{i = n} (\mf{A}_i', \bar{a}_i')$.
\end{enumerate}
\end{defn}

\begin{defn}[$n$-copy]
Given a vocabulary $\tau$, let $\tau_{\text{copy}, k_1, \ldots, k_n} =
\tau_{\mathsf{disj-un}, k_1, \ldots, k_n} \cup \{\sim\}$, where $\sim$
is a binary relation symbol not in $\tau$, and $k_1, \ldots, k_n \ge
0$.  Given a $\tau$-structure $\mf{A}$ and a $k_i$-tuple $\bar{a}_i$
from $\mf{A}$ for each $i \in \{1, \ldots, n\}$, let $(\mf{A}_i,
\bar{b}_i)$ be an isomorphic copy of $(\mf{A}, \bar{a}_i)$, with
universe $\{(i, a) \mid a \in \mathsf{U}_{\mf{A}}\}$.  Then the
\emph{$n$-copy of $\mf{A}$ with $\bar{a}_1, \ldots, \bar{a}_n$},
denoted $\ncopy(\mf{A}, \bar{a}_1, \ldots, \bar{a}_n)$, is the
$\tau_{\text{copy}, k_1, \ldots, k_n}$-structure defined as below:

\begin{enumerate}[nosep]
\item If $n = 1$, then $\ncopy(\mf{A}, \bar{a}_1) = (\mf{A},
  \bar{a}_1)$.
\item If $n > 1$, then $\ncopy(\mf{A}, \bar{a}_1, \ldots, \bar{a}_n)$
  is such that (i) the $\tau_{\mathsf{disj-un}, k_1, \ldots,
    k_n}$-reduct of \linebreak $\ncopy(\mf{A}, \bar{a}_1, \ldots,
  \bar{a}_n)$ is the structure $\bigoplus_{i = 1}^{i = n} (\mf{A}_i,
  \bar{a}_i)$, and (ii) $\sim$ is interpreted in\linebreak
  $\ncopy(\mf{A}, \bar{a}_1, \ldots, \bar{a}_n)$ as the set $ \{((i,
  a), (j, a)) \mid 1 \leq i, j \leq n, a \in \mathsf{U}_{\mf{A}}\}$.
\end{enumerate}
\end{defn}

The above definitions, which are given for structures expanded with
tuples of elements, instead of simply for structures that are not
expanded with tuples of elements, are given so because in the proofs
of our results below, we will need these general definitions. However,
for the statements of our results, we deal with $n$-disjoint sums and
$n$-copies of only structures that are not expanded with tuples of
elements, i.e. for the case when $k_1 = \cdots = k_n = 0$ in the
definitions above. In such a case, we denote $\tau_{\mathsf{disj-un},
  k_1, \ldots, k_n}$ and $\tau_{\mathsf{copy}, k_1, \ldots, k_n}$,
simply as $\tau_{\mathsf{disj-un}, n}$ and $\tau_{\mathsf{copy}, n}$
respectively. Likewise, we denote $\bigoplus_{i = 1}^{i = n}
(\mf{A}_i, \bar{a}_i)$ simply as $\bigoplus_{i = 1}^{i = n} \mf{A}_i$
(since each $\bar{a}_i$ is empty), and $\ncopy(\mf{A}, \bar{a}_1,
\ldots, \bar{a}_n)$ simply as $\ncopy(\mf{A})$.



Given classes $\cl{S}_1, \ldots, \cl{S}_n$ of $\tau$-structures, let
$\ndisjointsum(\cl{S}_1, \ldots, \cl{S}_n) = \{\bigoplus_{i = 1}^{i =
  n} \mf{A}_i \mid \mf{A}_i \in \cl{S}_i, 1 \leq i \leq n\}$.  Given a
quantifier-free $(t, \tau_{\text{disj-un}, n}, \tau, \fo)$-translation
scheme $\Xi_1$, let \linebreak $\Xi_1(\ndisjointsum(\cl{S}_1, \ldots,
\cl{S}_n)) = \{\Xi_1(\bigoplus_{i = 1}^{i = n} \mf{A}_i) \mid \mf{A}_i
\in \cl{S}_i, 1 \leq i \leq n\}$. Then $\Xi_1$ gives rise to an
$n$-ary operation $\mathsf{O}_1: \cl{S}_1 \times \cdots \times
\cl{S}_n \rightarrow \Xi_1(\ndisjointsum(\cl{S}_1, \ldots, \cl{S}_n))$
defined as $\mathsf{O}_1(\mf{A}_1, \ldots, \mf{A}_n) =
\Xi_1(\bigoplus_{i = 1}^{i = n} \mf{A}_i)$.  Likewise, given a class
$\cl{S}$ of structures, if $\ncopy(\cl{S}) = \{\ncopy(\mf{A}) \mid
\mf{A} \in \cl{S}\}$, then a quantifier-free $(t, \tau_{\text{copy},
  n}, \tau, \fo)$-translation scheme $\Xi_2$ gives rise to a unary
operation $\mathsf{O}_2: \cl{S} \rightarrow \Xi_2(\ncopy(\cl{S}))$
where $\Xi_2(\ncopy(\cl{S})) = \{ \Xi_2(\ncopy(\mf{A})) \mid \mf{A}
\in \cl{S}\}$ such that $\mathsf{O}_2(\mf{A}) =
\Xi_2(\ncopy(\mf{A}))$.  For the above cases, we say that
$\mathsf{O}_1$, resp. $\mathsf{O}_2$, is \emph{implementable using
  $\Xi_1$, resp. $\Xi_2$}. We say an operation is implementable using
a quantifier-free translation scheme if it is one of the two kinds of
operations $\mathsf{O}_1$ and $\mathsf{O}_2$ just
described.\sindex[term]{operation!implementable using
  quantifier-\\free translation scheme} The following two results,
which are our central results of this section, together show that
operations that are implementable using quantifier-free translation
schemes preserve the $\lebsp{\cdot}{\cdot}$ property of the classes
they operate on.

\begin{lemma}\label{lemma:lebsp-pres-under-n-disj-sum-and-n-copy}
Let $\cl{S}, \cl{S}_1, \ldots, \cl{S}_n$ be classes of structures for
$n \ge 1$. The following are true.
\begin{enumerate}[nosep]
\item If $\lebsp{\cl{S}_i}{k_i}$ is true for $k_i \in \mathbb{N}$ for
  each $i \in \{1, \ldots, n\}$, then so is \linebreak
  $\lebsp{\ndisjointsum(\cl{S}_1, \ldots, \cl{S}_n)}{l}$, where $l =
  \text{min}\{k_i \mid i \in \{1, \ldots, n\}\}$. Further, if there is
  a computable witness function for $\lebsp{\cl{S}_i}{k_i}$ for each
  $i \in \{1, \ldots, n\}$, then there is a computable witness
  function for $\lebsp{\ndisjointsum(\cl{S}_1, \ldots, \cl{S}_n)}{l}$
  as well.\label{lemma:lebsp-pres-under-n-disj-sum-and-n-copy:part-1}
\item If $\lebsp{\cl{S}}{k}$ is true for $k \in \mathbb{N}$, then so
  is $\lebsp{\ncopy(\cl{S})}{k}$. Further, if there is a computable
  witness function for $\lebsp{\cl{S}}{k}$, then there is a computable
  witness function for $\lebsp{\ncopy(\cl{S})}{k}$ as
  well.\label{lemma:lebsp-pres-under-n-disj-sum-and-n-copy:part-2}
\end{enumerate}
\end{lemma}

\begin{theorem}\label{theorem:ebsp-gebsp-and-transductions}
Let $\cl{S}$ be class of $\tau$-structures, and let $\Xi = (\xi,
(\xi_R)_{R \in \sigma})$ be a quantifier-free $(t, \tau, \sigma,
\fo)$-translation scheme.  Then the following hold for each $k \in
\mathbb{N}$.
\begin{enumerate}[nosep]
\item If ~$\febsp{\cl{S}}{k \cdot t}$ is true, then so is
  $\febsp{\Xi(\cl{S})}{k}$.\label{theorem:ebsp-pres-under-transductions}
\item If ~$\Xi$ is scalar and $\mebsp{\cl{S}}{k}$ is true, then so
  is $\mebsp{\Xi(\cl{S})}{k}$.
  \label{theorem:gebspmso-pres-under-transductions}
\end{enumerate}
In each of the implications above, a computable witness function for
the antecedent implies a computable witness function for the
consequent.
\end{theorem}

\begin{remark}
The quantifier-freeness of $\Xi$ in
Theorem~\ref{theorem:ebsp-gebsp-and-transductions} is necessary in
general. In fact, the presence of even a single quantifier in any one
of the formulas $\xi_R$ above can cause
Theorem~\ref{theorem:ebsp-gebsp-and-transductions} to fail. We show
this towards the end of this section.
\end{remark}

For an operation $\mathsf{O}$ that is implementable using a
quantifier-free translation scheme, define the \emph{dimension} of
$\mathsf{O}$ to be the minimum of the dimensions of the
quantifier-free translation schemes that implement $\mathsf{O}$.  We
say $\mathsf{O}$ is ``sum-like'' if its dimension is one, else we say
$\mathsf{O}$ is ``product-like''.  Examples of
sum-like\sindex[term]{operation!sum-like} operations include unary
graph operations like complement, transpose, across-connect and the
line-graph operation~\cite{diestel}, and binary operations like
disjoint union and join. Examples of
product-like\sindex[term]{operation!product-like} operations include
various kinds of products such as cartesian, tensor, lexicographic,
and strong products. We now have the following corollary which shows
that $\lebsp{\cdot}{\cdot}$ and $\febsp{\cdot}{\cdot}$ are indeed
preserved under sum-like and product-like operations respectively.
  
\begin{corollary}\label{corollaty:lebsp-closure-under-1-step-operations}
Let $\cl{S}_1, \ldots, \cl{S}_n, \cl{S}$ be classes of structures and
let $\mathsf{O}: \cl{S}_1 \times \cdots \times \cl{S}_n \rightarrow
\cl{S}$ be an $n$-ary operation that is implementable using a
quantifier-free translation scheme.  Let $\mathsf{O}(\cl{S}_1, \ldots,
\cl{S}_n)$ denote the class of structures that are in the range of
$\mathsf{O}$, and let $t$ be the dimension of $\mathsf{O}$. Then the
following are true.
\begin{enumerate}[nosep]
\item If $\lebsp{\cl{S}_i}{k_i}$ is true for $k_i \in \mathbb{N}$ for
  each $i \in \{1, \ldots, n\}$, then so is \linebreak
  $\lebsp{\mathsf{O}(\cl{S}_1, \ldots, \cl{S}_n)}{l}$, for $l =
  \text{min}\{k_i \mid i \in \{1, \ldots, n\}\}$, whenever
  $\mathsf{O}$ is sum-like.
\item If $\febsp{\cl{S}_i}{k_i \cdot t}$ is true for $k_i \in
  \mathbb{N}$ for each $i \in \{1, \ldots, n\}$, then so is \linebreak
  $\febsp{\mathsf{O}(\cl{S}_1, \ldots, \cl{S}_n)}{l}$, for $l =
  \text{min}\{k_i \mid i \in \{1, \ldots, n\}\}$, whenever
  $\mathsf{O}$ is product-like.
\end{enumerate}
In each of the implications above, if there are computable witness
functions for each of the conjuncts in the antecedent, then there is
a computable witness function for the consequent as well.
\end{corollary}

\begin{proof}
Follows easily from
Lemma~\ref{lemma:lebsp-pres-under-n-disj-sum-and-n-copy} and
Theorem~\ref{theorem:ebsp-gebsp-and-transductions}.
\end{proof}


The rest of this section is entirely devoted to proving
Lemma~\ref{lemma:lebsp-pres-under-n-disj-sum-and-n-copy} and
Theorem~\ref{theorem:ebsp-gebsp-and-transductions}.


Towards the proof of
Lemma~\ref{lemma:lebsp-pres-under-n-disj-sum-and-n-copy}, we present
the following simple facts about $n$-disjoint sum and $n$-copy. We
skip the proof.

\begin{lemma}\label{lemma:facts-about-disj-unions}
Let $(\mf{A}_i, \bar{a}_i)$ and $(\mf{B}_i, \bar{b}_i)$ be
$\tau_{k_i}$-structures for $i \in \{1, \ldots, n\}$. Let $m \in
\mathbb{N}$. Then the following are true.
\begin{enumerate}[nosep]
\item If $(\mf{B}_i, \bar{b}_i) \hookrightarrow (\mf{A}_i, \bar{a}_i)$
  for $i \in \{1, \ldots, n\}$, then $\bigoplus_{i = 1}^{i = n}
  (\mf{B}_i, \bar{b}_i) \hookrightarrow \bigoplus_{i = 1}^{i = n}
  (\mf{A}_i, \bar{a}_i)$.
\item If $(\mf{B}_i, \bar{b}_i) \lequiv{m} (\mf{A}_i, \bar{a}_i)$ for
  $i \in \{1, \ldots, n\}$, then $\bigoplus_{i = 1}^{i = n} (\mf{B}_i,
  \bar{b}_i) \lequiv{m} \bigoplus_{i = 1}^{i = n} (\mf{A}_i,
  \bar{a}_i)$.
\end{enumerate}
\end{lemma}
\begin{lemma}\label{lemma:facts-about-n-copies}
Let $(\mf{A}, \bar{a}_i)$ and $(\mf{B}, \bar{b}_i)$ be
$\tau_{k_i}$-structures for $i \in \{1, \ldots, n\}$.  Let $m \in
\mathbb{N}$. If $\mc{C} = \ncopy(\mf{B}, \bar{b}_1, \ldots,
\bar{b}_n)$ and $\mf{D} = \ncopy(\mf{A}, \bar{a}_1, \ldots,
\bar{a}_n)$, then the following are true.
\begin{enumerate}[nosep]
\item If $(\mf{B}, \bar{b}_i) \hookrightarrow (\mf{A}, \bar{a}_i)$ for
  $i \in \{1, \ldots, n\}$, then $\mc{C} \hookrightarrow \mc{D}$.
\item If $(\mf{B}, \bar{b}_i) \lequiv{m} (\mf{A}, \bar{a}_i)$ for $i
  \in \{1, \ldots, n\}$, then $\mc{C} \lequiv{m} \mc{D}$.
\end{enumerate}
\end{lemma}

We now prove Lemma~\ref{lemma:lebsp-pres-under-n-disj-sum-and-n-copy}.

\begin{proof}[Proof of Lemma~\ref{lemma:lebsp-pres-under-n-disj-sum-and-n-copy}]
\ul{Part \ref{lemma:lebsp-pres-under-n-disj-sum-and-n-copy:part-1}}:
Consider a structure $\mf{A} = (\bigoplus_{i = 1}^{i = n} \mf{A}_i)
\in \ndisjointsum(\cl{S}_1, \ldots, \cl{S}_n)$ and let $\bar{a}$ be an
$l$-tuple from $\mf{A}$. Let $\bar{a}_i$ be the sub-tuple of $\bar{a}$
consisting of all elements of $\bar{a}$ that belong to
$\univ{\mf{A}_i}$; clearly $|\bar{a}_i| \leq k_i$ for $i \in \{1,
\ldots, n\}$. Let $m \in \mathbb{N}$. Since $\lebsp{\cl{S}_i}{k_i}$ is
true, there exists $\mf{B}_i$ such that $\lebspcond(\cl{S}_i,
\mf{A}_i, \mf{B}_i, k_i, m, \bar{a}_i, \witfn{\cl{S}_i}{m}{\mc{L}})$
holds where $\witfn{\cl{S}_i}{m}{\mc{L}}$ is a witness function for
$\lebsp{\cl{S}_i}{k_i}$.  Then $(\mf{B}_i, \bar{a}_i) \subseteq
(\mf{A}_i, \bar{a}_i)$ and $(\mf{B}_i, \bar{a}_i) \lequiv{m}
(\mf{A}_i, \bar{a}_i)$. Then by
Lemma~\ref{lemma:facts-about-disj-unions}, we have that (i)
$\bigoplus_{i = 1}^{i = n} (\mf{B}_i, \bar{b}_i) \hookrightarrow
\bigoplus_{i = 1}^{i = n} (\mf{A}_i, \bar{a}_i)$, and (ii)
$\bigoplus_{i = 1}^{i = n} (\mf{B}_i, \bar{b}_i) \lequiv{m}
\bigoplus_{i = 1}^{i = n} (\mf{A}_i, \bar{a}_i)$. Then it is easy to
verify that (i) $((\bigoplus_{i = 1}^{i = n} \mf{B}_i), \bar{a})
\hookrightarrow ((\bigoplus_{i = 1}^{i = n} \mf{A}_i), \bar{a})$, and
(ii) $((\bigoplus_{i = 1}^{i = n} \mf{B}_i), \bar{a}) \lequiv{m}
((\bigoplus_{i = 1}^{i = n} \mf{A}_i), \bar{a})$. Observe that
$(\bigoplus_{i = 1}^{i = n} \mf{B}_i) \in \ndisjointsum(\cl{S}_1,
\ldots, \cl{S}_n)$, and that $|(\bigoplus_{i = 1}^{i = n} \mf{B}_i)|
\leq \theta(m) = \Sigma_{i = 0}^{i = n}
\witfn{\cl{S}_i}{m}{\mc{L}}(m)$. Taking $(\mf{B}, \bar{a})$ to be the
substructure of $(\mf{A}, \bar{a})$ that is isomorphic to
$((\bigoplus_{i = 1}^{i = n} \mf{B}_i), \bar{a})$, we see
$\lebspcond(\ndisjointsum(\cl{S}_1, \ldots, \cl{S}_n), $ $ \mf{A},
\mf{B}, l, m, \bar{a}, \theta)$ is true with witness function
$\theta$. Whereby $\lebsp{\ndisjointsum(\cl{S}_1, \ldots,
  \cl{S}_n)}{l}$ is true. It is easy to see that if
$\witfn{\cl{S}_i}{m}{\mc{L}}$ is computable for each $i \in \{1,
\ldots, n\}$, then so is $\theta$.

\ul{Part \ref{lemma:lebsp-pres-under-n-disj-sum-and-n-copy:part-2}}:
This is proved analogously as the previous part, and using
Lemma~\ref{lemma:facts-about-n-copies}.
\end{proof}

We now proceed to proving
Theorem~\ref{theorem:ebsp-gebsp-and-transductions}.  Towards the
proof, we first prove the following result that shows that
quantifier-free translation schemes preserve the substructure relation
between any two structures of $\cl{S}$. We use results mentioned in
Section~\ref{section:background-FMT-translation-schemes} in our proof.

\begin{lemma}\label{lemma:qt-free-transductions-preserve-substructure-prop}
Let $\cl{S}$ be a given class of finite structures. Let $\Xi = (\xi,
(\xi_R)_{R \in \sigma})$ be a quantifier-free $(t, \tau, \sigma,
\fo)$-translation scheme. Let $\mf{A}$ and $\mf{B}$ be given
structures from $\cl{S}$, and let $\bar{b}_1, \ldots, \bar{b}_n$ be
$n$ elements from $\Xi(\mf{B})$, for some $n \ge 0$.  If $(\mf{B},
\bar{b}_1, \ldots, \bar{b}_n) \subseteq (\mf{A}, \bar{b}_1, \ldots,
\bar{b}_n)$, then (i) $\bar{b}_1, \ldots, \bar{b}_n$ belong to
$\Xi(\mf{A})$ and (ii) $(\Xi(\mf{B}), \bar{b}_1, \ldots, \bar{b}_n)
\subseteq (\Xi(\mf{A}), \bar{b}_1, \ldots, \bar{b}_n)$.
\end{lemma}
\begin{proof}

Consider any element of $\Xi(\mf{B})$; it is a $t$-tuple $\bar{b}$ of
$\mf{B}$ such that $(\mf{B}, \bar{b}) \models \xi(\bar{x})$. Since
$\xi(\bar{x})$ is quantifier-free, it is preserved under extensions
over $\cl{S}$. Whereby $(\mf{A}, \bar{b}) \models \xi(\bar{x})$; then
$\bar{b}$ is an element of $\Xi(\mf{A})$. Since $\bar{b}$ is an
arbitrary element of $\Xi(\mf{B})$, we have $\mathsf{U}_{\Xi(\mf{B})}
\subseteq \mathsf{U}_{\Xi(\mf{A})}$. In particular therefore,
$\bar{b}_1, \ldots, \bar{b}_n$ belongs to $\Xi(\mf{A})$.

Consider a relation symbol $R \in \sigma$ of arity say $n$. Let
$\bar{d}_1, \ldots, \bar{d}_n$ be elements of $\Xi(\mf{B})$. Then we
have the following. Below $\bar{x}_i = (x_{i, 1}, \ldots, x_{i, t})$
for each $i \in \{1, \ldots, n\}$.
\[
\begin{array}{lllll}
& (\Xi(\mf{B}), \bar{d}_1, \ldots, \bar{d}_n) & \models & R(x_1,
  \ldots, x_n) & \\

\mbox{iff} & (\mf{B}, \bar{d}_1, \ldots, \bar{d}_n) & \models &
\Xi(R)(\bar{x}_1, \ldots, \bar{x}_n) & \mbox{(by
  Proposition~\ref{prop:relating-transductions-applications-to-structures-and-formulae})}\\

\mbox{iff} & (\mf{B}, \bar{d}_1, \ldots, \bar{d}_n) & \models &
\bigwedge_{i = 1}^{i = n}\xi(\bar{x}_i) ~\wedge~ \xi_R(\bar{x}_1,
\ldots, \bar{x}_n) & \mbox{(by defn. of $\Xi(R)$; see Section~\ref{section:background-FMT-translation-schemes})} \\
\end{array}
\]
Now since (i) each of $\xi$ and $\xi_R$ is quantifier-free, (ii) a
finite conjunction of quantifier-free formulae is a
quantifier-free formula, and (iii) a quantifier-free formula is
preserved under substructures as well as preserved under extensions
over any class, we have that
\[
\begin{array}{lllll}
& (\mf{B}, \bar{d}_1, \ldots, \bar{d}_n) & \models & \bigwedge_{i =
    1}^{i = n}\xi(\bar{x}_i) ~\wedge~ \xi_R(\bar{x}_1, \ldots,
  \bar{x}_n) &  \\

\mbox{iff} & (\mf{A}, \bar{d}_1, \ldots, \bar{d}_n) & \models &
\bigwedge_{i = 1}^{i = n}\xi(\bar{x}_i) ~\wedge~ \xi_R(\bar{x}_1,
\ldots, \bar{x}_n) & \\

\mbox{iff} & (\mf{A}, \bar{d}_1, \ldots, \bar{d}_n) & \models &
\Xi(R)(\bar{x}_1, \ldots, \bar{x}_n) & \mbox{(by definition of
  $\Xi(R)$)} \\

\mbox{iff} & (\Xi(\mf{A}), \bar{d}_1, \ldots, \bar{d}_n) & \models & R(x_1,
  \ldots, x_n) & \mbox{(by
  Proposition~\ref{prop:relating-transductions-applications-to-structures-and-formulae})}\\
\end{array}
\]
Since $R$ is an arbitrary relation symbol of $\sigma$, we have that
$\Xi(\mf{B}) \subseteq \Xi(\mf{A})$, whereby \linebreak $(\Xi(\mf{B}),
\bar{d}_1, \ldots, \bar{d}_n) \subseteq (\Xi(\mf{A}), \bar{d}_1,
\ldots, \bar{d}_n)$.
\end{proof}

\begin{proof}[Proof of Theorem~\ref{theorem:ebsp-gebsp-and-transductions}]

\ul{Part 1}: 
Consider a structure $\Xi(\mf{A}) \in \Xi(\cl{S})$ for some structure
$\mf{A} \in \cl{S}$. Let $(\bar{a}_1, \ldots, \bar{a}_k)$ be a
$k$-tuple from $\Xi(\mf{A})$ and let $m \in \mathbb{N}$.  For each $i
\in \{1, \ldots, k\}$, let $\bar{a}_i = (a_{i, 1}, \ldots, a_{i,
  t})$. Let $p = k \cdot t$ and consider the $p$-tuple $\bar{a}$ from
$\mf{A}$ given by $\bar{a} = (a_{1, 1}, \ldots, a_{1, t}, a_{2, 1},
\ldots, a_{2, t}, \ldots, a_{k, 1}, \ldots, a_{k, t})$.  Let $r = t
\cdot m$. Since $\febsp{\cl{S}}{p}$ is true, there exists a witness
function $\witfn{\cl{S}}{p}{\fo}:\mathbb{N} \rightarrow \mathbb{N}$
and a structure $\mf{B}$ such that \linebreak $\febspcond(\cl{S},
\mf{A}, \mf{B}, p, r, \bar{a}, \witfn{\cl{S}}{p}{\fo})$ is true. That
is (i) $\mf{B} \in \cl{S}$, (ii) $\mf{B} \subseteq \mf{A}$, (iii) the
elements of $\bar{a}$ are contained in $\mf{B}$, (iv) $|\mf{B}| \leq
\witfn{\cl{S}}{p}{\fo}(r)$ and (v) $(\mf{B}, \bar{a}) \equiv_r
(\mf{A}, \bar{a})$.

We now show that there exists a function
$\witfn{\Xi(\cl{S})}{k}{\fo}:\mathbb{N} \rightarrow \mathbb{N}$ such
that \linebreak $\febspcond(\Xi(\cl{S}), \Xi(\mf{A}), $ $\Xi(\mf{B}),
k, m, (\bar{a}_1, \ldots, \bar{a}_k), \witfn{\Xi(\cl{S})}{k}{\fo})$ is
true. This would show that $\febsp{\Xi(\cl{S})}{k}$ is true.
\begin{enumerate}[(i), leftmargin=*, nosep]
\item $\Xi(\mf{B}) \in \Xi(\cl{S})$: Obvious from the definition of
  $\Xi(\cl{S})$ and the fact that $\mf{B} \in \cl{S}$.

\item $\Xi(\mf{B}) \subseteq \Xi(\mf{A})$: Follows from
  Lemma~\ref{lemma:qt-free-transductions-preserve-substructure-prop}.

\item The element $\bar{a}_i$ is contained in
  $\mathsf{U}_{\Xi(\mf{B})}$ for each $i \in \{1, \ldots, k\}$: Since
  the elements of $\bar{a}$ are contained in $\mf{B}$, we have for
  each $i \in \{1, \ldots, k\}$, that $\bar{a}_i$ is a $t$-tuple from
  $\mf{B}$.  Now since $\bar{a}_i$ is an element of $\Xi(\mf{A})$, we
  have $(\mf{A}, \bar{a}_i) \models \xi(\bar{x})$. Since
  $\xi(\bar{x})$ is quantifier-free, it is preserved under
  substructures over $\cl{S}$. Whereby $(\mf{B}, \bar{a}_i) \models
  \xi(\bar{x})$; then $\bar{a}_i$ is an element of $\Xi(\mf{B})$, for
  each $i \in \{1, \ldots, k\}$.

\item $(\Xi(\mf{B}), \bar{a}_1, \ldots, \bar{a}_k) \equiv_m
  (\Xi(\mf{A}), \bar{a}_1, \ldots, \bar{a}_k)$: Since $(\mf{B},
  \bar{a}) \equiv_r (\mf{A}, \bar{a})$, it follows from
  Corollary~\ref{corollary:transferring-m-equivalence-across-transductions},
  that $(\Xi(\mf{B}), \bar{a}_1, \ldots, \bar{a}_k) \equiv_m
  (\Xi(\mf{A}), \bar{a}_1, \ldots, \bar{a}_k)$.

\item The existence of a function $\witfn{\Xi(\cl{S})}{k}{\fo}:
  \mathbb{N} \rightarrow \mathbb{N}$ such that $|\Xi(\mf{B})| \leq
  \witfn{\Xi(\cl{S})}{k}{\fo}(m)$: Define
  $\witfn{\Xi(\cl{S})}{k}{\fo}:\mathbb{N} \rightarrow \mathbb{N}$ as
  $\witfn{\Xi(\cl{S})}{k}{\fo}(m) =
  (\witfn{\cl{S}}{p}{\fo}(r))^t$. Since $|\mf{B}| \leq
  \witfn{\cl{S}}{p}{\fo}(r)$, we have that $|\Xi(\mf{B})| \leq
  \witfn{\Xi(\cl{S})}{k}{\fo}(m)$.
\end{enumerate}
It is clear that if $\witfn{\cl{S}}{p}{\fo}$ is computable, then so is
$\witfn{\Xi(\cl{S})}{k}{\fo}$.

\vspace{5pt}\ul{Part 2}: The proof of this part is similar to
the proof above.  
\end{proof}

\vspace{5pt} \tbf{Necessity of the condition on $\Xi$ of being
  quantifier-free in
  Theorem~\ref{theorem:ebsp-gebsp-and-transductions}}:

Let $\tau = \{\leq\}$ and $\sigma = \{E\}$ where $\leq, E$ are binary
relation symbols.  Consider the $(1, \tau, \sigma,
\text{FO})$-translation scheme $\Xi_1$ given by $\Xi_1 = (\xi_1,
\xi^1_E)$ where $\xi_1(x)$ is the formula $ (x = x)$ and $\xi^1_E(x,
y) = \forall z \big(\big((x \leq z) \wedge (x \neq z)\big) \rightarrow
(y \leq z)\big)$. Consider the class $\cl{S}$ of all finite linear
orders. We know from
Theorem~\ref{theorem:words-and-trees-and-nested-words-satisfy-lebsp}
that both $\febsp{\cl{S}}{l}$ and $\mebsp{\cl{S}}{l}$ hold for all $l
\in \mathbb{N}$. The (universal) formula $\xi^1_E(x, y)$ cannot be
$\cl{S}$-equivalent to a quantifier-free formula. To see this, suppose
$\xi^1_E(x, y)$ is $\cl{S}$-equivalent to a quantifier-free formula
$\beta(x, y)$. Consider the structure $\mf{A} = (\{1, 2, 3\},
\leq^{\mf{A}}) \in \cl{S}$ where $\leq^{\mf{A}}$ is the usual linear
order on $\{1, 2, 3\}$. Clearly $(\mf{A}, 1, 3) \models \neg
\xi^1_E(x, y)$ whereby $(\mf{A}, 1, 3) \models \neg \beta(x,
y)$. Since $\neg \beta$ is quantifier-free, it is preserved under
substructures whereby $\mf{B} = (\{1. 3\}, \{(1, 1), (1, 3), (3, 3)
\})$ is such that $\mf{B} \models \neg \beta(x, y)$ and hence $\mf{B}
\models \neg \xi^1_E(x, y)$. The latter is clearly not true. Then
$\Xi_1$ is not quantifier-free.

We now show that $\febsp{\Xi_1(\cl{S})}{k}$ is false for each $k \ge
2$. The class $\Xi_1(\cl{S})$ is the class of all finite directed
paths. It is easy to see that for $m, n \in \mathbb{N}$ such that $m
\ge 4$ and $n \ge 2$, the path $P_n$ of length $n$ (i.e. having $n+1$
vertices) is not $m$-equivalent to any substructure of $P_n$ that
contains both the end-points of $P_n$ and that has size at most
$n$. Then $\febsp{\Xi_1(\cl{S})}{k}$ is false for each $k \ge 2$.

Consider the $(1, \tau, \sigma, \text{FO})$-translation scheme $\Xi_2$
given by $\Xi_2 = (\xi_2, \xi^2_E)$ where $\xi_2 = \xi_1$ and $\xi^2_E
= \neg \xi^1_E$. Let $\mathsf{Neg}_{\sigma} = (\alpha, \alpha_E)$ be
the $(1, \sigma, \sigma, \text{FO})$-translation scheme that is
quantifier-free and such that $\alpha(x)$ is the formula $ (x = x)$
and $\alpha_E(x, y) = \neg E(x, y)$.  For the class $\cl{S}$ as above,
observe that $\Xi_1(\cl{S})$ is exactly the class
$\mathsf{Neg}_{\sigma}(\Xi_2(\cl{S}))$.  Whence if
$\febsp{\Xi_2(\cl{S})}{k}$ is true for some $k \ge 2$, then by Part
(\ref{theorem:ebsp-pres-under-transductions}) above,
$\febsp{\mathsf{Neg}_{\sigma}(\Xi_2(\cl{S}))}{k}$ is true,
contradicting the fact that $\febsp{\Xi_1(\cl{S})}{k}$ is false for
all $k \ge 2$.

%% file: lebsp-closure-under-regular-op-tree-languages.tex
\vspace{5pt}
\subsection{Closure under regular operation-tree languages}\label{subsection:closure-under-regular-op-tree-lang}
\newcommand{\optree}[1]{\ensuremath{\mathsf{Op}\text{-}\mathsf{tree}(#1)}}



Theorem~\ref{theorem:ebsp-gebsp-and-transductions} shows us that
operations that are implemented using quantifier-free translation
schemes, preserve the $\febsp{\cdot}{\cdot}$ or $\mebsp{\cdot}{k}$
property of the class of structures they are applied to. From this,
and from
Lemma~\ref{lemma:lebsp-closure-under-set-theoretic-ops}(\ref{lemma:lebsp-closure-under-set-theoretic-ops:part-2}),
it follows that finite unions of the classes obtained by applying
finite compositions of the aforesaid kind of operations to a given
class $\cl{S}$ of structures, also preserves the
$\febsp{\cdot}{\cdot}$ or $\mebsp{\cdot}{k}$ property of $\cl{S}$.
However, as already mentioned in the introduction, there are
interesting classes of structures that are produced only by taking
infinite such unions; examples include hamming graphs of the
$n$-clique, and the class of all $p$-dimensional grid posets, where
$p$ belongs to an $\mso$-definable (using a linear order) class of
natural numbers. In this section, we discuss the case of such infinite
unions.  Specifically, we show that under reasonable additional
assumptions on the aforementioned operations, that are satisfied by
the operations of disjoint union, join, across connect, and the
various kinds of products mentioned in
Section~\ref{subsection:closure-under-operations-implemented-using-translation-schemes},
it is the case that the property of $\lebsp{\cdot}{0}$ of a class is
preserved under taking the aforementioned infinite unions, provided
that these unions are ``regular'' in a sense that we make precise.
Indeed the infinite unions that produce the examples of hamming graphs
of the $n$-clique, and the class of $p$-dimensional grid posets
referred to above, are regular in our sense, whereby since the
examples are produced using the cartesian product operation, it
follows that each of these satisfies $\lebsp{\cdot}{0}$.

Let $\mathsf{Op}$ be a finite set of operations implementable using
quantifier-free translation schemes. We call the operations of
$\mathsf{Op}$ as \emph{quantifier-free operations}, and abusing
notation, use $\Xi$ to represent these operations.  Let $\rho:
\mathsf{Op} \rightarrow \mathbb{N}$ be such that $\rho(\Xi)$ is the
arity of $\Xi$, for $\Xi \in \mathsf{Op}$. An \emph{operation tree
over $\mathsf{Op}$} \sindex[term]{operation!-tree} is an ordered tree
ranked by $\rho$, in which each internal node is labeled with an
operation of $\mathsf{Op}$, and each leaf node is labeled with the
label $\diamond$, which is a place-holder for an ``input'' structure.
The singleton tree (without any internal nodes) in which the sole node
is labeled with a $\diamond$ is also an operation tree over
$\mathsf{Op}$ (treated as the ``no operation'' tree).  When the
$\diamond$ labels of the leaf nodes of an operation tree $\tree{t}$
are replaced with structures, then the resulting tree $\tree{s}$ can
naturally be seen as a representation tree of a structure
$\mf{A}_{\tree{s}}$.  Formally, the structure $\mf{A}_{\tree{s}}$ can
be defined (up to isomorphism) inductively as follows. If $\tree{s}$
is a singleton, then $\mf{A}_{\tree{s}}$ is the structure labeling the
sole node of $\tree{s}$. Else, let $a_1, \ldots, a_n$ be in increasing
order, the children of the root of $\tree{s}$. Let $\tree{t}_i
= \tree{s}_{\ge a_i}$ be the subtree of $\tree{s}$ rooted at $a_i$,
for $i \in \{1, \ldots, n\}$. Assume (as induction hypothesis) that
the structure $\mf{A}_{\tree{t}_i}$ represented (upto isomorphism) by
the tree $\tree{t}_i$ is already defined.  Let $\Xi$ be the operation
labeling the root of $\tree{s}$.  Then $\mf{A}_{\tree{s}}
= \Xi(\mf{A}_{\tree{t}_1}, \ldots,
\mf{A}_{\tree{t}_n})$ upto isomorphism.

Given an operation tree $\tree{t}$ over $\mathsf{Op}$ and a class
$\cl{S}$ of structures, let $\tree{t}(\cl{S})$ be the
\emph{isomorphism-closed} class of structures represented by the
representation trees obtained by simply replacing the labels of the
leaf nodes of $\tree{t}$, with structures from $\cl{S}$. By extension,
given a class $\mc{V}$ of operation trees over $\mathsf{Op}$, let
$\mc{V}(\cl{S}) = \bigcup_{\tree{t} \in \mc{V}} \tree{t}(\cl{S})$. The
class $\mc{V}(\cl{S})$ is then isomorphism-closed as well.  If
$\mc{V}$ is finite, then
Theorem~\ref{theorem:ebsp-gebsp-and-transductions} and
Lemma~\ref{lemma:lebsp-closure-under-set-theoretic-ops}(\ref{lemma:lebsp-closure-under-set-theoretic-ops:part-2})
show that $\lebsp{\cdot}{\cdot}$ property of $\cl{S}$ remains
preserved under $\mc{V}$, where $\mc{V}$ is seen as a transformation
of a class of structures.  While we are yet to investigate what
happens if $\mc{V}$ is an arbitrary infinite class, we show below that
if $\mc{V}$, seen as a language of ordered ranked trees over
$\mathsf{Op} \cup \{\diamond\}$, is regular (in the sense of
regularity used in the literature for ordered ranked trees), then the
truth of $\lebsp{\cdot}{0}$ is preserved in going from $\cl{S}$ to
$\mc{V}(\cl{S})$, provided that the operations in $\mathsf{Op}$
satisfy the additional properties of ``monotonicity'' and
``$\lequiv{m}$-preservation'' that we define below.

An $n$-ary operation $\Xi$ is said to be \emph{monotone} if for all
structures $\mf{A}_1, \ldots, \mf{A}_n$, we have $\mf{A}_i$ is
(isomorphically) embeddable in $\Xi(\mf{A}_1, \ldots, \mf{A}_n)$. We
say $\Xi$ is \emph{$\lequiv{m}$-preserving} if for all structures
$\mf{A}_1, \ldots, \mf{A}_n, \mf{B}_1, \ldots, \mf{B}_n$, it is the
case that if $\mf{A}_i \lequiv{m} \mf{B}_i$ for all $i \in \{1,
\ldots, n\}$, then $\Xi(\mf{A}_1, \ldots, \mf{A}_n) \lequiv{m}
\Xi(\mf{B}_1, \ldots, \mf{B}_n)$. The operations of disjoint union, join and
across connect seen in
Section~\ref{subsection:closure-under-operations-implemented-using-translation-schemes}
are monotone and $\mequiv{m}$-preserving, while each of the products
mentioned in
Section~\ref{subsection:closure-under-operations-implemented-using-translation-schemes},
like cartesian, tensor, lexicographic and strong products, is monotone
and $\fequiv{m}$-preserving.  The central result of this section can
now be stated as follows.

\begin{theorem}\label{theorem:lebsp(S,0)-pres-under-MSO-def-op-tree-lang}
Let $\mathsf{Op}$ be a finite set of operations, where each operation
in $\mathsf{Op}$ is quantifier-free, monotone and
$\lequiv{m}$-preserving. Let $\mc{V}$ be a class of operation trees
over $\mathsf{Op}$, that is regular. Let $\cl{S}$ be a class of
structures. If $\lebsp{\cl{S}}{0}$ is true, then so is
$\lebsp{\mc{V}(\cl{S})}{0}$. Further, if $\lebsp{\cl{S}}{0}$ has a
computable witness function, then so does $\lebsp{\mc{V}(\cl{S})}{0}$.
\end{theorem}

\begin{proof}
We assume familiarity with the notions and results of
Section~\ref{section:abstract-tree-theorem} for the present proof.

Let $\cl{S}_1 = \bigcup_{\tree{t} \in \cl{S}_2} \tree{t}(\cl{S})$,
where $\cl{S}_2$ be the class of all operation trees over
$\mathsf{Op}$. Then $\cl{S}$ and $\mc{V}$ are resp. subclasses of
$\cl{S}_1$ and $\cl{S}_2$.  Let $\sigmaint = \mathsf{Op}$ and
$\sigmaleaf = \{\mf{A} \mid \mathsf{U}_{\mf{A}} \subseteq \mathbb{N},
\mf{A} \cong \mf{B}, \mf{B} \in \cl{S}\}$. Observe that $\sigmaleaf$
is countable.  Let $\rho: \mathsf{Op} \rightarrow \mathbb{N}$ be such
that $\rho(\Xi)$ is the arity of $\Xi$.  Let $\mc{T}$ be the class of
all representation-feasible trees over $\sigmaint \cup \sigmaleaf$,
that are ranked by $\rho$; then $\mc{T}$ is closed under rooted
subtrees and under replacements with rooted subtrees.

We now construct two representation maps $\mathsf{Str}_i: \mc{T}
\rightarrow \cl{S}_i$ for $i \in \{1, 2\}$ such that for $\tree{s} \in
\mc{T}$, $\stri{1}{\tree{s}}$ is the structure $\mf{A}_{\tree{s}}$
represented by $\tree{s}$ (as defined earlier), while
$\stri{2}{\tree{s}}$ is the operation tree corresponding to $\tree{s}$
(i.e. the tree obtained by simply replacing the leaf nodes of
$\tree{s}$ with $\diamond$).  We now observe the following.
\begin{enumerate}[nosep]
\item The map $\mathsf{Str}_1$ satisfies conditions \ref{B.1} and
  \ref{A.1.b} of Section~\ref{section:abstract-tree-theorem}, for each
  $m \in \mathbb{N}$, because each operation in $\mathsf{Op}$ is
  assumed to be monotone and $\lequiv{m}$ preserving. That
  $\mathsf{Str}_1$ also satisfies \ref{A.1.a} is seen by observing
  that each operation in $\mathsf{Op}$ is implementable using a
  quantifier-free translation scheme (see the paragraph before
  Lemma~\ref{lemma:lebsp-pres-under-n-disj-sum-and-n-copy} for the
  precise meaning of implementability using quantifier-free
  translation schemes), and then using
  Lemmas~\ref{lemma:facts-about-disj-unions},~\ref{lemma:facts-about-n-copies}
  and~\ref{lemma:qt-free-transductions-preserve-substructure-prop}. Whereby,
  $\mathsf{Str}_1$ is $\mc{L}$-height-reduction favourable.
\item The map $\mathsf{Str}_2$ is easily seen to satisfy conditions
  \ref{A.1.a} and \ref{B.1}. That it also satisfies \ref{A.1.b} for
  $\mc{L} = \mso$ and for all $m \ge 2$ follows from the $\mso$
  composition lemma for ordered trees (see
  Lemma~\ref{lemma:mso-composition-lemma-for-ordered-trees}). Whereby,
  $\mathsf{Str}_2$ is $\mso$-height-reduction favourable.
\end{enumerate}

Let $\mf{A}$ be a structure in $\mc{V}(\cl{S})$, and let $m \ge 2$. We
show below the existence of a structure $\mf{B}$ such that
$\lebspcond(\mc{V}(\cl{S}), \mf{A}, \mf{B}, 0, m, \mathsf{null},
\witfn{\mc{V}(\cl{S})}{0}{\mc{L}})$ holds, where $\mathsf{null}$ is
the empty tuple and $\witfn{\mc{V}(\cl{S})}{0}{\mc{L}}$ is a function
from $\mathbb{N}$ to $\mathbb{N}$ such that
$\witfn{\mc{V}(\cl{S})}{0}{\mc{L}}(p) =
\witfn{\mc{V}(\cl{S})}{0}{\mc{L}}(2)$ for $p \leq 2$. It is obvious
then that $\lebspcond(\mc{V}(\cl{S}), \mf{A}, \mf{B}, 0, p,
\mathsf{null}, \witfn{\mc{V}(\cl{S})}{0}{\mc{L}})$ holds for $p \leq
2$. Then $\lebsp{\mc{V}(\cl{S})}{0}$ holds with
$\witfn{\mc{V}(\cl{S})}{0}{\mc{L}}$ being a witness function.

Since $\mf{A} \in \mc{V}(\cl{S})$, there exists $\tree{t} \in \mc{T}$
such that $\stri{1}{\tree{t}} \cong \mf{A}$ and $\stri{2}{\tree{t}}
\in \mc{V}$.  Since $\mc{V}$ is regular, it is defined by an MSO
sentence $\varphi$ (the sentence $\varphi$ exists since regularity
corresponds to MSO definability for ordered ranked trees; see
Section~\ref{section:words-trees-nestedwords}).  Then
$\stri{2}{\tree{t}} \models \varphi$. Let the rank of $\varphi$ be
$n$. Since $\mathsf{Str}_1$ is $\mc{L}$-height-reduction favourable
and $\mathsf{Str}_2$ is $\mso$-height-reduction favourable, we have by
Theorem~\ref{theorem:abstract-tree-theorem}, that there is a
computable function $\eta_2:\mathbb{N} \times \mathbb{N} \rightarrow
\mathbb{N}$ and a subtree $\tree{s}_2$ of $\tree{t}$ in $\mc{T}$, such
that (i) the height of $\tree{s}_2$ is at most $\eta_2(m, n)$, (ii)
$\stri{1}{\tree{s}_2} \hookrightarrow \stri{1}{\tree{t}}$, (iii)
$\stri{1}{\tree{s}_2} \lequiv{m} \stri{1}{\tree{t}}$, and (iv)
$\stri{2}{\tree{s}_2} \mequiv{n} \stri{2}{\tree{t}}$. Since
$\stri{2}{\tree{t}} \models \varphi$, we have $\stri{2}{\tree{s}_2}
\models \varphi$ whereby $\stri{2}{\tree{s}_2} \in \mc{V}$.

Now since $\lebsp{\cl{S}}{0}$ is true, we have for each structure
$\mf{C} \in \cl{S}$, a structure $\mf{C}' \in \cl{S}$ such that (i)
$\mf{C}' \subseteq \mf{C}$ (ii) $|\mf{C}'| \leq
\witfn{\cl{S}}{0}{\mc{L}}(m)$ and (iii) $\mf{C}' \lequiv{m} \mf{C}$,
where $\witfn{\cl{S}}{0}{\mc{L}}$ is a witness function for
$\lebsp{\cl{S}}{0}$.  Let $\tree{s}_1 \in \mc{T}$ be the tree obtained
from $\tree{s}_2$ by replacing each structure $\mf{C}$ labeling a leaf
of $\tree{s}_2$ with the structure $\mf{C}'$ described above.  Since
$\mathsf{Str}_1$ satisfies conditions \ref{A.1.a} and \ref{A.1.b} for
each $m \ge 2$, one can verify that (i) $\stri{1}{\tree{s}_1}
\hookrightarrow \stri{1}{\tree{s}_2} \hookrightarrow
\stri{1}{\tree{t}} \cong \mf{A}$, and (ii) $\stri{1}{\tree{s}_1}
\lequiv{m} \stri{1}{\tree{s}_2} \lequiv{m} \stri{1}{\tree{t}}$. It is
clear that $\stri{1}{\tree{s}_1} \in \mc{V}(\cl{S})$ since
$\stri{2}{\tree{s}_1} = \stri{2}{\tree{s}_2} \in \mc{V}$. Let $\mf{B}$
be the substructure of $\mf{A}$ such that $\mf{B} \cong
\stri{1}{\tree{s}_1}$. Then from the above discussion, we have $\mf{B}
\lequiv{m} \mf{A}$ and $\mf{B} \in \mc{V}(\cl{S})$.  We now show that
$\mf{B}$ is of bounded size. Let $d$ be the maximum arity of any
operation in $\mathsf{Op}$, and $t$ be the maximum of the dimensions
of the translation schemes implementing the operations in
$\mathsf{Op}$ (see the definition of dimension in
Section~\ref{section:background-FMT-translation-schemes}). Recall that
(i) the height of $\tree{s}_1$ is the same as the height of
$\tree{s}_2$ which in turn is at most $\eta_2(m, n)$, and (ii) the
size of any structure labeling a leaf node of $\tree{s}_1$ is at most
$\witfn{\cl{S}}{0}{\mc{L}}(m)$.  Then $|\mf{B}| =
|\stri{1}{\tree{s}_1}| \leq \witfn{\mc{V}(\cl{S})}{0}{\mc{L}}(m) =
f(0)$ where for $0 \leq j \leq \eta_2(m, n)$, $f(j)$ is as defined
below.
\[
f(j)  =  \left\{
\begin{array}{ll}
\witfn{\cl{S}}{0}{\mc{L}}(m) & \mbox{if}~ j = \eta_2(m, n)\\
(d \cdot f(j+1))^t & \mbox{if}~ j < \eta_2(m, n)\\
\end{array}  
\right.
\]

It is now easy to verify that $\lebspcond(\mc{V}(\cl{S}), \mf{A},
\mf{B}, 0, m, \mathsf{null}, \witfn{\mc{V}(\cl{S})}{0}{\mc{L}})$ is
true, where $\mathsf{null}$ denotes the empty tuple. Define
$\witfn{\mc{V}(\cl{S})}{0}{\mc{L}}(p) =
\witfn{\mc{V}(\cl{S})}{0}{\mc{L}}(2)$ for $p \leq 2$.  Then as
reasoned earlier, we have that $\lebsp{\mc{V}(\cl{S})}{0}$ holds with
$\witfn{\mc{V}(\cl{S})}{0}{\mc{L}}$ being a witness function. One can
see that if $\witfn{\cl{S}}{0}{\mc{L}}$ is computable, then so is
$\witfn{\mc{V}(\cl{S})}{0}{\mc{L}}$.
\end{proof}

Using the above theorem, we show below that each of the following
classes, that motivated this section, satisfies $\lebsp{\cdot}{0}$:
the class of hamming graphs of the $n$-clique, and the class of all
$p$-dimensional grid posets where $p$ belongs to an $\mso$ definable
(using a linear order) class of natural numbers.
\begin{enumerate}[leftmargin=*,nosep]
\item Let $\cl{S}$ be a class consisting of only the $n$-clique upto
  isomorphism. Let $\mathsf{Op} = \{ \times\}$ where $\times$ denotes
  cartesian product. Let $\mc{V}$ be the class of all trees over
  $\mathsf{Op}$; clearly $\mc{V}$ is defined by the sentence $\true$,
  and is trivially regular. Observe that the class $\mc{V}(\cl{S})$ is
  exactly the class of all hamming graphs of the $n$-clique. Since
  $\cl{S}$ is finite, $\lebsp{\cl{S}}{0}$, and hence
  $\febsp{\cl{S}}{0}$, is true with a computable witness function (see
  Chapter~\ref{chapter:lebsp}). Since $\times$ is quantifier-free,
  monotone and $\fequiv{m}$-preserving, we have by
  Theorem~\ref{theorem:lebsp(S,0)-pres-under-MSO-def-op-tree-lang}
  that $\febsp{\mc{V}(\cl{S})}{0}$ is true with a computable witness
  function.
\item Let $\cl{S}$ be the class of all linear orders. Let $\mathsf{Op}
  = \{\times\}$. Let $\mc{U}$ be the class of all operation-trees over
  $\mathsf{Op}$ in which each internal node has exactly two children,
  at least one of which is a leaf. It is easy to see that any tree in
  $\mc{U}$ has a ``spine'' consisting of the internal nodes of the
  tree. Let $\mc{V}$ be any $\mso$ definable (over the class of all
  trees over $\mathsf{Op}$) subclass of $\mc{U}$ (like for instance,
  the class of all trees of $\mc{U}$ having a spine of even
  length). Then $\mc{V}$ is clearly regular.  Since $\mathsf{Op}$ is a
  singleton, we can identify $\mc{V}$ with a set $Z$ of natural
  numbers that is definable in $\mso$ using a linear order.  Then
  $\mc{V}(\cl{S})$ can be seen as the class of all $p$-dimensional
  grid posets where $p \in Z$. Since $\lebsp{\cl{S}}{0}$, and hence
  $\febsp{\cl{S}}{0}$, is true with a computable witness function (by
  Theorem~\ref{theorem:words-and-trees-and-nested-words-satisfy-lebsp}),
  it follows from
  Theorem~\ref{theorem:lebsp(S,0)-pres-under-MSO-def-op-tree-lang},
  that $\febsp{\mc{V}(\cl{S})}{0}$ is also true with a computable
  witness function.
\end{enumerate}

One can ask what happens to
Theorem~\ref{theorem:lebsp(S,0)-pres-under-MSO-def-op-tree-lang} for
$k > 0$. From the very special cases we have managed to solve so far,
we believe that new techniques would be necessary, in addition to the ones
currently employed in proving 
Theorem~\ref{theorem:lebsp(S,0)-pres-under-MSO-def-op-tree-lang}.

%% file: additional-properties-of-lebsp.tex
\chapter{Additional studies on $\lebsp{\cdot}{k}$}\label{chapter:additional-studies-on-lebsp}
\input{lebsp-and-SAT}

\input{wqo-and-lebsp}

\input{other-pres-thms-entailed-by-lebsp}

%% file: lebsp-and-SAT.tex
\vspace{-12pt}
\section{$\lebsp{\cl{S}}{k}$ and the decidability of  $\mc{L}$-$\mathsf{Th}(\cl{S})$}\label{section:lebsp-and-SAT}

Denote by $\mc{L}$-$\mathsf{Th}(\cl{S})$ the $\mc{L}$-theory of
$\cl{S}$, i.e. the set of all $\mc{L}$ sentences that are true in all
structures of $\cl{S}$. We have the following result.

\begin{lemma}\label{lemma:lebsp-and-lSAT}
Let $\cl{S}$ be a class of structures such that $\lebsp{\cl{S}}{k}$
holds for some $k \in \mathbb{N}$. If there exists a computable
witness function for $\lebsp{\cl{S}}{k}$, then
$\mc{L}$-$\mathsf{Th}(\cl{S})$ is decidable.
\end{lemma}
\begin{proof}
Let $\varphi$ be an $\mc{L}$ sentence of rank $m$.  Let $\psi = \neg
\varphi$ be the negation of $\varphi$; then $\psi$ has rank $m$ as
well.  Suppose $\psi$ is satisfied in a structure $\mf{A} \in
\cl{S}$. Then since $\lebsp{\cl{S}}{k}$ is true, for any $k$-tuple
$\bar{a}$ from $\mf{A}$, there exists a structure $\mf{B}$ such that
$\lebspcond(\cl{S}, \mf{A}, \mf{B}, k, m, \bar{a},
\witfn{\cl{S}}{k}{\mc{L}})$ is true, where $\witfn{\cl{S}}{k}{\mc{L}}$
is a witness function for $\lebsp{\cl{S}}{k}$. Then, (i) $\mf{B} \in
\cl{S}$, (ii) $|\mf{B}| \leq \witfn{\cl{S}}{k}{\mc{L}}(m)$, and (iii)
$\mf{B} \lequiv{m} \mf{A}$. Whereby $\mf{B} \models \psi$ since
$\mf{A} \models \psi$ and the rank of $\psi$ is $m$.  Thus, if $\psi$
is satisfiable over $\cl{S}$, it is satisfied in a structure of
$\cl{S}$, of size $\leq \witfn{\cl{S}}{k}{\mc{L}}(m)$.  Whereby, if
$\witfn{\cl{S}}{k}{\mc{L}}$ is a computable function, the following
algorithm $\mc{A}$ decides membership in
$\mc{L}$-$\mathsf{Th}(\cl{S})$.

\ul{Algorithm $\mc{A}$:}
\begin{enumerate}[nosep]
\item Compute the rank $m$ of the input sentence $\varphi$, and
  compute the number $p = \witfn{\cl{S}}{k}{\mc{L}}(m)$.
\item Enumerate all the finitely many structures $\mf{C}$ in $\cl{S}$
  of size $\leq p$, and check if the sentence $\psi = \neg \varphi$ is
  true in all of them. Checking if $\psi$ is true in $\mf{C}$ is
  effective since $\mf{C}$ is finite.
\item If some structure $\mf{C}$ is found satisfying $\psi$ in the
  previous step, then output ``No'', else output ``Yes''.
\end{enumerate}
It is clear that $\mc{A}$ indeed decides $\mc{L}$-$\mathsf{Th}(\cl{S})$. 
\end{proof}

As seen in Chapter~\ref{chapter:classes-satisfying-lebsp}, a wide
array of classes $\cl{S}$ satisfy either $\febsp{\cl{S}}{k}$ or
$\mebsp{\cl{S}}{k}$ with computable witness functions, whereby
FO-$\mathsf{Th}(\cl{S})$ or MSO-$\mathsf{Th}(\cl{S})$ resp., is
decidable.

%% file: wqo-and-lebsp.tex
\section{$\lebsp{\cdot}{k}$ and well-quasi-ordering under embedding}\label{section:wqo-and-lebsp}



A pre-order $(A, \leq)$ is said to be a \emph{well-quasi-order}
(w.q.o.)\sindex[term]{well-quasi-order} if for every infinite sequence
$a_1, a_2, \ldots$ of elements of $A$, there exists $i < j$ such that
$a_i \leq a_j$ (see \cite{diestel}).  If $(A, \leq)$ is a w.q.o., we
say that ``$A$ is a w.q.o. under $\leq$''.  An elementary fact is that
if $A$ is a w.q.o. under $\leq$, then for every infinite sequence
$a_1, a_2, \ldots$ of elements of $A$, there exists an infinite
subsequence $a_{i_1}, a_{i_2}, \ldots$ such that $i_1 < i_2 < \ldots$
and $a_{i_1} \leq a_{i_2} \leq \ldots$.

Given a vocabulary $\tau$ and $k \in \mathbb{N}$, let $\tau_k$ be as
usual, the vocabulary obtained by expanding $\tau$ with $k$ fresh and
distinct constant symbols.  Let $\cl{S}$ be a class of
$\tau$-structures.  Denote by $\cl{S}^k$\sindex[symb]{$\cl{S}^k$} the
class of all $\tau_k$-structures whose $\tau$-reducts are structures
in $\cl{S}$.  Observe that $(\cl{S}^k, \hookrightarrow)$ is a
pre-order.  We now define the property $\wqo{\cl{S}}{k}$ via the
notion of w.q.o. mentioned above.

\begin{defn}\label{defn:pwqo} 
We say that \emph{$\wqo{\cl{S}}{k}$ holds} if $(\cl{S}^k,
\hookrightarrow)$ is a well-quasi-order.
\end{defn}

A simple example of a class $\cl{S}$ of structures satisfying
$\wqo{\cl{S}}{k}$ for every $k \in \mathbb{N}$ is a finite class of
finite structures. The celebrated results such as Higman's lemma and
Kruskal's tree theorem~\cite{diestel} state that
$\wqo{{\words}(\Sigma)}{0}$ and $\wqo{{\unorderedtrees}(\Sigma)}{0}$
respectively hold. Also, the results in~\cite{shrub-depth} show that
$\wqo{\ncographs}{0}$ holds, for every $n \in \mathbb{N}$, where
$\ncographs$ is the class of all $n$-partite cographs.

\emph{A priori}, there is no reason to expect any relation between the
$\wqo{\cdot}{k}$ and $\lebsp{\cdot}{k}$ properties.  Surprisingly, we
have the following result.




\begin{theorem}\label{theorem:wqo-implies-lebsp}
Let $\cl{S}$ be a class of structures that is closed under
isomorphisms, and let $k \in \mathbb{N}$. If $\wqo{\cl{S}}{k}$ holds,
then so does $\lebsp{\cl{S}}{k}$.
\end{theorem}
\begin{proof}
We prove the result by contradiction.  Suppose, if possible,
$\wqo{\cl{S}}{k}$ holds but \linebreak$\lebsp{\cl{S}}{k}$ fails.  Then
by Definition~\ref{defn:lebsp}, there exists $m \in \mathbb{N}$ such
that for all $p \in \mathbb{N}$, there exists a structure $\mf{A}_p$
in $\cl{S}$ and a $k$-tuple $\bar{a}_p$ from $\mf{A}_p$ such that for
any structure $\mf{B} \in \cl{S}$, we have $\big((\mf{B} \subseteq
\mf{A}_p) \,\wedge\, (\bar{a}_p \in \mathsf{U}^k_{\mf{B}}) \,\wedge\,
(|\mf{B}| \leq p)\big) \,\rightarrow\, (\mf{B}, \bar{a}_p)
\not\lequiv{m} (\mf{A}_p, \bar{a}_p)$.

For each $p \geq 1$, fix the structure $\mf{A}_p$ and the tuple
$\bar{a}_p$ that satisfy the above properties. Let $\mf{A}'_p$ be the
structure $(\mf{A}_p, \bar{a}_p) \in \cl{S}^k$.  Consider the sequence
$(\mf{A}'_i)_{i \geq 1}$.  Since $\wqo{\cl{S}}{k}$ holds, $\cl{S}^k$
is a w.q.o. under $\hookrightarrow$.  Therefore, there exists an
infinite sequence $I = (i_1, i_2, \ldots)$ of indices such that $i_1 <
i_2 < \ldots$ and $\mf{A}'_{i_1} \hookrightarrow \mf{A}'_{i_2}
\hookrightarrow \ldots$.  Consider $\Delta_{\mc{L}}(m, \cl{S}^k)$ --
the set of all equivalence classes of the $\lequiv{m}$ relation over
the structures of $\cl{S}^k$. From
Proposition~\ref{prop:finite-index-of-lequiv(m)-relation}, we see that
$\Delta_{\mc{L}}(m, \cl{S}^k)$ is a finite set. Therefore, there
exists an infinite subsequence $J = (j_1, j_2, \ldots)$ of $I$ such
that (i) $j_1 < j_2 < \ldots$ (ii) $\mf{A}'_{j_1} \hookrightarrow
\mf{A}'_{j_2} \hookrightarrow \ldots$, and (iii) $\mf{A}'_{j_1},
\mf{A}'_{j_2}, \ldots$ are all in the same $\lequiv{m}$ class.  Let $r
= |\mf{A}'_{j_1}|$, and let $n > 1$ be an index such that $j_n \geq
r$.  Then $\mf{A}'_{j_1} \hookrightarrow \mf{A}'_{j_n}$ and
$\mf{A}'_{j_1} \lequiv{m} \mf{A}'_{j_n}$.  Fix an embedding $\imath:
\mf{A}'_{j_1} \hookrightarrow \mf{A}'_{j_n}$.

Recall that $\mf{A}'_{j_n} = (\mf{A}_{j_n}, \bar{a}_{j_n})$ whereby
the image of $\mf{A}'_{j_1}$ under $\imath$ is a structure $(\mf{B},
\bar{a}_{j_n})$. Then $\mf{B}$ has the following properties: (i)
$\mf{B} \in \cl{S}$, since $\mf{A}_{j_1} \in \cl{S}$ and $\cl{S}$ is
closed under isomorphisms, (ii) $\mf{B} \subseteq \mf{A}_{j_n}$, (iii)
$\bar{a}_{j_n} \in \mathsf{U}^k_{\mf{B}}$, (iv) $|\mf{B}|=
|\mf{A}'_{j_1}| = r \leq j_n$, and (v) $(\mf{B}, \bar{a}_{j_n})
\lequiv{m} (\mf{A}_{j_n}, \bar{a}_{j_n})$.  This contradicts the
property of $\mf{A}_{j_n}$ stated at the outset, completing the proof.
\end{proof}

\begin{remark}\label{remark:wqo-does-not-imply-computable-lebsp} 
The implication given by Theorem~\ref{theorem:wqo-implies-lebsp} does
not in general, imply the existence of a computable witness function
for $\lebsp{\cl{S}}{k}$. Consider the class $\cl{S}$ of two
dimensional grid posets; $\cl{S}$ can be seen to be w.q.o. under
embedding, whereby $\lebsp{\cl{S}}{0}$ holds. But if there is a
computable witness function for $\lebsp{\cl{S}}{0}$, then by
Lemma~\ref{lemma:lebsp-and-lSAT}, it follows that
$\mc{L}$-$\mathsf{Th}(\cl{S})$ is decidable. Equivalently, the
satisfiability problem for $\mc{L}$ (the problem of deciding if a
given $\mc{L}$ sentence is satisfiable) is decidable over $\cl{S}$,
for both $\mc{L} = \fo$ and $\mc{L} = \mso$. However, this contradicts
the known result that the $\mso$ satisfiability is undecidable over
two dimensional grid posets. The latter class of posets is thus an
example of a class of structures that is w.q.o. under embedding, and
hence satisfies $\mebsp{\cdot}{0}$, but for which there is no
computable witness function for $\mebsp{\cdot}{0}$.
\end{remark}

We now show that the converse to
Theorem~\ref{theorem:wqo-implies-lebsp} does not hold in general.

\begin{proposition}\label{prop:febsp-does-not-imply-wqo}
There exists a class $\cl{S}$ of structures such that
$\febsp{\cl{S}}{0}$ holds but $\wqo{\cl{S}}{0}$ fails.
\end{proposition}
\begin{proof} 
Let $C_n$ (respectively, $P_n$) denote an undirected cycle
(respectively, path) of length $n$. Let $mP_n$ denote the disjoint
union of $m$ copies of $P_n$. Let $H_n = \bigsqcup_{i = 0}^{i = 3^n}
nP_i$ and $G_n = C_{3^n} \sqcup H_n$, where $\sqcup$ denotes disjoint
union.  Now consider the class $\cl{S}$ of undirected graphs that is
closed under isomorphisms, and is given upto isomorphisms by $\cl{S} =
\cl{S}_1 \cup \cl{S}_2$, where $\cl{S}_1 = \{H_n \mid n \ge 1\}$ and
$\cl{S}_2 = \{ G_n \mid n \ge 1\}$.  That $\wqo{\cl{S}}{0}$ fails is
easily seen by considering the sequence $(G_n)_{n \ge 1}$, and noting
that $C_{3^n}$ cannot embed in $C_{3^m}$ unless $m = n$.  We now show
$\febsp{\cl{S}}{0}$ holds with the witness function
$\witfn{\cl{S}}{0}{\fo}$ being given by $\witfn{\cl{S}}{0}{\fo}(m) =
|G_{m}|$. In other words, we show that for $\mf{A} \in \cl{S}$ and $m
\in \mathbb{N}$, there exists $\mf{B}$ such that (i) $\mf{B} \in
\cl{S}$, (ii) $\mf{B} \subseteq \mf{A}$, (iii) $|\mf{B}| \leq
\witfn{\cl{S}}{0}{\fo}(m)$ and (iv) $\mf{B} \equiv_m \mf{A}$. Towards
this, we first present some basic facts about $C_n$, $P_n$, $H_n$ and
$G_n$, that are easy to verify. Let $m \in \mathbb{N}$ be given. 
\begin{enumerate}[nosep]
\item[(F.1)] If $n_1, n_2 \ge 3^m$, then $P_{n_1} \equiv_m P_{n_2}$.
\item[(F.2)] If $n_1 \ge 3^m$ and $n_2 \ge m$, then $n_2P_{n_1} \equiv_m
  mP_{3^m}$.
\item[(F.3)] If $n_1 \leq n_2$, then $H_{n_1}$ always embeds in $H_{n_2}$.
\item[(F.4)] If $m \leq n_1 \leq n_2$, then $H_{n_1} \equiv_m
  H_{n_2}$. (follows from (1) and (2) above)
\item[(F.5)] If $n \ge m$, then $G_n \equiv_m H_n$.
\end{enumerate}

Consider a structure $\mf{A} \in \cl{S}$ and let $m \in
\mathbb{N}$. We have two cases: (a) $\mf{A} \in \cl{S}_1$ (b) $\mf{A}
\in \cl{S}_2$.

\underline{$\mf{A} \in \cl{S}_1$:} Then $\mf{A} = H_n$ for some
$n$. If $n \leq m$, then taking $\mf{B}$ to be $\mf{A}$, we see
that\linebreak $\febspcond(\cl{S}, \mf{A}, \mf{B}, 0, m,
\mathsf{null}, \witfn{\cl{S}}{0}{\fo})$ is true, where $\mathsf{null}$
is the empty tuple.  Else $n > m$. Then consider $H_m$. From the facts
F.3 and F.4, we have that $H_m$ embeds in $H_n$ and that $H_m \equiv_m
H_n$. Then taking $\mf{B}$ to be the isomorphic copy of $H_m$ that is
a substructure of $H_n$, we see that $\febspcond(\cl{S}, \mf{A},
\mf{B}, 0, m, \mathsf{null}, \witfn{\cl{S}}{0}{\fo})$ is indeed true.

\underline{$\mf{A} \in \cl{S}_2$}: Then $\mf{A} = G_n$ for some
$n$. If $n \leq m$, then taking $\mf{B}$ to be $\mf{A}$, we see
that\linebreak $\febspcond(\cl{S}, \mf{A}, \mf{B}, 0, m,
\mathsf{null}, \witfn{\cl{S}}{0}{\fo})$ is true.  Else $n > m$. Then
consider $H_m$. From the facts F.3, F.4 and F.5, we see that $H_m$
embeds in $G_n$ and that $H_m \equiv_m G_n$.  Then taking $\mf{B}$ to
be the isomorphic copy of $H_m$ that is a substructure of $G_n$, we
see that\linebreak $\febspcond(\cl{S}, \mf{A}, \mf{B}, 0, m,
\mathsf{null}, \witfn{\cl{S}}{0}{\fo})$ is indeed true.
\end{proof}

\tbf{Using Theorem~\ref{theorem:wqo-implies-lebsp} as a technique to
  show $\lebsp{\cdot}{k}$ for classes of structures}

Let $\cl{S}$ be the class of all $n$-dimensional grid posets
(i.e. cartesian product of linear orders, iterated $n$ times), for a
given $n \in \mathbb{N}$. From
Theorem~\ref{theorem:words-and-trees-and-nested-words-satisfy-lebsp},
it follows that $\febsp{\cdot}{k}$ holds of the class of all linear
orders (and with a computable witness function). Then using
Lemma~\ref{lemma:lebsp-pres-under-n-disj-sum-and-n-copy} and
Theorem~\ref{theorem:ebsp-gebsp-and-transductions}, we see that
$\febsp{\cl{S}}{k}$ is true for all $k \in \mathbb{N}$ (and with a
computable witness function). But these results do not tell us whether
$\mebsp{\cl{S}}{k}$ is true. We demonstrate below that we can use
Theorem~\ref{theorem:wqo-implies-lebsp} to show that
$\mebsp{\cl{S}}{k}$ is indeed true. Thus
Theorem~\ref{theorem:wqo-implies-lebsp} gives us a new technique to
show $\lebsp{\cdot}{k}$ for classes of structures for which the
$\lebsp{\cdot}{k}$ property cannot be inferred (at least prima facie)
using the results presented in
Chapter~\ref{chapter:classes-satisfying-lebsp}.

We show that $\mebsp{\cl{S}}{k}$ holds by showing that
$\wqo{\cl{S}}{k}$ holds. We show the latter for the case when $\cl{S}$
is the class of all 2-dimensional grid posets. The proof for the case
of $r$-dimensional grid posets for $r > 2$ can be done similarly.

Consider an infinite sequence $(G_i, \bar{a}_i)_{i \ge 0}$ of
structures of $\cl{S}^k$, where $G_i$ is a 2-dimensional grid poset
and $\bar{a}_i$ is a $k$-tuple from $G_i$, for $i \ge 1$. Let $G_i$ be
the cartesian product of linear orders $L_{i, 1}$ and $L_{i, 2}$, and
let $\bar{b}_i$ and $\bar{c}_i$ be the projections of $\bar{a}_i$ onto
$L_{i, 1}$ and $L_{i, 2}$ respectively (in other words, $\bar{b}_i$ is
the $k$-tuple of first components of the elements of $\bar{a}_i$,
while $\bar{c}_i$ is the $k$-tuple of the second components of the
elements of $\bar{a}_i$). Now $(L_{i, 1}, \bar{b}_i)$ can be looked at
as a word $w_{i, 1}$ over the powerset of $\{1,...,k\}$, such that (i)
the underlying linear order of $w_{i, 1}$ is $L_{i, 1}$, and (ii) each
position $e$ of $w_{i, 1}$ is labeled with the set of all those
indices $r$ in $\{1, ... k\}$ such that $e$ equals the $r^{\text{th}}$
component of $\bar{b}_i$. Similarly $(L_{i, 2}, \bar{c}_i)$ can be
looked at as a word $w_{i, 2}$.  Let $M_i$ be the cartesian product of
the words $w_{i, 1}$ and $w_{i, 2}$, and let $N_i$ be labeled grid
poset obtained from $M_i$ such that (i) the unlabeled grid underlying
$N_i$ is exactly the same as the unlabeled grid underlying $M_i$ (and
the latter is the same as $G_i$), and (ii) the label of any element
$(g_1, g_2)$ of $N_i$ is the intersection of the labels of $g_1$ and
$g_2$ in $w_{i, 1}$ and $w_{i, 2}$. It is easy to see that $N_i$ is
simply a ``coloured'' representation of $(G_i, \bar{a}_i)$. Whereby if
for $i, j \ge 0$, we have $N_i \hookrightarrow N_j$, then $(G_i,
\bar{a}_i) \hookrightarrow (G_j, \bar{a}_j)$. The proof that
$\wqo{\cl{S}}{k}$ holds is therefore completed by showing that indeed
there exist $i, j \ge 0$ such that $i < j$ and $N_i \hookrightarrow
N_j$.

Consider the sequences $(w_{i, 1})_{i \ge 0}$ and $(w_{i, 2})_{i \ge
  0}$.  Since words are w.q.o. under embedding (Higman's lemma) and
the cartesian product of two w.q.o. sets is also w.q.o. under the
point-wise order, there exist $i, j$ such that $i < j$ and the pair
$(w_{i, 1}, w_{i, 2}) \hookrightarrow (w_{j, 1}, w_{j, 2})$ where
$\hookrightarrow$ for pairs means point-wise $\hookrightarrow$. Then
$w_{i, 1} \hookrightarrow w_{j, 1}$ and $w_{i, 2} \hookrightarrow
w_{j, 2}$, whereby $M_i \hookrightarrow M_j$, and hence $N_i
\hookrightarrow N_j$.

\vspace{5pt} On a final note for this section, we observe that the
implication given by Theorem~\ref{theorem:wqo-implies-lebsp}, taken in
its contrapositive form, gives a \emph{logic-based tool} to show
non-w.q.o.-ness of a class of structures under isomorphic embedding.

%% file: other-pres-thms-entailed-by-lebsp.tex
\section[$\lebsp{\cdot}{k}$ and the homomorphism preservation theorem]{$\lebsp{\cdot}{k}$ and the homomorphism preservation \\theorem}\label{section:lebsp-entails-the-hpt}

The homomophism preservation theorem
($\hpt$)\sindex[term]{homomorphism preservation theorem} is one of the
important classical preservation theorems that has been of significant
interest in the finite model theory
setting~\cite{dawar-hom,dawar-quasi-wide,nowhere-dense-original,nowhere-dense-2}. While
the theorem was shown to be true over various special classes of
finite structures (such as those seen earlier in
Chapter~\ref{chapter:need-for-new-classes-for-GLT}, namely classes
that are acyclic, of bounded degree or of bounded
tree-width~\cite{dawar-hom}), its status over the class of all finite
structures was open for a long time. In a landmark
paper~\cite{rossman-hom}, Rossman proved that this theorem is indeed a
rare classical preservation theorem that holds over the class of all
finite structures. However, Rossman's result does not imply anything
about the truth of the $\hpt$ over the aforementioned special classes
of finite structures, since restricting the theorem to special classes
weakens both the hypothesis and the conclusion of the theorem.  In
this section, we show that the homomorphism preservation theorem, in
fact a parameterized generalization of it along the lines of
$\mathsf{GLT}(k)$, holds over classes that satisfy $\lebsp{\cdot}{k}$.

We first formally define the notion of homomorphism and state the
$\hpt$. While the $\hpt$ holds for arbitrary vocabularies, we restrict
our discussion to vocabularies $\tau$ containing only relation
symbols.  Given a vocabulary $\tau$ and $\tau$-structures $\mf{A}$ and
$\mf{B}$, a \emph{homomorphism} from $\mf{A}$ to $\mf{B}$, denoted $h:
\mf{A} \rightarrow \mf{B}$, is a function $h: \mathsf{U}_{\mf{A}}
\rightarrow \mathsf{U}_{\mf{B}}$ such that $(a_1, \ldots, a_n) \in
R^{\mf{A}}$ implies $(h(a_1), \ldots, h(a_n)) \in R^{\mf{B}}$ for
every $n$-ary relation symbol $R \in \tau$.  We say that an $\fo$
sentence $\varphi$ is \emph{preserved under homomorphisms} over a
class $\cl{S}$ of structures if for all structures $\mf{A}, \mf{B} \in
\cl{S}$, if $\mf{A} \models \varphi$ and there is a homomorphism from
$\mf{A}$ to $\mf{B}$, then $\mf{B} \models \varphi$. We say an FO
formula is \emph{existential-positive} if it is built up from
un-negated atomic formulas using conjunction, disjunction and
existential quantification. The $\hpt$ characterizes preservation
under homomorphisms using existential-positive sentences.

\begin{theorem}[$\hpt$]
A first order sentence is preserved under homomorphisms over all
structures iff it is equivalent over all structures to an
existential-positive sentence.
\end{theorem}

We now define a parameterized generalization of the notion of
preservation under homomorphisms, along the lines of preservation
under $k$-ary covered extensions seen in
Section~\ref{section:PCE(k)}. For this, we first define the notion of
\emph{$k$-ary homomorphic covering} as a parameterized generalization
of the notion of homomorphism. Recall, for a vocabulary $\tau$, that
$\tau_k$ is the vocabulary obtained by expanding $\tau$ with $k$ fresh
and distinct constants $c_1, \ldots, c_k$.

\begin{defn}\label{defn:k-ary-hom-cover}
Let $\mf{A}$ be a $\tau$-structure, and $k \in \mathbb{N}$.  Let
$\bar{a}_1, \ldots, \bar{a}_t$ be an enumeration of the $k$-tuples of
$\mf{A}$, and let $I = \{1, \ldots, t\}$.  Let $R = \{(\mf{B}_i,
\bar{b}_i) \mid i \in I\}$ be a (non-empty) set of
$\tau_k$-structures. A \emph{$k$-ary homomorphic covering from $R$ to
  $\mf{A}$} is a set $\mc{H}$ of homomorphisms $\{h_i:(\mf{B}_i,
\bar{b}_i) \rightarrow (\mf{A}, \bar{a}_i) \mid i \in I\}$. If
$\mc{H}$ exists, then we call $R$ a \emph{$k$-ary homomorphic cover}
of $\mf{A}$.
\end{defn}
\begin{remark}
Observe that if $R = \{\mf{B}\}$ for some $\tau$-structure $\mf{B}$,
then $R$ is a 0-ary homomorphic cover of $\mf{A}$ iff there is a
homomorphism from $\mf{B}$ to $\mf{A}$. Also for a structure $\mf{A}$,
if $t$ is the number of $k$-tuples of elements of $\mf{A}$, then for a
set $R$ of $\tau_k$-structures, if there exists a $k$-ary homomorphic
covering from $R$ to $\mf{A}$, then we require that $|R| = t$.
\end{remark}
We now define the notion of \emph{preservation under $k$-ary
  homomorphic coverings}. Recall from
Section~\ref{section:wqo-and-lebsp} that for a class $\cl{U}$ of
$\tau$-structures, $\cl{U}^k$ denotes the class of all
$\tau_k$-structures whose $\tau$-reducts are structures in $\cl{U}$.

\begin{defn}\label{defn:h-PC}
Let $\cl{S}$ be a class of structures and $k \in \mathbb{N}$. A
subclass $\cl{U}$ of $\cl{S}$ is said to be \emph{preserved under
  $k$-ary homomorphic coverings} over $\cl{S}$, abbreviated as
\emph{$\cl{U}$ is $\hpc{k}$ over $\cl{S}$}, if for every collection
$R$ of structures of $\cl{U}^k$, if there is a $k$-ary homomorphic
covering from $R$ to $\mf{A}$ and $\mf{A} \in \cl{S}$, then $\mf{A}
\in \cl{U}$. Given an $\mc{L}$-sentence $\phi$, we say \emph{$\phi$ is
  $\hpc{k}$ over $\cl{S}$} if the class of models of $\phi$ in
$\cl{S}$ is $\hpc{k}$ over $\cl{S}$.
\end{defn}

A class of sentences that is $\hpc{k}$ over any class of structures is
the class of, what we call, \emph{$(\forall^k
  \exists^*)\text{-positive}$ sentences}. A formula $\varphi$ is said
to be $(\forall^k \exists^*)\text{-positive}$ if it is of the form
$\forall x_1 \ldots \forall x_k \psi(x_1, \ldots, x_k)$ where
$\psi(x_1, \ldots, x_k)$ is an existential positive formula. Observe
that for $k = 0$, the class of $(\forall^k \exists^*)\text{-positive}$
formulae is exactly the class of existential-positive formulae. We say
that the \emph{generalized $\hpt$ for $\mc{L}$ and parameter $k$},
abbreviated $\lghpt{k}$, holds over a class $\cl{S}$ if the following
is true: An $\mc{L}$ sentence $\phi$ is $\hpc{k}$ over $\cl{S}$ iff
$\phi$ is equivalent over $\cl{S}$ to a $(\forall^k
\exists^*)$-positive (FO) sentence.  Observe that $\fghpt{0}$ holds
over a class $\cl{S}$ iff $\hpt$ holds over $\cl{S}$.  We show below
that $\lghpt{k}$ holds over classes of structures that satisfy
$\lebsp{\cdot}{k}$.  We in fact show something more general as we
describe below. Towards this, we first present a ``homomorphic''
version of $\mc{L}$-$\mathsf{EBSP}$.

\begin{defn}[$\hlebsp{\cl{S}}{k}$]
Let $\cl{S}$ be a class of finite structures and $k$ be a natural
number.  We say that $\cl{S}$ satisfies the \emph{homomorphic
  $\mc{L}$-$\mathsf{EBSP}$} for parameter $k$, abbreviated
\emph{$\hlebsp{\cl{S}}{k}$ is true}, if there exists a function
$\witfn{\cl{S}}{k}{\mc{L}}: \mathbb{N} \rightarrow \mathbb{N}$ such
that for each $m \in \mathbb{N}$, for each structure $\mf{A}$ of
$\cl{S}$ and for every $k$-tuple $\bar{a}$ from $\mf{A}$, there exists
$(\mf{B}, \bar{b}) \in \cl{S}^k$ such that (i) there is a homomorphism
$h: (\mf{B}, \bar{b}) \rightarrow (\mf{A}, \bar{a})$, (ii) $|\mf{B}|
\leq \witfn{\cl{S}}{k}{\mc{L}}(m)$, and (iii)
$\ltp{\mf{B}}{\bar{b}}{m}(\bar{x}) = \ltp{\mf{A}}{\bar{a}}{m}(\bar{x})
$.  We call $\witfn{\cl{S}}{k}{\mc{L}}$ a \emph{witness function} of
$\hlebsp{\cl{S}}{k}$.
\end{defn}


The following lemma is easy to see and the proof is skipped.
\begin{lemma}\label{lemma:relating-ebsp-and-gebsp-with-their-hom-versions}
Let $\cl{S}$ be a class of structures. Then for each $k \in
\mathbb{N}$, we have the following.
\begin{enumerate}[nosep]
\item $\lebsp{\cl{S}}{k}$ implies $\hlebsp{\cl{S}}{k}$
\item $\hmebsp{\cl{S}}{k}$ implies $\hfebsp{\cl{S}}{k}$
\end{enumerate}
Further, in each of the implications above, any witness function for
the antecedent is also a witness function for the consequent.
\end{lemma}


We now state and prove the central result of this section.

\begin{theorem}\label{theorem:h-ebsp-implies-ghpt(k)}
Let $\cl{S}$ be a class of finite structures and $k \in \mathbb{N}$ be
such that $\hlebsp{\cl{S}}{k}$ is true. Then $\lghpt{k}$, and hence
$\hpt$, holds over $\cl{S}$.
Further, if there is a computable witness function for
$\hlebsp{\cl{S}}{k}$, then the translation from an $\mc{L}$ sentence
that is $\hpc{k}$ over $\cl{S}$ to an $\cl{S}$-equivalent
$(\forall^k\exists^*)$-positive sentence, is effective.

The same statement as above holds when $\hlebsp{\cl{S}}{k}$ is
replaced with $\lebsp{\cl{S}}{k}$.
\end{theorem}

Towards the proof of Theorem~\ref{theorem:h-ebsp-implies-ghpt(k)}, we
recall the notion of canonical conjunctive query from the
literature. Given a $\tau$-structure $\mf{A}$ of size $n$, the
\emph{canonical conjunctive query} associated with $\mf{A}$, denoted
$\xi_{\mf{A}}$, is the sentence given by $\xi_{\mf{A}} = \exists x_1
\ldots \exists x_n \beta(x_1, \ldots, x_n)$ where $\beta(x_1, \ldots,
x_n)$ is the conjunction of all atomic formulae of the form
$R(x_{i_1}, \ldots, x_{i_r})$ where $R \in \tau$, $r$ is the arity of
$R$, $i_1, \ldots, i_r \in \{1, \ldots, n\}$ and $(a_{i_1}, \ldots,
a_{i_r}) \in R^{\mf{A}}$. Observe that $\xi_{\mf{A}}$ is an
existential-positive sentence. The following theorem by Chandra and
Merlin characterizes when a homomorphism exists from a structure
$\mf{A}$ to a structure $\mf{B}$, in terms of $\xi_{\mf{A}}$.

\begin{theorem}[Chandra-Merlin, 1977]\label{theorem:chandra-merlin}
Let $\mf{A}$ and $\mf{B}$ be two finite structures. Then there is a
homomorphism from $\mf{A}$ to $\mf{B}$ iff $\mf{B} \models
\xi_{\mf{A}}$.
\end{theorem}

We now prove Theorem~\ref{theorem:h-ebsp-implies-ghpt(k)}.

\begin{proof}[Proof of Theorem~\ref{theorem:h-ebsp-implies-ghpt(k)}] 
Suppose $\hlebsp{\cl{S}}{k}$ is true.

\ul{`If' part of $\lghpt{k}$}: Let $\phi$ be an $\mc{L}$ sentence that
is equivalent over $\cl{S}$ to the $(\forall^k \exists^*)$-positive
sentence $\varphi = \forall^k \bar{x} \psi(\bar{x})$ where $\psi$ is
an existential-positive formula. Let $R = \{(\mf{B}_i, \bar{b}_i) \in
\cl{S}^k \mid i \in I\}$ be a set of structures from $\cl{S}^k$ such
that $\mf{B}_i \models \phi$ for each $i \in I$. Let $\mf{A} \in
\cl{S}$ and suppose there exists a $k$-ary homomorphic covering
$\mc{H}$ from $R$ to $\mf{A}$. Consider a $k$-tuple $\bar{a}$ from
$\mf{A}$. Since $\mc{H}$ is a $k$-ary homomorphic covering, there
exists $i \in I$ such that there is a homomorphism $h: (\mf{B}_i,
\bar{b}_i) \rightarrow (\mf{A}, \bar{a}) \in \mc{H}$.  Since $\mf{B}
\models \phi$, we have $\mf{B} \models \varphi$ and hence $(\mf{B},
\bar{b}_i) \models \psi(\bar{x})$. Since existential-positive formulas
are preserved under homomorphisms, we have $(\mf{A}, \bar{a}) \models
\psi(\bar{x})$.  Since $\bar{a}$ is arbitrary, we have $\mf{A} \models
\varphi$, whence $\mf{A} \models \phi$. Then $\phi$ is $\hpc{k}$ over
$\cl{S}$.

\ul{`Only if' part of $\lghpt{k}$}: Let $\phi$ be an $\mc{L}$ sentence
that is $\hpc{k}$ over $\cl{S}$. Let the rank of $\phi$ be $m$ and let
$p = \witfn{\cl{S}}{k}{\mc{L}}(m)$, where $\witfn{\cl{S}}{k}{\mc{L}}$
is a \emph{witness function} of $\hlebsp{\cl{S}}{k}$. Let
$\text{Mod}(\cl{S}^k, \phi, p)$ be the set (upto isomorphism) of all
models of $\phi$ in $\cl{S}^k$ that have size $\leq p$. For $(\mf{B},
\bar{b}) \in \text{Mod}(\cl{S}^k, \phi, p)$, let $\xi_{(\mf{B},
  \bar{b})}$ be the canonical conjunctive query associated with
$(\mf{B}, \bar{b})$. Observe that $\xi_{(\mf{B}, \bar{b})}$ is an
FO$(\tau_k)$ sentence.  Let $\xi_{(\mf{B}, \bar{b})}[c_1 \mapsto x_1;
  \ldots; c_k \mapsto x_k]$ be the formula whose free variables are
among $x_1, \ldots, x_k$, that is obtained by substituting $x_i$ for
the free occurrences of $c_i$ in $\xi_{(\mf{B}, \bar{b})}$ for $i \in
\{1, \ldots, k\}$, where $c_!, \ldots, c_k$ are the constants of
$\tau_k \setminus \tau$. We abuse notation slightly and denote
$\xi_{(\mf{B}, \bar{b})}[c_1 \mapsto x_1; \ldots; c_k \mapsto x_k]$
simply as $\xi_{(\mf{B}, \bar{b})}(x_1, \ldots, x_k)$.  Now consider
the sentence $\varphi = \forall x_1 \ldots \forall x_k \alpha(x_1,
\ldots, x_k)$ where $\alpha(x_1, \ldots, x_k) = \bigvee_{(\mf{B},
  \bar{b}) \in \text{Mod}(\cl{S}^k, \phi, p)} \xi_{(\mf{B},
  \bar{b})}(x_1, \ldots, x_k)$. Clearly $\varphi$ is $(\forall^k
\exists^*)$-positive. We show below that $\phi$ is equivalent to
$\varphi$ over $\cl{S}$.

\begin{itemize}[leftmargin=*,nosep]
\item \ul{$\phi \rightarrow \varphi$:} Let $\mf{A} \in \cl{S}$ be such
  that $\mf{A} \models \phi$.  Let $t = |\mf{A}|^k$ and let
  $\bar{a}_1, \ldots, \bar{a}_{t}$ be an enumeration of the $k$-tuples
  of $\mf{A}$. Let $I = \{1, \ldots, t\}$.  Since $\hlebsp{\cl{S}}{k}$
  is true, we have for each $k$-tuple $\bar{a}_i$ from $\mf{A}$ where
  $i \in I$, that there exists a structure $(\mf{B}_i, \bar{b}_i) \in
  \cl{S}^k$ such that (i) there is a homomorphism $h_i: (\mf{B}_i,
  \bar{b}_i) \rightarrow (\mf{A}, \bar{a}_i)$, (ii) $|\mf{B}_i| \leq
  \witfn{\cl{S}}{k}{\mc{L}}(m) = p$, and (iii)
  $\ltp{\mf{B}_i}{\bar{b}_i}{m}(\bar{x}) =
  \ltp{\mf{A}}{\bar{a}_i}{m}(\bar{x})$. Then $\mf{B}_i \lequiv{m}
  \mf{A}$. Since the rank of $\phi$ is $m$, we have $\mf{B}_i \models
  \phi$; then $(\mf{B}_i, \bar{b}_i) \in \text{Mod}(\cl{S}^k, \phi,
  p)$. Let $\xi_{(\mf{B}_i, \bar{b}_i)}$ be the canonical conjunctive
  query associated with $(\mf{B}_i, \bar{b}_i)$, where $i \in I$. By
  Theorem~\ref{theorem:chandra-merlin}, we have for each $i \in I$,
  that $(\mf{A}, \bar{a}_i) \models \xi_{(\mf{B}_i, \bar{b}_i)}$
  whereby $(\mf{A}, \bar{a}_i) \models \alpha(x_1, \ldots, x_k)$.
  Then $\mf{A} \models \varphi$.

\item \ul{$\varphi \rightarrow \phi$:} Let $\mf{A} \in \cl{S}$ be such
  that $\mf{A} \models \varphi$.  As before, let $I = \{1, \ldots,
  t\}$ and $\bar{a}_1, \ldots, \bar{a}_{t}$ be an enumeration of the
  $k$-tuples of $\mf{A}$. Since $\mf{A} \models \varphi$, we have,
  recalling the form of $\varphi$, that for each $i \in I$, $(\mf{A},
  \bar{a}_i) \models \xi_{(\mf{B}_i, \bar{b}_i)}(x_1, \ldots, x_k)$
  for some $(\mf{B}_i, \bar{b}_i) \in \text{Mod}(\cl{S}^k, \phi,
  p)$. By Theorem~\ref{theorem:chandra-merlin}, there is a
  homomorphism $h_i: (\mf{B}_i, \bar{b}_i) \rightarrow (\mf{A},
  \bar{a}_i)$. Then the set $\{h_i \mid i \in I\}$ is a $k$-ary
  homomorphic covering from $R$ to $\mf{A}$.  Since $(\mf{B}_i,
  \bar{b}_i) \in \text{Mod}(\cl{S}^k, \phi, p)$, we have $(\mf{B}_i,
  \bar{b}_i) \models \phi$ for each $i \in I$. Then since $\phi$ is
  $\hpc{k}$ over $\cl{S}$, we have $\mf{A} \models \phi$.
\end{itemize}

\vspace{2pt}
That the above result holds when $\hlebsp{\cl{S}}{k}$ is replaced with
$\lebsp{\cl{S}}{k}$ follows directly from
Lemma~\ref{lemma:relating-ebsp-and-gebsp-with-their-hom-versions}.
\end{proof}

%% file: open-questions-finite-model-theory.tex
\chapter{Directions for future work}\label{chapter:open-questions-finite-model-theory}

The results seen so far naturally motivate various questions that we
propose as future work.

\tbf{A. Questions regarding $\lebsp{\cdot}{k}$:}
\begin{enumerate}[leftmargin=*, nosep]
\item {[}{Model-theoretic}{]} The {\lt} theorem and the homomorphism
  preservation theorem are true over any class satisfying
  $\lebsp{\cdot}{k}$ (Theorems~\ref{theorem:lebsp-implies-glt(k)} and
  \ref{theorem:h-ebsp-implies-ghpt(k)}).  What other theorems of
  classical model theory are true of classes satisfying
  $\lebsp{\cdot}{k}$?  For instance, are Lyndon's positivity theorem
  and Craig's interpolation theorem true of $\lebsp{\cdot}{k}$ classes?
\item {[}{Poset-theoretic}{]} Theorem~\ref{theorem:wqo-implies-lebsp}
  shows us that w.q.o. under embedding entails $\lebsp{\cdot}{0}$. The
  converse however is not true:
  Proposition~\ref{prop:febsp-does-not-imply-wqo} gives a class that
  is not w.q.o. under embedding but for which $\lebsp{\cdot}{0}$
  holds. However, this class is not hereditary. This motivates the
  following question: Under what reasonable closure assumptions on a
  class $\cl{S}$ does $\lebsp{\cl{S}}{0}$ become equivalent to
  w.q.o. under embedding?  Another natural question, given
  Remark~\ref{remark:wqo-does-not-imply-computable-lebsp}, is the
  following: what strengthing of the w.q.o. under embedding property
  entails $\lebsp{\cdot}{0}$ with \emph{computable} witness functions?
\item {[}{Relational structures whose Gaifman graphs are $n$-partite
    cographs}{]} Given a $\tau$-structure $\mf{A}$ where $\tau$ is
  relational, the \emph{Gaifman graph} of $\mf{A}$ is an undirected
  graph $\mc{G}(\mf{A}) = (V, E)$ such that $V$ is exactly
  $\univ{\mf{A}}$, and for $a, b \in V$, the pair $(a, b) \in E$ iff
  for some $r$-ary relation $R \in \tau$ and some $r$-tuple $\bar{c}
  \in R^{\mf{A}}$, it is the case that $\bar{c}$ contains $a$ and $b$
  as components.  We can now ask whether our results showing
  $\lebsp{\cdot}{k}$ for graphs
  (cf. Theorem~\ref{theorem:n-partite-cographs-satisfy-lebsp}) can be
  lifted to relational structures via the Gaifman graphs of the
  latter. Specifically, is it the case under suitable assumptions,
  that a class of finite relational structures whose Gaifman graphs
  form a subclass of $n$-partite cographs, satisfies
  $\lebsp{\cdot}{k}$, and further, with a computable witness function?
  As a step in this direction, we indeed have been able to show that a
  hereditary class of relational structures whose Gaifman graphs are
  of bounded tree-depth, satisfies $\febsp{\cdot}{k}$ with a
  computable witness function (see Theorem 4 of~\cite{abhisekh-mfcs}).
\item {[}{Computational}{]} For what classes of structures that
  satisfy $\lebsp{\cdot}{k}$, is it the case that there are elementary
  witness functions (as opposed to just computable witness functions)?
  For the case of words, trees and nested words, $\lebsp{\cdot}{k}$
  holds with necessarily non-elementary witness functions. This is
  because for any class $\cl{S}$ satisfying $\lebsp{\cdot}{k}$ with
  witness function $\witfn{\cl{S}}{k}{\mc{L}}$, since any structure in
  the class is $(m, \mc{L})$-similar to a structure of size $\leq
  \witfn{\cl{S}}{k}{\mc{L}}(m)$, the index of the $\lequiv{m}$
  relation over $\cl{S}$ is elementary if $\witfn{\cl{S}}{k}{\mc{L}}$
  is elementary. However over words, the index of the $\fequiv{m}$
  relation itself is non-elementary~\cite{frick-grohe}.

  Since the model checking problem for $\mso$ is fixed parameter
  tractable with elementary dependence on formula size, over classes
  of structures of bounded tree-depth or bounded shrub-depth, we would
  like to investigate if this elementariness shows up as the
  elementariness of the witness functions for the $\lebsp{\cdot}{k}$
  properties of the aforementioned classes. If so, this would also
  show that the index of the $\lequiv{m}$ relation is elementary over
  these classes, as reasoned above.

\item {[}{Concerning closure under operations}{]} Is there a syntactic
  characterization of operations that are quantifier-free, monotone
  and $\lequiv{m}$-preserving?
  (cf. Theorem~\ref{theorem:lebsp(S,0)-pres-under-MSO-def-op-tree-lang})
  Also, can
  Theorem~\ref{theorem:lebsp(S,0)-pres-under-MSO-def-op-tree-lang} be
  generalized to $k > 0$?
\item {[}{Structural}{]} Is there a structural characterization of
  posets/graphs that satisfy $\lebsp{\cdot}{k}$? If not in general,
  then under reasonable closure assumptions on the classes (like say
  hereditariness)?  As a step in this direction,
  Theorem~\ref{theorem:glt(k)-implies-bounded-ind-path-lengths} shows
  that any hereditary class $\cl{S}$ of directed graphs for which
  $\lebsp{\cdot}{k}$ holds for any $k \ge 2$ (and hence over which
  $\glt{k}$ holds by Theorem~\ref{theorem:lebsp-implies-glt(k)}) must
  be such that the underlying undirected graphs of the graphs of
  $\cl{S}$ must have bounded induced path lengths.  The converse of
  this statement is a technical challenging question, that we wish to
  investigate. Given the ``empirical evidence'' that many interesting
  classes of structures of interest in computer science satisfy
  $\lebsp{\cdot}{\cdot}$, a structural characterization of the latter,
  even under reasonable assumptions (like hereditariness), might
  ``give back'' notions/new classes of structures of use and relevance
  to computer science. (As a very successful recent example of such a
  ``give back'', a structural characterization under the assumption of
  hereditariness, of the notion of \emph{quasi-wideness} that was
  introduced in~\cite{dawar-hom} in the context of the homomorphism
  preservation theorem, yielded the notion of \emph{nowhere dense
    graphs}~\cite{nowhere-dense-original,nowhere-dense-2}, and this
  class of graphs has turned out to be widely useful from the
  combinatorial and algorithmic points of
  view~\cite{dawar-domination,kernelization,sparsegraphs}.)

\item {[}{Probabilistic}{]} One can define a ``probabilistic version''
  of $\lebsp{\cl{S}}{0}$ in which, instead of asserting that for any
  structure in $\cl{S}$, a bounded $(m, \mc{L})$-similar substructure
  of it that is in $\cl{S}$, exists with probability 1, one asserts
  the same ``with high probability''. One can define an analogous
  probabilistic version of $\lebsp{\cl{S}}{k}$. It would be
  interesting to investigate what classes of graphs satisfy this
  version of $\lebsp{\cdot}{k}$.

\end{enumerate}

\vspace{5pt}
\tbf{B. Questions regarding $\fghpt{k}$:}

\begin{enumerate}[leftmargin=*,nosep]
\item Using techniques very similar to those presented in
  Section~\ref{section:glt-for-sentences}, and using \emph{special
    models} (see Chp. 5 of~\cite{chang-keisler}) instead of
  $\lambda$-saturated models, we can show the following
  result. Observe that $\fghpt{0}$ is exactly $\hpt$.

\begin{futureworktheorem}[The generalized $\hpt$]\label{theorem:ghpt-over-all-structures}
Let $\cl{S}$ be a class of arbitrary structures, that is
elementary. Then $\fghpt{k}$ holds over $\cl{S}$ for each $k \in
\mathbb{N}$.
\end{futureworktheorem}

Over all finite structures, we know that $\fghpt{0}$, which is $\hpt$,
is true by the results of Rossman~\cite{rossman-hom}.  Given that the
$\hpt$ is amongst the very rare theorems from classical model theory
to hold over all finite structures, it would be interesting to
investigate if $\fghpt{k}$ holds over all finite structures for $k >
0$.

\end{enumerate}

\vspace{5pt}
\tbf{C. Questions concerning $\glt{k}$:}

\begin{enumerate}[leftmargin=*,nosep]
\item As Proposition~\ref{prop:glt(k)-failure-for-special-classes}
  demonstrates, each of the classes of structures that are acyclic, or
  of bounded degree (more generally, wide), or of bounded tree-width
  fails to satisfy $\glt{k}$ for $k \ge 2$. A natural question to
  investigate is the case of $k = 1$.
\item Proposition~\ref{prop:failure-of-glt(k)-in-the-finite} shows for
  each $l \ge 0$, that $PSC(l)$ sentences cannot be equivalent to
  $\exists^k \forall^*$ sentences for any fixed $k \ge 0$. In
  particular, for each $k \ge 0$,
  Proposition~\ref{prop:failure-of-glt(k)-in-the-finite} gives a
  sentence $\psi_k$ that is $PS$, and hence $PSC(l)$ for each $l \ge
  0$, over all finite structures, but that is not equivalent in the
  finite, to any $\exists^k \forall^*$ sentence. However, $\psi_k$ is
  itself an $\exists^{k+1} \forall^*$ sentence, i.e. a $\Sigma^0_2$
  sentence
  (cf. Remark~\ref{remark:suggestions-for-PSC=Sigma_2-in-the-finite}).

  This raises the following question: Is it the case that for each $l
  \ge 0$, any sentence that is $PSC(l)$ in the finite is equivalent in
  the finite, to a $\Sigma^0_2$ sentence? Recall that $PSC =
  \bigvee_{l \ge 0} PSC(l)$, and that every $\Sigma^0_2$ sentence is
  $PSC$ over any class of structures. We can then reframe the
  aforesaid question as: Over all finite structures, is it the case
  that a sentence is $PSC$ iff it is equivalent to a $\Sigma^0_2$
  sentence?  We conjecture that this is indeed the case.

\begin{futureworkconj}\label{conjecture}
Over the class of all finite structures, a sentence is $PSC$ if, and
only if, it is equivalent to a $\Sigma^0_2$ sentence.
\end{futureworkconj}

  Over arbitrary structures, $PSC$ \emph{is} characterized by
  $\Sigma^0_2$ as shown by Corollary~\ref{corollary:glt-PSC-PCE}.
  Then proving Conjecture~\ref{conjecture} in the affirmative would
  give us a preservation theorem that is not only true over arbitrary
  structures but also true over all finite structures. It would be
  interesting to investigate (the relativized version of) this
  conjecture over the special classes of structures mentioned in the
  previous point, and also over the classes considered in the context
  of the homomorphism preservation theorem (such as nowhere dense
  classes).
\end{enumerate}

%% file: conclusion.tex
\chapter{A summary of our contributions}\label{chapter:conclusion}

We conclude by summarizing the contributions of this thesis in the
classical and finite model theory settings. In each of these settings,
our contributions are of three kinds: notions, results and techniques.

\vspace{3pt} \underline{Classical model theory}:

\begin{enumerate}[leftmargin=*,nosep]
\item[A.] Notions: We introduce the properties of \emph{preservation
  under substructures modulo $k$-cruxes ($PSC(k)$)} and
  \emph{preservation under $k$-ary convered extensions ($PCE(k)$)} as
  natural parameterized generalizations of the classical properties of
  preservation under substructures and preservation under extensions
  (Definitions~\ref{defn:PSC(k)} and ~\ref{defn:PCE(k)}). Our
  properties are finitary and combinatorial, and are non-trivial both
  over arbitrary structures as well as over finite structures.


\item[B.] Results:
 \begin{enumerate}[leftmargin=*,nosep]
  \item The generalized {\lt} theorem for sentences ($\glt{k}$): This
    result provides semantic characterizations of the $\exists^k
    \forall^*$ and $\forall^k \exists^*$ classes of sentences
    (Theorem~\ref{theorem:glt(k)}). Whereby, we get \emph{finer
      characterizations} of the $\Sigma^0_2$ and $\Pi^0_2$ fragments
    of FO sentences than those in the literature, which are via
    notions like unions of ascending chains, intersections of
    descending chains, Keisler's 1-sandwiches, etc. \emph{None} of the
    latter notions relates the \emph{count} of quantifiers to any
    model-theoretic properties. As a consequence of $\glt{k}$, we
    obtain new semantic characterizations of the $\Sigma^0_2$ and
    $\Pi^0_2$ classes of FO sentences
    (Corollary~\ref{corollary:glt-PSC-PCE}).
  \item New semantic characterizations of the $\Sigma^0_2$ and
    $\Pi^0_2$ classes of FO theories: These characterizations are
    obtained via ``infinitary'' variants of $PSC(k)$ and $PCE(k)$,
    namely, the notions of preservation under substructures modulo
    $\lambda$-cruxes and preservation under $\lambda$-ary covered
    extensions respectively, for infinite cardinals $\lambda >
    \aleph_0$
    (Theorems~\ref{theorem:subst-char-for-theories}(\ref{theorem:PSC(lambda)-char-for-theories})
    and~\ref{theorem:ext-chars}(\ref{theorem:char-of-PCE(lambda)-theories})).
  \item Applications in proving inexpressibility results in FO: We
    give new and simple proofs of well-known inexpressibility results
    in FO, such as inexpressibility of acyclicity, connectedness,
    bipartiteness, etc., using our preservation theorems
    (Section~\ref{subsection:inexp}).
  \end{enumerate}

\item[C.] Techniques: We introduce a novel technique of getting a
  syntactically defined FO theory equivalent to a given FO theory
  satisfying a semantic property, \emph{by going outside of FO}
  (Lemma~\ref{lemma:inf-characterization-of-PSC-var(k)} and
  Proposition~\ref{prop:char-of-inf-formulae-defining-elem-classes}). The
  idea is to first express the semantic property in a syntactically
  defined fragment of an infinitary logic , and then use a
  ``compiler-result'' to translate the aforementioned infinitary
  sentences to equivalent FO theories, when these sentences are known
  to be equivalent to FO theories. The latter FO theories are obtained
  from suitable \emph{finite approximations} of the infinitary
  sentences, that are defined syntactically in terms of the latter.
  We believe this technique of accessing the descriptive power of an
  infinitary logic followed by accessing the translation power of a
  compiler result, may have other applications.
\end{enumerate}


\vspace{7pt}\underline{Finite model theory}:

\begin{enumerate}[leftmargin=*,nosep]
\item[A.] Notions: We define a new logic based combinatorial property
  of finite structures that we call the \emph{$\mc{L}$-Equivalent
    Bounded Substructure Property $\lebsp{\cl{S}}{k}$}
  (Definition~\ref{defn:lebsp}).  

  
\item[B.] Results:
  \begin{enumerate}[leftmargin=*,nosep]
  \item
    A strengthening of the classical result showing the failure of the
    {\lt} theorem in the finite: We show that there is a vocabulary
    $\tau$ such that for each $k$, there is an $\fo(\tau)$ sentence
    that is preserved under substructures over the class $\cl{S}$ of
    all finite structures but that is not equivalent over $\cl{S}$, to
    any $\exists^k \forall^*$ sentence
    (Theorem~\ref{prop:failure-of-glt(k)-in-the-finite}). The case of
    $k = 0$ of this result is the classical failure of the {\lt}
    theorem in the finite.
    
  \item
    A preservation theorem that imposes structural restrictions: We
    show that under the assumption that a given class $\cl{S}$ of
    graphs is hereditary, if $\glt{k}$ holds over $\cl{S}$, then
    $\cl{S}$ must have bounded induced path lengths
    (Theorem~\ref{theorem:glt(k)-implies-bounded-ind-path-lengths}).
    
  \item
    Characterizing prenex FO sentences with two blocks of quantifiers:
    The preservation theorems studied over well-behaved classes,
    namely the {\lt} theorem and the homomorphism preservation
    theorem, characterize $\Sigma^0_1$ and $\Pi^0_1$ sentences --
    sentences that contain only one block of quantifiers -- or
    subclasses of these. We characterize over various interesting
    classes of finite structures, $\Sigma^0_2$ and $\Pi^0_2$
    sentences -- sentences which contain two blocks of quantifiers.

  \item
    Strong connections of $\lebsp{\cl{S}}{k}$ with classical model
    theory: The property of $\lebsp{\cl{S}}{k}$ entails $\glt{k}$ (and
    hence the {\lt} theorem) as well as a generalization of the
    homomorphism preservation theorem, and even ``effective'' versions
    of all these theorems if the witness function for
    $\lebsp{\cl{S}}{k}$ is computable
    (Theorems~\ref{theorem:lebsp-implies-glt(k)}
    and~\ref{theorem:h-ebsp-implies-ghpt(k)}).  Furthermore, from the
    very close resemblance of its definition to that of the {\dls}
    property, $\lebsp{\cl{S}}{k}$ can very well be regarded as a
    finitary analogue of the latter
    (Section~\ref{section:lebsp-and-dlsp}). To the best of our
    knowledge, finitary analogues of intrinsically infinitary
    properties from classical model theory have rarely been studied
    earlier.

  \item
    Strong connections of $\lebsp{\cl{S}}{k}$ with computer science:
    We show that a variety of classes of interest in computer science
    satisfy $\lebsp{\cdot}{k}$, and further, with computable witness
    functions. These include the classes of words, trees (unordered,
    ordered, or ranked), nested words, cographs, graphs of bounded
    tree-depth, graph classes of bounded shrub-depth and $n$-partite
    cographs
    (Theorems~\ref{theorem:words-and-trees-and-nested-words-satisfy-lebsp}
    and~\ref{theorem:n-partite-cographs-satisfy-lebsp}). We show that
    $\lebsp{\cdot}{\cdot}$ remains preserved under finite unions and
    finite intersections, and under taking subclasses that are
    hereditary or $\mc{L}$-definable
    (Lemma~\ref{lemma:lebsp-closure-under-set-theoretic-ops}). Again,
    $\lebsp{\cdot}{\cdot}$ remains preserved under various
    well-studied operations from the literature that are implementable
    using quantifier-free translation schemes; these include unary
    operations like complementation, transpose and the line-graph
    operation, binary ``sum-like'' operations like disjoint union,
    join and pointed substitution, and binary ``product-like''
    operations that include various kinds of products like cartesian,
    tensor, lexicographic and strong products
    (Corollary~\ref{corollaty:lebsp-closure-under-1-step-operations}). While
    it follows that $\lebsp{\cdot}{\cdot}$ remains preserved under
    finite unions of classes obtained by finite compositions of the
    above operations, we show that $\lebsp{\cdot}{0}$ remains
    preserved even infinite unions of such classes, provided these
    unions are ``regular''
    (Theorem~\ref{theorem:lebsp(S,0)-pres-under-MSO-def-op-tree-lang}).
    These various closure properties enables us to construct a wide
    spectrum of classes of finite structures that satisfy
    $\lebsp{\cdot}{\cdot}$, and that are hence ``well-behaved''
    model-theoretically. All of these classes are different from the
    well-behaved classes considered in the
    literature~\cite{dawar-pres-under-ext,dawar-hom,nicole-lmcs-15},
    and were earlier not known to enjoy the many model-theoretic
    properties that they do.

  \item
    New composition results for nested words and $n$-partite cographs:
    Composition results allow inferring the formulas that are
    satisfied in a structure that is built up from smaller structures,
    from the formulas satisfied in the latter
    structures~\cite{makowsky}. Composition results for FO and MSO
    have traditionally been known for words. These have natural
    extensions to (unordered, ordered and ranked) trees. (We prove
    these extensions in this thesis.) We provide new FO and MSO
    composition results for nested words and $n$-partite cographs by
    defining the operations of ``insert'' and ``merge'' for these
    classes respectively, and showing that these operations possess
    the FO and MSO composition properties
    (Lemmas~\ref{lemma:nested-words-composition-lemma}
    and~\ref{lemma:composition-lemma-for-n-partite-cographs}).
  \item
    A new connection between well-quasi-ordering and logic: We show
    that any class of structures that is well-quasi-ordered (w.q.o.)
    under embedding satisfies $\lebsp{\cdot}{0}$
    (Theorem~\ref{theorem:wqo-implies-lebsp}). In contrapositive form,
    this result gives a logic-based tool to show that a class of
    structures is not w.q.o. under embedding. This result also shows that
    classes that are w.q.o. under embedding satisfy the {\lt}
    preservation theorem. This fact does not seem to be
    well-known~\cite{goubault16}.
    
  \end{enumerate}

\item[C.] Techniques: We prove an abstract result concerning tree
  representations (Theorem~\ref{theorem:abstract-tree-theorem}), that
  takes as input a tree-representation of a structure and produces as
  output, a small subtree that represents a small and logically
  similar substructure of the original structure. The output structure
  is obtained by iteratively performing appropriate ``prunings'' of,
  and ``graftings'' within, the input tree representation, in a manner
  that preserves the substructure and ``$(m, \mc{L})$-similarity''
  relations between the structures represented by the trees before and
  after the pruning and grafting operations.  Two key technical
  elements that are employed to perform the aforementioned operations
  are the finiteness of the index of the ``$(m, \mc{L})$-similarity''
  relation and the \emph{type-transfer property} of the
  tree-representation.  We utilize our abstract result in showing the
  $\lebsp{\cdot}{\cdot}$ property for the variety of classes of
  structures that we mentioned earlier. Given that many interesting
  classes of finite structures have natural representations using
  trees, it is possible that our abstract result has more applications
  than the ones indicated in this thesis.
 
\end{enumerate}


%% file: ack-4.0.tex
I have a number of people to thank, who in their own unique ways have
contributed to my growth during the long journey of my Ph.D., and an
even longer stay in the IIT Bombay (IITB henceforth) campus. I have
organized my acknowledgments in two parts: ``academic'' and
``cultural''. In case I have missed out on thanking anyone, I
apologize sincerely, and offer my sincere thanks to them at the outset
itself.

\vspace{7pt} {\large{\tbf{A. ``Academic'' acknowledgments}}}

\vspace{5pt} First and foremost, I am indebted to my advisors 
Supratik Chakraborty and Bharat Adsul for the immense freedom and
unconditional support they gave me to pursue my research for the
length of time that I have. Indeed, but for these, I would never have
been able to express myself through my research, to my satisfaction.
For the many penetrating technical discussions, the valuable feedback
on my results, writing and presentations, the sustained stipend (which
I never had to bother about at all) and generous financial support for
internships and conferences, and for all their appreciation of my work
(which I feel is so vital for any creative activity), I offer my
deepest gratitude to them.


I would like thank my panel members Ajit A. Diwan, Paritosh Pandya,
and Akshay S., for providing me with their valuable comments over the
years in my annual progress seminars. I thank Ajit A. Diwan
additionally for various insightful discussions and for indeed being
an ``oracle'' for my queries on graph theory and combinatorics.  I
thank Akshay S. again, along with Nutan Limaye, for always being
helpful and appreciative, and for the discussions on inexpressibility
results using preservation theorems~\cite{inexp-TR}.

It was my great fortune that I met Anand Pillay at the Logic and Set
Theory satellite meeting of ICM (International Congress of
Mathematicians) 2010, held at CMI (Chennai Mathematical Institute) in
August 2010. I express my sincerest gratitude to him for having
directed me to the subject of classical model theory for the
characterization that I was after then ($\glt{k}$ indeed), and for
providing me, time and again, with the most valuable directions and
advice, extended discussions and historical perspectives on the
subject.


My heartfelt gratitude and thanks go to Anuj Dawar for the precious
interactions and guidance that he provided me with on so many
occasions on the subject of finite model theory. I thank him earnestly
again for giving me the memorable opportunity of doing an internship
with him during June-August 2012 at the University of Cambridge,
UK. His deeply insightful observations and keen directions have proven
to be of vital importance to many aspects of my thesis, particularly
the finite model theory part of it.

I would like to thank the Association for Logic in India (ALI) for
organizing the biennial Indian Conference on Logic and Applications
(ICLA) and the biennial Indian School on Logic and Applications (ISLA)
in alternate years.  Likewise, I thank the Calcutta Logic Circle (CLC)
for organizing the annual meets of CLC.  Having been a regular at
these meets for many years, I have benefitted a lot from them in terms
of exposure to various aspects of logic, getting to meet a number of
stalwarts, and getting platforms to present my work and obtain
important feedback and opportunities.  Members of ALI and CLC whom I
would especially like to thank include Mohua Banerjee, Mihir
Chakraborty, Soma Dutta, Sujata Ghosh, Kamal Lodaya, Ranjan
Mukhopadhyaya, R. Ramajuman, Pulak Samanta, Jayanta Sen and Sourav
Tarafder. I thank them again for all the warmth and touching affection
that they have always given me.  I thank R. Ramanajum further for many
valuable discussions during the time I spent with him at the Institute
of Mathematical Sciences (IMSc), Chennai in September-October 2010.

It is through the above mentioned meets that I got the privilege to
meet and interact with Johan van Benthem and Rohit Parikh.  I am very
grateful to Johan van Benthem for sponsoring my first international
trip to attend the European Summer School for Logic, Language and
Information (ESSLLI) in July 2009 at Bordeaux, France.  I am deeply
thankful to Rohit Parikh for inviting me to spend a month at the
Graduate Center, CUNY (City University of New York), USA in
September-October 2012, and for so many memorable hours that we spent
at New York and during his visits to Mumbai, discussing a variety of
subjects, technical and otherwise. It was during my visit to CUNY that
I had a providential meeting with Saharon Shelah. I offer him my very
sincere thanks for appreciating $\glt{k}$, and for pointing me
to~\cite{keisler:infinitary-logic-book} whose importance I realized
only much later when contemplating on
Theorem~\ref{theorem:conditional-refinement}.  I am very grateful to
Thomas Scanlon for many valuable discussions on classical model
theory, and sincerely thank him, Anand Pillay and Byunghan Kim, for
inviting me to the ICM Satellite Conference on Classification Theory
and its Applications at Daejeon, South Korea in August 2014, which
exposed me to the frontier of one of the currently very active themes
of classical model theory.

I am most thankful to Haim Gaifman, Erich Gr{\"a}del, Martin Grohe,
Yuri Gurevich, Lauri Hella, Phokion Kolaitis, Janos Makowsky, Jaroslav
Ne\v{s}et\v{r}il, Benjamin Rossman and Jouko V\"{a}\"{a}n\"{a}nen for
giving me their valuable time, encouragement and discussions on finite
model theory. I earnestly thank further, Erich Gr{\"a}del for his warm
hospitality during my visit to RWTH Aachen, Germany in September 2015,
and Martin Grohe for raising the possibility that preservation under
substructures may not be capturable by $\exists^k \forall^*$ sentences
for any fixed $k$, over all finite structures -- indeed I was able to
confirm this later
(Theorem~\ref{prop:failure-of-glt(k)-in-the-finite}).  I express my
special gratitude to Benjamin Rossman for the interactions we had at
FSTTCS 2011 and ICM 2014, for sharing his unpublished result
(Theorem~\ref{theorem:rossman-los-tarski}), and for being such an
inspiration for me ever since my Ph.D. took its turn to model theory.

The company of all of the above people is what I regard as, what is
called in Sanskrit as, \emph{satsang}, meaning the company of the wise
and the elevated. I am grateful to all the people above for their
satsang, and I look forward to continue to benefit from it in the
future as well.

The international trips mentioned above, along with trips to other
international conferences, were made possible due to financial
supports from different sources. For their generous support, I
gratefully acknowledge ACM India for sponsoring my travel to the ACM
Turing Centenary Celebrations 2012 (USA), Indian Association for
Research in Computing Sciences (IARCS) and IITB for my travel to MFCS
2012 (Hungary), and the CSE department of IITB for my travel to WoLLIC
2012 (Argentina), ICM 2014 (South Korea) and HLF 2015 (Germany).

In connection with my published papers, I thank the anonymous
reviewers of~\cite{abhisekh-apal,abhisekh-mfcs,abhisekh-wollic} for
their careful reading of the earlier versions of these, and for their
many keen suggestions on improving the presentation of the results
therein.  I also thank Pritish Kamath and Vivek Madan, my
collaborators in~\cite{abhisekh-wollic}, for the many stimulating
discussions we had as a part of their Bachelor's theses (and
appreciate them for their mathematical maturity even as
undergraduates) that led to more insights about $\glt{k}$, and also
proving it in special cases (before it was shown to be true in
general, over arbitrary structures, in~\cite{abhisekh-apal}).

I would like to thank Amey Karkare for making the {\LaTeX} thesis
style file for IITB, that has been used in preparing this thesis.  I
also thank him sincerely for his technical help in preparing the two
indices that can be found towards the end of this document.  I am very
grateful to the IITB CSE department for providing me a robust laptop
to complete the preparation of the final version of my thesis, at the
critical juncture when my laptop went out of order just as I had begun
to give my thesis its final edits.

I express many sincere thanks to the following people of the IITB CSE
department for their affection, support and encouragement: Pushpak
Bhattacharyya, Rushikesh Joshi, Saraswati Krithivasan, Krithi
Ramamritham, Krishna S., Amitabha Sanyal, G. Sivakumar, S. Sudarshan
and Sundar Viswanathan. I express my special gratitude to S. Sudarshan
for his exceptional support during his tenure as the Head of the CSE
Department (and for being, in my opinion, amongst the coolest HoDs the
CSE department has ever had).  Finally, I heartily thank the
administrative staff of CFDVS (Center for Formal Design and
Verification of Software), in particular Chandrakant Talekar,
Kakasaheb Killedar and Pravin Jadhav, and the administrative staff of
the CSE department, in particular Alpana Athavankar, Vijay Ambre,
Sunanda Ghadge, Homcy, and Nitin Peje for completely taking care of
all administratitve matters in connection with my Ph.D., and further,
always with so much warmth.

On a concluding note of my academic acknowledgments, I thank
Prof. Amitabha Gupta and Mrs. Chaitali Gupta for always being so
affectionate, and appreciative of my work.  I hold them in high regard
for their untiring efforts to make possible two surveys of modern
logic (``Logic and Philosophy Today'', Vol. I and II published by
College Publications) containing expository articles written by world
experts on various aspects of logic, which I am sure will prove to be
of immense value to students in India pursuing mathematics, computer
science or philosophy.



\vspace{7pt} {\large{\tbf{B. ``Cultural'' acknowledgments}}}

\vspace{5pt} 
Music has been and will always be an inseparable part of my life.  I
express my deepest gratitude to all my teachers of Indian classical
music (specifically, North Indian classical music or
Hindust$\bar{\text{a}}$ni Sh$\bar{\text{a}}$striya Sangeet). My
sincerest pran$\bar{\text{a}}$m (a Sanskrit word that means bowing in
reverence) to Shri Ajay Thosar for teaching me vocal music from
absolute basics and with so much love always, for giving me the
freedom to explore on my own right from day one, the ocean of Indian
classical music, for being exceptionally considerate and patient with
me, and for being a constant guide in my musical journey. My sincere
pran$\bar{\text{a}}$m to Shri Shripad Zail from whom I took my initial
lessons in tabla (a North Indian classical percussion instrument), for
initiating me formally into the profoundly abstract world of Indian
rhythm, so admirably and with such warmth and fondness.  My sincere
pran$\bar{\text{a}}$m to Shri Upendra Patnaik who took forward my
tabla training, and so affectionately enriched me with several
rhythmically sophisticated compositions, as well as anecdotes from the
lives of Indian classical musicians. Finally, it has been my great
privilege to learn tabla under the tutelage of Pandit Nayan Ghosh. I
am deeply indebted to him for imparting a wide repertoire of
technically involved compositions and their corresponding playing
techniques, a balanced understanding of both the rhythmic and melodic
aspects of music, umpteen anecdotes from the lives of great masters,
and for all the affection that he has given me; I offer him my most
earnest pran$\bar{\text{a}}$m. Finally, my heartfelt
pran$\bar{\text{a}}$m to the Spirit of Indian classical music for the
enormous influence that it has had on my life and for the most
unconditional benevolence that It has showered on me on so many
occasions. To It, and to all the people above who show the way to It,
I pledge my life-long studentship and service.

My fervent thanks to my friends who made my stay at IITB so rich and
full. I can never thank enough Meenarli, Prashant, Poornima, Vishal
(Anvekar), Iman, Kasi and Russell for the committed friendship,
whole-hearted affection, and unstinted support and care they have
shown since the time I have known them. The limited space here
restricts me to write so much more about each one of you. I would just
say to you all -- life has become so beautiful just because of you
guys, and each one of you occupies a very very special and unique
place in my heart.

I am grateful to my department seniors Hrishi, Ravi-g (or simply,
Ravig), Shetal, Sobhan and Sreyash for all their affection (indeed,
pampering), and words of wisdom they have given me over the years. I
thank my department mates Anindya, Durgesh, Girish, Jagdish, Jinesh,
Preeti, Rijurekha, Soumitra and Sujesha for being such wonderful
people. My thanks to my lab mates for the awesome times we've had in
the old and new labs. My special thanks and gratitude go to Rakesh for
being ever present to help me out with problems arising out of my
pedestrian understanding of technology, such as those pertaining to my
laptop, mobile phone and useful services on the Internet (like
Dropbox), and for the memorable times we spent doing creative design
work during the organization of ALC 2015 and ICLA 2015 that took place
at IITB.

Outside of the department, I would like to earnestly thank
Dr. S. N. Jha, my football coach and former in-charge of the IITB
Gymkhana, for his heart-warming care and affection, and personal
attention that he always gave me, and for making me realize the
crucial role of physical education in academics and research.  I thank
sincerely the members of the ``Kannada gang'' for the memorable times
we have shared: Ajay, Karthik (SRK), Nanditha, Prasanna, Smita,
Vineet, and Vinita.  My thanks to my hostel mates Dev-Sir, Dushyant,
Jai (JNT), Lalit, Mithun, Parkash, Pranab, Riju, Sathish and
Vishvendra for making my stay in the hostel so much enjoyable and so
lively. I thank my dearest canine friends Damroo, Kali, Kalia and Kalu
for being so very loving and caring, and for making life in the hostel
such a celebration. I thank the people who fed me for all these years:
my friends in the hostel mess (Bhagwat, Kailas, Munnabhai, Pandeyji,
Rajesh and Shiva), in the hostel canteen (Dilip, Ganesh, Manoj,
Sachin, Satish and Soma), at Staff Canteen (Hemant, Krishna Mama,
Nilesh, Sanjay, Shankar Mama and Yogesh) and at Sunrise Dhaba (Mirza
Uncle and Mukesh). The results in this thesis have a strong connection
with the food that they so fondly prepared always. A special word of
thanks to Mirza Uncle for his painstaking efforts in constructing the
lovely ambience at Sunrise Dhaba, contemplations in which have led to
insights for some results in this thesis (such as
Theorem~\ref{theorem:subst-char-for-theories}(\ref{theorem:PSC(lambda)-char-for-theories})).

As I near the end of my acknowledgments, I return to music since a
majority of my friendships outside of the department, hostel and
special ``gangs'' are those that music made possible.  Surbahaar and
Swarsandhya, the two annual music shows of IITB, have been an integral
part of my life at IITB, and enabled me to bring about the much needed
metamorphosis for my Ph.D. life from a very different undergraduate
life. Of the people who were a part of it, there are so many to thank,
but I would especially mention Amitji, Anubhav, Arpita (Mondal), Balu,
Kolte, Sayantani, Shantanu, Srabasti, Sukhada and Suman. I look back
at the time spent with such talented musicians and awesome people as
you guys, with much nostalgia.  Ujaan and Bijoya-Sammelani, the annual
events of the Bengali Cultural Association of IITB, gave me an
opportunity to know Madhumanti and Avradeep, two highly trained and
proficient classical musicians, whom I thank for their great music and
friendship. Performances at the cultural evenings of the conferences
of various departments gave me an opportunity to know Hemali, and I
thank her so much for being the wonderful friend that she is.
Finally, I thank Milind Malshe for his Thursday \emph{baithaks}
(musical sessions), for generously sharing his deep insights into
r$\bar{\text{a}}$g$\bar{\text{a}}$s (Indian melodic scales) and for
giving me opportunites to perform at his annual \emph{gurupoornima}
(an annual Indian festival dedicated to teachers) programmes.


Om a concluding note, I thank my family members, cousins, aunts,
uncles and grandparents, for their support and encouragement over the
years. I especially mention my brother Kannan and sister-in-law
Andrea, aunts Bhuvana Ramanathan and Subhadra Sethuraman, and
grandmother Rajalakshmi Chandran for being constant sources of warmth
and affection. Last but never the least, I express my indebtedness to
my parents. I can never put into words what you mean for me. Your
immeasurable sacrifices, unfathomable love, and being permanently
there for me at all times, are things that I will never ever be able
to repay. This thesis is dedicated to you.

\vspace{-5pt}

%% file: arxiv-thesis-abhi.bbl
\begin{thebibliography}{82}
\providecommand{\natexlab}[1]{#1}
\providecommand{\url}[1]{\texttt{#1}}
\expandafter\ifx\csname urlstyle\endcsname\relax
  \providecommand{\doi}[1]{doi: #1}\else
  \providecommand{\doi}{doi: \begingroup \urlstyle{rm}\Url}\fi

\bibitem[Ajtai and Gurevich(1987)]{gurevich-ajtai}
Miklos Ajtai and Yuri Gurevich.
\newblock Monotone versus positive.
\newblock \emph{J. ACM}, 34\penalty0 (4):\penalty0 1004--1015, October 1987.

\bibitem[Ajtai and Gurevich(1994)]{ajtai-gurevich94}
Miklos Ajtai and Yuri Gurevich.
\newblock {Datalog vs first-order logic}.
\newblock \emph{J. Comput. Sys. Sci.}, 49\penalty0 (3):\penalty0 562 -- 588,
  1994.

\bibitem[Alechina and Gurevich(1997)]{gurevich-alechina}
Natasha Alechina and Yuri Gurevich.
\newblock Syntax vs. semantics on finite structures.
\newblock In \emph{Structures in Logic and Computer Science. A Selection of
  Essays in Honor of A. Ehrenfeucht}, pages 14--33. Springer-Verlag, 1997.

\bibitem[Alur and Madhusudan(2004)]{visibly-pushdown}
Rajeev Alur and Parthasarathy Madhusudan.
\newblock Visibly pushdown languages.
\newblock In \emph{Proceedings of the 36th Annual {ACM} Symposium on Theory of
  Computing, Chicago, IL, USA, June 13-16, 2004}, pages 202--211, 2004.

\bibitem[Alur and Madhusudan(2009)]{alur-madhu}
Rajeev Alur and Parthasarathy Madhusudan.
\newblock Adding nesting structure to words.
\newblock \emph{J. {ACM}}, 56\penalty0 (3), 2009.

\bibitem[Atserias et~al.(2006)Atserias, Dawar, and Kolaitis]{dawar-hom}
Albert Atserias, Anuj Dawar, and Phokion~G. Kolaitis.
\newblock On preservation under homomorphisms and unions of conjunctive
  queries.
\newblock \emph{J. ACM}, 53\penalty0 (2):\penalty0 208--237, 2006.

\bibitem[Atserias et~al.(2008)Atserias, Dawar, and Grohe]{dawar-pres-under-ext}
Albert Atserias, Anuj Dawar, and Martin Grohe.
\newblock Preservation under extensions on well-behaved finite structures.
\newblock \emph{SIAM J. Comput.}, 38\penalty0 (4):\penalty0 1364--1381, 2008.

\bibitem[Badesa(2004)]{badesa}
Calixto Badesa.
\newblock \emph{The Birth of Model Theory: L{\"o}wenheim's Theorem in the Frame
  of the Theory of Relatives}.
\newblock Princeton University Press, 2004.

\bibitem[Blumensath and Courcelle(2010)]{blumensath}
Achim Blumensath and Bruno Courcelle.
\newblock On the monadic second-order transduction hierarchy.
\newblock \emph{Log. Meth. Comp. Sci.}, 6\penalty0 (2), 2010.

\bibitem[B{\"{o}}rger et~al.(1997)B{\"{o}}rger, Gr{\"{a}}del, and
  Gurevich]{classical-decision-problem}
Egon B{\"{o}}rger, Erich Gr{\"{a}}del, and Yuri Gurevich.
\newblock \emph{The Classical Decision Problem}.
\newblock Perspectives in Mathematical Logic. Springer-Verlag, 1997.

\bibitem[Cal{\`{\i}} et~al.(2011)Cal{\`{\i}}, Gottlob, and Pieris]{datalog}
Andrea Cal{\`{\i}}, Georg Gottlob, and Andreas Pieris.
\newblock New expressive languages for ontological query answering.
\newblock In \emph{Proceedings of the 25th {AAAI} Conference on Artificial
  Intelligence, {AAAI} 2011, San Francisco, California, USA, August 7-11,
  2011}, 2011.

\bibitem[Chang and Keisler(1990)]{chang-keisler}
Chen~C. Chang and Howard~J. Keisler.
\newblock \emph{Model Theory}.
\newblock Elsevier Science Publishers, $3^{\text{rd}}$ edition, 1990.

\bibitem[Comon et~al.(2007)Comon, Dauchet, Gilleron, L\"oding, Jacquemard,
  Lugiez, Tison, and Tommasi]{tata}
Hubert Comon, Max Dauchet, Remi Gilleron, Christof L\"oding, Florent
  Jacquemard, Denis Lugiez, Sophie Tison, and Marc Tommasi.
\newblock Tree automata techniques and applications.
\newblock Available at: \url{http://www.grappa.univ-lille3.fr/tata}, 2007.
\newblock release October 12, 2007.

\bibitem[Courcelle(1990)]{courcelle}
Bruno Courcelle.
\newblock The monadic second-order logic of graphs. {I.} {Recognizable} sets of
  finite graphs.
\newblock \emph{Inf. Comput.}, 85\penalty0 (1):\penalty0 12--75, 1990.

\bibitem[Dawar(2007)]{dawar-survey}
Anuj Dawar.
\newblock On preservation theorems in finite model theory.
\newblock In \emph{Invited talk at the 6th Panhellenic Logic Symposium-Volos,
  Greece}, 2007.

\bibitem[Dawar(2010)]{dawar-quasi-wide}
Anuj Dawar.
\newblock Homomorphism preservation on quasi-wide classes.
\newblock \emph{J. Comput. Sys. Sci.}, 76\penalty0 (5):\penalty0 324--332,
  2010.

\bibitem[Dawar and Kreutzer(2009)]{dawar-domination}
Anuj Dawar and Stephan Kreutzer.
\newblock Domination problems in nowhere-dense classes.
\newblock In \emph{Proceedings of the 29th {IARCS} Annual Conference on
  Foundations of Software Technology and Theoretical Computer Science, {FSTTCS}
  2009, {IIT} Kanpur, India, December 15-17, 2009}, pages 157--168, 2009.

\bibitem[Dawar et~al.(2007)Dawar, Grohe, Kreutzer, and
  Schweikardt]{dawar-model-theory-large}
Anuj Dawar, Martin Grohe, Stephan Kreutzer, and Nicole Schweikardt.
\newblock Model theory makes formulas large.
\newblock In \emph{Proceedings of the International Colloquium on Automata,
  Languages and Programming, ICALP 2007, Wroclaw, Poland, July 9 -- 13, 2007},
  pages 913--924, 2007.

\bibitem[Diestel(2010)]{diestel}
Reinhard Diestel.
\newblock \emph{Graph Theory}.
\newblock Springer-Verlag, 2010.

\bibitem[Downey and Fellows(2013)]{downey-fellows}
Rodney~G. Downey and Michael~R. Fellows.
\newblock \emph{Fundamentals of Parameterized Complexity}.
\newblock Texts in Computer Science. Springer-Verlag, 2013.

\bibitem[Duris(2010)]{duris-ext}
David Duris.
\newblock Extension preservation theorems on classes of acyclic finite
  structures.
\newblock \emph{SIAM J. Comput.}, 39\penalty0 (8):\penalty0 3670--3681, 2010.

\bibitem[Dvorak et~al.(2010)Dvorak, Kr{\'a}l, and Thomas]{sparsegraphs}
Zdenek Dvorak, Daniel Kr{\'a}l, and Robin Thomas.
\newblock Deciding first-order properties for sparse graphs.
\newblock In \emph{Proceedings of the 51st Annual IEEE Symposium on Foundations
  of Computer Science, {FOCS} 2010, Las Vegas, USA, October 23-26, 2010}, pages
  133--142. IEEE Computer Society, 2010.

\bibitem[Ebbinghaus and Flum(2005)]{ebbinghaus-flum}
Heinz-Dieter Ebbinghaus and J{\"o}rg Flum.
\newblock \emph{Finite Model Theory}.
\newblock Springer-Verlag, 2005.

\bibitem[Elberfeld et~al.(2012)Elberfeld, Grohe, and
  Tantau]{tree-depth-FO-equals-MSO}
Michael Elberfeld, Martin Grohe, and Till Tantau.
\newblock Where first-order and monadic second-order logic coincide.
\newblock In \emph{Proceedings of the 27th Annual {IEEE} Symposium on Logic in
  Computer Science, {LICS} 2012, Dubrovnik, Croatia, June 25-28, 2012}, pages
  265--274, 2012.

\bibitem[Emmer et~al.(2010)Emmer, Khasidashvili, Korovin, and Voronkov]{emmer}
Moshe Emmer, Zurab Khasidashvili, Konstantin Korovin, and Andrei Voronkov.
\newblock Encoding industrial hardware verification problems into effectively
  propositional logic.
\newblock In \emph{Proceedings of Formal Methods in Computer Aided Design,
  FMCAD 2010, Lugano, Switzerland, October 20 - 23, 2010}, pages 137--144,
  2010.

\bibitem[Fagin et~al.(2005)Fagin, Kolaitis, Miller, and Popa]{data-exch}
Ronald Fagin, Phokion~G. Kolaitis, Ren{\'e}e~J. Miller, and Lucian Popa.
\newblock Data exchange: semantics and query answering.
\newblock \emph{Theor. Comput. Sci.}, 336\penalty0 (1):\penalty0 89--124, 2005.

\bibitem[Frick and Grohe(2004)]{frick-grohe}
Markus Frick and Martin Grohe.
\newblock The complexity of first-order and monadic second-order logic
  revisited.
\newblock \emph{Ann. Pure Appl. Logic}, 130\penalty0 (1-3):\penalty0 3--31,
  2004.

\bibitem[Gajarsk{\'{y}} and
  Hlinen{\'{y}}(2012)]{gajarsky-MSO-FPT-bdd-tree-depth}
Jakub Gajarsk{\'{y}} and Petr Hlinen{\'{y}}.
\newblock Faster deciding {MSO} properties of trees of fixed height, and some
  consequences.
\newblock In \emph{Proceedings of the 32nd {IARCS} Annual Conference on
  Foundations of Software Technology and Theoretical Computer Science, {FSTTCS}
  2012, Hyderabad, India, December 15-17, 2012}, pages 112--123, 2012.

\bibitem[Gajarsky and Hlinen{\'{y}}(2015)]{shrub-depth-FO-equals-MSO}
Jakub Gajarsky and Petr Hlinen{\'{y}}.
\newblock Kernelizing {MSO} properties of trees of fixed height, and some
  consequences.
\newblock \emph{Log. Meth. Comp. Sci.}, 11\penalty0 (19):\penalty0 1--26, 2015.

\bibitem[Gajarsk{\'y} et~al.(2013)Gajarsk{\'y}, Hlin{\v{e}}n{\'y},
  Obdr{\v{z}}{\'a}lek, Ordyniak, Reidl, Rossmanith, S.~Villaamil, and
  Sikdar]{kernelization}
Jakub Gajarsk{\'y}, Petr Hlin{\v{e}}n{\'y}, Jan Obdr{\v{z}}{\'a}lek, Sebastian
  Ordyniak, Felix Reidl, Peter Rossmanith, Fernando S.~Villaamil, and Somnath
  Sikdar.
\newblock Kernelization using structural parameters on sparse graph classes.
\newblock In Hans~L. Bodlaender and Giuseppe~F. Italiano, editors,
  \emph{Proceedings of the 21st Annual European Symposium on Algorithms, ESA
  2013, Sophia Antipolis, France, September 2-4, 2013}, pages 529--540.
  Springer-Verlag, 2013.

\bibitem[Ganian et~al.(2012)Ganian, Hlinen{\'{y}}, Ne\v{s}et\v{r}il,
  Obdrz{\'{a}}lek, de~Mendez, and Ramadurai]{shrub-depth}
Robert Ganian, Petr Hlinen{\'{y}}, Jaroslav Ne\v{s}et\v{r}il, Jan
  Obdrz{\'{a}}lek, Patrice~Ossona de~Mendez, and Reshma Ramadurai.
\newblock When trees grow low: Shrubs and fast {MSO1}.
\newblock In \emph{Proceedings of the 37th International Symposium on
  Mathematical Foundations of Computer Science, {MFCS} 2012, Bratislava,
  Slovakia, August 27-31, 2012}, pages 419--430, 2012.

\bibitem[Goubault-Larrecq and Schmitz(2016)]{goubault16}
Jean Goubault-Larrecq and Sylvain Schmitz.
\newblock Deciding piecewise testable separability for regular tree languages.
\newblock working paper or preprint, February 2016.
\newblock URL \url{https://hal.inria.fr/hal-01276119}.

\bibitem[Gr{\"{a}}del and Rosen(1999)]{gradel-rosen}
Erich Gr{\"{a}}del and Eric Rosen.
\newblock On preservation theorems for two-variable logic.
\newblock \emph{Math. Log. Quart.}, 45:\penalty0 315--325, 1999.

\bibitem[Gr\"{a}del et~al.(2005)Gr\"{a}del, Kolaitis, Libkin, Marx, Spencer,
  Vardi, Venema, and Weinstein]{gradel-vardi-et-al}
Erich Gr\"{a}del, Phokion~G. Kolaitis, Leonid Libkin, Maarten Marx, Joel
  Spencer, Moshe~Y. Vardi, Yde Venema, and Scott Weinstein.
\newblock \emph{Finite Model Theory and Its Applications}.
\newblock Texts in Theoretical Computer Science. An EATCS Series.
  Springer-Verlag, 2005.

\bibitem[Grohe and Makowsky(2011)]{model-theoretic-methods}
Martin Grohe and Johann~A. Makowsky.
\newblock \emph{Model Theoretic Methods in Finite Combinatorics}, volume 558 of
  \emph{Contemporary Mathematics}.
\newblock AMS, 2011.

\bibitem[Gulwani(2010)]{gulwani}
Sumit Gulwani.
\newblock Dimensions in program synthesis.
\newblock In \emph{Proceedings of the 12th International ACM SIGPLAN Symposium
  on Principles and Practice of Declarative Programming}, PPDP '10, pages
  13--24. ACM, 2010.

\bibitem[Gurevich(1984)]{gurevich84}
Yuri Gurevich.
\newblock {Toward logic tailored for computational complexity}.
\newblock In Michael M.~Richter et~al., editor, \emph{{Computation and Proof
  Theory: Proceedings of the Logic Colloquium held in Aachen, July 18–23,
  1983, Part II}}, page 175–216. Springer-Verlag, 1984.

\bibitem[Harwath et~al.(2015)Harwath, Heimberg, and
  Schweikardt]{nicole-lmcs-15}
Frederik Harwath, Lucas Heimberg, and Nicole Schweikardt.
\newblock Preservation and decomposition theorems for bounded degree
  structures.
\newblock \emph{Log. Meth. Comp. Sci.}, 11\penalty0 (4), 2015.

\bibitem[Higman(1952)]{higman}
Graham Higman.
\newblock {Ordering by Divisibility in Abstract Algebras}.
\newblock \emph{Proc. London Math. Soc.}, s3-2\penalty0 (1):\penalty0 326--336,
  January 1952.

\bibitem[Hodges(1997)]{hodges}
Wilfrid Hodges.
\newblock \emph{A Shorter Model Theory}.
\newblock Cambridge University Press, 1997.

\bibitem[Hodges(2000)]{hodges-history}
Wilfrid Hodges.
\newblock {~Model~Theory~(Draft~20~Jul~00)}, 2000.
\newblock URL \url{http://wilfridhodges.co.uk/history07.pdf}.

\bibitem[Imrich and Klavzar(1997)]{hamming-graphs}
Wilfried Imrich and Sandi Klavzar.
\newblock Recognizing hamming graphs in linear time and space.
\newblock \emph{Inf. Process. Lett.}, 63\penalty0 (2):\penalty0 91--95, 1997.

\bibitem[Joseph et~al.(2014)Joseph, Kuper, and Serafini]{mathew-1}
Mathew Joseph, Gabriel~M. Kuper, and Luciano Serafini.
\newblock Query answering over contextualized {RDF} knowledge with
  forall-existential bridge rules: Attaining decidability using acyclicity.
\newblock In \emph{Proceedings of the 29th Italian Conference on Computational
  Logic, CILC 2014, Torino, Italy, June 16-18, 2014.}, pages 210--224, 2014.

\bibitem[Joseph et~al.(2016)Joseph, Kuper, Mossakowski, and Serafini]{mathew-2}
Mathew Joseph, Gabriel~M. Kuper, Till Mossakowski, and Luciano Serafini.
\newblock Query answering over contextualized {RDF/OWL} knowledge with
  forall-existential bridge rules: Decidable finite extension classes.
\newblock \emph{Semantic Web}, 7\penalty0 (1):\penalty0 25--61, 2016.

\bibitem[Keisler(1960)]{keisler-sandwich}
Howard~J. Keisler.
\newblock Theory of models with generalized atomic formulas.
\newblock \emph{J. Symb. Logic}, 25:\penalty0 1--26, 1960.

\bibitem[Keisler(1965)]{keisler:finite-approx}
Howard~J. Keisler.
\newblock Finite approximations of infinitely long formulas.
\newblock In J.~Addison, editor, \emph{The Theory of Models}, pages 158--169.
  North-Holland Pub. Co., Amsterdam, 1965.

\bibitem[Keisler(1971)]{keisler:infinitary-logic-book}
Howard~J. Keisler.
\newblock \emph{Model Theory for Infinitary Logic. Logic with Countable
  Conjunctions and Finite Quantifiers}, volume~62 of \emph{{Studies in Logic
  and the Foundations of Mathematics}}.
\newblock North-Holland Pub. Co., Amsterdam-London, 1971.

\bibitem[Kelley(1975)]{kelley}
John~L. Kelley.
\newblock \emph{General Topology}.
\newblock Graduate Texts in Mathematics. Springer, 1975.

\bibitem[Kr{\"{o}}tzsch and Rudolph(2011)]{forall-exis-rules}
Markus Kr{\"{o}}tzsch and Sebastian Rudolph.
\newblock Extending decidable existential rules by joining acyclicity and
  guardedness.
\newblock In \emph{Proceedings of the 22nd International Joint Conference on
  Artificial Intelligence, {IJCAI} 2011, Barcelona, Catalonia, Spain, July
  16-22, 2011}, pages 963--968, 2011.

\bibitem[Kruskal(1972)]{wqosurveykruskal}
Joseph~B. Kruskal.
\newblock The theory of well-quasi-ordering: A frequently discovered concept.
\newblock \emph{J. Comb. Theory Ser. A}, 13\penalty0 (3):\penalty0 297 -- 305,
  1972.

\bibitem[Kunen(1980)]{kunen}
Kenneth Kunen.
\newblock \emph{Set Theory: An Introduction to Independence Proofs}.
\newblock Mathematical Programming Study. North-Holland Pub. Co., Amsterdam,
  1980.

\bibitem[Lampis(2014)]{lampis}
Michael Lampis.
\newblock Model checking lower bounds for simple graphs.
\newblock \emph{Log. Meth. Comp. Sci.}, 10\penalty0 (1), 2014.

\bibitem[Lenzerini(2002)]{data-integ}
Maurizio Lenzerini.
\newblock Data integration: A theoretical perspective.
\newblock In \emph{Proceedings of the 21st ACM SIGMOD-SIGACT-SIGART Symposium
  on Principles of Database Systems}, PODS '02, pages 233--246. ACM, 2002.

\bibitem[Libkin(2004)]{libkin}
Leonid Libkin.
\newblock \emph{Elements of Finite Model Theory}.
\newblock Springer-Verlag, 2004.

\bibitem[{Lindstr{\"o}m}(1973)]{lindstrom}
{Per} {Lindstr{\"o}m}.
\newblock A characterization of elementary logic.
\newblock In {S{\"o}ren} {Halld{\'e}n}, editor, \emph{Modality, Morality and
  Other Problems of Sense and Nonsense}, pages 189--191. Lund,Gleerup, 1973.

\bibitem[L{\"o}wenheim(1915)]{lowenheim}
Leopold L{\"o}wenheim.
\newblock {\"U}ber m{\"o}glichkeiten im relativkalk{\"u}l.
\newblock \emph{Mathematische Annalen}, 76\penalty0 (4):\penalty0 447--470,
  1915.

\bibitem[Makowsky(2004)]{makowsky}
Johann~A. Makowsky.
\newblock Algorithmic uses of the {Feferman-Vaught} theorem.
\newblock \emph{Ann. Pure Appl. Logic}, 126\penalty0 (1-3):\penalty0 159--213,
  2004.

\bibitem[Maltsev(1936)]{maltsev}
Anatoly~I. Maltsev.
\newblock {Untersuchungen aus dem Gebiete der mathematischen Logik}.
\newblock \emph{Matematicheskii Sbornik}, n.s.\penalty0 (1):\penalty0 323--336,
  1936.

\bibitem[Milner(1985)]{basicwqotheory}
Eric~C. Milner.
\newblock Basic {WQO-} and {BQO-}theory.
\newblock In Ivan Rival, editor, \emph{Graphs and Order: The Role of Graphs in
  the Theory of Ordered Sets and Its Applications}, pages 487--502.
  Springer-Verlag, 1985.

\bibitem[Ne{\v{s}}et{\v{r}}il and De~Mendez(2010)]{nowhere-dense-original}
Jaroslav Ne{\v{s}}et{\v{r}}il and Patrice~Ossona De~Mendez.
\newblock First order properties on nowhere dense structures.
\newblock \emph{J. Symb. Logic}, 75\penalty0 (03):\penalty0 868--887, 2010.

\bibitem[Ne\v{s}et\v{r}il and de~Mendez(2011)]{nowhere-dense-2}
Jaroslav Ne\v{s}et\v{r}il and Patrice~Ossona de~Mendez.
\newblock On nowhere dense graphs.
\newblock \emph{Eur. J. Comb.}, 32\penalty0 (4):\penalty0 600 -- 617, 2011.

\bibitem[Piskac et~al.(2010)Piskac, de~Moura, and Bj{\o}rner]{bjorner}
Ruzica Piskac, Leonardo~Mendon{\c{c}}a de~Moura, and Nikolaj Bj{\o}rner.
\newblock Deciding effectively propositional logic using {DPLL} and
  substitution sets.
\newblock \emph{J. Autom. Reasoning}, 44\penalty0 (4):\penalty0 401--424, 2010.

\bibitem[Ramsey(1930)]{ramsey}
Frank~P. Ramsey.
\newblock On a problem of formal logic.
\newblock \emph{Proc. London Math. Soc.}, s2-30:\penalty0 264--286, 1930.

\bibitem[Robertson and Seymour(1986)]{graphminortreewidth}
Neil Robertson and Paul~D. Seymour.
\newblock Graph minors. {V. Excluding} a planar graph.
\newblock \emph{J. Comb. Theory Ser. B}, 41\penalty0 (1):\penalty0 92--114,
  1986.

\bibitem[Robertson and Seymour(2004)]{robertson-seymour}
Neil Robertson and Paul~D. Seymour.
\newblock {Graph Minors. XX. Wagner's Conjecture}.
\newblock \emph{J. Comb. Theory Ser. B}, 92\penalty0 (2):\penalty0 325--357,
  November 2004.

\bibitem[Rosen(1995)]{rosen-thesis}
Eric Rosen.
\newblock \emph{Finite model theory and finite variable logics}.
\newblock PhD thesis, University of Pennsylvania, 1995.

\bibitem[Rosen(1997)]{rosen-van-benthem}
Eric Rosen.
\newblock {Modal logic over finite structures}.
\newblock \emph{Journal of Logic, Language and Information}, 6:\penalty0
  427--439, 1997.

\bibitem[Rosen(2002)]{rosen}
Eric Rosen.
\newblock Some aspects of model theory and finite structures.
\newblock \emph{Bull. Symbolic Logic}, 8\penalty0 (3):\penalty0 380--403, 2002.

\bibitem[Rosen and Weinstein(1994)]{rosen-weinstein}
Eric Rosen and Scott Weinstein.
\newblock Preservation theorems in finite model theory.
\newblock In \emph{Logical and Computational Complexity. Selected Papers. Logic
  and Computational Complexity, International Workshop {LCC} '94, Indianapolis,
  Indiana, USA, October 13-16, 1994}, pages 480--502, 1994.

\bibitem[Rossman(2008)]{rossman-hom}
Benjamin Rossman.
\newblock Homomorphism preservation theorems.
\newblock \emph{J. ACM}, 55\penalty0 (3):\penalty0 15:1--15:53, 2008.

\bibitem[Rossman(2012)]{rossman-los-tarski}
Benjamin Rossman.
\newblock \textit{Personal Communication}.
\newblock 2012.

\bibitem[Sankaran and Chakraborty(2010)]{arxiv-self-2010}
Abhisekh Sankaran and Supratik Chakraborty.
\newblock On semantic generalizations of the {Bernays-Sch\"onfinkel-Ramsey}
  class with finite or co-finite spectra.
\newblock \emph{CoRR}, abs/1002.4334, 2010.

\bibitem[Sankaran et~al.(2012{\natexlab{a}})Sankaran, Adsul, Madan, Kamath, and
  Chakraborty]{abhisekh-wollic}
Abhisekh Sankaran, Bharat Adsul, Vivek Madan, Pritish Kamath, and Supratik
  Chakraborty.
\newblock Preservation under substructures modulo bounded cores.
\newblock In \emph{Proceedings of the 19th International Workshop on Logic,
  Language, Information and Computation, WoLLIC 2012, Buenos Aires, Argentina,
  September 3-6, 2012}, pages 291--305, 2012{\natexlab{a}}.

\bibitem[Sankaran et~al.(2012{\natexlab{b}})Sankaran, Limaye, Sundararaman, and
  Chakraborty]{inexp-TR}
Abhisekh Sankaran, Nutan Limaye, Akshay Sundararaman, and Supratik Chakraborty.
\newblock Using preservation theorems for inexpressibility results in first
  order logic.
\newblock Technical report, IIT Bombay, 2012{\natexlab{b}}.
\newblock URL \url{http://www.cfdvs.iitb.ac.in/reports/index.php}.

\bibitem[Sankaran et~al.(2014)Sankaran, Adsul, and Chakraborty]{abhisekh-mfcs}
Abhisekh Sankaran, Bharat Adsul, and Supratik Chakraborty.
\newblock A generalization of the {{\L}o{\'{s}}-Tarski} preservation theorem
  over classes of finite structures.
\newblock In \emph{Proceedings of the 39th International Symposium on
  Mathematical Foundations of Computer Science, {MFCS} 2014, Budapest, Hungary,
  August 25-29, 2014, Part {I}}, pages 474--485, 2014.

\bibitem[Sankaran et~al.(2016)Sankaran, Adsul, and Chakraborty]{abhisekh-apal}
Abhisekh Sankaran, Bharat Adsul, and Supratik Chakraborty.
\newblock A generalization of the {{\L}o{\'{s}}-Tarski} preservation theorem.
\newblock \emph{Ann. Pure Appl. Logic}, 167\penalty0 (3):\penalty0 189--210,
  2016.

\bibitem[Seese(1996)]{seese}
Detlef Seese.
\newblock Linear time computable problems and first-order descriptions.
\newblock \emph{Math. Struct. Comp. Sci.}, 6:\penalty0 505--526, 1996.

\bibitem[Skolem(1920)]{skolem}
Thoralf Skolem.
\newblock Logisch-kombinatorische untersuchungen {\"u}ber die erf{\"l}lbarkeit
  oder beweisbarkeit mathematischer s{\"t}ze nebst einem theoreme {\"u}ber
  dichte mengen.
\newblock \emph{Videnskapsselskapet Skrifter, I. Matematisk-naturvidenskabelig
  Klasse}, 4:\penalty0 1--36, 1920.

\bibitem[Srivastava and Gulwani(2009)]{srivastava}
Saurabh Srivastava and Sumit Gulwani.
\newblock Program verification using templates over predicate abstraction.
\newblock In \emph{Proceedings of the 30th Annual ACM SIGPLAN conference on
  Programming Language Design and Implementation, {PLDI} 2009, Dublin, Ireland,
  June 15-21, 2009}, pages 223--234, 2009.

\bibitem[Stolboushkin(1995)]{stolboushkin}
Alexei~P. Stolboushkin.
\newblock Finitely monotone properties.
\newblock In \emph{Proceedings of the 10th Annual IEEE Symposium on Logic in
  Computer Science, LICS 1995, San Diego, USA, June 26-29, 1995}, pages
  324--330. IEEE Computer Society, 1995.

\bibitem[Tait(1959)]{tait}
William~W. Tait.
\newblock A counterexample to a conjecture of {S}cott and {S}uppes.
\newblock \emph{J. Symb. Logic}, 24\penalty0 (1):\penalty0 15--16, 1959.

\bibitem[Tarski(1936)]{tarski-truth}
Alfred Tarski.
\newblock The concept of truth in formalized languages.
\newblock In A.~Tarski, editor, \emph{Logic, Semantics, Metamathematics}, pages
  152--278. Oxford University Press, 1936.

\end{thebibliography}
